\documentclass{amsart}

\usepackage{caption}
    \usepackage{enumitem}
    \setlist[enumerate,1]{label=\bf  (\roman*)}
    
\usepackage{multirow}
\usepackage{hhline}

\usepackage{tabularray}
\UseTblrLibrary{booktabs}

\usepackage{adjustbox}
\usepackage{tcolorbox}
\newtcolorbox{standout}{
  colback=gray!15,
  boxrule=0pt,
  left=.3cm,
  right=.3cm,
  top=.18cm,
  bottom=.18cm,
  boxsep=0pt
}

\usepackage{mathtools}
\usepackage{amssymb}
\usepackage{amsthm}
\usepackage{tensor}
\usepackage{stmaryrd} 
\usepackage[safe]{tipa} 

\theoremstyle{plain}
\newtheorem{theorem}{Theorem}[section]
\newtheorem{lemma}[theorem]{Lemma}
\newtheorem{proposition}[theorem]{Proposition}
\newtheorem{corollary}[theorem]{Corollary}
\theoremstyle{definition}
\newtheorem{definition}[theorem]{Definition}
\newtheorem{example}[theorem]{Example}
\newtheorem{notation}[theorem]{Notation}
\newtheorem{fact}[theorem]{Fact}
\theoremstyle{remark}
\newtheorem{remark}[theorem]{Remark}
\usepackage[
  colorlinks=true,
  linkcolor=darkgreen,
  citecolor=darkgreen,
  urlcolor=darkgreen
]{hyperref}

\usepackage{cleveref}
\crefname{equation}{}{}
\crefname{section}{\S}{\S\S}
\crefname{subsection}{\S}{\S\S}
\crefname{subsubsection}{\S}{\S\S}
\crefname{definition}{Def.}{Defs.}
\crefname{theorem}{Thm.}{Thms.}
\crefname{corollary}{Cor.}{Cors.}
\crefname{lemma}{Lem.}{Lems.}
\crefname{proposition}{Prop.}{Props.}
\crefname{remark}{Rem.}{Rems.}
\crefname{notation}{Ntn.}{Ntns.}
\crefname{fact}{Fact}{Fact}
\crefname{example}{Ex.}{Exs.}
\crefname{figure}{Fig.}{Figs.}
\crefname{table}{Tab.}{Tabs.}

\usepackage[
  backend=biber,
  style=alphabetic,
  backref=true
]{biblatex}
\usepackage{tikz}
\usetikzlibrary{
  cd,
  calc,
  arrows.meta,
  backgrounds,
  decorations,
  decorations.pathmorphing, 
}

\tikzcdset{
  arrow style=tikz, 
  diagrams={
    trim left=
    {([xshift=5pt]current bounding box.west)},
    trim right=
    {([xshift=-5pt]current bounding box.east)},
    baseline=
    {([yshift=-3pt]current bounding box.center)},
    >={
      Computer Modern Rightarrow[
        length=4pt, width=4pt
      ]
    },
    hook/.style={
      {Hooks[right, length=2pt, width=8pt]}->
    },
    hook'/.style={
      {Hooks[left, length=2pt, width=8pt]}->
    },
    shorten=0pt,
  }
}

\definecolor{darkblue}{rgb}{0.05,0.25,0.65}
\definecolor{darkgreen}{RGB}{20,140,10}
\definecolor{lightgray}{rgb}{0.9,0.9,0.9}
\definecolor{darkorange}{RGB}{200,100,5}
\definecolor{darkyellow}{rgb}{.91,.91,0}

\let\originalsslash\sslash
\renewcommand{\sslash}{\mathord{\originalsslash}}

\renewcommand{\setminus}{-}

\newcommand{\backsslash}{\backslash\mspace{-6.5mu}\backslash}

\tikzset{
  snake left/.style={
    rounded corners,
    to path={
      let \p1 = (\tikztostart.east),
          \p2 = (\tikztotarget.west),
          \p3 = ($(\p1)!0.5!(\p2)$),
          \n1 = {8pt} 
      in
      (\p1)
      -- (\x1 + \n1, \y1)
      -- (\x1 + \n1, \y3)
      -- (\x2 - \n1, \y3) \tikztonodes
      -- (\x2 - \n1, \y2)
      -- (\p2)
    }
  }
}

\tikzset{
  uphordown/.style={
    rounded corners,
    to path={
      let \p1 = (\tikztostart.north),
          \p2 = (\tikztotarget.north),
          \n1 = {max(\y1,\y2) + 8pt}
      in
      (\p1)
      -- (\x1, \n1)
      -- (\x2, \n1) \tikztonodes 
      -- (\p2)
    }
  }
}

\tikzset{
  downhorup/.style={
    rounded corners,
    to path={
      let \p1 = (\tikztostart.south),
          \p2 = (\tikztotarget.south),
          \n1 = {min(\y1,\y2) - 8pt}
      in
      (\p1)
      -- (\x1, \n1)
      -- (\x2, \n1) \tikztonodes 
      -- (\p2)
    }
  }
}

\tikzset{
  rightvertleft/.style={
    rounded corners,
    to path={
      let \p1 = (\tikztostart.east),
          \p2 = (\tikztotarget.east),
          \n1 = {max(\x1,\x2) + 8pt}
      in
      (\p1)
      -- (\n1, \y1)
      -- (\n1, \y2) \tikztonodes 
      -- (\p2)
    }
  }
}

\tikzset{
  leftvertright/.style={
    rounded corners,
    to path={
      let \p1 = (\tikztostart.west),
          \p2 = (\tikztotarget.west),
          \n1 = {min(\x1,\x2) - 8pt}
      in
      (\p1)
      -- (\n1, \y1)
      -- (\n1, \y2) \tikztonodes 
      -- (\p2)
    }
  }
}

\newcommand{\DTopology}[1]{{#1}_D}

\newcommand{\shape}{%
  \hspace{.7pt}%
  \raisebox{0.8pt}{\rm\normalfont\textesh}%
  \hspace{1pt}%
}

\newcommand{\defneq}{\equiv}

\newcommand{\hotype}[1]{\mathcal{#1}}

\newcommand{\cpt}{\hspace{.8pt}{\adjustbox{scale={.5}{.77}}{$\cup$} \{\infty\}}}

\newcommand{\grayunderbrace}[2]{\mathcolor{gray}{\underbrace{\mathcolor{black}{#1}}}_{\mathcolor{gray}{#2}}}

\newcommand{\grayoverbrace}[2]{\mathcolor{gray}{\overbrace{\mathcolor{black}{#1}}}^{\mathcolor{gray}{#2}}}

\newcommand{\RealHopfFibration}{p_{{}_{\hspace{-.5pt}\mathbb{R}}}}
\newcommand{\ComplexHopfFibration}{p_{{}_{\hspace{-.5pt}\mathbb{C}}}}
\newcommand{\QuaternionicHopfFibration}{p_{{}_{\hspace{-.5pt}\mathbb{H}}}}

\newcommand{\TwistorFibration}{t_{{}_{\hspace{-.5pt}\mathbb{H}}}}

\newcommand{\FactoredCHopfFibration}{t_{{}_{\hspace{-.5pt}\mathbb{C}}}}

\newcommand{\EUnit}{1^{\!E}}

\newcommand{\acts}{%
  \hspace{1.3pt}%
  \raisebox{1.2pt}{%
    \rotatebox[origin=c]{90}{$%
      \curvearrowright%
    $}%
  }%
  \hspace{.7pt}%
}

\newcommand{\UnitaryGroup}{%
  \mathrm{U}%
}

\newcommand{\OrthogonalGroup}{%
  \mathrm{O}%
}

\newcommand{\HilbertSpace}{%
  \mathcal{H}%
}

\newcommand{\GradedHilbertSpace}{%
  \HilbertSpace_{\mathrm{gr}}%
}

\newcommand{\HilbertSpaces}{%
  \mathrm{Hilb}%
}

\newcommand{\BoundedOperators}{%
  \mathcal{B}%
}

\newcommand{\CompactOperators}{%
  \mathcal{K}%
}

\newcommand{\FredholmOperator}{%
  F%
}

\newcommand{\FredholmOperators}{%
  \mathrm{Fred}%
}

\newcommand{\GradedFredholmOperators}{%
  \FredholmOperators_{\mathrm{gr}}%
}

\newcommand{\GradedVectorSpaces}{%
  \mathbb{C}\mathrm{Vec}^{\mathrm{gr}}%
}

\newcommand{\UnitaryOperator}{%
  U%
}

\newcommand{\UH}{%
  \UnitaryGroup(\HilbertSpace)%
}

\newcommand{\antiUH}{%
  \UnitaryGroup_{\mathrm{anti}}%
  (\HilbertSpace)%
}

\newcommand{\GradedUH}{%
  \UnitaryGroup_{\mathrm{gr}}(\GradedHilbertSpace)%
}

\newcommand{\OH}{%
  \OrthogonalGroup\big(%
    \HilbertSpace
      ^{\mathbb{Z}^{\tSymmetry}_2}
  \big)%
}

\newcommand{\PUH}{%
  \mathrm{P}\UH%
}

\newcommand{\antiPUH}{%
  \mathrm{P}\antiUH%
}

\newcommand{\QuantumSymmetries}{%
  \mathrm{QS}%
}

\newcommand{\POH}{%
  \mathrm{P}\OH%
}

\newcommand{\GradedPUH}{%
  \mathrm{P}\UnitaryGroup_{\mathrm{gr}}(\GradedHilbertSpace)%
}

\newcommand{\CliffordElement}{\gamma}

\newcommand{\pSymmetry}{(p)}
\newcommand{\cSymmetry}{(c)}
\newcommand{\tSymmetry}{(t)}
\newcommand{\pOperator}{P}
\newcommand{\cOperator}{C}
\newcommand{\tOperator}{T}

\newcommand{\ComplexConjugation}[1]
  { #1^\ast }

\begin{document}

\title
[Orientations of Orbi-K-Theory Measuring Charges]
{Orientations of Orbi-K-Theory measuring 
Topological Phases and Brane Charges}

\thanks{\emph{Funding} by Tamkeen UAE under the 
NYU Abu Dhabi Research Institute grant {\tt CG008}.}

\author{Hisham Sati}
\address{Mathematics Program and Center for Quantum and Topological Systems, New York University Abu Dhabi}
\curraddr{}
\email{hsati@nyu.edu}
\thanks{}

\author{Urs Schreiber           }
\address{Mathematics Program and Center for Quantum and Topological Systems, New York University Abu Dhabi}
\curraddr{}
\email{us13@nyu.edu}
\thanks{}



\date{\today}

\dedicatory{
  \href{https://ncatlab.org/nlab/show/Center+for+Quantum+and+Topological+Systems}{\includegraphics[width=3.1cm]{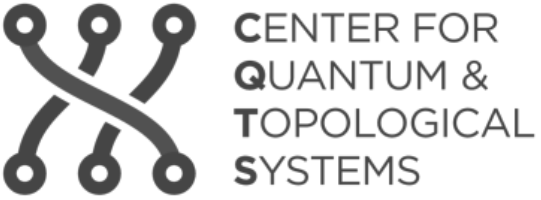}}
}

\begin{abstract}
Topological Chern phases of quantum materials, as well as brane charges on M-theory orbifolds, have famously been argued to be classified by (orbi) topological K-theory, or possibly by other stable and, notably, complex-oriented cohomology theories, such as elliptic cohomology or Morava K-theory.

However, closer inspection reveals that the most fine-grained “fragile” microscopic classification in both cases is in (orbi) Cohomotopy, which is the primordial “unstable” or nonabelian generalized cohomology. Coarsening takes the latter (fragile) to the former (stable) cohomology along nonabelian cohomology operations. But what then is the role of complex orientation on the stable side?

We observe here (i) that over gapped nodal lines in the 2D Brillouin torus and on probe M5-branes in 11D spacetime, the cohomotopical phases/charges lift through the complex/quaternionic Hopf fibration, and (ii) that measuring this fragile situation in stable cohomology means equivalently to ask for universal complex/quaternionic orientation on stable cohomology in four/ten dimensions!

Then we give an explicit realization of such unstable four/ten-dimensional complex/quaternionic orientation in $\mathrm{U}(2)/\mathrm{Sp}(2)$-equivariant K-theory, using real division-algebraic tools within a new model of twisted orbifold K-theory in cohesive homotopy theory; and we explain this as an extraordinary character map from orbi Cohomotopy-twisted Cohomotopy to relative orbi K-theory.

Finally, we discuss an application to the classification of 2-band crystalline topological insulator phases 
sensitive to the topology in the gapping process of their nodal line semimetal parent phase, and to the measurement of M-string charges inside M5-brane probes.
\end{abstract}

\maketitle

\setcounter{tocdepth}{3}
\tableofcontents

\newpage

\textbf{Outline:}
\begin{itemize}
  \item[\cref{OverviewChargesInCohomology}:]
  Bird's eye view on the role of twisted relative nonabelian cohomology, specifically of Cohomotopy, in classifying fragile crystalline Chern phases and microscopic M-brane charges, and of how their measurement in abelian cohomology theories corresponds to
  universal low-dimensional orientations. 

  \item[\cref{SomeCohesiveHomotopyTheory}:] Streamlined introduction to topological groupoids/ stacks, culminating in a general definition of (generalized nonabelian) twisted orbifold cohomology. 

  \item[\cref{OrientationsInOrbiKTheory}:] Based on this, a slick construction of twisted orbi-orientifold K-theory and of its four/ten-dimensional $\mathbb{C}$/$\mathbb{H}$-orientation in the $\mathrm{KU}$-sector.

  \item[\cref{FragilePhasesAndMicroscopicCharges}:] Novel applications, using this machinery, to the classification of gappings of nodal line semimetal phases to Chern topological insulators, and to the measurement of M-string charges on M5-brane worldvolumes.
\end{itemize}



\section{Overview: Charges in Nonabelian Cohomology}
\label{OverviewChargesInCohomology}

\subsection{Cohomology in Quantum Systems}
\label{CohomologyAndPhysics}

In algebraic topology and homotopy theory, \emph{cohomology} is, quite generally, about \emph{deformation classes} of structures fibered over spaces \parencites{nLab:AlgebraicTopology}[\S I]{FSS23-Char}{SS25-Bun}{SS26-Orb}. Our central perspective is that this is ultimately encoded via \emph{classifying maps} (cf. \cref{TableOfNotionsOfCohomology} below) in the notion of \emph{homotopy} (recalled in \cref{OnHomotopy}), namely of continuous deformations between continuous maps $f,g :  X \to \hotype{A}$, 
making them have the same \emph{homotopy class} $\pi_0(-)$:
\begin{equation}
  \label{Homotopy}
  \begin{tikzcd}
    X
    \ar[
      rr,
      bend left=30,
      "{
        f
      }",
      "{\ }"{swap, name=s}
    ]
    \ar[
      rr,
      bend right=30,
      "{
        g
      }"{swap},
      "{\ }"{name=t}
    ]
    \ar[
      from=s,
      to=t,
      dashed,
      Rightarrow
    ]
    &&
    Y
  \end{tikzcd}
  \;\;\;\;\;
  \mbox{exists iff}
  \;\;\;\;\;
  [f] = [g]
  \;\in\;
  \pi_0\, \mathrm{Map}(X,Y)
  \,.
\end{equation}

In physics, cohomology describes global dynamical invariants of quantum systems \cite{nlab:TopologicalPhasesOfMatter}: Strongly interacting quantum materials in their ground states may occupy globally non-trivial configurations classified by cohomology classes of the space occupied by the sample (or dually of its Brillouin space of crystal momenta). Moreover, in the ``geometric engineering'' of such quantum systems on higher dimensional gravitating ``branes'' \parencites{nLab:GeometricEngineering}{SS25-Srni}{GSS25-Embedding},  older arguments suggest that some form of generalized cohomology measures the charges of higher gauge field fluxes sourced by such branes (cf. \cref{ChoicesOfChargeCohomology} for more). 

\smallskip 
These suggestions follow up on the classical observation, going back to a famous insight by Dirac from almost a century ago, that the totality of ordinary magnetic flux through surfaces is classified (in a modern parlance) by ordinary integral cohomology \parencites{Alvarez1985}[\S16.4e]{Frankel2011}[\S 2.1]{SS25-Flux}. 

\smallskip 
It was a major development (first in brane physics \cite{nlab:DBraneChargeInKTheory}, then in topological quantum materials \cite{nlab:KClassificationOfTopologicalPhases}) to realize that more general topological charges may plausibly be classified more accurately in ``extraordinary'' or ``generalized'' cohomology theories (in a sense going back to Whitehead, cf. \cite{nLab:WhiteheadGeneralizedCohomology}), more fine-grained than ordinary cohomology --- such as notably in the famous example of topological ``K-cohomology theory'' (traditionally just called topological \emph{K-theory}), cf. \cite{nLab:TopologicalKTheory}. 

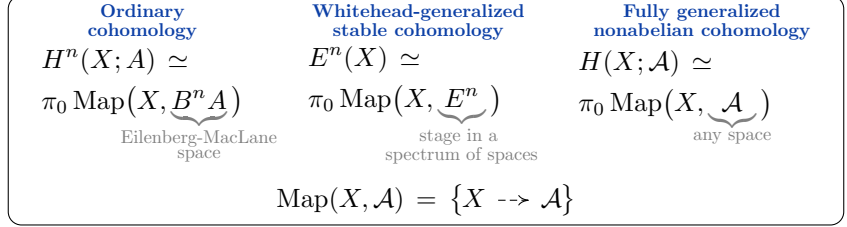
\begin{figure}[htb]
\caption{
  \label{TableOfNotionsOfCohomology}
  The unifying perspective on general \emph{cohomology} is as homotopy classes $\pi_0(-)$ of maps $\mathrm{Map}(-,-)$ into a \emph{classifying space} $\hotype{A}$.
  \protect\footnotemark
}

\centering

\adjustbox{
  rndfbox=5pt
}{
$
\begin{array}{c}
\begin{array}{ccc}
  \substack{
    \scalebox{.7}{\bf\color{darkblue}Ordinary}
    \\
    \scalebox{.7}{\bf\color{darkblue}cohomology}
  }
  &
  \qquad 
  \substack{
    \scalebox{.7}{\bf\color{darkblue}Whitehead-generalized}
    \\
    \scalebox{.7}{\bf\color{darkblue}stable cohomology}
  }
  &
  \qquad 
  \substack{
    \scalebox{.7}{\bf\color{darkblue}Fully generalized}
    \\
    \scalebox{.7}{\bf\color{darkblue}nonabelian cohomology}
  }
  \\[+3pt]
  \begin{aligned}
  &H^n(
    X;
    A
  )
  \simeq
  \\
  &
  \pi_0 \, \mathrm{Map}\big(
    X,
    \grayunderbrace{
      B^n A
    }{
      \mathclap{
        \substack{
          \scalebox{.7}{Eilenberg-MacLane}
          \\
          \scalebox{.7}{space}
        }
      }
    }
  \big)
  \end{aligned}
  &
  \begin{aligned}
  & \quad 
  E^n(
    X
  )
  \simeq
  \\
  &
  \quad \pi_0 \, \mathrm{Map}\big(
    X,
    \grayunderbrace{
      E^n
    }{
      \mathclap{
      \substack{
        \scalebox{.7}{stage in a }
        \\
        \scalebox{.7}{spectrum of spaces}
      }
      }
    }
  \big)
  \end{aligned}
  &
  \begin{aligned}
  &H(
    X
    ;
    \hotype{A}
  )
  \simeq
  \\
  &
  \pi_0 \, \mathrm{Map}\big(
    X,
    \grayunderbrace{
      \hotype{A}
    }{
      \mathclap{
      \substack{
        \scalebox{.7}{any space}
        \\
        \phantom{\scalebox{.7}{A}}
      }
      }
    }
  \big)
  \end{aligned}
\end{array}
\\
\\[-7pt]
\mathrm{Map}(
  X, 
  \hotype{A}
)
=
\big\{
\begin{tikzcd}[sep=12pt]
  X 
  \ar[
    r,
    dashed
  ]
  &
  \hotype{A}
\end{tikzcd}
\big\}
\end{array}
$
}

\end{figure}

\footnotetext{
  More generally (in \cref{TableOfNotionsOfCohomology}), we are to consider \emph{pointed} spaces --- the base point of the domain $X$ regarded as its ``point at infinity'' and the base point of the coefficient space $\hotype{A}$ regarded as zero ---  and maps $c \in \mathrm{Map}^\ast(X,\hotype{A})$ that preserve these base-points --- thus literally implementing the condition that solitonic charges \emph{vanish at infinity}, $c(\infty) = 0$ (cf. \parencites[\S 2.2]{SS25-Flux}[\S A.2]{SS24-Obs}[Ntn. 3.3]{SS23-Mf}). We disregard this extra structure here just for brevity of the exposition; the pointed generalization of all statements is straightforward. 

  Similarly, another important generalization: For \emph{geometric} cohomology (such as {\'e}tale cohomology or differential cohomology) the  classifying space $\hotype{A}$ is generalized to a \emph{moduli stack} $\mathbf{A}$, cf. \cite{SS26-Orb}. This is particularly relevant for refining the discussion here from topological sectors to actual gauge field configurations (cf. \parencites[\S 3.3]{SS25-Flux}{SS24-Phase}), but again we disregard it here for brevity of exposition.
}

But there are yet more fine-grained \emph{unstable} generalized cohomology theories called  \emph{non-abelian cohomology} (\parencites[Def. 6.0.6]{Toen2002}[Def. 2.3]{Schreiber2009OWR}[Def. 6]{Lurie2014}, cf. \parencites[\S 1]{FSS23-Char}[\S 1]{SS25-TEC}), which properly capture also non-linear Gauss laws of flux densities \cite{SS25-Flux, SS24-Obs}. In full generality, all these notions of cohomology are neatly understood as being about homotopy classes
\cref{Homotopy} of classifying maps to some \emph{classifying space}, see \cref{TableOfNotionsOfCohomology}.

\begin{figure}[htb]
\caption{
  \label{NotionsOfTwistedRelativeCohomology}
  The general notion of \emph{twisted relative non-abelian cohomology}, in evident generalization of \cref{TableOfNotionsOfCohomology}, has cocycles given by squares of maps commuting up to specified homotopy.\protect\footnotemark
  \\
  On the right, the expression ``$(-)\sqcup_{(-)} (-)$'' denotes a homotopy pushout and ``$(-)\times_{(-)} (-)$'' a homotopy pullback.
  The ordinary notion of \emph{relative cohomology} is recovered for $\phi$ a cofibration and $\hotype{B} \defneq \ast$, whence $X \sqcup_{\Sigma} \hotype{B} \simeq X/\Sigma$. But in general the homotopy pushout is richer (cf. \cref{MeasuringRelativeChargesInQCohom}).
}
\centering

\adjustbox{
  rndfbox=5pt,
  scale=.94
}{
$
\;
\begin{aligned}
  \underset{
    \mathclap{
      \adjustbox{
        scale=.7,
        raise=-5pt
      }{
        \bf\color{darkblue}%
        \renewcommand{\arraystretch}{.9}
        \begin{tabular}{c}
          Twisted relative 
          \\
          cohomology
        \end{tabular}
      }
    }
  }{
  H(
    \phi; p
  )
  }
  &
  \defneq
  \pi_0
  \left\{
  \adjustbox{raise=2.5pt}{
  \begin{tikzcd}[
    sep=20pt,
  ]
    \Sigma
    \ar[
      dd,
      "{ 
        \phi 
      }"{swap}
    ]
    \ar[
      rr,
      dashed
    ]
    &&
    \hotype{B}
    \ar[
      dd,
      "{ 
        p 
      }"
    ]
    \ar[
      ddll,
      Rightarrow,
      dashed,
      shorten=13pt,
      "{ 
        \sim 
      }"{sloped, swap},
    ]
    \\
    \\
    X
    \ar[
      rr,
      dashed
    ]
    &&
    \hotype{A}
  \end{tikzcd}
  }
  \right\}
  \;\xleftrightarrow{}\hspace{5pt}
  \begin{tikzcd}[row sep=20pt, column sep=24pt]
    \substack{
      \scalebox{.7}{\color{darkblue} \bf Probe}
      \\
      \scalebox{.7}{\color{darkblue}\bf brane}
    }
    \ar[
      dd,
      "{
        \scalebox{.7}{\color{darkgreen}``embbeding}
      }"{sloped},
      "{
        \scalebox{.7}{\color{darkgreen}field''}
      }"{sloped, swap}
    ]
    \ar[
      rr,
      "{
        \substack{
          \scalebox{.7}{\color{darkgreen}charges}
          \\
          \scalebox{.7}{\color{darkgreen}on brane}
        }          
      }"
    ]
    &&
    \substack{
      \scalebox{.7}{\color{darkblue}\bf Brane}
      \\
      \scalebox{.7}{\color{darkblue}\bf coeffs}
    }    
    \ar[
      dd,
      "{
        \scalebox{.7}{\color{darkgreen}fibered over}
      }"{sloped, yshift=2pt}
    ]
    \\
    \\
    \substack{
      \scalebox{.7}{\color{darkblue}\bf Bulk}
      \\
      \scalebox{.7}{\color{darkblue}\bf space}
    }    
    \ar[
      rr,
      "{
        \substack{
          \scalebox{.7}{\color{darkgreen}charges}
          \\
          \scalebox{.7}{\color{darkgreen}in bulk}
        }          
      }"
    ]
    &&
    \substack{
      \scalebox{.7}{\color{darkblue}\bf Bulk}
      \\
      \scalebox{.7}{\color{darkblue}\bf coeffs}
    }    
  \end{tikzcd}
  \\
  \underset{
    \mathclap{
      \adjustbox{
        scale=.7,
        raise=-3pt,
      }{
        \bf\color{darkblue}%
        Relative cohomology
      }
    }
  }{
    H_{p_\ast\mathcolor{purple}\rho}\big(
      X,\Sigma; \hotype{A}
    \big)
  }
  &
  \defneq
  \pi_0
  \left\{
  \adjustbox{raise=2.5pt}{
  \begin{tikzcd}[
    sep=20pt,
  ]
    \Sigma
    \ar[
      dd,
      "{ 
        \phi 
      }"{swap}
    ]
    \ar[
      rr,
      "{
        \mathcolor{purple}{\rho}
      }"
    ]
    &&
    \hotype{B}
    \ar[
      dd,
      "{ 
        p 
      }"
    ]
    \ar[
      ddll,
      Rightarrow,
      dashed,
      shorten=13pt,
      "{ 
        \sim 
      }"{sloped, swap},
    ]
    \\
    \\
    X
    \ar[
      rr,
      dashed
    ]
    &&
    \hotype{A}
  \end{tikzcd}
  }
  \right\}
  \simeq
  \pi_0
  \left\{
  \adjustbox{raise=2.5pt}{
  \begin{tikzcd}[
    row sep=13pt,
    column sep=8
  ]
    \Sigma
    \ar[
      dd,
      "{ 
        \phi 
      }"{swap},
      "{\ }"{name=t}
    ]
    \ar[
      rr,
      "{
        \mathcolor{purple}{\rho}
      }",
      "{\ }"{name=s, swap}
    ]
    &&
    \hotype{B}
    \ar[
      dd,
      "{ 
        p 
      }"
    ]
    \ar[
      dl
    ]
    \\
    &
    X 
      \smash{\underset{\Sigma}{\sqcup}}
    \hotype{B}
    \ar[
      dr,
      dashed
    ]
    \\
    X
    \ar[
      ur
    ]
    \ar[
      rr,
      dashed
    ]
    &&
    \hotype{A}
  \end{tikzcd}
  }
  \right\}
  \\
  \underset{
    \mathclap{
      \adjustbox{
        scale=.7,
        raise=-3pt
      }{%
        \bf\color{darkblue}%
        Twisted cohomology
      }
    }
  }{
    H^{\phi^\ast\mathcolor{purple}\tau}\big(
      \Sigma; \hotype{B}, \hotype{A}
    \big)
  }
  &
  \defneq
  \pi_0
  \left\{
  \adjustbox{raise=2.5pt}{
  \begin{tikzcd}[
    sep=20pt,
  ]
    \Sigma
    \ar[
      dd,
      "{ 
        \phi 
      }"{swap}
    ]
    \ar[
      rr,
      dashed
    ]
    &&
    \hotype{B}
    \ar[
      dd,
      "{ 
        p 
      }"
    ]
    \ar[
      ddll,
      Rightarrow,
      dashed,
      shorten=13pt,
      "{ 
        \sim 
      }"{sloped, swap},
    ]
    \\
    \\
    X
    \ar[
      rr,
      "{ \mathcolor{purple}{\tau} }"
    ]
    &&
    \hotype{A}
  \end{tikzcd}
  }
  \right\}
  \simeq
  \pi_0
  \left\{
  \adjustbox{raise=2.5pt}{
  \begin{tikzcd}[
    row sep=13pt,
    column sep=8
  ]
    \Sigma
    \ar[
      dd,
      "{ 
        \phi 
      }"{swap},
      "{\ }"{name=t}
    ]
    \ar[
      rr,
      dashed,
      "{\ }"{name=s, swap}
    ]
    \ar[
      dr,
      dashed
    ]
    &&
    \hotype{B}
    \ar[
      dd,
      "{ 
        p 
      }"
    ]
    \ar[
      dl
    ]
    \\
    &
    X 
      \smash{\underset{\hotype{A}}{\times}}
    \hotype{B}
    \\
    X
    \ar[
      ur
    ]
    \ar[
      rr,
      "{
        \mathcolor{purple}{\tau}
      }"
    ]
    &&
    \hotype{A}
  \end{tikzcd}
  }
  \right\}
\end{aligned}
$
}

\end{figure}

\footnotetext{
  All 2-dimensional diagrams we show, here and in the following,  are filled by homotopies \cref{Homotopy}, but we display only some of these homotopies explicitly, for emphasis.
}

\subsection{Choices of Charge Cohomology}
\label{ChoicesOfChargeCohomology}

The mathematical situation in \cref{CohomologyAndPhysics} highlights a general question (largely open) in physics theory building: 
\begin{standout}
\textit{Which generalized (nonabelian) cohomology theories reflect the topological phases/charges of a given quantum system, both microscopically as well as at some level of coarse graining?}
\end{standout}

This is the question for the choice of \emph{flux quantization} \cite{SS25-Flux}. Here we discuss this question in parallel for three situations, shown in \cref{TheThreeTopics},  which are quite distinct at face value but turn out to be intimately related.

\begin{figure}[htb]
\captionsetup{width=.95\linewidth}
\caption{
  \label{TheThreeTopics}
  We discuss here the refined cohomological description \cite{SS25-WilsonLoops}/\parencites{FSS20-H,FSS21-Hopf} of microscopic charges in 5D/11D Supergravity with probe L1/M5-branes, minimally flux-quantized in Cohomotopy relatively twisted by the Hopf fibration \cref{TheHopfFibration}. 
  We find (cf. \cref{MeasuringRelativeChargesInQCohom}) that measurement of these unstable charges in a stable cohomology theory $E$ is equivalent to complex/quaternionic four/ten-dimensional orientation in $E$ -- such as exists in particular on complex K-theory, $E = \mathrm{KU}$ (discussed in \S\ref{OrientationsInOrbiKTheory}).
}
\centering
\adjustbox{
  rndfbox=5pt
}{
\begin{tabular}{@{}l|l|c@{}}
  {\bf Topological phases} of 
  & 
  {\bf Brane charges} on 
  &
  {\bf Brane charges} on 
  \\
  crystalline
  2-band insulators
  &
  $D$=5 SuGra orbifolds 
  &
  $D$=11 SuGra orbifolds
  \\
  from gapping nodal lines
  &
  with probe L1-branes
  &
  with probe M5-branes
  \\
\multicolumn{2}{c|}{
$
  \begin{tikzcd}[
    ampersand replacement=\&,
    sep=24pt
  ]
    \Sigma^1
    \ar[
      dd,
      "{ 
        \phi_{\mathrlap{\mathrm{L1}}} 
      }"{description},
      "{
        \scalebox{.7}{\color{darkgreen} 
          L1 in 5D SuGra
        }  
      }"{sloped, swap, yshift=-6pt}
    ]
    \ar[
      rr,
      dashed,
      "{ 
        H_1^\pi 
      }"
    ]
    \&\&
    S^3
    \ar[
      dd,
      "{
        \ComplexHopfFibration
      }"{description},
      "{
        \scalebox{.7}{\color{darkgreen} 
          \scalebox{1.35}{$\mathbb{C}$}-Hopf fib}
      }"{sloped, yshift=3pt}
    ]
    \\
    \\
    X^4
    \ar[
      rr,
      dashed,
      "{
        (F_2^\pi, F_3^\pi)
      }"
    ]
    \&\&
    S^2
  \end{tikzcd}
$
}
&
$
  \begin{tikzcd}[
    ampersand replacement=\&,
    sep=22pt
  ]
    \Sigma^5
    \ar[
      dd,
      "{ 
        \phi_{\mathrlap{\mathrm{M5}}} 
      }"{description},
      "{
        \scalebox{.7}{\color{darkgreen} 
          M5 in 11D SuGra
        }  
      }"{sloped, swap, yshift=-7pt}
    ]    
    \ar[
      rr,
      dashed,
      "{ 
        H_3^\pi 
      }"
    ]
    \&\&
    S^7
    \ar[
      dd,
      "{
        \QuaternionicHopfFibration
      }"{description},
      "{
        \scalebox{.7}{\color{darkgreen} 
          \scalebox{1.35}{$\mathbb{H}$}-Hopf fib}
      }"{sloped, yshift=3pt}
    ]
    \\
    \\
    X^{10}
    \ar[
      rr,
      dashed,
      "{
        (F_4^\pi, F_7^\pi)
      }"
    ]
    \&\&
    S^4
  \end{tikzcd}
$
\end{tabular}
}
\end{figure}
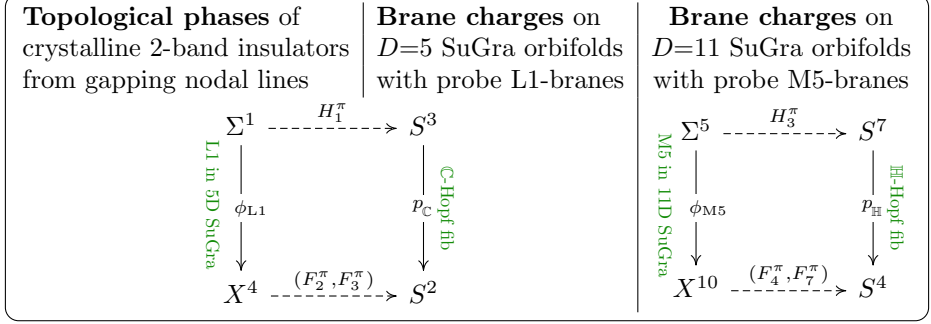

Traditionally, topological phases and brane charges have been conjectured to be classified by stable generalized cohomology theories (with classifying spaces given by stages of spectra $E$, cf. again \cref{TableOfNotionsOfCohomology}) such as: 
\begin{itemize}
\item \emph{ordinary cohomology} (cf. \cite[Ex. 3.10]{SS25-Flux}),

\item \emph{complex K-cohomology} (cf. \cite{nlab:DBraneChargeInKTheory,nLab:KTheoryClassOfTopPhases}),
\
\item \emph{elliptic cohomology} (cf. \parencites[\S6]{Segal1988}{SK04-IIA}{SK05-IIB}{Sati2006}{StolzTeichner2011}{Huan2020}{Huan2025})

\item \emph{Morava K-cohomology} (cf. \cite{SK04-IIA, SatiWesterland05}). 
\end{itemize}
Curiously, all of these proposed cohomology theories are \emph{complex oriented} (cf. \cref{NotionsOfOrientations}) --- which is remarkable in view of our first main observation in \cref{OrientationsMeasuringRelativeCharges} below.

\begin{figure}[htb]
\caption{
  \label{NotionsOfOrientations}
  Notions of universal orientations
  of fibers of vector bundles 
  in a multiplicative stable cohomology theory $E$ (cf. \cite[\S3.8]{SS23-Mf}):
  \\
  \textbf{Top left:} A \emph{complex orientation} is a choice of 2-class $c_1^{\!E}$ on $\mathbb{C}P^\infty$ which restricts to the unit on $\mathbb{C}P^1 \simeq S^2$ (cf. \parencites[\S3.2]{TamakiKono2006}[\S 4.3]{Kochman1996}).
  \\
  \textbf{Top right:}
  A \emph{quaternionic orientation} is a choice of 4-class $\tfrac{1}{2}p_1^{\!E}$ on $\mathbb{H}P^\infty$ restricting to the unit on $\mathbb{H}P^1 \simeq S^4$ (cf. \cite[\S3.9]{TamakiKono2006}).
  \\
  \textbf{Bottom row:} A complex orientation \emph{over 4-dimensions} is a choice of 2-class $\gamma_1^{\!E}$ on $\mathbb{C}P^2$ restricting to the unit on $\mathbb{C}P^1$; and a quaternionic  orientation \emph{over 10-dimensions} is a choice of 4-class $\kappa_1^{\!E}$ on $\mathbb{H}P^2$ restricting to the unit on $\mathbb{H}P^1$ \parencites[\S1.2]{Hopkins1984}[\S3.8]{SS23-Mf}.
  \\
  Note that complex orientations in $4k+2$ dimensions induce quaternionic orientations in the same dimension \cite[Thm. 3.99]{SS23-Mf}.
}
\centering
\adjustbox{
  rndfbox=5pt,
  scale=.95
}{
\setlength\tabcolsep{0pt}
\begin{tabular}{ccc}
\begin{tikzcd}[
  sep=15pt
]
  &
  \mathbb{C}P^\infty
  \ar[
    dr,
    dashed,
    shorten=-2pt,
    "{
      \,\mathclap{c_1^{\!E}}\;
    }"
  ]
  \\
  \mathbb{C}P^1
  \ar[
    ur,
    hook
  ]
  \ar[
    rr,
    "{
      \Sigma^2 \EUnit
    }"{description}
  ]
  &&
  E^2
\end{tikzcd}
&
\quad 
$\Rightarrow$
&
\quad 
\begin{tikzcd}[
  sep=15pt
]
  &
  \mathbb{H}P^\infty
  \ar[
    dr,
    dashed,
    shorten=-2pt,
    "{
      \;\mathclap{\scalebox{.7}{$\tfrac{1}{2}p_1^E$}}\;
    }"
  ]
  \\
  \mathbb{H}P^1
  \ar[
    ur,
    hook
  ]
  \ar[
    rr,
    "{
      \Sigma^4 \EUnit
    }"{description}
  ]
  &&
  E^4
\end{tikzcd}
\\[+23pt]
\rotatebox[origin=c]{-90}{$\Rightarrow$}
&&
\rotatebox[origin=c]{-90}{$\Rightarrow$}
\\[+4pt]
\begin{tikzcd}[
  sep=15pt
]
  &
  \mathbb{C}P^2
  \ar[
    dr,
    dashed,
    shorten=-2pt,
    "{
      \,\mathclap{\gamma_1^E}\;
    }"
  ]
  \\
  \mathbb{C}P^1
  \ar[
    ur,
    hook
  ]
  \ar[
    rr,
    "{
      \Sigma^2 \EUnit
    }"{description}
  ]
  &&
  E^2
\end{tikzcd}
&&
\begin{tikzcd}[
  sep=15pt
]
  &
  \mathbb{H}P^2
  \ar[
    dr,
    dashed,
    shorten=-2pt,
    "{
      \,\mathclap{\kappa_1^E}\;
    }"
  ]
  \\
  \mathbb{H}P^1
  \ar[
    ur,
    hook
  ]
  \ar[
    rr,
    "{
      \Sigma^2 \EUnit
    }"{description}
  ]
  &&
  E^2
\end{tikzcd}
\end{tabular}
}
\end{figure}

\smallskip 
On the other hand, closer analysis reveals \cite{SS25-FQAH, SS25-WilsonLoops}/\cite{FSS20-H, GSS24-SuGra, GSS25-M5} (following \parencites[\S 2.5]{Sati2018}{FSS17-Sphere}, reviewed in  \parencites{FSS19-RationalM}[\S 4.4]{SS25-Flux}) that the most fine-grained microscopic description of topological insulator phases and of brane charges in D=5/D=11 supergravity is in a non-abelian cohomology theory, namely in \emph{Cohomotopy} $\pi^n(-)$ \parencites[\S VII]{STHu59}[Ex. 2.7]{FSS23-Char}, 
\begin{equation}
  \pi^n(-)
  :=
  \pi_0
  \,
  \mathrm{Map}(
    -;
    S^n
  )
  \mathrlap{\,,}
\end{equation}
specifically in 2-Cohomotopy/4-Cohomotopy, whose classifying space is the 2-sphere/ 4-sphere $S^2/S^4$ --- and that in the presence of gapped nodal lines and of probe L1/M5-branes this becomes \emph{twisted Cohomotopy relative to the brane} (recalling \cref{NotionsOfTwistedRelativeCohomology}) 
classified by the complex/quaternionic Hopf fibration $\ComplexHopfFibration/\QuaternionicHopfFibration$ \cref{TheHopfFibration} -- this is shown in \cref{CohomotopicalFluxQuantization}.

\smallskip 
This ``proper'' flux quantization in Cohomotopy --- properly reflecting the non-linear Gauss laws of the flux densities --- has the remarkable consequence that it reflects the actual phase space structure of the (globally completed) gauge fields \cite{SS24-Phase}. Moreover, when spacetime has an ``M-fiber'' $\mathbb{R}^1_{\mathrm{M}}$ and the probe $p$-brane ``wraps'' it, in that
\begin{equation}
  \label{WrappingTheMFiber}
  \phi
  :
  \begin{tikzcd}
    \grayoverbrace{
    \mathbb{R}^1_{\mathrm{M}}
    \times
    \Sigma^{p-1}
    }{
      \scalebox{.8}{\color{black}$\Sigma^p$}
    }
    \ar[
      rr,
      "{
        \mathrm{id}
        \times
        \phi'
      }"
    ]
    &&
    \grayoverbrace{
    \mathbb{R}^1
    \times
    X^{d-1}
    }{
      \mathcolor{black}{X^d}
    }
  \end{tikzcd}
\end{equation}
then:
\begin{standout}
\emph{The topological quantum observables on solitonic field configurations are completely determined from the flux quantization law.} \cite{SS24-Obs,SS25-WilsonLoops}
\end{standout}
(Concretely, with the flux quantization given by its classifying space $\hotype{A}$, the algebra of topological quantum observables is the Pontrjagin algebra of the mapping space of a Cauchy surface of spacetime into $\hotype{A}$.)

In summary then, the microscopic flux quantization of 5D/11D supergravity with probe L1/M5 branes that we consider corresponds to the following set of charges in twisted relative Cohomotopy:
\begin{equation}
  \label{MicroscopicChargesWithMFiber}
  H\big(
    \phi'_{\mathrm{L1}}
    ;
    \ComplexHopfFibration
  \big)
  =
  \pi_0
  \left\{
  \begin{tikzcd}[
    column sep=12pt, 
    row sep=7pt
  ]
    \Sigma^0
    \ar[
      rr,
      dashed,
      "{
        H_1^\pi
      }"
    ]
    \ar[
      dd,
      "{ \phi'_{\mathrm{L1}} }"
    ]
    &&
    S^3
    \ar[
      dd,
      "{
        \ComplexHopfFibration
      }"
    ]
    \\
    \\
    X^{3}
    \ar[
      rr,
      dashed,
      "{
        F_2^\pi
      }",
      "{
        F_3^\pi
      }"{swap}
    ]
    &&
    S^2
  \end{tikzcd}
  \right\}
  ,
  \;\;
  H\big(
    \phi'_{\mathrm{M5}}
    ;
    \QuaternionicHopfFibration
  \big)
   =
  \pi_0
  \left\{
  \begin{tikzcd}[
    column sep=12pt, 
    row sep=7pt]
    \Sigma^4
    \ar[
      rr,
      dashed,
      "{
        H_3^\pi
      }"
    ]
    \ar[
      dd,
      "{ \phi'_{\mathrm{M5}} }"
    ]
    &&
    S^7
    \ar[
      dd,
      "{
        \QuaternionicHopfFibration
      }"
    ]
    \\
    \\
    X^{9}
    \ar[
      rr,
      dashed,
      "{
        F_4^\pi
      }",
      "{
        F_7^\pi
      }"{swap}
    ]
    &&
    S^4
  \end{tikzcd}
  \right\}
  \mathrlap{.}
\end{equation}

\begin{figure}[htb]
\captionsetup{width=.95\linewidth}
\caption{
  \label{CohomotopicalFluxQuantization}
  \textbf{1st row}:
  On a Cauchy surface $X^4/X^{10}$ of spacetime and compatibly of the probe brane, $\Sigma^1/\Sigma^5 \hookrightarrow X^4/X^{10}$,  the relative Gauss laws of the flux densities 
  on probe L1/M5-branes in 5D/11D supergravity are equivalent \cite{SS24-Phase} to the closure condition (cl) on differential forms $\Omega_{\mathrm{dR}}$ with coefficients in the real Whitehead $L_\infty$-algebra $\mathfrak{l}\ComplexHopfFibration/\mathfrak{l}\QuaternionicHopfFibration$ of the $\mathbb{C}/\mathbb{H}$-Hopf fibration (cf. \cite[\S5, \S12]{FSS23-Char}).
  \\
  \textbf{2nd row}: The proper quantization of these fluxes is therefore \cite{SS25-Flux} in those twisted relative nonabelian cohomology theories whose classifying fibration $p$ has the same $\mathfrak{l}p$. The minimal and hence most fine-grained choice among all these is the $\mathbb{C}/\mathbb{H}$ Hopf fibration $\ComplexHopfFibration/\QuaternionicHopfFibration$ itself. The dashed maps shown classify the charges under this twisted relative cohomotopical flux quantization \cite[\S 4.1.2]{SS25-WilsonLoops}/\parencites[\S 3.7]{FSS20-H}{FSS21-Hopf}{FSS21-StrStruc}.
  \\
  \textbf{3rd row}: But other admissible proper flux quantization laws exist. For instance, the factorization $\FactoredCHopfFibration$ \cref{FactoringTheCHopfFibration} of the $\mathbb{C}$-Hopf fibration through $\mathbb{R}P^3$ has the same relative real Whitehead $L_\infty$-algebra, $\mathfrak{l}\FactoredCHopfFibration \simeq \mathfrak{l}\ComplexHopfFibration$. Choosing this for flux quantization turns out to classify topological classes of Bloch Hamiltonians $H_{(-)}$ of 2-band topological insulators (TI) sensitive to the topology of the mass term $M_{(-)}$ which creates the TI phase from a parent nodal line topological semimetal phase (NLSM) -- we discuss this in \cref{GappedNodalLinesIn2BandChernInsulators}.
}
\centering
\adjustbox{
  rndfbox=5pt,
  scale=.9
}{
\begin{tabular}{@{}r|r@{}}
    \begin{tikzcd}[
      ampersand replacement=\&,
      column sep=2pt,
      row sep=14pt
    ]
      \smash{
      \left.
      \mathrm{d}
      H_1 
        \!=\! 
      \phi^\ast F_2\;\;\;
      \rule{0pt}{10pt}
      \right\}
      }
      \&[-5pt]
      \Sigma^1
      \ar[
        rr,
        dashed,
        shorten=-3pt,
        "{ 
          H_1 
        }"
      ]
      \ar[
        dd,
        hook,
        "{ \phi }"{description}
      ]
      \&\hspace{2pt}\&
      \Omega^1_{\mathrm{dR}}\big(
        -;
        \mathfrak{l}S^3
      \big)_{\mathrm{cl}}
      \ar[
        dd,
        ->>,
        "{
          (
            \mathfrak{l}
            \ComplexHopfFibration
          )_\ast
        }"
      ]
      \\
      \\
      \smash{
      \left.
      \begin{aligned}
        \mathrm{d} F_2 & \! =\! 0
        \\
        \mathrm{d} F_3 & 
        \!=\! 
        \tfrac{1}{2} F_2 \wedge F_2
      \end{aligned}
      \right\}
      }
      \&
      X^4
      \ar[
        rr,
        dashed,
        shorten=-3pt,
        "{
          F_2
        }"{pos=.4},
        "{
          F_3
        }"{pos=.4,swap}
      ]
      \&\&
      \Omega^1_{\mathrm{dR}}\big(
        -;
        \mathfrak{l}S^2
      \big)_{\mathrm{cl}}
    \end{tikzcd}
   &
    \begin{tikzcd}[
      ampersand replacement=\&,
      column sep=2pt,
      row sep=17pt
    ]
      \smash{
      \left.
      \mathrm{d}
      H_3 
      = \phi^\ast F_4\;\;\;
      \rule{0pt}{10pt}
      \right\}
      }
      \&[-5pt]
      \Sigma^5
      \ar[
        rr,
        dashed,
        shorten=-3pt,
        "{ 
          H_3 
        }"
      ]
      \ar[
        dd,
        hook,
        "{ 
          \phi 
        }"{description}
      ]
      \&\hspace{2pt}\&
      \Omega^1_{\mathrm{dR}}\big(
        -;
        \mathfrak{l}S^7
      \big)_{\mathrm{cl}}
      \ar[
        dd,
        ->>,
        "{
          (
            \mathfrak{l}
            \QuaternionicHopfFibration
          )_\ast
        }"
      ]
      \\
      \\
      \smash{
      \left.
      \begin{aligned}
        \mathrm{d}F_4 & \!= \! 0
        \\
        \mathrm{d} F_7 & 
        \!= \!
        \tfrac{1}{2} F_4 \wedge F_4
      \end{aligned}
      \right\}
      }
      \&
      X^{10}
      \ar[
        rr,
        dashed,
        shorten=-3pt,
        "{
          F_4
        }"{pos=.4},
        "{
          F_7
        }"{swap, pos=.4}
      ]
      \&\&
      \Omega^1_{\mathrm{dR}}\big(
        -;
        \mathfrak{l}S^4
      \big)_{\mathrm{cl}}
    \end{tikzcd}
  \\
  &
  \\
  \hline
  &
  \\[-10pt]
  $
  \underset{
   \mathclap{
    \substack{
      \scalebox{.7}{\bf\color{darkblue}Microscopic charges}
      \\
      \scalebox{.7}{\bf\color{darkblue}of L1 in 5D SuGra}
    }
    }
  }{
  H\big(
    \phi_{\mathrm{L1}}
    ;
    \ComplexHopfFibration
  \big)
  =
  }
  \pi_0
  \left\{
  \begin{tikzcd}[
    column sep=36pt,
    row sep=28pt
  ]
    \Sigma^1
    \ar[
      r,
      dashed,
      "{
        H_1^{\pi}
      }"
    ]
    \ar[
      d,
      hook,
      "{ 
        \phi_{\mathrm{L1}} 
      }"
    ]
    &
    S^3
    \ar[
      d,
      ->>,
      "{
        \ComplexHopfFibration
      }"
    ]
    \\
    X^4
    \ar[
      r,
      dashed,
      "{
        {F_2^{\pi}}
      }",
      "{
        F_3^{\pi}
      }"{swap}
    ]
    &
    S^2
  \end{tikzcd}
  \right\}
  $
  &
  $
  \underset{
    \mathclap{
      \substack{
      \scalebox{.7}{\bf\color{darkblue}Microscopic charges}
      \\
      \scalebox{.7}{\bf\color{darkblue}of M5 in 11D Sugra}
      }
    }
  }{
  H\big(
    \phi_{\mathrm{M5}}
    ;
    \QuaternionicHopfFibration
  \big)
   =
  }
  \pi_0
  \left\{
  \begin{tikzcd}[
    column sep=36pt,
    row sep=28pt
  ]
    \Sigma^5
    \ar[
      r,
      dashed,
      "{
        H_3^{\pi}
      }"
    ]
    \ar[
      d,
      hook,
      "{ 
        \phi_{\mathrm{M5}} 
      }"
    ]
    &
    S^7
    \ar[
      d,
      ->>,
      "{
        \QuaternionicHopfFibration
      }"
    ]
    \\
    X^{10}
    \ar[
      r,
      dashed,
      "{
        F_4^{\pi}
      }",
      "{
        F_7^{\pi}
      }"{swap}
    ]
    &
    S^4
  \end{tikzcd}
  \right\}
  $
  \\[-6pt]
  &
  \\
  \hline
  &
  \\[-8pt]
  $
  \underset{
   \mathclap{
    \;\;
    \substack{
      \scalebox{.7}{\bf\color{darkblue}
      TI Bloch Hamltn, rela-}
      \\
      \scalebox{.7}{\bf\color{darkblue}
      tive to NLSM parent}
    }
    }
  }{
  H\big(
    \phi_{\mathrm{NL}}
    ;
    \FactoredCHopfFibration
  \big)
  =
  }
  \pi_0
  \left\{
  \adjustbox{
    raise=3pt
  }{
  \begin{tikzcd}[
    column sep=36pt,
    row sep=28pt
  ]
    \Sigma^1
    \ar[
      r,
      dashed,
      "{
        M_{(-)}
      }"
    ]
    \ar[
      d,
      hook,
      "{ 
        \phi_{\mathrm{NL}} 
      }"
    ]
    &
    \mathbb{R}P^3
    \ar[
      d,
      ->>,
      "{
        \FactoredCHopfFibration
      }"
    ]
    \\
    \widehat{T}{}^d
    \ar[
      r,
      dashed,
      "{
        H_{(-)}
      }"
    ]
    &
    S^2
  \end{tikzcd}
  }
  \right\}
  $
  &
\end{tabular}
}
\end{figure}

\subsection{Cohomology Operations on Charges}

 Cohomotopy is rich and may contain more information than necessary in a given situation. A nonabelian \emph{cohomology operation} \cite[Def. 2.3]{FSS23-Char} from $n$-Cohomotopy to a stable cohomology theory $E$, hence a natural transform $\pi^n(-) \to E^n(-)$, may be understood as a \emph{coarse graining} or extraordinary \emph{character map}, which retains less but potentially more pertinent information. By the Yoneda lemma, such cohomology operations are given simply by postcomposition with maps 
$o : S^n \to E^n$ between the corresponding classifying spaces:
\begin{equation}
  \begin{tikzcd}[row sep=-3pt, column sep=0pt]
    \pi^n(X)
    \ar[
      rr,
      "{ o_\ast }"
    ]
    &&
    E^n(X)
    \\
  \big[
      X 
        \dashrightarrow
      S^n
    \big]
    &\longmapsto&
    \big[
      X 
        \dashrightarrow
      S^n
        \overset{o}{\longrightarrow}
      E^n
    \big]
    \mathrlap{\,.}
  \end{tikzcd}
\end{equation}

The fundamental (but most coarse) example is the (real, for our purpose) \emph{Chern-Dold character} \parencites{nLab:ChernDoldCharacter}[\S 7]{FSS23-Char} on Cohomotopy (seen through its stabilization), which extracts its degree=$n$ class in $\mathbb{R}$-rational cohomology  
\begin{equation}
  \label{ChernDoldCharacterOnCohomotopy}
  \begin{array}{c}
  \begin{tikzcd}[sep=0pt]
    \pi^n(X)
    \ar[
      rr,
      "{  }"
    ]
    &\phantom{---}&
    H^n(X;\mathbb{R})
  \end{tikzcd}
  \\
    \phantom{--}
    \big[
      \begin{tikzcd}[sep=15pt]
      X 
      \ar[r, dashed]
      &
      S^n
      \end{tikzcd}
    \big]
    \mapsto
    \big[
      \begin{tikzcd}[sep=18pt]
      X 
      \ar[r, dashed]
      &[+2pt]
      S^n
      \ar[
        r,
        "{
          \Sigma^n 1
        }"
      ]
      &
      B^n \mathbb{R}
      \end{tikzcd}
    \big]
  \end{array}
\end{equation}
by composing with the $\mathbb{R}$-rational unit class
\begin{equation}
  \Big[
  \begin{tikzcd}[sep=20pt]
     S^n
     \ar[
       r,
       "{
         \Sigma^n 1
       }"
     ]
     &
     B^n \mathbb{R}
  \end{tikzcd}
  \Big]
  =
  1 
  \in
  \mathbb{R}
  \simeq
  H^n(
    S^n
    ;
    \mathbb{R}
  )
  \mathrlap{\,.}
\end{equation}
In view of flux quantization of the $F_2/F_4$ flux in 5D/11D supergravity, this character map \cref{ChernDoldCharacterOnCohomotopy} witnesses how Cohomotopy indeed quantizes the total flux, in that it forces its de Rham class to be the rational image of a Cohomotopy class (cf. \cite[\S 3]{SS25-Flux}).

This notion of nonabelian cohomology operations has an evident generalization \cite[Def. 3.6]{FSS23-Char} to twisted relative nonabelian cohomology (\cref{NotionsOfTwistedRelativeCohomology}), where a cohomology operation is thus given by ``pasting'' \cref{PastingOfHomotopies} of homotopy-commuting squares, in our case as follows:
\begin{equation}
  \label{TwistedRelativeCharacterMap}
  \begin{array}{c}
  \begin{tikzcd}
    H(
      \phi
      ;
      p
    )
    \ar[
      rr
    ]
    &&
    E^n\big(
      X^d
      ,
      \Sigma^p
    \big)
  \end{tikzcd}
  \\
  \left[
  \begin{tikzcd}[
    column sep=12pt
  ]
    \Sigma^p
    \ar[
      d,
      hook,
      "{ \phi }"{swap}
    ]
    \ar[
      r,
      dashed
    ]
    &
    S^{2n-1}
    \ar[
      d,
      ->>,
      "{ p }"{swap,pos=.4}
    ]
    \\
    X^d
    \ar[
      r,
      dashed
    ]
    &
    S^n
  \end{tikzcd}
  \right]
  \longmapsto
  \left[
  \begin{tikzcd}[
    column sep=12pt
  ]
    \Sigma^p
    \ar[
      d,
      hook,
      "{ \phi }"{swap}
    ]
    \ar[
      r,
      dashed
    ]
    &
    S^{2n-1}
    \ar[
      d,
      ->>,
      "{ p }"{swap, pos=.4}
    ]
    \ar[
      r
    ]
    &
    \ast
    \ar[
      d,
      "{ 0 }"
    ]
    \ar[
      dl,
      shorten=8pt,
      Rightarrow,
      "{\sim}"{sloped, swap, pos=.46},
    ]
    \\
    X^d
    \ar[
      r,
      dashed
    ]
    &
    S^n
    \ar[
      r,
      "{ o }"
    ]
    &
    E^n
  \end{tikzcd}
  \right]
  .
  \end{array}
\end{equation}
With the left-hand side of \cref{TwistedRelativeCharacterMap} understood as the set of microscopic brane charges (\cref{CohomotopicalFluxQuantization}), we may think of the character map \cref{TwistedRelativeCharacterMap} as ``measuring'' or coarse-graining these to the extent reflected in the given stable $E$-cohomology.

\newpage

\subsection{Orientations measuring Brane Charges}
\label{OrientationsMeasuringRelativeCharges}

The first main observation we highlight now (following \cite[\S2.8, \S3.8]{SS23-Mf}) is the following, whose proof is shown in \cref{MeasuringRelativeChargesInQCohom}:

\begin{standout}\textit{Measuring \cref{TwistedRelativeCharacterMap}
in a stable cohomology theory $E$ the fragile charges of 2-band insulators in the presence of nodal lines, or the microscopic brane charges in 5D/11D supergravity in the presence of probe L1/M5-branes \emph{(see \cref{CohomotopicalFluxQuantization})}, is equivalent to having four/ten-dimensional complex/quaternionic \emph{$E$-orientation} \textup{(\cref{NotionsOfOrientations})}.}
\end{standout}

\begin{figure}[htb]
\captionsetup{width=.9\linewidth}
\caption{
  \label{MeasuringRelativeChargesInQCohom}
  How measuring \cref{TwistedRelativeCharacterMap} in stable cohomology $E$ the {\color{purple}chiral flux} on L1/M5-branes, microscopically in relative twisted Cohomotopy (\cref{CohomotopicalFluxQuantization}), is equivalently a four/ten-dimensional $\mathbb{C}/\mathbb{H}$-orientation in $E$-cohomology (following \cite[\S2.8, \S3.8]{SS23-Mf}):
  \\
  \textbf{Top row:} The generator $\mathcolor{purple}{h_1}/\mathcolor{purple}{h_3}$ of the Sullivan model of the $\mathbb{C}/\mathbb{H}$ Hopf fibration $\ComplexHopfFibration/\QuaternionicHopfFibration$, relative to that of the base, exhibits a null homotopy of the pullback of the generator $\mathcolor{darkblue}{f_2}/\mathcolor{darkblue}{f_4}$, the latter giving the unit map to the rational classifying space.
  \\
  \textbf{2nd row:} Lifting this situation from rational cohomology to any multiplicative stable cohomology theory $E$ means to ask for a null homotopy $\mathcolor{purple}{h_1^{\!E}}/\mathcolor{purple}{h_3^{\!E}}$ of the pullback of the $E$-unit $\mathcolor{darkblue}{\Sigma^2\EUnit}/\mathcolor{darkblue}{\Sigma^4\EUnit}$. Indicated in gray is how this defines a character cohomology operation 
  \cref{TwistedRelativeCharacterMap}
  from $\ComplexHopfFibration$-twisted Cohomotopy, to $E$-cohomology relative to the probe L1/M5-brane, by forming pasting composites of homotopy squares \cref{PastingOfHomotopies}.
  \\
  \textbf{3rd row:} Factoring through the homotopy pushout (po) 
  exhibits (\cref{CofibersOfHopfFibration})
  these null homotopies as equivalent to maps $\mathcolor{purple}{\gamma_1^E}/\mathcolor{purple}{\kappa_1^E}$ from $\mathbb{C}P^2/\mathbb{H}P^2$ to $E^2/E^4$, whose restriction to $\mathbb{C}P^1/\mathbb{H}P^1$ is (homotopic to) the $E$-unit: These are four/ten-dimensional 
  $\mathbb{C}/\mathbb{H}$-orientations in $E$-cohomology (cf. \cref{NotionsOfOrientations}). We spell this out for $E = \mathrm{KU}$ below in \cref{TheEquivariantOrientation}.
  \\
  \textbf{4th row:} By the pasting law (\cref{PastingLaw}) and by \cref{CofibersOfHopfFibration}, this statement remains true when the microscopic brane flux is quantized instead in $\mathbb{R}P^3/\mathbb{R}P^7$ (cf. 3rd row of \cref{CohomotopicalFluxQuantization}), if now the stable coefficients are taken to be the pullback of the $E$-orientation to $\mathbb{R}^4/\mathbb{R}^8$:
}
\centering
\adjustbox{
  rndfbox=5pt,
  scale=.92
}{
\begin{tabular}{@{}r|r@{}}
$
  \begin{tikzcd}[
    ampersand replacement=\&,
    column sep=0pt
  ]
    \smash{
    \left.
    \begin{aligned}
      \mathrm{d}\, 
      \mathcolor{purple}{h_1}
      & 
      = 
      \ComplexHopfFibration^\ast f_2
    \end{aligned}%
    \rule{0pt}{12pt}%
    \right\}
    }
    \&
    S^3
    \ar[
      dd,
      ->>,
      "{
        \ComplexHopfFibration
      }"{description}
    ]
    \ar[
      rr
    ]
    \&\phantom{---}\&
    \ast
    \ar[
      dd,
      "{ 0 }"{description}
    ]
    \ar[
      ddll,
      Rightarrow,
      shorten=10pt,
      "{
        \mathcolor{purple}{h_1} 
      }"{sloped, description}
    ]
    \\
    \\
    \smash{
    \left.
    \begin{aligned}
      \mathrm{d}
      \,
      \mathcolor{darkblue}{f_2} 
      &
      \!=\! 0 
      \\
      \mathrm{d}
      \,
      f_3
      &
      \!=\! \tfrac{1}{2} f_2 \wedge f_2
    \end{aligned}
    \right\}
    }
    \&
    S^2
    \ar[
      rr,
      "{
        \mathcolor{darkblue}{f_2}
      }"{description}
    ]
    \&\&
    B^2 \mathbb{Q}
  \end{tikzcd}
$
&
$
  \begin{tikzcd}[
    ampersand replacement=\&,
    column sep=0pt
  ]
    \smash{
    \left.
    \begin{aligned}
      \mathrm{d}\, 
      \mathcolor{purple}{h_3}
      & 
      = 
      \QuaternionicHopfFibration^\ast
      f_4
    \end{aligned}%
    \rule{0pt}{12pt}%
    \right\}
    }
    \&
    S^7
    \ar[
      dd,
      ->>,
      "{
        \QuaternionicHopfFibration
      }"{description}
    ]
    \ar[
      rr
    ]
    \&\phantom{---}\&
    \ast
    \ar[
      dd,
      "{ 0 }"{description}
    ]
    \ar[
      ddll,
      Rightarrow,
      shorten=10pt,
      "{
        \mathcolor{purple}{h_3} 
      }"{sloped, description}
    ]
    \\
    \\
    \smash{
    \left.
    \begin{aligned}
      \mathrm{d}
      \,
      \mathcolor{darkblue}{f_4} 
      &
      \!=\! 0 
      \\
      \mathrm{d}
      \,
      f_7
      &
      \!=\! \tfrac{1}{2} f_4 \wedge f_4
    \end{aligned}
    \right\}
    }
    \&
    S^4
    \ar[
      rr,
      "{
        \mathcolor{darkblue}{f_4}
      }"{description}
    ]
    \&\&
    B^4 \mathbb{Q}
  \end{tikzcd}
$
\\
&
\\
\hline
  \begin{tikzcd}[
    ampersand replacement=\&,
    column sep=0pt
  ]
  \mathcolor{gray}{%
    \Sigma^1%
  }
  \ar[
    dd, 
    gray,
    "{
      \phi_{\mathrlap{\mathrm{L1}}}
    }"{description}
  ]
  \ar[
    rr,
    dashed,
    gray,
    "{
      H_3^\pi
    }"{description}
  ]
  \&\phantom{---}\&
  S^3
  \ar[
    rr
  ]
  \ar[
    dd,
    ->>,
    "{
      \ComplexHopfFibration
    }"{description}
  ]
  \ar[
    ddll,
    Rightarrow,
    gray,
    shorten=15pt
  ]
  \&\phantom{---}\&
  \ast
  \ar[
    dd,
    "{ 0 }"{description}
  ]
  \ar[
    ddll,
    Rightarrow,
    dashed,
    shorten=10pt,
    "{
      \mathcolor{purple}{h_1^{\!E}}
    }"{sloped, description, pos=.46}
  ]
  \\
  \\
  \mathcolor{gray}{X^4}
  \ar[
    rr,
    gray,
    dashed,
    "{
      F_2^\pi
    }", 
    "{
      F_3^\pi
    }"{swap}
  ]
  \&\&
  S^2
  \ar[
    rr,
    "{
      \mathcolor{darkblue}{\Sigma^2 \EUnit}
    }"{description}
  ]
  \&\&
  E^2
  \&{}
\end{tikzcd}
&
  \begin{tikzcd}[
    ampersand replacement=\&,
    column sep=0pt
  ]
  \mathcolor{gray}{\Sigma^5}
  \ar[
    rr,
    gray,
    dashed,
    "{
      H_7^{\pi}
    }"{description}
  ]
  \ar[
    dd,
    gray,
    "{ 
      \phi_{\mathrlap{\mathrm{M5}}} 
    }"{description}
  ]
  \&\phantom{---}\&
  S^7
  \ar[
    ddll,
    Rightarrow,
    dashed,
    gray,
    shorten=15pt
  ]
  \ar[
    rr
  ]
  \ar[
    dd,
    ->>,
    "{
      \QuaternionicHopfFibration
    }"{description}
  ]
  \&\phantom{---}\&
  \ast
  \ar[
    dd,
    "{ 0 }"{description}
  ]
  \ar[
    ddll,
    Rightarrow,
    dashed,
    shorten=10pt,
    "{
      \mathcolor{purple}{h_3^{\!E}}
    }"{sloped, description, pos=.46}
  ]
  \\
  \\
  \mathcolor{gray}{X^{10}}
  \ar[
    rr,
    dashed,
    gray,
    "{
      F_4^\pi
    }",
    "{
      F_7^\pi
    }"{swap}
  ]
  \&\&
  S^4
  \ar[
    rr,
    "{
      \mathcolor{darkblue}{\Sigma^4 \EUnit}
    }"{description}
  ]
  \&\&
  E^4
  \&{}
\end{tikzcd}
\\
\hline
$
\simeq
\;
\begin{tikzcd}[
  row sep=17,
  column sep=14
]
  S^3 
  \ar[
    rr
  ]
  \ar[
    dd,
    "{
      \ComplexHopfFibration
    }"{description}
  ]
    \ar[
      dr,
      phantom,
      "{
        \mathrm{(po)}
      }"{scale=.75, pos=.55}
    ]
    && 
  \ast
  \ar[
    dd,
    "{
      0
    }"{description}
  ]
  \ar[
    dl
  ]
  \\
  &
  \mathbb{C}P^{\mathrlap{2}}
  \ar[
    dr,
    dashed,
    shorten=-1.5pt,
    "{
      \,\mathclap{
      \mathcolor{purple}{
        \gamma_1^E
      }
      }
      \;
    }"{description, pos=.43}
  ]
  \\
  S^2
  \ar[
    rr,
    "{
      \Sigma^2 \EUnit
    }"{description}
  ]
  \ar[
    ur
  ]
  &&
  E^2
\end{tikzcd}
$
&
$
\simeq
\;
\begin{tikzcd}[
  row sep=17,
  column sep=14
]
  S^7
  \ar[
    rr
  ]
  \ar[
    dd,
    "{
      \QuaternionicHopfFibration
    }"{description}
  ]
    \ar[
      dr,
      phantom,
      "{
        \mathrm{(po)}
      }"{scale=.75, pos=.55}
    ]
    && 
  \ast
  \ar[
    dd,
    "{
      0
    }"{description}
  ]
  \ar[
    dl
  ]
  \\
  &
  \mathbb{H}P^{\mathrlap{2}}
  \ar[
    dr,
    dashed,
    shorten=-1.5pt,
    "{
      \,\mathclap{
      \mathcolor{purple}{
        \kappa_1^E
      }
      }
      \;
    }"{description, pos=.43}
  ]
  \\
  S^4
  \ar[
    rr,
    "{
      \Sigma^4 \EUnit
    }"{description}
  ]
  \ar[
    ur
  ]
  &&
  E^4
\end{tikzcd}
$
\\
\hline
& 
\\[-8pt]
$\simeq$
\begin{tikzcd}[
  column sep=10pt,
  row sep=15pt
]
  S^3 
  \ar[
    d,
    ->>
  ]
  \ar[
    rr
  ]
  \ar[
    drr,
    phantom,
    "{
      \mathrm{(po)}
    }"{pos=.7, scale=.65}
  ]
  &&
  \ast
  \ar[
    d
  ]
  \\[-5pt]
  \mathbb{R}P^3
  \ar[
    dr,
    phantom,
    "{ 
      \mathrm{(po)} 
    }"{pos=.7, scale=.65}
  ]
  \ar[
    dd,
    "{ 
      \FactoredCHopfFibration 
    }"{description}
  ]
  \ar[
    rr,
    hook
  ]
  &&
  \mathbb{R}P^4
  \ar[
    dd
  ]
  \ar[
    dl,
    ->>
  ]
  \\[-2pt]
  &
  \mathbb{C}P^2
  \ar[
    dr,
    shorten <=-2pt,
    dashed,
    "{ 
      \,
      \mathclap{%
        \mathcolor{purple}{\gamma_1^E}%
      }
      \,
    }"{description, pos=.4}
  ]
  \\[-2pt]
  S^2
  \ar[
    ur
  ]
  \ar[
    rr,
    "{
      \Sigma^2 \EUnit
    }"{description}
  ]
  &&
  E^2
\end{tikzcd}
\hspace{-8pt}
&
$\simeq$
\begin{tikzcd}[
  column sep=10pt,
  row sep=15pt
]
  S^7 
  \ar[
    d,
    ->>
  ]
  \ar[
    rr
  ]
  \ar[
    drr,
    phantom,
    "{
      \mathrm{(po)}
    }"{pos=.7, scale=.65}
  ]
  &&
  \ast
  \ar[
    d
  ]
  \\[-5pt]
  \mathbb{R}P^7
  \ar[
    dr,
    phantom,
    "{ 
      \mathrm{(po)} 
    }"{pos=.7, scale=.65}
  ]
  \ar[
    dd,
    "{ 
    }"{description}
  ]
  \ar[
    rr,
    hook
  ]
  &&
  \mathbb{R}P^8
  \ar[
    dd
  ]
  \ar[
    dl,
    ->>
  ]
  \\[-2pt]
  &
  \mathbb{H}P^2
  \ar[
    dr,
    shorten <=-2pt,
    dashed,
    "{ 
      \,
      \mathclap{%
        \mathcolor{purple}{\kappa_1^E}%
      }
      \,
    }"{description, pos=.4}
  ]
  \\[-2pt]
  S^4
  \ar[
    ur
  ]
  \ar[
    rr,
    "{
      \Sigma^4 \EUnit
    }"{description}
  ]
  &&
  E^4
\end{tikzcd}
\hspace{-8pt}
\end{tabular}
}

\end{figure}

\begin{example}
Every elliptic curve $C$ (over any base ring) entails an elliptic cohomology theory $E_C$ (cf. \cite[(5.2)]{Segal1988}), which is complex oriented (cf. \cite[Ex. p. 197]{Segal1988}) and hence also quaternionic-oriented (cf. \cite[Prop. 3.98]{SS23-Mf}). Therefore the result of \cref{MeasuringRelativeChargesInQCohom},
in the situation \cref{WrappingTheMFiber} of an M5-brane wrapped on the M-fiber, 
says that a choice of an elliptic curve and of a sub-4-manifold   
determines an extraordinary character map from microscopic M-brane charges to the elliptic cohomology of the bulk spacetime relative to the brane locus:
\begin{equation}
  \left.
    \begin{array}{rl}
      \mbox{elliptic curve} 
      &
      C
      \\
      \mbox{4-manifold}
      &
      \Sigma^4 \xrightarrow{\phi} X^9
    \end{array}
  \!\right\}
  \;\;
  \Rightarrow
  \;\;
  \begin{tikzcd}[
  ]
    H\big(
      \phi'_{\mathrm{M5}}
      ;
      \QuaternionicHopfFibration
    \big)
    \ar[r]
    &
    E_C^0\big(
      X^{9},
      \Sigma^4
    \big)    
    \mathrlap{\,.}
  \end{tikzcd}
\end{equation}
This result is similar to the situation expected by informal arguments in \cite{Gukov2021}.
\end{example}
\begin{example}
Consider specifically the situation of an M5-brane worldvolume of the form $\Sigma^{1,5} =  \mathbb{R}^{1,2} \times S^3$ inside 11D Minkowski spacetime $\mathbb{R}^{1,10}$. Since the latter is equivalent to the point, in this situation the microscopically quantized charges of the 3-form field $H_3$ on $\Sigma^{1,5}$ are in plain 3-Cohomotopy, classified by the fiber of the $\mathbb{H}$-Hopf fibration, hence are given by some integer $c \in \mathbb{Z} \simeq \pi^3(S^3)$. Measuring this situation in a stable cohomology theory $E$ which is complex-oriented in 10d by some $h_3^E$ (according to \cref{MeasuringRelativeChargesInQCohom}) sees a charge $H_3^E$ of the 3-form field $H_3$ classified by $\Omega E^4 \simeq E^3$ and being the $c$ fold multiple of the generator there:
\begin{equation}
  \begin{tikzcd}[
   row sep=40pt, 
   column sep=65pt
  ]
    & 
    S^3
    \ar[
      r,
      "{
        \Sigma^3 \EUnit
      }"{description, pos=.46}
    ]
    \ar[
     d,
     shorten >=-3pt
    ]
    &
    \Omega E^4
    \ar[d]
    \\[-10pt]
    \mathbb{R}^2 \times S^3
    \ar[
      ur,
      "{ c \cdot p_{S^3} }"{description}
    ]
    \ar[
      urr,
      crossing over,
      "{
        H_3^E
      }"{description, pos=.67}
    ]
    \ar[d]
    \ar[
      r,
      "{
        H_3^\pi
      }"{description}
    ]
    &
    S^7
    \ar[r]
    &[+5pt]
    \ast
    \ar[
      dl,
      Rightarrow,
      "{
        h_3^E
      }"{description},
      shorten=5pt
    ]
    \ar[
      dll,
      shorten <=12pt,
      Rightarrow,
      "{ 
        H_3^E 
      }"{description, pos=.8}
    ]
    \ar[
      d, 
      "{
        0
      }"{description}
    ]
    \\
    \ast 
      \simeq
    \mathbb{R}^{10}
    \ar[
      r, 
      "{ 
        F^\pi_4 
      }"{description}
    ]
    &
    S^4
    \ar[
      from=u,
      crossing over,
      "{
        \QuaternionicHopfFibration
      }"{description, pos=.2}
    ]
    \ar[
      r,
      "{
        \Sigma^4 \EUnit
      }"{description}
    ]
    &
    E^4
    \mathrlap{\,.}
  \end{tikzcd}
\end{equation}
This means that, in this simple situation, the coarsened charges seen in $E$ cohomology still reflect the full microscopic charges on the M5 in Cohomotopy iff the unit class $[\EUnit]$ is non-torsion.
\end{example}

\clearpage

\subsection{Tangential Twisting and Orbifolding}
\label{TangentialTwistingAndOrbifolding}

Finally, all these considerations are to be generalized to include \emph{tangential twisting} of the charge cohomology theory by tangential $G$-structure of the spacetime domain.

Concretely, the $\mathbb{C}/\mathbb{H}$-Hopf fibration $\ComplexHopfFibration/\QuaternionicHopfFibration$ is equivariant with respect to a canonical $\mathrm{Spin}^3(3)/\mathrm{Spin}(5) \simeq \mathrm{U}(2)/\mathrm{Sp}(2)$ action (Def. \cref{EquivarianceOfHopfFibration} below), so that the microscopic brane charges may and should be \cite{FSS20-H, FSS21-Hopf, SS21-M5Anomaly}
twisted by tangential $\mathrm{U}(2)/\mathrm{Sp}(2)$-structure $\tau$ (\cref{TwistedEquivariantOrientation}). By \cite[Thm. 6.2.6]{SS26-Orb} this implies that the charges are in the correspondingly $\mathrm{RO}$-graded equivariant cohomology in the vicinity of orbifold singularities.

It is this tangentially twisted/equivariantized version of the construction of orientations as extraordinary characters that we establish below in \cref{OrientationsInOrbiKTheory}.

\begin{figure}[htb]
\captionsetup{width=.95\linewidth}
\caption{
  \label{TwistedEquivariantOrientation}
  {\bf Top row:}
  With coupling to background gravity taken into account --- whose topological charges are encoded in the class of the tangent bundle $T X$ --- the relative brane charges are to be further twisted by tangential $\mathrm{SU}(2)/\mathrm{Sp}(2)$-structure $\tau$, hence equivalently by $\mathrm{Spin}(3)/\mathrm{Spin}(5)$-structure on spacetime. \protect\footnotemark
  \\
  If we understand the homotopy quotients $(-)\sslash(-)$ and deloopings $\mathbf{B}(-) \simeq\!\ast\!\sslash (-)$ in topological groupoids (stacks), then this tangentially twisted cohomology automatically reduces to RO-graded equivariant cohomology in the vicinity of orbi-singularities \cite[Thm. 6.2.6]{SS26-Orb}.
  \\
  {\bf Bottom row:} In this situation, the $E$-valued orientation character maps (\cref{MeasuringRelativeChargesInQCohom}) are to be equivariantized accordingly. This is what we construct, for $E = \mathrm{KU}$, in \cref{OrientationsInOrbiKTheory}.
}
\centering

\adjustbox{
  rndfbox=5pt
}{
\begin{tabular}{cc}
\begin{tikzcd}[
  column sep=20pt,
  row sep=18pt
]
  \Sigma^1
  \ar[
    dd,
    "{ 
      \phi
    }"{description}
  ]
  \ar[
    rr,
    dashed,
    "{
      H_3^\pi
    }"{description}
  ]
  &&[-20pt]
  S^3 \sslash \mathrm{U}(2)
  \ar[
    dd,
    "{
      \ComplexHopfFibration
      \sslash 
      \mathrm{U}(2)
    }"{description}
  ]
  \\
  \\
  X^{4}
  \ar[
    d,
    "{
      \vdash T X^4
    }"{sloped, swap, yshift=-2pt}
  ]
  \ar[
    dr,
    "{ \tau }"{description}
  ]
  \ar[
    rr,
    dashed,
    "{
      F_2^\pi
    }",
    "{
      F_3^\pi
    }"{swap}
  ]
  &&
  S^2 \sslash \mathrm{U}(2)
  \ar[dl]
  \\[+4pt]
  \mathbf{B}\mathrm{GL}(4)
  \ar[
    r,
    <-,
    shorten=-1pt
  ]
  & 
  \mathbf{B}\mathrm{U}(2)
\end{tikzcd}
&
\begin{tikzcd}[
  column sep=20pt,
  row sep=18pt
]
  \Sigma^5
  \ar[
    dd,
    "{ 
      \phi
    }"{description}
  ]
  \ar[
    rr,
    dashed,
    "{
      H_3^\pi
    }"{description}
  ]
  &&[-20pt]
  S^7 \sslash \mathrm{Sp}(2)
  \ar[
    dd,
    "{
      \QuaternionicHopfFibration
      \sslash 
      \mathrm{Sp}(2)
    }"{description}
  ]
  \\
  \\
  X^{10}
  \ar[
    d,
    "{
      \vdash T X^{10}
    }"{sloped, swap, yshift=-2pt}
  ]
  \ar[
    dr,
    "{ \tau }"{description}
  ]
  \ar[
    rr,
    dashed,
    "{
      F_4^\pi
    }",
    "{
      F_7^\pi
    }"{swap}
  ]
  &&
  S^4 \sslash \mathrm{Sp}(2)
  \ar[dl]
  \\[+4pt]
  \mathbf{B}\mathrm{GL}(10)
  \ar[
    r,
    <-,
    shorten=-1pt,
    "{
      \scalebox{.6}{
        \cref{Triality}
      }
    }"
  ]
  & 
  \mathbf{B}\mathrm{Sp}(2)
\end{tikzcd}
\\
\hline
\begin{tikzcd}[
  column sep=16pt
]
  S^3 \sslash \mathrm{U}(2)
  \ar[
    rr
  ]
  \ar[
    dd,
    "{
      \ComplexHopfFibration
      \sslash
      \mathrm{U}(2)
    }"{description}
  ]
  &&
  \ast \sslash \mathrm{U}(2)
  \ar[
    dd,
    "{
      0 \sslash \mathrm{U}(2)
    }"{description}
  ]
  \ar[
    ddll,
    Rightarrow,
    dashed, 
    shorten=12pt,
    "{
      \mathcolor{purple}{h_1^E}
      \sslash 
      \mathrm{U}(2)
    }"{description}
  ]
  \\
  \\
  S^2 \sslash \mathrm{U}(2)
  \ar[
    rr,
    "{
      \Sigma^2 \EUnit
      \sslash 
      \mathrm{U}(2)
    }"{description}
  ]
  \ar[
    dr,
  ]
  &&
  E^2 \!\sslash \mathrm{U}(2)
  \ar[
    dl
  ]
  \\
  &
  \mathbf{B}\mathrm{U}(2)
\end{tikzcd}
&
\;\;\;
\begin{tikzcd}[
  column sep=16pt
]
  S^7 \sslash \mathrm{Sp}(2)
  \ar[
    rr
  ]
  \ar[
    dd,
    "{
      \QuaternionicHopfFibration
      \sslash
      \mathrm{Sp}(2)
    }"{description}
  ]
  &&
  \ast \sslash \mathrm{Sp}(2)
  \ar[
    dd,
    "{
      0 \sslash \mathrm{Sp}(2)
    }"{description}
  ]
  \ar[
    ddll,
    Rightarrow,
    dashed, 
    shorten=12pt,
    "{
      \mathcolor{purple}{h_3^E}
      \sslash 
      \mathrm{Sp}(2)
    }"{description}
  ]
  \\
  \\
  S^4 \sslash \mathrm{Sp}(2)
  \ar[
    rr,
    "{
      \Sigma^4 \EUnit
      \sslash 
      \mathrm{Sp}(2)
    }"{description}
  ]
  \ar[
    dr,
  ]
  &&
  E^4 \!\sslash \mathrm{Sp}(2)
  \ar[
    dl
  ]
  \\
  &
  \mathbf{B}\mathrm{Sp}(2)
\end{tikzcd}\end{tabular}
}

\end{figure}

\footnotetext{
  Beware the crucial subtlety (cf. \cref{TwistedEquivariantOrientation}) that $\mathrm{Spin}(5) \simeq \mathrm{Sp}(2)$ as abstract Lie groups, but that as subgroups of $\mathrm{Spin}(8)$ (and hence of the full $\mathrm{Spin}(1,10)$) they are \emph{not} equivalent --- but related by the less widely appreciated form of \emph{triality} \cite[\S 2.3]{FSS20-H}:  
  \begin{equation}
   \label{Triality}
    \begin{tikzcd}[
      ampersand replacement=\&,
      column sep=25pt, 
      row sep=9pt
    ]
      \mathrm{Spin}(5)
      \ar[
        r,
        "{
          \sim
        }"
      ]
      \ar[
        d,
        hook
      ]
      \&
      \mathrm{Sp}(2)
      \ar[
        d, 
        hook,
        "{ \iota \, }"{swap}
      ]
      \\
      \mathrm{Spin}(8)
      \ar[
        r,
        "{ \sim }",
        "{
          \mathrm{triality}
        }"{swap, yshift=-1pt}
      ]
      \&
      \mathrm{Spin}(8)
      \mathrlap{\,.}
    \end{tikzcd}%
  \end{equation}%
  \vspace{-.3cm}%
}



\section{Topological Stacks and Orbifold Cohomology}
\label{SomeCohesiveHomotopyTheory}

We give a pedagogical and practical new account of a streamlined theory of \emph{topological stacks} (in \cref{TopologicalGroupoidsAndStacks}, as a faithful fragment of the smooth $\infty$-groupoids briefly recalled in \cref{SmoothInfinityGroupoids}) neatly supporting a general notion of \emph{twisted orbifold cohomology} (in \cref{TwistedOrbifoldCohomology} below).

Here ``topological stacks'' (cf. \parencites{Carchedi2012}[\S 4.2]{SS25-Bun}) refers to the \emph{geometric homotopy theory} (cf. \cref{SmoothInfinityGroupoids}) of \emph{groupoids} (cf. \cite{IbortRodriguez2021} and \cref{TopologicalGroupoids}) with topological structure (\emph{topological groupoids}, cf. \parencites[\S II]{Mackenzie1987}[\S 2.2.1]{SS25-Bun}), subject to \emph{Morita equivalences} (discussed in \cref{TopologicalStacks}). This may be viewed as the first-stage enhancement of classical general topology to include \emph{gauge transformations} between points of topological spaces: such as the isotropy group actions in orbifolds and the quantum symmetry actions in spaces of Fredholm operators --- which in fibered combination makes for the \emph{twisted orbifold K-theory} discussed in \cref{OrbifoldKTheory}.

For perspective, afterwards we briefly indicate (in \cref{SmoothInfinityGroupoids}) how this theory of topological stacks is a full fragment of the more encompassing \emph{cohesive homotopy theory} of \emph{smooth $\infty$-groupoids} (\cite[\S 4.3]{SS25-Bun}\cite[\S 4.1]{SS26-Orb}, going back to \cite[\S 3.1]{SSS12}\cite{Sc13-dcct}).
This is in the general context of ``geometric homotopy theory'' (\emph{$\infty$-topos theory} \cite{ToenVezzosi2005, Lurie2009}\cite[\S 1]{FSS23-Char}) --- see exposition for mathematical physicists in \cite{Schreiber2025}.

\subsection{Topological Groupoids and Orbifolds}
\label{TopologicalGroupoidsAndStacks}

Discussion of topological stacks in traditional literature may tend to look mysterious to the newcomer and cumbersome to the expert. We spell out an approach (in specialization of \cref{SmoothInfinityGroupoids}) which is transparent and practically useful.

\subsubsection{Topological Spaces}

To set up notation, first some quick paragraphs on topological spaces in general, cf. \parencites{Sc2017-Topology}. (Beware that from \cref{DTopologicalSpaces} on we will be entirely concerned only with the special case of \emph{D-topological spaces}.)

For $X$ a topological space, we write 
\begin{itemize}
\item $\pi_0 X$ for its set of path-connected components,
\item $\flat X$ for its underlying set of points, 
\end{itemize}
both regarded as  discrete topological spaces, if necessary.
We write: 
\begin{itemize}
\item 
``$\ast$'' for the singleton space (the \emph{point}),
\item ``$\varnothing$'' for the empty set regarded as a topological space.
\end{itemize}

The archetypical topological spaces for our purposes (cf. the next \cref{DTopologicalSpaces}) are the \emph{Cartesian spaces} $\mathbb{R}^n$ with their usual Euclidean topology.
For  a pair of topological spaces $(X, Y)$, an arrow $X \to Y$ denotes a \emph{continuous function} between them, called a \emph{map}, for short. On every space $X$ there is the identity map 
$X \xrightarrow{\mathrm{id}} X$. A \emph{homeomorphism} is a map $f$ which has an inverse map $f^{-1}$, denoted $f : X \xrightarrow{ \sim } Y$, hence such that $f^{-1} \circ f = \mathrm{id}$ and $f \circ f^{-1} = \mathrm{id}$. A most basic but important example for our purpose is the homeomorphy of Cartesian spaces with their own open balls of any radius $\epsilon \in \mathbb{R}_{> 0}$:
\begin{equation}
  \label{CartesianSpaceHomeomorphicToOpenBall}
  \begin{tikzcd}[
    row sep=-3pt,
    column sep=0pt
  ]
    \mathbb{R}^n
    \ar[rr, "{ \sim }"]
    &&
    \mathbb{D}^n_\epsilon
    :=
    \big\{
      x \in \mathbb{R}^n
      \,\big\vert\;
      \vert x \vert < \epsilon
    \big\}
    \\
    x 
      &\longmapsto&
    \epsilon
    \frac{x}{ 1 + \vert x \vert }
    \mathrlap{\,.}
  \end{tikzcd}
\end{equation}

A \emph{diagram} of maps is always meant/understood to \emph{commute} in that all composite maps between any pair of spaces in the diagram are equal. Notably a \emph{commuting square} is:
\begin{equation}
  \label{CommutingSquare}
  \begin{tikzcd}[
   row sep=15pt,
   column sep=30pt
  ]
    Y
    \ar[
      d,
      "{
        \phi
      }"{swap}
    ]
    \ar[
      r,
      "{ c }"
    ]
    & 
    B
    \ar[
      d,
      "{ p }"
    ]
    \\
    X 
    \ar[
      r,
      "{ b }"
    ]
      & 
    A
  \end{tikzcd}
  \;\;\;\;\;\;\;
  \Leftrightarrow
  \;\;\;\;\;\;\;
  p \circ c
  =
  b \circ \phi
  \mathrlap{\,.}
\end{equation}

The following are some \emph{universal constructions} on topological spaces (cf. \cite[\S I.6]{Sc2017-Topology}) that we need:

\begin{definition}
  \label[definition]{FiberProduct}
  For $\begin{tikzcd}[sep=small] X \ar[r, "{f_1}"] & B \ar[r, <-, "{f_2}"] & Y 
  \end{tikzcd}$ a pair of coincident maps, their \emph{fiber product}, or \emph{pullback} (pb) of one along the other, is $\begin{tikzcd} X \underset{B}{\times} Y \ar[r, "{(p_1, p_2)}"] & X \times Y \end{tikzcd}$, unique up to compatible homeomorphism, which makes the following bottom right square of maps commute \cref{CommutingSquare}, and universally so in that it uniquely factors (shown by the dashed map) every other completion to a commuting square:
  \begin{equation}
  \label{PullbackSquare}
   \begin{tikzcd}[row sep=small, column sep=large]
     Q
     \ar[
       drr,
       bend left=30,
       "{ \forall }"{description}
     ]
     \ar[
       ddr,
       bend right=30,
       "{ \forall }"{description}
     ]
     \ar[
       dr,
       dashed,
       "{
         \exists !
       }"{description}
     ]
     \\
     &
     X
     \underset{B}{\times}
     Y
     \ar[
       dr,
       phantom,
       "{ 
         \lrcorner 
        }"{pos=.1}
     ]
     \ar[
       d,
       "{ 
         p_1 
       }"{swap, pos=.4}
     ]
     \ar[
       r,
       "{
         p_2
       }"
     ]
     &
     Y
     \ar[
       d,
       "{
         f_2
       }"
     ]
     \\
     &
     X
     \ar[
       r,
       "{ f_1 }"{swap}
     ]
     &
     B
     \mathrlap{\,.}
   \end{tikzcd}
  \end{equation}

  Dually, for $\begin{tikzcd}[sep=small] X \ar[r, <-, "{f_1}"] & T \ar[r,  "{f_2}"] & Y 
  \end{tikzcd}$ a pair of co-emanent maps, their \emph{cofiber product}, or \emph{pushout} (po) of one along the other, is $\begin{tikzcd} X \overset{B}{\sqcup} Y \ar[r, <-, "{(q_1, q_2)}"] & X \times Y \end{tikzcd}$, unique up to compatible homeomorphism, which makes the following top left square of maps commute \cref{CommutingSquare}, and universally so in that it uniquely factors (shown by the dashed map) every other completion to a commuting square:
  \begin{equation}
    \label{PushoutSquare}
   \begin{tikzcd}[row sep=small, column sep=large]
     T
     \ar[
       dr,
       phantom,
       "{ 
         \ulcorner
       }"{pos=.9}
     ]
     \ar[
       r, "{ f_2 }"
     ]
     \ar[
       d,
       "{ 
         f_1 
       }"{swap}
     ]
     & 
     Y
     \ar[
       d,
       "{
         q_2
       }"
     ]
     \ar[
       ddr,
       bend left=30pt,
       "{ \forall }"{description}
     ]
     \\
     X
     \ar[
       r,
       "{ 
         q_1 
       }"{swap}
     ]
     \ar[
       drr,
       bend right=30pt,
       "{ \forall }"{description}
     ]
     &
     X 
       \overset{Y}{\sqcup}
     Y
     \ar[
       dr,
       dashed,
       "{ 
         \exists! 
        }"{description}
     ]
     \\
     & & 
     Q
     \mathrlap{\,.}
   \end{tikzcd}
  \end{equation}

\end{definition}

\begin{definition}
For a pair of parallel maps $\begin{tikzcd} T \ar[r, shift left=4pt, "{f_1}"{description, pos=.4}] \ar[r, shift right=4pt, "{f_2}"{description, pos=.4}] & X \end{tikzcd}$, their \emph{coequalizer}
is $\begin{tikzcd}[sep=small] X \!\ar[r, "{q}"] & \! X/T \end{tikzcd}$, unique up to compatible homeomorphism, which makes the following horizontal composites agree, and universally so in that it uniquely factors (shown by the dashed map) every other such coequalizing completion: 
\begin{equation}
  \label{CoequalizerDiagram}
  \begin{tikzcd}[row sep=10pt, column sep=huge]
    T 
    \ar[
      r, 
      shift left=4pt, 
      "{f}"{description}
    ] 
    \ar[
      r, 
      shift right=4pt,
      "{g}"{description}
    ] 
    \ar[
      dr,
      shift right=2pt,
      shorten <=2pt,
      "{ \forall }"
      {swap}
    ]
    & 
    X
    \ar[
      r, "{q}"
    ]
    \ar[
      d,
      "{ 
        \forall 
      }"{pos=.4, swap}
    ]
    &
    X/T\mathrlap{\,.}
    \ar[
      dl,
      dashed,
      "{ \exists! }"{description}
    ]
    \\
    & 
    Q
  \end{tikzcd}
\end{equation}
Concretely, $X/T$ is the quotient space by the smallest equivalence relation $\!\!\tensor[_{f_1}]{\sim}{_{f_2}}$ on $X$ for which $f_1(x) \sim f_2(x)$ for all $x \in X$.
\end{definition}

\begin{example}[{cf. \parencites[p. 217]{Bott1982}[\S3.1]{AguilarGitlerPrieto2002}}]
  \label[example]
    {CellAttachment}
  For $n \in \mathbb{N}$ a pushout \cref{PushoutSquare} along the boundary inclusion of the closed $n$-dimensional unit ball, 
    $S^{n-1} =  S(\mathbb{R}^n) \hookrightarrow  D(\mathbb{R}^n) = D^n$, 
  is called an \emph{$n$-cell attachment}
  \begin{equation}
    \label{PushoutForCellAttachment}
    \begin{tikzcd}[row sep=small, column sep=large]
      S(\mathbb{R}^n)
      \ar[r, hook]
      \ar[
        d,
        "{ f }"{swap}
      ]
      \ar[
        dr,
        phantom,
        "{ 
          \ulcorner 
        }"{pos=.7}
      ]
      &
      D(\mathbb{R}^n)
      \ar[d]
      \\
      X
      \ar[
        r,
        hook
      ]
      &
      X \cup_f D^n
    \end{tikzcd}
  \end{equation}
  with \emph{attaching map} $f$.
  A topological space that arises from $\varnothing$ via (possibly transfinite) sequences of such $n$-cell attachments is called a \emph{cell complex} and a \emph{CW-complex} if the cell dimension $n$ is increasing monotonically in the process.
\end{example}

\begin{proposition}[Pasting law, {cf. \cite{nLab:PastingLawForPullbacks}}]
  \label[proposition]{PastingLaw}
  Given a ``pasting'' diagram 
  of commuting squares, \cref{CommutingSquare}
  $
  \adjustbox{scale=.65}{
    \begin{tikzcd}
      {}
      \ar[r]
      \ar[d]
      &
      {}
      \ar[r]
      \ar[d]
      &
      {}
      \ar[d]
      \\
      {}
      \ar[r]
      &
      {}
      \ar[r]
      &
      {}
    \end{tikzcd}
    }
  $
  then:
  \begin{enumerate}
    \item
    If the right square is a pullback \cref{PullbackSquare} then the left square is so iff the total rectangle is.
    \item
    If the left square is a pushout \cref{PushoutSquare} then the right square is so iff the total rectangle is.
  \end{enumerate}
\end{proposition}

\subsubsection{D-Topological Spaces}
\label{DTopologicalSpaces}

The topological spaces that one actually cares about in geometry are \emph{probeable} (cf. \cite{Schreiber2025}) by \emph{Cartesian spaces} $\mathbb{R}^n$ ($n \in \mathbb{N}$) with their standard Euclidean topology, in that they are \emph{D-topological} (cf. \parencites{nLab:DeltaGeneratedTopologicalSpace}[Ntn. 4.3.19]{SS25-Bun}):

\begin{definition}
\label[definition]{DTopologicalSpace}
A topological space $X$ is called \emph{D-topological} (traditionally: \emph{Delta-generated} or \emph{numerically generated}) if their subsets $S \subset X$ are open iff their preimages under all maps of the form $\mathbb{R}^n \to X$ are open in Cartesian space.
\end{definition}

Basic classes of examples D-topological spaces include, in increasing generality:
\begin{itemize}
  \item topological manifolds,
  \item cell complexes (\cref{CellAttachment}),
  \item retracts of cell complexes,
\end{itemize}
and thereby all the cofibrant spaces of algebraic topology (cf. \cite{nLab:ClassicalModelStructureOnTop}).

More generally, all pushouts \cref{PushoutSquare} (and generally: colimits) of D-topological spaces are themselves D-topological. 

Moreover, the topology of topological spaces $A$ may be refined to their induced D-topology 
\begin{equation}
  \label{DTopologization}
  \begin{tikzcd}[sep=15pt]
    \DTopology{X} 
    \ar[
      rr, 
      "{ x \,\mapsto\, x}"
     ] 
     &&
     X
\end{tikzcd}
\end{equation}
without changing the maps into it out of D-topological spaces:
\begin{equation}
  \mbox{$X$ is D-topological}
  \;\;\;
  \Rightarrow
  \;\;\;
  \big\{
    \begin{tikzcd}[sep=small]
      X
      \ar[
        r, 
        dashed
      ]
      &
      A
    \end{tikzcd}
  \big\}
  \simeq
  \big\{
    \begin{tikzcd}[sep=small]
      X
      \ar[
        r, 
        dashed
      ]
      &
      \DTopology{A}
    \end{tikzcd}
  \big\}
  \mathrlap{\,.}
\end{equation}
In particular this means that for the discussion of cohomology in terms of maps from manifolds/CW-complexes into classifying spaces (as surveyed in \cref{CohomologyAndPhysics}) the topology on the classifying spaces may without restriction be taken to be refined to their D-topology.

For some pairs $(X,Y)$ of exotic D-topological spaces, their fiber products $X \times_B Y$ \cref{PullbackSquare}  may fail to be D-topological --- but the D-topologization $\DTopology{(X \times_B Y)}$ \cref{DTopologization} still satisfies the universal property \cref{PullbackSquare} among D-topological spaces, and hence is the correct fiber product (generally: limit) in the category of D-topological spaces.

Finally, the central property of the category of D-topological spaces is its \emph{cartesian closure}, meaning that for $(X,Y)$ a pair of D-topological spaces, the set 
\begin{equation}
  \label{MappingSpace}
  \mathrm{Map}(X,Y)
  \defneq
  \big\{
    \begin{tikzcd}[column sep=small]
      X
      \ar[r, dashed]
      &
      Y
    \end{tikzcd}
  \big\}
\end{equation}
of all maps $ X \to Y$ becomes a D-topological space (with the D-topologization of the \emph{compact-open topology}) such that there are natural homeomorphisms (cf. \cite[Rem. 2.0.17]{SS25-Bun})
\begin{equation}
  \label{InternalHomAdjunction}
  \begin{tikzcd}[row sep=-3pt, column sep=0pt]
  \mathrm{Map}\big(
    Z \times X
    ,
    Y
  \big)
  \ar[
    rr,
    "{
      \widetilde{(-)}
    }",
    "{ \sim }"{swap}
  ]
  &&
  \mathrm{Map}\big(
    Z 
    ,
    \mathrm{Map}(X,Y)
  \big)
  \\
  f(-,-) &\longmapsto& f(-)(-)
  \mathrlap{\,.}
  \end{tikzcd}
\end{equation}

In summary, this says that, in this technical sense, D-topological spaces form a \emph{convenient category of topological spaces} (cf. \parencites{nLab:ConvenientCategoryOfSpaces}[\S 1.1.1]{SS25-Bun}). Therefore we declare that:
\begin{notation}
\label[notation]{SpacesAreDTopological}
From now on,
\emph{by ``\emph{topological spaces}'' we mean ``\emph{D-topological spaces'' (\cref{DTopologicalSpace})}}. In consequence, we say ``\emph{topological group}'' \cref{TopologicalGroup} for ``\emph{D-topological group}'' and ``\emph{topological groupoid}'' (in \cref{TopologicalGroupoids}) for ``\emph{D-topological groupoid}'', etc.
\end{notation}

\subsubsection{Spaces with group action}
\label{SpacesWithGroupAction}

\begin{definition}[{cf. \cite[\S 2.1]{SS25-Bun}}]
\label[definition]
  {TopologicalGroupAction}
For $G$ a \emph{topological group}, hence a topological space equipped with maps of the form
\begin{equation}
  \label{TopologicalGroup}
  \begin{tikzcd}
    G \times G
    \ar[
      rr,
      "{ (-)\cdot(-) }"
    ]
    &&
    G
    \mathrlap{\,,}
  \end{tikzcd}
  \;\;\;\;\;
  \begin{tikzcd} 
     \ast 
       \ar[
         r,
         "{ \mathrm{e} }"
       ] 
     & 
    G
    \mathrlap{\,,}
  \end{tikzcd}
\end{equation}
and for $X$ a topological space, a (left) \emph{topological $G$-action on $X$} is a map 
\begin{equation}
  \label{TopologicalGroupActionMap}
  \begin{tikzcd}[sep=0pt]
  G \times X
  \ar[
    rr,
    "{
      G \acts X
    }"
  ]
  &&
  X
  \\[-2pt]
  (g, x) &\mapsto& g \cdot x
  \mathrlap{\,,}
  \end{tikzcd}
  \;\;\;\;
  \mbox{s.t.}
  \;\;
  \forall_{x \in X}
  \left\{\,
  \begin{aligned}
    & 
    \mathrm{e}\cdot x = x
    \\
    &
    \forall_{g_1, g_2 \in G}
    :
    (g_2 \cdot g_1) \cdot x
    =
    g_2 \cdot (g_1 \cdot x)
    \mathrlap{\,.}
  \end{aligned}
  \right.
\end{equation}
One also says that $G \acts \, X$ is a \emph{$G$-space}, for short. For a pair of these, $G \acts \,X$ and
$G \acts \,Y$, an \emph{equivariant map} $f$ between them is
\begin{equation}
  \label{Equivariance}
  f \in \mathrm{Map}(X,Y)
  \;\;
  \mbox{s.t.}
  \;\;
  \forall_{
    \substack{
      x \in X
      \\
      g \in G
    }
  }
  \;
  f(g \cdot x)
  =
  g \cdot f(x)
  \mathrlap{\,.}
\end{equation}
We denote the subspace of the mapping space \cref{MappingSpace} on the $G$-equivariant maps as follows:
\begin{equation}
  \label{EquivariantMappingSpace}
  \Bigg\{\!
  \adjustbox{raise=-8pt}{
  \begin{tikzcd}
    X
    \ar[
      in=60,
      out=180-60,
      looseness=4,
      "{
        \,\mathclap{G}\,
      }"{description}
    ]
    \ar[
      r, 
      dashed,
      "{ f }"
    ]
    &
    Y
    \ar[
      in=60,
      out=180-60,
      looseness=4,
      "{
        \,\mathclap{G}\,
      }"{description}
    ]
  \end{tikzcd}
  }
  \Bigg\}
  \defneq
  \mathrm{Map}(X,Y)^G
  \subset
  \mathrm{Map}(X,Y)
  \mathrlap{\,.}
\end{equation}

\end{definition}

Here for $G \acts X$ a $G$-space, the notation
\begin{equation}
  \label{FixedSubspace}
  X^G
  :=
  \big\{
    x \in X
    \,\big\vert\,
    \forall_{g \in G}
    \, 
    g \cdot x = x
  \big\}
  \subset
  X
\end{equation}
indicates the \emph{$G$-fixed subspace}.

For example, the plain mapping space between $G$-spaces becomes itself a $G$-space by the \emph{conjugation action}
\begin{equation}
  \begin{tikzcd}[row sep=-3pt, 
    column sep=0pt
  ]
    G \times
    \mathrm{Map}(X,Y)
    \ar[
      rr
    ]
    &&
    \mathrm{Map}(X,Y)
    \\
    \big(
      g, f(-)
    \big)
    &\longmapsto&
    g^{-1}\cdot f(g\cdot -)
  \end{tikzcd}
\end{equation}
and its $G$-fixed points \cref{FixedSubspace} are precisely the $G$-equivariant maps \cref{EquivariantMappingSpace}.

The most basic examples of $G$-spaces are the \emph{coset spaces} for subgroups $H \subset G$,
\begin{equation}
  \label{CosetSpace}
  G/H
  :=
  \big\{
    g \cdot H
    \,\big\vert\,
    g \in G
  \big\}
\end{equation}
equipped with their inherited $G$-action:
\begin{equation}
  \begin{tikzcd}[row sep=-3pt, column sep=0pt]
    \label{ActionOnCosetSpace}
    G \times G/H
    \ar[
      rr
    ]
    &&
    G/H
    \\
    (g',\, g \cdot H)
    &\mapsto&
    g' \cdot g \cdot H
    \mathrlap{\,.}
  \end{tikzcd}
\end{equation}
We come back to this in \cref{HoQuotientOfGModHByG} below.

\subsubsection{Homotopy}
\label{OnHomotopy}

A \emph{homotopy} (cf. \cite[\S 3]{Fomenko2016}) between a pair of parallel maps $f,g : X \to Y$ is a continuous deformation between them,
\begin{equation}
  \label{AHomotopy}
  \begin{tikzcd}
  X
  \ar[
    rr,
    bend left=30,
    "{ f }"{description, name=s}
  ]
  \ar[
    rr,
    bend right=30,
    "{ g }"{description, name=t}
  ]
  \ar[
    from=s,
    to=t,
    Rightarrow,
    dashed,
    "{ \eta }"
  ]
  &&
  Y
  \end{tikzcd}
 \;\; : \;\; 
  \begin{tikzcd}
    {[0,1]}
    \ar[r]
    &
    \mathrm{Map}(X,Y)
    \mathrlap{\,,}
  \end{tikzcd}
\end{equation}
namely a continuous path between the corresponding points $(\widetilde{f}, \widetilde{g})$ in the mapping space \cref{MappingSpace},  hence a map $\eta$ fitting into this commuting diagram of maps:
\begin{equation}
  \label{HomotopyBetweenTopologicalSpaces}
  \begin{tikzcd}[
    row sep=18pt
  ]
    X
    \ar[
      d,
      hook,
      "{
        (\mathrm{id},0)
      }"{swap}
    ]
    \ar[
      drr,
      "{ f }"
    ]
    \\
    X \times [0,1]
    \ar[
      rr,
      dashed,
      "{ \eta }"{description, pos=.4}
    ]
    &&
    Y
    \\
    X
    \ar[
      u,
      hook',
      "{
        (\mathrm{id},1)
      }"
    ]
    \ar[
      urr,
      "{ g }"{swap}
    ]
  \end{tikzcd}
  \hspace{.4cm}
  \underset{
    \scalebox{.7}{
      \eqref{InternalHomAdjunction}
    }
  }{
    \Leftrightarrow
  }
  \hspace{.6cm}
  \begin{tikzcd}[row sep=15pt]
    \{0\}
    \ar[d, hook]
    \ar[
      drr,
      "{ 
        \widetilde{f} 
      }"
    ]
    \\
    {[0,1]}
    \ar[
      rr,
      dashed,
      "{
        \widetilde{\eta}
      }"{description}
    ]
    &&
    \mathrm{Map}(X,Y)
    \mathrlap{\,.}
    \\
    \{1\}
    \ar[u, hook']
    \ar[
      urr,
      "{ \widetilde{g} }"{swap}
    ]
  \end{tikzcd}
\end{equation}

When the spaces are equipped with $G$-action \cref{TopologicalGroupActionMap}, then an \emph{equivariant homotopy} between equivariant maps \cref{Equivariance} is a homotopy \cref{HomotopyBetweenTopologicalSpaces} running inside the equivariant mapping space \cref{EquivariantMappingSpace}:
\begin{equation}
  \label{EquivariantHomotopy}
  \begin{tikzcd}[row sep=small,
    column sep=20pt
  ]
    \{0\}
    \ar[d, hook]
    \ar[
      drr,
      "{ 
        \widetilde{f} 
      }"
    ]
    \\
    {[0,1]}
    \ar[
      rr,
      dashed,
      "{
        \widetilde{\eta}
      }"{description, pos=.4}
    ]
    &&
    \mathrm{Map}(X,Y)^G
    \ar[
      r, 
      hook
    ]
    &
    \mathrm{Map}(X,Y)
    \mathrlap{\,.}
    \\
    \{1\}
    \ar[
      u,
      hook'
    ]
    \ar[
      urr,
      "{ \widetilde{g} }"
      {swap}
    ]
  \end{tikzcd}
\end{equation}

  A map is a \emph{homotopy equivalence}, to be denoted 
  \begin{equation}
    \label{AHomotopyEquivalence}
    \begin{tikzcd}
      f 
        : 
      X 
      \ar[
        r,
        "{ \sim }",
        "{ \mathrm{hmtpy} }"{swap}
      ] 
      &
      Y
    \end{tikzcd}
  \end{equation}  
   if there exists a reverse map $\begin{tikzcd}[sep=small] \overline{f} : Y \ar[r] & X \end{tikzcd}$ and homotopies \cref{AHomotopy} of this form:
  \begin{equation}
    \label{HomotopyEquivalenceOfSpaces}
    \begin{tikzcd}[
      row sep=15pt, 
      column sep=25pt
    ]
      & 
      Y
      \ar[
        dr,
        "{ \overline{f} }"{description}
      ]
      \ar[
        d,
        shorten=5pt,
        Rightarrow,
      ]
      \ar[
        rr,
        "{
          \mathrm{id}
        }"{description}
      ]
      &
      {}
      \ar[
        d,
        Rightarrow,
        shorten=2pt
      ]
      &
      Y
      \\
      X
      \ar[
        ur,
        "{ f }"{description}
      ]
      \ar[
        rr,
        "{ \mathrm{id} }"{description}
      ]
      &
      {}
      &
      X
      \mathrlap{\,.}
      \ar[
        ur,
        "{
          f
        }"{description}
      ]
    \end{tikzcd}
  \end{equation}

The \emph{vertical composition} of a composable pair of homotopies \cref{AHomotopy} is the evident concatenation of these paths of maps
\begin{equation}
  \begin{tikzcd}
  X
  \ar[
    rr,
    bend left=60,
    "{ f }"{description, name=s}
  ]
  \ar[
    rr,
    "{ g }"{description, name=m}
  ]
  \ar[
    rr,
    bend right=60,
    "{ h }"{description, name=t}
  ]
  \ar[
    from=s,
    to=m,
    Rightarrow,
    dashed,
    "{\, \eta_1 }"
  ]
  \ar[
    from=m,
    to=t,
    Rightarrow,
    dashed,
    "{\, \eta_2 }"
  ]
  &&
  Y
  \end{tikzcd}
 \;\; : \;\;
  \begin{tikzcd}
    (x,s)
    \mapsto
    \left\{
    \begin{aligned}
      \eta_1(x,s) & \;
      \mbox{if $s\in [0,\tfrac{1}{2}]$}
      \\
      \eta_2(x,s) & 
      \mbox{if $s\in [\tfrac{1}{2}, 1]$\rlap{,}}
    \end{aligned}
    \right.
  \end{tikzcd}
\end{equation}
while \emph{horizontal composition} of homotopies by maps is the evident actual composition of component maps:
\begin{equation}
  \begin{tikzcd}[column sep=20pt]
  X'
  \ar[r, "{l}"]
  &
  X
  \ar[
    rr,
    bend left=30,
    "{ f }"{description, name=s}
  ]
  \ar[
    rr,
    bend right=30,
    "{ g }"{description, name=t}
  ]
  \ar[
    from=s,
    to=t,
    Rightarrow,
    dashed,
    "{ \eta }"
  ]
  &&
  Y
  \ar[r, "{ r }"]
  &
  Y'
  \end{tikzcd}
  \;: \;
  (x',s) 
    \mapsto
  r \circ \eta(-,s) \circ l (x')
  \mathrlap{\,.}
\end{equation}
Combining this, one obtains horizontal composition of homotopies themselves, as 
\begin{equation}
  \begin{tikzcd} 
    {}
    \ar[
      rr, bend left=40,
      "{}"{swap, name=s1}
    ]%
    \ar[
      rr, bend right=40,
      shorten=-3pt,
      "{}"{name=t1}
    ]%
    \ar[
      from=s1,
      to=t1,
      dashed,
      Rightarrow
    ]
    &&
    {}
    \ar[
      rr, bend left=40,
      "{}"{swap, name=s2}
    ]
    \ar[
      rr, bend right=40,
      shorten=-3pt,
      "{}"{name=t2}
    ]
    \ar[
      from=s2,
      to=t2,
      dashed,
      Rightarrow
    ]
    &&
    {}
  \end{tikzcd}
 \; := \;
  \begin{tikzcd}[
    row sep=-2pt
  ]
    {}
    \ar[
      rr, bend left=40,
      "{}"{swap, name=s1}
    ]
    &&
    {}
    \ar[
      rr, bend left=40,
      "{}"{swap, name=s2}
    ]
    \ar[
      rr, bend right=40,
      shorten=-3pt,
      "{}"{name=t2}
    ]
    \ar[
      from=s2,
      to=t2,
      dashed,
      Rightarrow
    ]
    &&
    {}
    \\
    {}
    \ar[
      rr, bend left=40,
      "{}"{swap, name=s1}
    ]
    \ar[
      rr, bend right=40,
      shorten=-3pt,
      "{}"{name=t1}
    ]
    \ar[
      from=s1,
      to=t1,
      dashed,
      Rightarrow
    ]
    &&
    {}
    \ar[
      rr, bend right=40,
      shorten=-2pt,
      "{}"{name=t2}
    ]
    &&
    {}
  \end{tikzcd}
\end{equation}
(or the other way around, which is different but higher-order homotopic),
such as in the important special case of \emph{pasting composites} of ``square'' homotopies:
\begin{equation}
  \label{PastingOfHomotopies}
  \begin{tikzcd}[
    sep=30pt
  ]
    {}
    \ar[r]
    \ar[d]
    & 
    {}
    \ar[d]
    \ar[r]
    \ar[
      dl,
      shorten=8pt,
      dashed,
      Rightarrow
    ]
    & 
    {}
    \ar[d]
    \ar[
      dl,
      shorten=8pt,
      dashed,
      Rightarrow
    ]
    \\
    {}
    \ar[r]
    & 
    {}
    \ar[r]
    & 
    {}
  \end{tikzcd}
 \;\; := \;
  \begin{array}{c}
  \begin{tikzcd}[
    sep=30pt
  ]
    {}
    \ar[r]
    &[+8pt] 
    {}
    \ar[d]
    \ar[r]
    & 
    {}
    \ar[d]
    \ar[
      dl,
      shorten=8pt,
      dashed,
      Rightarrow
    ]
    \\
    {}
    & 
    {}
    \ar[r]
    & 
    {}
  \end{tikzcd}
  \\[-40pt]
  \begin{tikzcd}[
    sep=30pt
  ]
    {}
    \ar[r]
    \ar[d]
    & 
    {}
    \ar[d]
    \ar[
      dl,
      shorten=8pt,
      dashed,
      Rightarrow
    ]
    &[8pt] 
    {}
    \\
    {}
    \ar[r]
    & 
    {}
    \ar[r]
    & 
    {}
    \mathrlap{\;\,.}
  \end{tikzcd}
  \end{array}
\end{equation}

The higher homotopy theory of such 2- and higher-dimensional diagrams of maps and homotopies is usefully captured by \emph{model category structure on topological spaces} (cf. \cite{nLab:ClassicalModelStructureOnTop}, for review in our context see \cite[\S 1]{FSS23-Char}).  Here we proceed with making explicit only the most minimum amount of technology necessary at this point.

\subsubsection{Topological Groupoids}
\label{TopologicalGroupoids}

A \emph{groupoid} (cf. \parencites{Weinstein1996Groupoids}[\S 2.1]{Sc2017-Topology}{IbortRodriguez2021}[p. 6]{Schreiber2025}) is a ``set with gauge transformations'' between its elements. 
For example, the phase space of a gauge theory is a groupoid, whose ``objects'' are the gauge field configurations and whose ``morphisms'' are the actual gauge transformations between them. This example is actually a \emph{Lie groupoid}, hence with smooth structure on its sets of objects and morphism (in physics this is best known for infinitesimal gauge transformations only, which gives the underlying \emph{Lie algebroid} whose Chevalley-Eilenberg algebra is known as the \emph{BRST complex}, cf. \cite[\S 10]{Sc2017-QFT}.) 

For the time being, we disregard smooth structure (we turn to this instead in \cref{SmoothInfinityGroupoids}) and consider groupoids in the broad generality where they are equipped with any topological structure (which here means: any \emph{D-topological structure}, by \cref{SpacesAreDTopological}, whence the following is about \emph{D-topological groupoids}):

\begin{definition}
\label[definition]{TopologicalGroupoid}
A \emph{topological groupoid} $\mathcal{X}$ (cf. \parencites[\S II.1]{Mackenzie1987}[Ntn. 2.2.1]{SS25-Bun}) is a topological space $\mathrm{Mor}(\mathcal{X})$  of ``morphisms'' and a subspace $\mathrm{Obj}(\mathcal{X})$ of ``objects'' (identity morphisms)  --- equipped with continuous maps of this form:
\begin{equation}
  \label{StructureOfATopologicalGroupoid}
    \begin{tikzcd}[column sep=large]
    \mathrm{Mor}(\mathcal{X})
      \tensor[_s]{\times}{_t}
    \mathrm{Mor}(\mathcal{X})
    \ar[
      r, 
      "{ (-)\circ(-) }"
    ]
    &
    \mathrm{Mor}(\mathcal{X})
    \ar[
      out=180-60,
      in=60,
      looseness=4,
      "{ i }"{description}
    ]
    \ar[
      r, shift left=8pt,
      "{ s }"{description}
    ]
    \ar[
      r, 
      shift right=8pt,
      "{ t }"{description}
    ]
    &
    \mathrm{Obj}(\mathcal{X})
    \ar[
      l, 
      hook',
      "{ \mathrm{e} }"{description}
    ]
    \mathrlap{\,,}
    \end{tikzcd}
\end{equation}
such that 
\begin{enumerate}
\item 
$s \circ e = t \circ e = \mathrm{id}$, which means that the morphisms form a reflexive graph over the subspace of objects,
\item 
the \emph{composition} operation $(-) \circ (-)$ --- of a morphism $f$ whose \emph{target} object $t(f)$ coincides with the \emph{source} object $s(g)$ of another morphism $g$ --- is associative, and unital with respect to the \emph{identity morphisms} $\mathrm{e}_x$ on objects $x$, and has \emph{inverse morphisms} given by $i$.
\end{enumerate}

For notational transparency, when identities and inverses are understood, it is often useful to denote topological groupoids by the set of generic pairs of composable morphisms and their composites, like this:
\begin{equation}
  \label{SimplicialNotationForGroupoids}
  \mathcal{X}
  \defneq
  \left\{
  \begin{tikzcd}[
    row sep=7pt, 
    column sep=25pt
  ]
    & 
    y
    \ar[dr, "{ g }"]
    \\
    x
    \ar[rr, "{ g \circ f }"]
    \ar[ur, "{ f }"]
    &&
    z
  \end{tikzcd}
  \;\middle\vert\;
    \begin{aligned}
      & 
      x,y,z \in \mathrm{Obj}
      \\
      & f,g \in \mathrm{Mor}
    \end{aligned}
  \right\}.
\end{equation}
Given a topological groupoid $\mathcal{X}$ we say:
\begin{enumerate}
\item The coequalizer \cref{CoequalizerDiagram} of the source and taget map is its space of \emph{isomorphism classes}:
\begin{equation}
 \label{IsomorphismClassesOfTopGroupoid}
  [\mathcal{X}]_0
  :=
  \mathrm{Obj}(\mathcal{X})\big/
  (\!\!
    \tensor[_s]{\sim}{_t}
  )
  \,.
\end{equation}

\item
The fiber of the combined source/target map
\begin{equation}
  \label{CombinedSourceAndTargetMap}
  \begin{tikzcd}
    \mathrm{Mor}(\mathcal{X})
    \ar[r, "{ (s,t) }"]
    &
    \mathrm{Obj}(\mathcal{X})^2
  \end{tikzcd}
\end{equation}
over a single object $x := (x,x) \in \mathrm{Obj}(\mathcal{X})^2$ is the \emph{isotropy group} (or \emph{automorphism group} or \emph{stabilizer group}) of $x$ 
\begin{equation}
  \label{IsotropyGroupOfTopGrpd}
  \mathcal{X}_x
  :=
  \left\{\!
    \adjustbox{
      raise=-5pt
    }{
    \begin{tikzcd}
      x
      \ar[
        in=50,
        out=180-50,
        looseness=4,
        shift right=2pt,
       "{ \,\mathclap{g}\, }"{description}
      ]
    \end{tikzcd}
    }
   \middle\vert
   \;
     g \in \mathrm{Mor}(\mathcal{X})
  \right\}
  \,,
\end{equation}
with topological group structure \cref{TopologicalGroup} inherited from the restriction of the topological groupoid structure. 
\end{enumerate}
\end{definition}

\begin{remark}
\label[remark]
 {LieGroupoidsAndSmoothGroupoids}
There are several natural variants and generalizations of the \emph{topological groupoids} of \cref{TopologicalGroupoid}:
\begin{enumerate}
\item A \emph{Lie groupoid} (cf. \parencites{Mackenzie1987}{Moerdijk2003}) is a topological groupoid whose spaces of objects and morphisms are equipped with the structure of smooth manifolds, whose structure maps are smooth maps, and whose source and target maps are submersions, so that their fiber product (of composable morphisms) also inherits the structure of a smooth manifold.

Most of the following examples and discussion apply to Lie groupoids just as well. Exceptions are mapping objects, starting with \cref{ProbingTopologicalGroupoidByRns}, which may be ``too large'' to be smooth manifolds (nor even Fr{\'e}chet manifolds, for that matter).

\item More generally, one may consider \emph{diffeological groupoids}, whose morphism space is equipped with the structure of a \emph{diffeological space} (cf. \parencites{IglesiasZemmour2013}[Ntn. 4.3.15]{SS25-Bun}{nLab:DiffeologicalSpace}) and whose structure maps are smooth maps with respect to that diffeological structure. This class faitfully subsumes both D-topological groupoids and Lie groupoids and is closed under all operations discussed here. 

\item Fully generally, as far as groupoids in differential topology are concerned, one may consider \emph{smooth groupoids} \parencites[Def. 1.2.252]{Sc13-dcct}{Eggertsson2014}, whose space of morphisms is equipped with the structure of a \emph{smooth set} \parencites[Def. 1.2.16, 1.3.58]{Sc13-dcct}{GS25-FieldsI}{Schreiber2025}[Ntn. 4.3.15]{SS25-Bun}{IbortMas2025} and whose structure maps are smooth with respect to that. This class faithfully subsumes all of the above but contains also ``non-concrete'' groupoids, like the moduli stacks of $\mathbf{B}\Gamma_{\mathrm{conn}}$ of $\Gamma$-principal bundles with connection (\parencites[Prop. 1.2.107]{Sc13-dcct}{FSS15-Stacky}[Ex. 2.11]{BeniniSchenkelSchreiber2018}, in variation of the plain moduli stack $\mathbf{B}\Gamma$ discussed below in \cref{DeloopingGroupoid,GroupoidOfPrincipalBundles,OrdinaryNonabelianCechCohomology}). 

\item From this point on it is natural to generalize, finally, to higher smooth groupoids, namely to \emph{smooth $\infty$-groupoids} \parencites[\S 4.4]{Sc13-dcct}[\S 4.3]{SS25-Bun}{SS26-Orb}, which is what we indicate in \cref{SmoothInfinityGroupoids}.
\end{enumerate}
  For the remainder of this section on topological groupoids we disregard all this further generality just for pedagogy of the exposition. The inclined reader is invited to make the evident substitutions.
\end{remark}

\begin{example}
  \label[example]{SpaceAsGroupoid}
  Given a topological space $X$, it may be regarded as a topological groupoid whose only morphisms are identities:
  \begin{equation}
    \bigg(
    \begin{tikzcd}[
      column sep=15pt
    ]
      X 
      \ar[
        rr,
        "{ \mathrm{id} }"
      ]
      &&
      X
      \ar[
        rr,
        shift left=6pt,
        "{\mathrm{id}}"
      ]
      \ar[
        rr,
        <-,
        "{\mathrm{id}}"{description}
      ]
      \ar[
        rr,
        shift right=6pt,
        "{\mathrm{id}}"{swap}
      ]
      && 
      X
    \end{tikzcd}
    \bigg)
    =
    \big\{
      x
      \,\big\vert\,
      x \in X
    \big\}
    \,,
  \end{equation}
  and we denote this groupoid still by ``$X$''.
\end{example}

\begin{example}
  The \emph{interval groupoid} $I$ has two objects, $\mathrm{Obj} = \{0,1\}$, and a single morphism and its inverse between these (hence, with the identity morphisms, a total of four morphisms):
  \begin{equation}
    \label{ArrowPictureOfIntervalGroupoid}
    I 
    =
    \big\{
    \begin{tikzcd}
      0
      \ar[
        r,
        bend left=20
      ]
      \ar[
        r,
        <-,
        bend right=20
      ]
      &
      1
    \end{tikzcd}
    \big\}
    \mathrlap{\,.}
  \end{equation}
\end{example}

\begin{example}[{cf. \cite[Ex. 2.2.6]{SS25-Bun}}]
  \label[example]{DeloopingGroupoid}
  For $\Gamma$ a topological group, its \emph{delooping groupoid} is the topological groupoid with a single object, $\Gamma$ worth of morphisms, composition given by the group operation $(-)\cdot(-) : \Gamma \times \Gamma \to \Gamma$ and inversion given by group inverses 
  $(-)^{-1} : \Gamma \to \Gamma$, hence:
  \begin{equation}
    \label{TheDeloopingGroupoid}
    \mathbf{B}G
    :=
    \Bigg(
    \begin{tikzcd}[column sep=30pt]
      \Gamma \times \Gamma
      \ar[
        r,
        "{
          (-)\cdot(-)
        }"
      ]
      &[+5pt]
      \;
      \Gamma
      \ar[
        out=60,
        in=180-60,
        looseness=4,
        "{
          \;\mathclap{(-)^{\mathrlap{-1}}}\;
        }"{description}
      ]
      \;
      \ar[
        r, 
        shift left=6pt,
      ]
      \ar[
        r,
        <-,
        "{ 
          \mathrm{e} 
        }"{description}
      ]
      \ar[
        r, 
        shift right=6pt,
      ]
      &
      \ast
    \end{tikzcd}
    \Bigg)
    =
    \left\{
    \begin{tikzcd}[
      row sep=small, 
      column sep=15pt,
    ]
      & 
      \bullet
      \ar[
        dr,
        "{ \gamma_2 }"
      ]
      \\
      \bullet
      \ar[
        ur,
        "{ \gamma_1 }"
      ]
      \ar[
        rr,
        "{
          \gamma_2 \cdot \gamma_1
        }"
      ]
      &&
      \bullet
    \end{tikzcd}
    \middle\vert\,
    \gamma_i \in \Gamma
    \right\}.
  \end{equation}
  When regarded as a \emph{topological stack}, below in \cref{TopologicalStacks}, this simple delooping groupoid is (a representation of) the \emph{moduli stack of principal $\Gamma$-bundles} (cf. \cref{GroupoidOfPrincipalBundles,OrdinaryNonabelianCechCohomology}), in fact of \emph{equivariant principal $\Gamma$-bundles} (cf. \cref{EquivariantNonabelianCechCohomology}).
\end{example}

\begin{example}
  \label{GroupoidOfPrincipalBundles}
  For $X$ a topological space $\Gamma$ a topological group \cref{TopologicalGroup}, the groupoid
  \begin{equation}
    \label{TheGroupoidOfPrincipalBundles}
    \Gamma \mathrm{PrnBdl}(X)
    \defneq
    \left\{
    \begin{tikzcd}[
      row sep=7pt
    ]
      &
      P_2
      \ar[
        dr,
        "{
          \gamma'
        }"        
      ]
      \\
      P_1
      \ar[
        ur,
        "{
          \gamma
        }"
      ]
      \ar[
        rr,
        "{
          \gamma' \circ \gamma
        }"
      ]
      &&
      P_3
    \end{tikzcd}
    \right\}
  \end{equation}
  has as objects the (discrete set of) $\Gamma$-principal bundles $P$ over $X$  (cf. \parencites[\S 4.3]{Husemoller1994}[\S 1.1]{RudolphSchmidt2017}[\S 9]{Nakahara2018}[Ntn. 2.0.25]{SS25-Bun}) and as morphisms their $\Gamma$-equivariant bundle homomorphisms,
  \begin{equation}
    \begin{tikzcd}[
      row sep=-1pt
    ]
      P_1
      \ar[
        in=60,
        out=180-60,
        looseness=4,
        "{ 
          \,\mathclap{\Gamma}\, 
        }"{description}
      ]
      \ar[
        rr,
        "{ \gamma }"
      ]
      \ar[dr]
      && 
      P_2
      \ar[
        in=60,
        out=180-60,
        looseness=4,
        "{ 
          \,\mathclap{\Gamma}\, 
        }"{description}
      ]
      \ar[dl]
      \\
      & X
    \end{tikzcd}
  \end{equation}
  hence the \emph{gauge transformations} if we think of these bundles as charge sectors of gauge fields.

  Over the point, this reduces to the delooping groupoid from \cref{DeloopingGroupoid}:
  \begin{equation}
    \Gamma\mathrm{PrnBdl}(\ast)
    =
    \mathbf{B}\Gamma
    \mathrlap{\,.}
  \end{equation}
\end{example}

\begin{example}[{cf. \cite[Ex. 2.2.6]{SS25-Bun}}]
\label[definition]{ActionGroupoid}
Given an action $G \acts X$ (\cref{TopologicalGroupAction}), its \emph{action groupoid} (or \emph{homotopy quotient}) is:
\begin{equation}
  \label{TheActionGroupoid}
  \begin{aligned}
  G \backsslash X
  &
  \simeq
  \left(
    \begin{tikzcd}[row sep=23pt,
      column sep=45pt
    ]
      G \times G \times X
      \ar[
        rr,
        "{
          (g_2, g_1, x)
          \,\mapsto\,
          (g_2 \cdot g_1, x)
        }"{yshift=+1pt}
      ]
      &&
      G \times X
      \ar[
        r,
        shift left=8pt,
        "{
          (g,x) \,\mapsto\, x
        }"
      ]
      \ar[
        r,
        <-,
        "{
          (\mathrm{e},x)
          \,\mapsfrom\,
          x
        }"{description}
      ]
      \ar[
        r,
        shift right=8pt,
        "{ 
          (g,x) \,\mapsto\, g \cdot x 
        }"{swap}
      ]
      &[+5pt]
      X
    \end{tikzcd}
  \right)
  \\
  & =
  \left\{
  \begin{tikzcd}[
    row sep=9pt, 
    column sep={between origins, 32pt}
  ]
    & 
    g_1 \cdot x
    \ar[
      dr,
      "{ g_2 }"
    ]
    \\
    x
    \ar[
      ur,
      "{ g_1 }"
    ]
    \ar[
      rr,
      "{
        g_2 \cdot g_1
      }"
    ]
    &&
    g_2 \cdot g_1 \cdot x
  \end{tikzcd}
  \;\middle\vert\;
  \begin{aligned}
    & x \in X
    \\
    & g_1, g_2 \in G
  \end{aligned}
  \right\}
  \mathrlap{.}
  \end{aligned}
\end{equation}
For an action $G \acts \ast$ on the point, the corresponding action groupoid \cref{TheActionGroupoid} is the delooping groupoid \cref{TheDeloopingGroupoid}:
\begin{equation}
  \label{DeloopingGroupoidIsPointQuotient}
  G \backsslash \ast
  =
  \mathbf{B}G  
  =
  \left\{
  \adjustbox{raise=-7pt}{
  \begin{tikzcd}
    \bullet
    \ar[
      in=55,
      out=180-55,
      looseness=5,
      shift right=2pt,
      "{ g }"{description}
    ]
  \end{tikzcd}
  }
  \Big\vert\;
    g \in G
  \right\}
.
\end{equation}
\end{example}

\begin{example}
  \label[example]
  {GeometricallyDiscreteUnderlyingGroupoid}
  Given a topological groupoid $\mathcal{X}$, its \emph{underlying topologically discrete} groupoid $\flat \mathcal{X}$ has the same objects, morphisms and structure maps as $\mathcal{X}$, but for the discrete topology:
  \begin{equation}
    \label{UnderlyingTopologicallyDiscreteGroupoid}
    \flat \mathcal{X}
    :=
    \Bigg(
    \begin{tikzcd}
      \flat \mathrm{Mor}(\mathcal{X})
      \tensor[_s]{\times}{_t}
      \flat \mathrm{Mor}(\mathcal{X})
      \ar[r, "{\circ}"]
      &
      \flat \mathrm{Mor}(\mathcal{X})
      \ar[
        out=180-60,
        in=60,
        looseness=4,
        "{
          \,\mathclap{i}\,
        }"{description}
      ]
      \ar[
        r,
        shift left=7pt,
        "{ s }"
      ]
      \ar[
        r,
        <-,
        "{ \mathrm{e} }"{description}
      ]
      \ar[
        r,
        shift right=7pt,
        "{ t }"{swap}
      ]
      &
      \flat \mathrm{Obj}(\mathcal{X})
    \end{tikzcd}
    \Bigg)
    \mathrlap{\,.}
  \end{equation}
\end{example}

\begin{example} 
  \label[example]{FundamentalGroupoid}
  Given a topological space $X$, its \emph{fundamental groupoid}, $\shape_{\!1} X$, is the topological groupoid with objects  the discrete set of points of $X$, and morphisms the homotopy classes of continuous paths between fixed endpoints, with composition by concatenation $(-) \star (-)$ of paths:
  \begin{equation}
    \begin{tikzcd}[
      column sep=50pt
    ]
      \underset{
        {
          \substack{
            x,y,z 
            \\
            \in X
          }
        }
      }{\coprod}
      \,
      \substack{
          \pi_0 
          \Big(
          \mathrm{Map}\big(
           [0,1], 
           X
          \big)
            ^y
            _x
          \Big)
          \\
          \overset{ 
            \substack{
              \times
              \\
              {}
            }
          }{\pi_0} 
          \Big(
          \mathrm{Map}\big(
           [0,1], 
           X
          \big)
            ^z
            _y
          \Big)
      }
      \ar[
        r,
        "{
          ([\gamma_1], [\gamma_2])
        }",
        "{ \mapsto }"{description},
        "{
          [\gamma_1 \star \gamma_2]        
        }"{swap}
      ]
      &[-8pt]
      \underset{
        \mathclap{x,y \in X}
      }{\coprod}
      \pi_0 
      \Big(
      \mathrm{Map}\big(
       [0,1], 
       X
      \big)
        ^y
        _x
      \Big)
      \ar[
        r,
        shift left=8pt,
        "{
            x 
              \overset{[\gamma]}{\rightsquigarrow} 
            y
          \,\mapsto\,
          x
        }"
      ]
      \ar[
        r,
        <-,
        "{
          [\mathrm{cnst}_x]
          \,\mapsfrom\,
          x
        }"{description}
      ]
      \ar[
        r,
        shift right=8pt,
        "{
            x 
              \overset{[\gamma]}{\rightsquigarrow} 
            y
          \,\mapsto\,
          y
        }"{swap}
      ]
      &[+6pt]
      \flat X
      \mathrlap{\,,}
    \end{tikzcd}
  \end{equation}
  where
  \begin{equation}
    \begin{tikzcd}[
      column sep=30pt,
      row sep=15pt,
    ]
      \mathrm{Map}\big(
        [0,1],
        \,
        X
      \big)_x^y
      \ar[d]
      \ar[r]
      \ar[
        dr,
        phantom,
        "{
          \lrcorner
        }"{pos=.3}
      ]
      &
      \mathrm{Map}\big(
        [0,1],
        \,
        X
      \big)
      \ar[
        d,
        "{
          (\mathrm{ev}_0, \mathrm{ev}_1)
        }"
      ]
      \\
      \ast
      \ar[
        r, 
        "{ (x,y) }"{description}
      ]
      &
      X \times X
    \end{tikzcd}
  \end{equation}
  is the space of paths with endpoints $(x,y)$.
  Hence in the notation \cref{SimplicialNotationForGroupoids} a fundamental groupoid looks like this:
  \begin{equation}
   \shape_{\!1} X
   =
   \left\{
   \begin{tikzcd}[
     row sep=12pt, 
     column sep=35pt
  ]
     & 
     y
     \ar[
       dr,
       "{
         y \overset{[\gamma_2]}{\rightsquigarrow} z
       }"{sloped}
     ]
     \\
     x
     \ar[
       ur,
       "{
         x \overset{[\gamma_1]}{\rightsquigarrow} y
       }"{sloped}
     ]
     \ar[
       rr,
       "{
         x \overset{
             [\gamma_1 \!\star\! \gamma_2]
           }{\rightsquigarrow} 
         y
       }"{sloped}
     ]
     &&
     z
   \end{tikzcd}
   \right\}
   \mathrlap{\,.}
  \end{equation}

  The isomorphism classes \cref{IsomorphismClassesOfTopGroupoid} of a fundamental groupoid are the connected components of the topological space
  \begin{equation}
    \label{IsoClassesOfFundamentalGroupoid}
    \big[
      \shape_{\! 1}
      X
    \big]_0
    \,=\,
    \pi_0 X
    \mathrlap{\,.}
  \end{equation}
\end{example}

\begin{example}
  \label[example]
    {ProbingTopologicalGroupoidByRns}
  For $\mathcal{X}$ a topological groupoid \cref{StructureOfATopologicalGroupoid} and $U$ a space, then forming mapping spaces \cref{MappingSpace} from $U$ into the component spaces of $\mathcal{X}$,
  \begin{equation}
    \label{MapsFromSpaceIntoGroupoid}
    \mathrm{Map}(U,\mathcal{X})
    :=
    \left(
    \begin{tikzcd}[
      column sep={between origins, 17pt}
    ]
      {
      \mathrm{Map}\big(
        U,
        \mathrm{Mor}(\mathcal{X})
      \big)
      \!\!
      \tensor[_{s_\ast}]{\times\!}{_{t_\ast}}
      \mathrm{Map}\big(
        U,
        \mathrm{Mor}(\mathcal{X})
      \big)
      }
      \ar[
        d,
        "{ \circ_\ast }"
      ]
      \\
      \mathrm{Map}\big(
        U,
        \mathrm{Mor}(\mathcal{X})
      \big)
      \ar[
        r,
        shift left=6pt,
        "{
          s_\ast
        }"
      ]
      \ar[
        r,
        <-,
        "{
          \mathrm{e}_\ast
        }"{description}
      ]
      \ar[
        r,
        shift right=6pt,
        "{
          t_\ast
        }"{swap}
      ]
      &
      \mathrm{Map}\big(
        U,
        \mathrm{Obj}(\mathcal{X})
      \big)
    \end{tikzcd}
    \hspace{-2pt}
    \right)
    \mathrlap{,}
  \end{equation}
  gives a topological groupoid which may be thought of as the groupoid of \emph{$U$-parameterized} objects of $\mathcal{X}$.

  If here $U \defneq \mathbb{R}^n$ is a Cartesian space \cref{CartesianSpaceHomeomorphicToOpenBall}, we also call $\mathrm{Map}(\mathbb{R}^n,\mathcal{X})$ the space of \emph{$n$-dimensional plots} of $\mathcal{X}$. Consider the quotient spaces 
  \begin{equation}
    \mathrm{Map}\big(
      \mathbb{G}^n,
      -
    \big)
    :=
    \mathrm{Map}\big(
      \mathbb{R}^n,
      -
    \big)\big/\!\sim_n
  \end{equation}
  (where ``$\mathbb{G}$'' is for ``germ'', which is not an actual topological space itself, but defined via the above formula)
  by the equivalence relation $\sim_n$ which identifies a pair of maps $\phi, \phi' :   \mathbb{R}^n \to (-)$ if they agree on any open ball $\mathbb{D}^n_\epsilon$ \cref{CartesianSpaceHomeomorphicToOpenBall} around the origin:
  \begin{equation}
    \phi \sim_n \phi'
    \;\;\;
    \Leftrightarrow
    \;\;\;
    \exists \; {\epsilon \in \mathbb{R}_{> 0}},
    \;\;
    \phi \vert_{ \mathbb{D}^n_\epsilon}
    =
    \phi' \vert_{ \mathbb{D}^n_\epsilon}
    \,.
  \end{equation}
  This construction extends to topological groupoids $\mathcal{X}$ as in \cref{MapsFromSpaceIntoGroupoid} to yield what we may call the \emph{groupoid of stalks of $n$-dimensional plots} of $\mathcal{X}$:
  \begin{equation}
    \label{StalkOfPlotsOfTopologicalGroupoid}
    \mathrm{Map}(\mathbb{G}^n,\mathcal{X})
    :=
    \!
    \left(
    \begin{tikzcd}[
      column sep={between origins, 12pt}
    ]
      {
      \mathrm{Map}\big(
        \mathbb{G}^n,
        \mathrm{Mor}(\mathcal{X})
      \big)
      \!\!
      \tensor[_{s_\ast}]{\times\!}{_{t_\ast}}
      \mathrm{Map}\big(
        \mathbb{G}^n,
        \mathrm{Mor}(\mathcal{X})
      \big)
      \hspace{-6pt}
      }
      \ar[
        d,
        "{ \circ_\ast }"
      ]
      \\
      \mathrm{Map}\big(
        \mathbb{G}^n,
        \mathrm{Mor}(\mathcal{X})
      \big)
      \ar[
        r,
        shift left=6pt,
        "{
          s_\ast
        }"
      ]
      \ar[
        r,
        <-,
        "{
          \mathrm{e}_\ast
        }"{description, pos=.55}
      ]
      \ar[
        r,
        shift right=6pt,
        "{
          t_\ast
        }"{swap}
      ]
      &
      \mathrm{Map}\big(
        \mathbb{G}^n,
        \mathrm{Obj}(\mathcal{X})
      \big)
    \end{tikzcd}
    \right).
  \end{equation}
  
\end{example}

\begin{definition}
\label[definition]{TopologicalFunctor}
Given a pair $(\mathcal{X}, \mathcal{Y})$ of topological groupoids (\cref{TopologicalGroupoid}), a \emph{continuous functor} (or \emph{topological functor}) $\mathcal{X} \xrightarrow{F} \mathcal{Y}$ between them is a map 
$\mathrm{Mor}(\mathcal{X}) \xrightarrow{ F_1 } \mathrm{Mor}(Y)$ which homomorphically respect all the structure \cref{StructureOfATopologicalGroupoid}:
\begin{equation}
  \begin{tikzcd}[
    row sep=2pt,
    column sep=15pt
  ]
    \mathcal{X}
    \ar[
      rr,
      "{ F }"
    ]
    &&
    \mathcal{Y}
    \\
    x
    \ar[
      dd,
      "{
        f
      }"
    ]
    \ar[
      dddd,
      leftvertright,
      "{
        g \circ f
      }"{
        sloped,
        rotate=180,
        description
      }
    ]
    &&
    F_0(x)
    \ar[
      dd,
      "{
        F(f)
      }"{swap}
    ]
    \ar[
      dddd,
      rightvertleft,
      "{
        F(g) \,\circ\, F(f)
      }"{
        sloped,
        rotate=180,
        description
      }
    ]
    \\
    & \mapsto
    \\
    y
    \ar[
      dd,
      "{ 
        g
      }"
    ]
    &&
    F_0(y)
    \ar[
      dd,
      "{ 
        F(g)
      }"{swap}
    ]
    \\
    & \mapsto
    \\
    z
    &&
    F_0(z)
  \end{tikzcd}
\end{equation}
(Here $F_0:\mathrm{Obj}(\mathcal{X}) \to \mathrm{Obj}(\mathcal{Y})$ is the restriction of $F$ to identity morphisms identified with objects.)
\end{definition}

\begin{definition}
\label[definition]{NaturalTransformations}
A continuous (``natural'') \emph{transformation} between parallel continuous functors (\cref{TopologicalFunctor}) is a map 
\begin{equation}
  \label{NaturalTransformation}
  \begin{tikzcd}
    \mathcal{X}
    \ar[
      rr,
      bend left=30,
      "{ F_1 }"{description, name=s}
    ]
    \ar[
      rr,
      bend right=30,
      "{ F_2 }"{description, name=t}
    ]
    &&
    \mathcal{Y}
    \ar[
      from=s,
      to=t,
      Rightarrow,
      "{ \eta }"
    ]
  \end{tikzcd}
  \;:\;
  \begin{tikzcd}
    \mathrm{Obj}(\mathcal{X})
    \ar[
      r
    ]
    &
    \mathrm{Mor}(\mathcal{Y})
  \end{tikzcd}
\end{equation}
whose values make these diagrams commute:
\begin{equation}
  \label{NaturalityForContinuousTransformations}
  \begin{tikzcd}[row sep=2pt,
   column sep=16pt
  ]
    x_1
    \ar[
      dd,
      "{ 
        \forall \; f 
      }"{swap}
    ]
    &[-10]&[-10]
    F_1(x_1)
      \ar[
        dd,
        "{ F_1(f) }"{swap}
      ]
      \ar[
        rr,
        "{ \eta(x_1) }"
      ]
      &&
      F_2(x_1)
      \ar[
        dd,
        "{ F_2(f) }"
      ]
    \\
    & :\;\;
    \\
    x_2
    &&
    F_1(X_2) 
    \ar[
      rr, 
      "{\eta(x_2)}" 
    ]
    &&
    F_2(x_2)
    \mathrlap{\,.}
  \end{tikzcd}
\end{equation}
The \emph{vertical composition} of composable such transformations is by composition of their component morphisms:
\begin{equation}
  \begin{tikzcd}
    \mathcal{X}
    \ar[
      rr,
      bend left=50,
      "{ F_1 }"{description, name=s}
    ]
    \ar[
      rr,
      "{ F_2 }"{description, name=m}
    ]
    \ar[
      rr,
      bend right=50,
      "{ F_3 }"{description, name=t}
    ]
    \ar[
      from=s,
      to=m,
      Rightarrow,
      "{ \, \eta_1 }"
    ]
    \ar[
      from=m,
      to=t,
      Rightarrow,
      "{ \, \eta_2 }"
    ]
    &&
    \mathcal{Y}
  \end{tikzcd}
    \;\;:\;
    x 
      \;\longmapsto\;
    \begin{tikzcd}
      F_1(x)
      \ar[r, "{ \eta_1(x) }"]
      &
      F_2(x)
      \ar[r, "{ \eta_2(x) }"]
      &
      F_3(3)
      \mathrlap{\,,}
    \end{tikzcd}
\end{equation}
while the \emph{horizontal composition} with functors is
\begin{equation}
  \label{HorizontalWhiskeringOfTransformations}
  \begin{tikzcd}
    \mathcal{X}'
    \ar[r, "{ L }"]
    &
    \mathcal{X}
    \ar[
      r,
      bend left=45,
      "{ F_1 }"{description, name=s}
    ]
    \ar[
      r,
      bend right=45,
      "{ F_2 }"{description, name=t}
    ]
    \ar[
      from=s,
      to=t,
      Rightarrow,
      "{ \eta }"
    ]
    &
    \mathcal{Y}
    \ar[r, "{ R }"]
    &
    \mathcal{Y}'
  \end{tikzcd}
  \;:\;
  x' 
    \;\longmapsto\;
  R\Big(\eta\big(L_0(x')\big)\!\Big)
  \mathrlap{\,.}
\end{equation}
\end{definition}

\begin{example}
\label{RelationBetweenHomotopiesAndTransformations}
There is a close relation between homotopies between topological spaces \cref{AHomotopy} and transformations between topological groupoids \cref{NaturalTransformation}.
Concretely, any homotopy between topological spaces gives a transformation between their fundamental groupoids (\cref{FundamentalGroupoid}).
\end{example}

\begin{example}
  Between delooping groupoids (\cref{DeloopingGroupoid}), continuous functors (Def. \ref{TopologicalFunctor}) are continuous group homomorphisms, hence continuous linear representations if the second group is linear:
  \begin{equation}
    \Big\{
    \begin{tikzcd}
      \mathbf{B}G
      \ar[r, "{ \mathbf{B}\rho }"]
      &
      \mathbf{B}\mathrm{U}(\HilbertSpace)
    \end{tikzcd}
    \Big\}
    \simeq
    \Big\{
      \rho \in \mathrm{Rep}(
        G, \HilbertSpace
      )
    \Big\}
    \mathrlap{\,.}
  \end{equation}
Transformations (\cref{NaturalTransformations}) between these are \emph{intertwiners} of representations.
\end{example}

\begin{definition}
  \label[definition]{FunctorGroupoid}
  For a pair of topological groupoids $\mathcal{X}$, $\mathcal{Y}$ (\cref{TopologicalGroupoid}), their \emph{functor groupoid}, $\mathrm{Func}(\mathcal{X}, \mathcal{Y})$, is the topological groupoid whose objects are the continuous functors $\mathcal{X} \to \mathcal{Y}$ (\cref{TopologicalFunctor}), topologized as a subspace of $\mathrm{Map}\big(\mathrm{Mor}(\mathcal{X}), \mathrm{Mor}(\mathcal{Y})\big)$ \cref{MappingSpace}, and whose morphisms are the continuous transformations \cref{NaturalTransformation}, topologized as the product space of that with $\mathrm{Map}\big(\mathrm{Obj}(\mathcal{X}), \mathrm{Mor}(\mathcal{Y})\big)$. Composition and inversion of transformations is given by composition of inversion of their component functions \cref{NaturalityForContinuousTransformations}.
  Hence in the notation \cref{SimplicialNotationForGroupoids}:
  \begin{equation}
    \label{FunctorGroupoidInArrowNotatiom}
    \mathrm{Func}(\mathcal{X}, \mathcal{Y})
    =
    \left\{\hspace{-3pt}
    \begin{tikzcd}[
      column sep=5pt,
      row sep=40pt
    ]
      & 
      \mathcolor{gray}{\mathcal{Y}}
      \\
     \mathcolor{gray}{\mathcal{X}}
     \ar[
       ur,
       bend left=95,
       gray,
       "{ F_1 }"{description, name=one}
     ]
     \ar[
       ur,
       bend right=95,
       gray,
       "{ F_3 }"{description, name=three}
     ]
     \ar[
       from=one,
       to=three,
       Rightarrow,
       shorten <=-4pt,
       shorten >=-3pt,
       shift right=5pt,
       bend right=15,
       "{
         \eta_2 \circ \eta_1
       }"{swap, sloped, scale=.8, pos=.8}
     ]
     \ar[
       ur,
       bend left=10,
       crossing over,
       gray,
       "{ F_2 }"{description, name=two}
     ]
     \ar[
       from=one,
       to=two,
       shorten=-2pt,
       bend left=10,
       Rightarrow,
       "{
         \hspace{-3pt}\eta_1
       }"{pos=0}
     ]
     \ar[
       from=two,
       to=three,
       bend left=2,
       Rightarrow,
       "{
         \eta_2
       }"{pos=.1}
     ]
    \end{tikzcd}
  \hspace{-3pt} \right\}
 .
  \end{equation}
\end{definition}

\subsubsection{Topological Stacks}
\label{TopologicalStacks}

If we think --- as we may and should --- of topological groupoids as topologized ``sets with gauge transformations between their elements'', then some of them ought to be ``the same up to gauge fixing'' and yet no invertibe continuous functors exist between them (\cref{IllustratingEquivalenceOfTopologicalGroupoids} illustrates a simple example, a special case of \cref{CechGroupoids} below).

\begin{figure}[htb]
\caption{  \label{IllustratingEquivalenceOfTopologicalGroupoids}
  Indicated on the left is a topological groupoid whose space of objects is the disjoint union of two intervals, but whose morphisms uniquely connect --- and thereby uniquely identify --- a subinterval of points in either component (a \emph{{\v C}ech groupoid}, cf. \cref{CechGroupoids}). Indicated on the right is the topological groupoid whose space of objects is the result of gluing these two interval along this subinterval, and which has no non-identity morphisms. The evident topological functor from the left to the right is an \emph{equivalence} of topological groupoids 
  (\cref{EquivalenceOfTopGroupoids})
  but no continuous functor can serve as its inverse.
}
\centering

\adjustbox{
  rndfbox=5pt
}{
\begin{tikzcd}
\Bigg\{
\begin{tikzpicture}
  \draw[
    gray,
    line width=.9
  ]
    (-1.5,-.4) to 
    (.5,-.4);
  \draw[
    gray,
    line width=.9
  ]
    (+1.5,+.4) to 
    (-.5,+.4);

\foreach \n in {-1,...,+1}{
  \draw[
    gray,
    Stealth-Stealth 
  ]
    (\n*.4,+.4) --
    (\n*.4,-.4);
}
\node[scale=.8] at (-.2,-.07) {$\cdots$};
\node[scale=.8] at (+.205,-.07) {$\cdots$};

\end{tikzpicture}
\Bigg\}
\ar[
  r,
  "{ \sim }"
]
&
\Big\{
\begin{tikzpicture}[
  baseline=(current bounding box.center)]
]
  \draw[
    gray,
    line width=.9
  ]
    (-1.5,0) to 
    (.5,0);
\end{tikzpicture}
\Big\}
\end{tikzcd}
}
\end{figure}
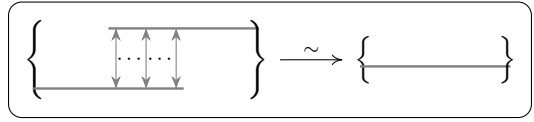

This means that topological functors by themselves are too rigid as a notion of ``maps'' between topological groupoids, and that the actual \emph{maps} (\cref{MoritaMap} below) must subsume inverses to those topological functors that ought to be equivalences of topological groupoids:
\begin{definition}
\label
  [definition]
  {EquivalenceOfTopGroupoids}
A topological functor (\cref{TopologicalFunctor}) is an \emph{equivalence}, denoted 
\begin{equation}
  \label{AnEquivalenceOfTopologicalGroupoids}
  \begin{tikzcd}[sep=8pt]
    F : 
    \mathcal{X} 
    \ar[
      rr, 
      "{\sim}"
    ] 
    &&
    \mathcal{Y}
    \,,
  \end{tikzcd}
\end{equation}
if for all $n \in \mathbb{N}$ the induced functor on stalks of $n$-dimensional plots \cref{StalkOfPlotsOfTopologicalGroupoid},
\begin{equation}
  \begin{tikzcd}[sep=small]
    F_\ast
    :
    \flat
    \mathrm{Map}\big(
      \mathbb{G}^n
      ,
      \mathcal{X}
    \big)
    \ar[rr] 
    && 
    \flat
    \mathrm{Map}\big(
      \mathbb{G}^n
      ,
      \mathcal{X}
    \big)
    \,,
  \end{tikzcd}
\end{equation}
is:
\begin{enumerate}
\item surjective on isomorphism classes \cref{IsomorphismClassesOfTopGroupoid} of objects (``essentially surjective''):
\begin{equation}
  \label{EssentiallySurjectiveOnStalks}
  \underset{n \in \mathbb{N}}{\forall}
  \;\;:\;\;
    \begin{tikzcd}[
      column sep=25pt
    ]
      \big[
        \flat
        \mathrm{Map}\big(
          \mathbb{G}^n
          ,
          \mathcal{X}
         \big)
      \big]_0
      \ar[
        rr, 
        ->>, 
        "{ 
          F_\ast
         }"
      ]
      &&
      \big[
        \flat
        \mathrm{Map}\big(
          \mathbb{G}^n
          ,
          \mathcal{Y}
         \big)
      \big]_0
      \mathrlap{\,,}
    \end{tikzcd}
\end{equation}
\item 
bijective on morphisms between pairs of objects (``fully faithful''):
\begin{equation}
  \label{FullyFaithfulOnStalks}
  \underset{
    n \in \mathbb{N}
  }{\forall}
  \;
  \underset
    {
      \substack{
        x,y \in 
        \\
        \mathrm{Obj}(\mathcal{X})
      }
    }
    {\forall}
    :
    \begin{tikzcd}
    \mathrm{Mor}\Big(
      \flat
      \mathrm{Map}(\mathbb{G}^n,\mathcal{X})
    \Big)\rule[-4pt]{0pt}{14pt}_{\mathbf{x}}^{\mathbf{y}}
    \ar[
      rr, 
      "{  
        F_\ast
      }",
      "{\sim}"{swap}
    ]
    &&
    \mathrm{Mor}\Big(
      \flat
      \mathrm{Map}(\mathbb{G}^n,\mathcal{Y})
    \Big)\rule[-4pt]{0pt}{14pt}
      _{F_\ast(\mathbf{x})}
      ^{F_\ast(\mathbf{y})}
    \mathrlap{\,,}
  \end{tikzcd}
\end{equation}
where
\begin{equation}
  \mathrm{Mor}(-)_{x}^y
  :=
  \ast \tensor[_x]{\times}{_s} 
    \mathrm{Mor}(-)
  \tensor[_t]{\times}{_y} \ast    
  \mathrlap{\,.}
\end{equation}
\end{enumerate}
denotes the subspace of morphisms between a given pair of objects.
\end{definition}

\begin{definition}
\label[definition]{MoritaMap}
A (Morita) \emph{map} between topological groupoids (as opposed to a plain continuous functor, \cref{TopologicalFunctor}) is a span of continuous functors (\cref{TopologicalFunctor}), with the left one an equivalence (\cref{EquivalenceOfTopGroupoids}):
\begin{equation}
  \label{AMoritaMap}
  \begin{tikzcd}[sep=18pt]
    \mathcal{X}
    \ar[
      r,
      <-,
      "{ \sim }"
    ]
    &
    \widehat{\mathcal{X}}
    \ar[
      r,
      "{ F }"
    ]
    &
    \mathcal{Y}
    \mathrlap{\,.}
  \end{tikzcd}
\end{equation}
In particular, a (Morita) \emph{equivalence} between topological groupoids is a span of continuous functors which are equivalences
\begin{equation}
  \label{AMoritaEquivalence}
  \begin{tikzcd}[sep=18pt]
    \mathcal{X}
    \ar[
      r,
      <-,
      "{ \sim }"
    ]
    &
    \widehat{\mathcal{X}}
    \ar[
      r,
      "{ \sim }"
    ]
    &
    \mathcal{Y}
    \mathrlap{\,.}
  \end{tikzcd}
\end{equation}
When considered up to (Morita) equivalence, topological groupoids are also referred to as \emph{topological stacks} (which is terrible terminology, but completely standard),cf. \S\ref{SmoothInfinityGroupoids}.
\end{definition}
\begin{remark}
  Equivalent topological groupoids (\cref{MoritaMap}) have homeomorphic spaces $[-]_0$ of isomorphism classes \cref{IsomorphismClassesOfTopGroupoid}:
  \begin{equation}
    \label{EquivalentGroupoidsHaveHomeomorphismIsomorphismClasses}
    \begin{tikzcd}[sep=18pt]
      \mathcal{X} 
      \ar[r, "{ \sim }"]
      &
      \mathcal{Y}
    \end{tikzcd}
    \;\;\;\;
    \Rightarrow
    \;\;\;\;
    \begin{tikzcd}[sep=18pt]
      {[\mathcal{X}]_0}
      \ar[r, "{ \sim }"]
      &
      {[\mathcal{Y}]_0}
    \end{tikzcd}
  \end{equation}
\end{remark}
\begin{example}
  \label[example]{HomotopyEquivalenceOfTopGrpds}
  If a continuous functor $\begin{tikzcd}[sep=small]F : \mathcal{X} \ar[r] & \mathcal{Y}\end{tikzcd}$ is a \emph{homotopy equivalence} (better terminology would be ``transformation equivalence'', cf. \cref{RelationBetweenHomotopiesAndTransformations,HomotopyEquivalenceOfSpaces}) in that there exists a reverse continuous functor $\begin{tikzcd}[sep=small]\overline{F} : \mathcal{Y} \ar[r] & \mathcal{X}\end{tikzcd}$ and continuous transformations (\cref{NaturalTransformations}) of this form:
  \begin{equation}
    \begin{tikzcd}[
      row sep=15pt, 
      column sep=25pt
    ]
      & 
      \mathcal{Y}
      \ar[
        dr,
        "{ \overline{F} }"{description}
      ]
      \ar[
        d,
        shorten=5pt,
        Rightarrow,
      ]
      \ar[
        rr,
        "{
          \mathrm{id}
        }"{description}
      ]
      &
      {}
      \ar[
        d,
        Rightarrow,
        shorten=2pt
      ]
      &
      \mathcal{Y}
      \\
      \mathcal{X}
      \ar[
        ur,
        "{ F }"{description}
      ]
      \ar[
        rr,
        "{ \mathrm{id} }"{description}
      ]
      &
      {}
      &
      \mathcal{X}
      \ar[
        ur,
        "{
          F
        }"{description}
      ]
    \end{tikzcd}
  \end{equation}
  then it is an equivalence \textup{in the sense of Def. \ref{EquivalenceOfTopGroupoids}}.
\end{example}
\begin{proof}
  After passage to $n$-dimensional stalks of plots \cref{StalkOfPlotsOfTopologicalGroupoid}, this reduces to the basic statement of category theory that essentially surjective and fully faithful functors are equivalences (cf. \cite[p. 93]{MacLane1998}).
\end{proof}

\begin{example}
For $\begin{tikzcd}[sep=small]H \ar[r, hook, "{\iota}"] & G\end{tikzcd}$ a topological subgroup \cref{TopologicalGroup}, consider the coset space \cref{CosetSpace}
\begin{equation}
  G/H 
    = 
  \big\{ 
    g \cdot H 
    \,\big\vert\, 
    g \in G 
  \big\}
  \,,
\end{equation}
with its canonical left $G$-action \cref{ActionOnCosetSpace}.
Then the corresponding homotopy quotient \cref{TheActionGroupoid} is equivalent (\cref{EquivalenceOfTopGroupoids}) to the delooping of $H$ (\cref{DeloopingGroupoid}):
\begin{equation}
  \label{HoQuotientOfGModHByG}
  \renewcommand{\arraystretch}{1.4}  
  \setlength{\arraycolsep}{-2pt}
  \begin{array}{ccc}
  \mathbf{B}H
  &
  \begin{tikzcd}
    \ar[
      r,
      "{ \sim }"
    ]
    &
    {}
  \end{tikzcd}
    &
    G \backsslash G/H
    \\[2pt]
    \left(
    \begin{tikzcd}
     \bullet
     \ar[
       d,
       "{ h }"{swap}
     ]
     \\
     \bullet
   \end{tikzcd}
   \right)
   &\longmapsto&
   \left(
    \begin{tikzcd}
     \mathrm{e}\cdot H
     \ar[
       d,
       shorten >=-1pt,
       "{ h }"{swap}
     ]
     \\
     \mathrm{e}\cdot H
   \end{tikzcd}   
   \right)
 \!.
  \end{array}
\end{equation}
\end{example}

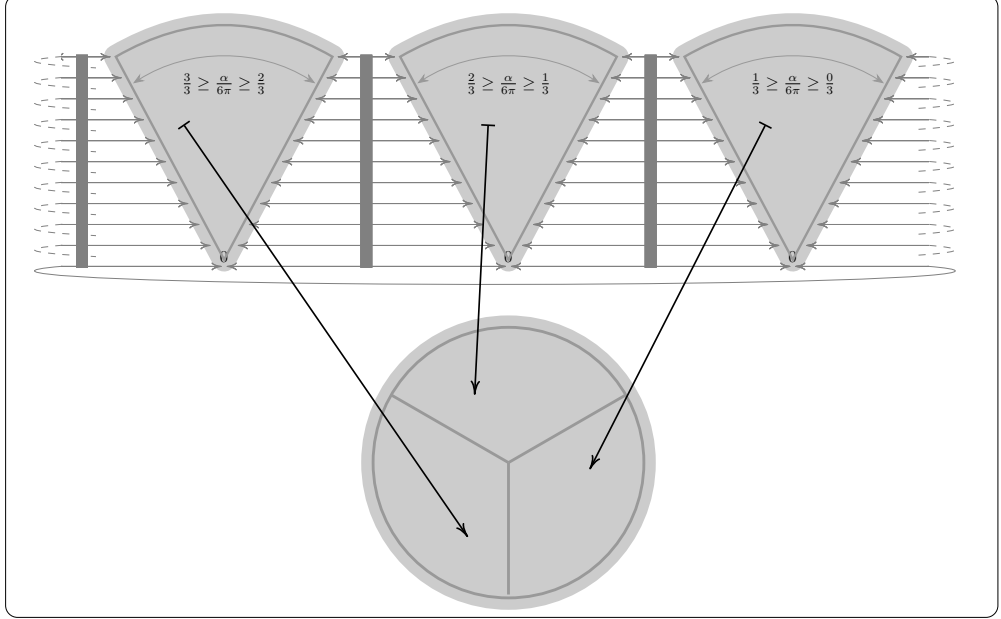
\begin{figure}[htb]
\caption{
 \label{ASimpleCechGroupoid}
 {\bf Top row:}
 The {\v C}ech groupoid (\cref{CechGroupoids}) of a (good, \cref{GoodOpenCover}) open cover of the open disk $\mathbb{D}^2_{1 + \epsilon}$. The space of objects is shown in light gray, the space of non-identity morphisms is shown in dark gray ($\alpha$ denotes an angular coordinate function).
 \\
 {\bf Bottom row:}
 The original disk, as a groupoid with only identity morphisms (\cref{SpaceAsGroupoid}), receiving the canonical projection functor \cref{CechGroupoidEquivalentToOriginalSpace} from the {\v C}ech groupoid.
}

\centering

\adjustbox{
  rndfbox=5pt,
  scale=.9
}{
\hspace{-.8cm}
\begin{tikzpicture}

\begin{scope}
\clip
 (-.8,.8) rectangle 
 (12,-3.3);

\foreach \n [count=\np] in {0,1,2,3} {

\begin{scope}[
 shift={(10-\n*4.2,0)}
]
\filldraw[
    gray!40,
    line width=10pt,
    line join=round
  ]
    (-1.6,0) to[bend left=30]
    (+1.6,0) --
    (0,-3) -- cycle;
  \node[
    scale=.7
  ] at (0,-2.95) {$0$};
  \draw[
    gray!80,
    line width=1pt
  ]
    (-1.6,0) to[bend left=30]
    (+1.6,0) --
    (0,-3) -- cycle;
  \draw[
    shift={(0,-2.5)},
    gray!80,
    line width=.4,
    Stealth-Stealth
  ]
    (90-32.6:2.5) arc 
    (90-32.6:90+32.6:2.5);
  \node[scale=.7]
    at (0,-.4) {
      $
        \tfrac{\np}{3}
        \!\geq\!
        \frac{
          \alpha
        }{6\pi}
        \!\geq\!
        \tfrac{\n}{3}
      $
    };

\end{scope}

\foreach \n [count=\np] in {-1,0,1,2,3} {
\begin{scope}[
 shift={(10-\n*4.2,0)}
]
\begin{scope}[shift={(-2.1,0)}]
 \filldraw[
     black!50,
     line width=5pt,
     line join=round
   ]
     (0,.03) --
     (0,-3.12)
      -- cycle;
 
 \foreach \k in {0,...,10} {
   \draw[
     {<[length=4pt,width=3pt]}-{>[length=4pt,width=3pt]},
     black!50,
     line width=.2
   ]
     ({-\k*.164-.42},{-\k*.31}) --
     ({+\k*.164+.42},{-\k*.31});
 }
\end{scope}
\end{scope}

}}

\end{scope}

\draw[
   shift={(5.6,-3.166)},
   black!50,
   line width=.4
]
  (90+70:6.8 and .2) arc
  (90+70:360+90-70:6.8 and .2);

\begin{scope}
\clip
  (-1.3,.2)  rectangle
  (-.3,-3.4);

\foreach \n in {0,...,9} {
\draw[
   shift={(5.6,-2.855+\n*.31)},
   black!50,
   line width=.2,
   dashed
]
  (90+70:6.8 and .2) arc
  (90+70:360+90-70:6.8 and .2);
}
\end{scope}

\begin{scope}
\clip[
  shift={(13.1,0)}
]
  (-1.3,.2)  rectangle
  (-.3,-3.4);

\foreach \n in {0,...,9} {
\draw[
   shift={(5.6,-2.855+\n*.31)},
   black!50,
   line width=.2,
   dashed
]
  (90+70:6.8 and .2) arc
  (90+70:360+90-70:6.8 and .2);
}
\end{scope}

\begin{scope}[
  shift={(5.8,-6)}
]

\filldraw[
    gray!40,
    line width=10pt,
    line join=round
  ]
 (0,0) circle (2);

\draw[
    gray!80,
    line width=1.2pt,
]
  (0,0) circle (2);

\draw[
    gray!80,
    line width=1.2pt,
]
 (0,0) -- (30:2);
\draw[
    gray!80,
    line width=1.2pt,
]
 (0,0) -- (30+120:2);
\draw[
    gray!80,
    line width=1.2pt,
]
 (0,0) -- (30+240:1.95);
  
\end{scope}

\draw[
  |-{>[length=6pt,width=4pt]},
  line width=.7pt
]
  (5.5,-1) -- 
  (5.3,-5);

\draw[
  |-{>[length=6pt,width=4pt]},
  line width=.7pt
]
  (1,-1) -- 
  (5.2,-7.1);
\draw[
  |-{>[length=6pt,width=4pt]},
  line width=.7pt
]
  (9.6,-1) -- 
  (7,-6.1);

\end{tikzpicture}
}

\end{figure}

A key class of examples of equivalences of topological groupoids (\cref{EquivalenceOfTopGroupoids}) which are not homotopy equivalences (\cref{HomotopyEquivalenceOfTopGrpds}) are projections out of {\v C}ech groupoids (cf. \cref{IllustratingEquivalenceOfTopologicalGroupoids} and \cref{ProjectionFromCechGroupoidIsEquivalence}):
\begin{example}[{\v C}ech groupoids]
\label[example]{CechGroupoids}
For a topological manifold $X$ equipped with an \emph{open cover} 
\begin{equation}
  \label{AnOpenCover}
  \mathcal{U} 
    := 
  \big\{
    \begin{tikzcd}[sep=20] 
      U_i 
     \ar[
       r, 
       hook, 
       shorten=-1pt, 
       "{ \iota_i }",
       "{ \mathrm{open} }"{
         swap, yshift=-1pt, scale=.9
       }
     ] 
     & 
     X
  \end{tikzcd}
  \big\}_{i \in I}
  ,\,\;\;\;
  \begin{tikzcd}[
   column sep=35
  ]
    \underset{i}{\bigsqcup}
    \,
    U_i
    \ar[
      r,
      ->>,
      "{
        (\iota_i)_{i \in i}
      }"
    ]
    &
    X
    \mathrlap{\,,}
  \end{tikzcd}
\end{equation}
then the corresponding \emph{{\v C}ech groupoid} (cf. \cref{ASimpleCechGroupoid}) is the topological groupoid \cref{StructureOfATopologicalGroupoid} given by
\begin{equation}
  \label{StructureOfChechGroupoid}
  X_{\mathcal{U}}
  :=
  \left(
  \begin{tikzcd}[column sep=38pt]
    \underset{i, j, k}{\bigsqcup} 
      U_{i j k}
    \ar[
      rr,
      "{
        (x,i,j,k) \,\mapsto\, (x,i,k)
      }"
    ]
    &&    
    \quad
    \underset{i,j}{\bigsqcup} 
    \, U_{i j}
    \ar[
      out=180-60,
      in=60,
      looseness=4,
      "{\scalebox{0.7}{$
        (x,i,j) 
        \mapsto
        (x,j,i)$}
      }"
    ]
    \quad 
    \ar[
      rr,
      shift left=8pt,
      "{ 
        (x,i,j) \,\mapsto\, (x,i) 
      }"
    ]
    \ar[
      rr,
      <-,
      "{
        (x,i,i) \,\mapsfrom\, (x,i)
      }"{description}
    ]
    \ar[
      rr,
      shift right=8pt,
      "{ 
         (x,i,j) \,\mapsto\, (x,j) 
      }"{swap}
    ]
    &&
    \;
    \underset{i}{\bigsqcup} 
    \,
    U_i
  \end{tikzcd}
  \right)
  ,
\end{equation}
where we abbreviate $U_{i j} := U_i \cap U_j$ and $U_{i j k} := U_i \cap U_j \cap U_k$, and where $(x,i)$ denotes a point $x \in X$ but regarded as the corresponding point of $U_i$, etc. Hence in the notation \cref{SimplicialNotationForGroupoids}, a {\v C}ech groupoid looks like this:
\begin{equation}
  X_{\mathcal{U}}
  =
  \left\{
  \begin{tikzcd}[
    row sep=8pt, 
    column sep=25pt
  ]
    &
    (x,j)
    \ar[
      dr,
      "{
        (x,j,k)
      }"{sloped}
    ]
    \\
    (x,i)
    \ar[
      ur,
      "{ 
        (x,i,j) 
      }"{sloped}
    ]
    \ar[
      rr,
      "{
        (x,i,k)
      }"
    ]
    &&
    (x,k)
  \end{tikzcd}
  \right\}.
\end{equation}
This comes with a continuous functor (\cref{TopologicalFunctor}) to the original space (regarded as a groupoid per \cref{SpaceAsGroupoid}); which is (see \cref{ProjectionFromCechGroupoidIsEquivalence}) an equivalence
(\cref{EquivalenceOfTopGroupoids}):
\begin{equation}
  \label{CechGroupoidEquivalentToOriginalSpace}
  \begin{tikzcd}[
    row sep=0pt,
    column sep=14pt
  ]
    X_{\mathcal{U}}
    \ar[
      rr,
      "{ \sim }"
    ]
    &&
    X
    \\[1pt]
    (x, i)
    \ar[
      dd,
      "{
        (x,i,j)
      }"
    ]
    &&
    x
    \ar[
      dd,
      equals
    ]
    \\
    & \longmapsto
    \\
    (x,j)
    &&
    x
    \mathrlap{\,.}
  \end{tikzcd}
\end{equation}
  To appreciate the relevance of this fact, consider a topological group $\Gamma$ and observe that there is a unique continuous functor $\begin{tikzcd}[sep=small]X \ar[r] & \mathbf{B}\Gamma\end{tikzcd}$ to its  delooping groupoid (\cref{DeloopingGroupoid}), and that unique functor is trivial (constant). But, due to the equivalence \cref{CechGroupoidEquivalentToOriginalSpace}, \emph{maps} of this form, in the sense of \cref{MoritaMap}, subsume the continuous functors out of the {\v C}ech groupoid \cref{StructureOfChechGroupoid} of any open cover, and these are identified with the cocycles, relative to $\mathcal{U}$, of nonabelian \emph{{\v C}ech cohomology} $H^1(X;\Gamma)$ (for which cf. \parencites[\S 7]{Wedhorn2016}[\S 4]{Alvarez1985}):
  \begin{equation}
    \label{CechCocycleAsFunctorOnCechGroupoid}
    \begin{array}{c}
    \begin{tikzcd}[
      row sep=2pt,
      column sep=27pt
    ]
      X
      \ar[
        rr,
        <-,
        "{ \sim }"
      ]
      &&
      X_{\mathcal{U}}
      \ar[
        rr,
        "{ \gamma }"
      ]
      &&
      \mathbf{B}\Gamma
      \\
      x
      \ar[
       dd,
       equals
      ]
      && 
      (x,i)
      \ar[
        dddd,
        leftvertright,
        "{
          (x,i,k)
        }"{description, sloped, rotate=180}
      ]
      \ar[
        dd,
        "{
          (x,i,j)
        }"
      ]
      &&
      \bullet
      \ar[
        dd,
        "{
          \gamma_{i j}(x)
        }"{swap}
      ]
      \ar[
        dddd,
        rightvertleft,
        "{
          \gamma_{ik}(x)
        }"{description, sloped, rotate=180}
      ]
      \\
      & \mapsfrom && \mapsto &
      \\
      x 
      \ar[
        dd,
        equals
      ]
      &&
      (x,j)
      \ar[
        dd,
        "{
          (x,j,k)
        }"
      ]
      &&
      \bullet
      \ar[
        dd,
        "{
          \gamma_{j k}(x)
        }"{swap}
      ]
      \\
      &\mapsfrom& & \mapsto &
      \\
      x 
        && 
      (x,k)
        &&
      \bullet
    \end{tikzcd}
    \end{array}
  \end{equation}
  while their transformations (\cref{NaturalTransformations}) are identified with {\v C}ech coboundaries:
\begin{equation}
  \label{CechCoboundaryAsFunctorOnCechGroupoid}
  \begin{tikzcd}[
    row sep=1,
    column sep=10pt
  ]
    X_{\mathcal{U}}
    \ar[
      rrr,
      bend left=22,
      "{ g }"{description, name=s}
    ]
    \ar[
      rrr,
      bend right=22,
      "{ g' }"{description, name=t}
    ]
    &&&
    \mathbf{B}G
    \ar[
      from=s,
      to=t,
      Rightarrow,
      "{ h }"
    ]
    \\[9pt]
    (x,i)
    \ar[
      dd,
      "{
        (x,i,j)
      }"
    ]
    &\qquad \mapsto &
 \bullet 
    \ar[
      rr,
      "{
        h_i(x)
      }"{description}
    ]
    \ar[
      dd,
      "{
        \gamma_{ij}(x)
      }"
    ]
      && 
    \bullet
    \ar[
      dd,
      "{
        \gamma'_{ij}(x)
      }"
    ]
    \\
    & \phantom{-}
    \\
    (x,j)
    &\qquad \mapsto&
    \bullet 
    \ar[
      rr,
      "{ 
        h_j(x) 
      }"{swap}
    ]
    &&
    \bullet
    \mathrlap{\,.}
  \end{tikzcd}
\end{equation}
Below, we generalize this example further as \cref{EquivariantCechGroupoids}. But first we make explicit the argument we just used, since it is an instructive illustration of the general machinery at play here:
\end{example}
\begin{lemma}
\label[lemma]
{ProjectionFromCechGroupoidIsEquivalence}
 The above continuous functor \textup{\cref{CechGroupoidEquivalentToOriginalSpace}} is indeed an equivalence of topological groupoids \textup{(\cref{EquivalenceOfTopGroupoids})}. 
\end{lemma}
\begin{proof}
\label{ProofThatFunctorFromCechGroupoidIsToSpaceIsEquivalence}
First to see that the functor is essentially surjective \cref{EssentiallySurjectiveOnStalks} on $n$-dimensional stalks: For $\mathbf{x} : \mathbb{R}^n \to X$ a continuous map (a \emph{plot}), the open cover property of $\mathcal{U}$ \cref{AnOpenCover} implies that there exists $i \in I$ and an open neighborhood $U_x \subset X$ of $\mathbf{x}(0)$ such that $U_x \subset U_i$. Therefore the preimage $\mathbf{x}^{-1}(U_x)$ is an open neighborhood of $0$ in $\mathbb{R}^n$, and by the Euclidean topology on $\mathbb{R}^n$ this contains an open ball $D^n_\epsilon \subset \mathbf{x}^{-1}(U_x)$ around 0, such that $\mathbf{x}\vert_{ D^n_\epsilon}$ factors through $U_i$. Precomposed with any map 
$\mathbb{R}^n \xrightarrow{\sim} \mathbb{D}^n_\epsilon$ which is the identity on a smaller open ball $D^n_{\epsilon'} \subset D^n_\epsilon$, this is a plot $(\mathbf{x},i)$ of objects of the {\v C}ech groupoid whose germ maps to the germ of the given $\mathbf{x}$.

Then to see that the functor is fully faithful \cref{FullyFaithfulOnStalks} on $n$-dimensional stalks: Given a pair of plots of the {\v C}ech groupoid, $\mathbf{x}_i : \mathbb{R}^n \to U_i$ and 
$\mathbf{x}_j : \mathbb{R}^n \to U_j$, there are two cases: regarded as maps to $X$ their germs either coincide --- in which case there is a unique morphism between them in $X$ regarded as a groupoid (\cref{SpaceAsGroupoid}), namely the identity --- or their germs do not coincide, in which case the set of morphisms between them in $X$ is empty. We need to see that the same two cases hold for morphisms between the germs of these plots regarded in the {\v C}ech groupoid. In the second case this is immediate from the definition, while in the first case it follows from the open cover property that the germs of $\mathbf{x}_i$ and $\mathbf{x}_j$ both factor through $U_i \cap U_j$, which constitutes the required unique morphism between them in the {\v C}ech groupoid.
\end{proof}

\subsubsection{Mapping Stacks and Nonabelian Cohomology}
\label{MappingStacksAndNonabelianCohomology}

\begin{definition}[Good open cover]
  \label[definition]
    {GoodOpenCover}
  An open cover  
  $
    \big\{
    \begin{tikzcd}[sep=20] 
      U_i 
     \ar[
       r, 
       hook, 
       shorten=-1pt, 
       "{ \iota_i }",
       "{ \mathrm{open} }"{
         swap, yshift=-1pt, scale=.9
       }
     ] 
     & 
     X
  \end{tikzcd}
  \big\}_{i \in I}
  $
  \cref{AnOpenCover}
  of an $n$-dimensional manifold is \emph{good} if all finite intersections of its patches are either empty or homeomorphic to $\mathbb{R}^n$:
  \begin{equation}
    \label{HomeomorphismsOfIntersectionsOfGoodOpenCovers}
    \forall_{k \in \mathbb{N}_{\geq 1}}
    \;
    \forall_{
      i_i, \cdots, i_k \in I
    }
    \;:\;
    U_{i_1}
    \cap 
      \cdots 
    \cap
    U_{i_k}
    \simeq
    \left\{
    \begin{array}{l}
      \varnothing \; \mbox{or}
      \\
      \mathbb{R}^n
      \mathrlap{\,.}
    \end{array}
    \right.
  \end{equation}
  We say that this is \emph{differentially good} if these homeomorphisms \cref{HomeomorphismsOfIntersectionsOfGoodOpenCovers} exist even as \emph{diffeomorphisms}.
\end{definition}
For general topological manifolds the existence of good open covers is not known, but we have:
\begin{lemma}[{cf. \parencites[Thm. 5.1]{Bott1982}[Prop. A.1]{FSSt12-DiffClasses}}]
  \label[lemma]{ExistenceOfGoodOpenCovers}
  Every \emph{smooth} manifold admits a differentiably good open cover \textup{(\cref{GoodOpenCover})}.
\end{lemma}

It turns out (in \cref{SmoothInfinityGroupoids}) that the {\v C}ech groupoids (\cref{CechGroupoids}) of \emph{good} open covers (\cref{GoodOpenCover}) are ``fine enough'' (technical term: \emph{cofibrant}) to represent all maps out of topological manifolds into delooping groupoids (\cref{DeloopingGroupoid}). Therefore:
\begin{definition}
  \label[definition]
  {MappingStackFromManifoldToDeloopingGroupoid}
  Let $X$ be a topological manifold which admits a good open cover $\mathcal{U}$ (\cref{GoodOpenCover}, such as any smooth manifold does, by \cref{ExistenceOfGoodOpenCovers}), and let $\Gamma$ be a topological group, then the \emph{mapping stack} 
  \footnote{
    In the literature, the mapping stack \cref{MappingStackOutOfGoodOpenCoverIntoBG} may also be called the \emph{derived internal hom}, or similar. We are tacitly using here that $X_{\mathcal{U}}$ is cofibrant (\cite[Ex. 4.3.42]{SS25-Bun}) and that $\mathbf{B}\Gamma$ is fibrant (\cite[Lem 4.3.30]{SS25-Bun}) in the local projective model structure of simplicial presheaves over the site of Cartesian spaces. This is is discussed in \cref{SmoothInfinityGroupoids} below.
  }
  from $X$ to the delooping groupoid $\mathbf{B}\Gamma$ \cref{TheDeloopingGroupoid} is, up to equivalence, the functor groupoid $\mathrm{Func}(-,-)$ (\cref{FunctorGroupoid}) into $\mathbf{B}\Gamma$ out of the {\v C}ech groupoid $X_{\mathcal{U}}$ (\cref{CechGroupoids}) of the good open cover $\mathcal{U}$:
  \begin{equation}
    \label{MappingStackOutOfGoodOpenCoverIntoBG}
    \begin{aligned}
    \mathrm{Map}\big(
      X,\,
      \mathbf{B}\Gamma
    \big)
    & :=
    \mathrm{Func}\big(
      X_{\mathcal{U}}
      ,\,
      \mathbf{B}\Gamma
    \big)
    \\[-2pt]
    & =
    \big\{
      \begin{tikzcd}[sep=15pt]
        X 
        \ar[
          r,
          dashed
        ]
        &
        \mathbf{B}\Gamma
      \end{tikzcd}
    \big\}
    \mathrlap{\,,}
    \end{aligned}
  \end{equation}
  where in the second line we are showing a more suggestive notation which highlights again that this is to be thought of as maps  out of $X$ itself, in the sense of \cref{MoritaMap}.
  
  The underlying topologically discrete groupoid \cref{UnderlyingTopologicallyDiscreteGroupoid} we denote by a boldface $\mathbf{H}(-,-)$:
  \begin{equation}
    \label{MappingGroupoidOutOfGoodOpenCoverIntoBG}
    \mathbf{H}(X,\mathbf{B}\Gamma)
    :=
    \flat \, \mathrm{Map}(X,\mathbf{B}\Gamma)
    \mathrlap{\,,}
  \end{equation}
  since this is (cf. \cite{Jardine2009}) the \emph{cocycle groupoid} (objects are {\v C}ech cocycles, morphisms are coboundaries)  whose isomorphism classes \cref{IsomorphismClassesOfTopGroupoid} are the \emph{nonabelian cohomology sets} $H(-,-)$, as discussed in \cref{OrdinaryNonabelianCechCohomology}.
\end{definition}

\begin{example}[Ordinary nonabelian cohomology]
  \label[example]{OrdinaryNonabelianCechCohomology}
  We have seen in \cref{CechGroupoids} \cref{CechCocycleAsFunctorOnCechGroupoid,CechCoboundaryAsFunctorOnCechGroupoid} that \eqref{MappingStackOutOfGoodOpenCoverIntoBG} is the groupoid whose objects are 1-cocycles and whose morphisms are coboundaries in nonabelian {\v C}ech cohomology $H^1(X;\Gamma)$, relative to the cover $\mathcal{U}$ (for which cf. \parencites[\S 7]{Wedhorn2016}[\S 4]{Alvarez1985}). In general, for these sets of cocycles modulo coboundaries to give the full cohomology set $H^1(X;\Gamma)$ one has to take their colimit over cover refinements. But in the presence of a good open cover it is sufficient to evaluate there!
  (For the case of abelian groups $\Gamma$ this is \cite[\S 4.18]{Bredon1997}, in general it follows immediately from the cofibrancy of the {\v C}ech groupoid.) Therefore we have:
  \begin{equation}
    \label{CechCohomologyAsIsoclassesOfMaps}
    H^1(X; \Gamma)
    \simeq
    \big[
      \mathbf{H}(X,\mathbf{B}\Gamma)
    \big]_0
    \defneq
    \big[
      \flat 
      \mathrm{Map}(X, \mathbf{B}\Gamma)
    \big]_0
    \mathrlap{\,.}
  \end{equation}
  In words: The isomorphism classes of the mapping stack from a space $X$ to $\mathbf{B}\Gamma$ is identified with the nonabelian 1-cohomology of $X$ with coefficients in $\Gamma$, and in components this identification incarnates as {\v C}ech cohomology.

  These classes are equivalently the \emph{isomorphism classes} of \emph{$\Gamma$-principal bundles} over $X$. In fact, before passing to isomorphism classes, 
  \begin{equation}
    \Gamma \mathrm{PrnBdl}(X)
    \simeq
    \mathbf{H}(X,\mathbf{B}\Gamma)
  \end{equation}
  is equivalent (cf. \cref{PrincipalBundleAsPullbackAlongCechCocycle} and \cite[\S 5.1.2]{SS25-Bun}) to the groupoid of $\Gamma$-principal bundles (\cref{GroupoidOfPrincipalBundles}).

  This means that the delooping groupoid $\mathbf{B}\Gamma$ serves not just as a coefficient object for nonabelian cohomology, in \cref{CechCohomologyAsIsoclassesOfMaps}, but as the \emph{moduli stack} of $\Gamma$-principal bundles (cf. \parencites{NSS2015a}{FSS15-Stacky}{SS25-Bun}).

  On the other hand, there is also the  \emph{classifying space} of $\Gamma$ (\cite{Milgram1967}, cf. \cite[Prop. 3.3.4]{SS25-Bun}), the ``bar construction'' or ``topological realization'' $\vert-\vert$ of $\mathbf{B}\Gamma$ (cf. \cite[Ntn. 2.2.28]{SS25-Bun}), traditionally denoted:
  \begin{equation}
    \label{ClassifyingSpaceBG}
    B\Gamma 
      := 
    \vert \mathbf{B}\Gamma\vert
    \mathrlap{\,.}
  \end{equation}
Over paracompact topological spaces $X$ (such as manifolds and CW complexes, \cref{CellAttachment}) we have that homotopy classes of maps into $B G$ coincide with transformation classes of maps into $\mathbf{B}G$ and hence give the same nonabelian cohomology (cf. \parencites[Thm. 3.5.1]{RudolphSchmidt2017}[Thm. 5.1.13]{SS25-Bun}):
\begin{equation}
  \label{NonabelianCohomologyRepresentedByBGAndbfBG}
  \begin{tikzcd}
    \pi_0 
    \mathrm{Map}\big(
      X
      ,\,
      B G
    \big)
    \underset{
      \scalebox{.7}{\cref{IsoClassesOfFundamentalGroupoid}}
    }{=}
    \big[
      \shape \mathrm{Map}(X, B G)
    \big]_0
    \simeq
    \big[
      \flat \mathrm{Map}(X,\mathbf{B}G)
    \big]_0
    \underset{
      \scalebox{.7}{\cref{CechCohomologyAsIsoclassesOfMaps}}
    }{=}
    H^1(X,G)
    \mathrlap{\,.}
  \end{tikzcd}
\end{equation}
To note here the two kind of classifying objects that appear:
\begin{itemize}
  \item $\mathbf{B}\Gamma$ is a topological groupoid with trivial space of objects,
  \item $B \Gamma$ is, as a topological space, a topological groupoid with only identity morphisms.
\end{itemize}
Hence, general topological groupoids, regarded as classifying objects for cohomology, unify these two extremes. We will see this in action from \cref{SlicingAndTwisting} on.
\end{example}

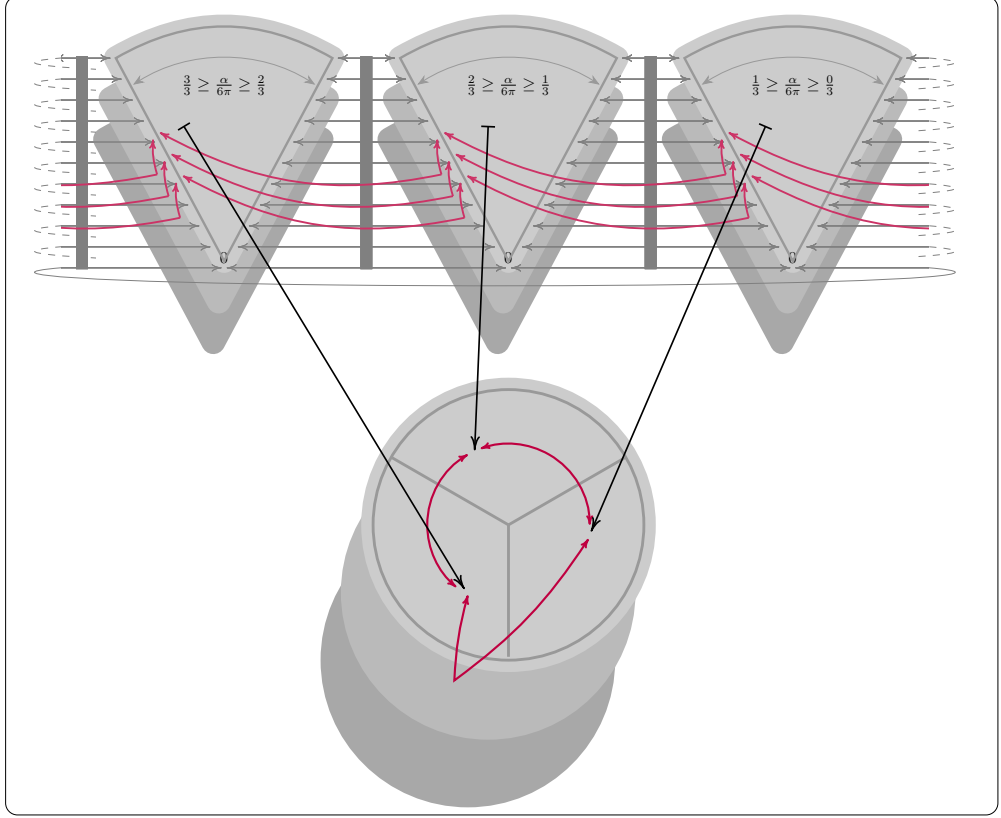
\begin{figure}[htb]
\caption{
 \label{ASimpleEquivariantCechGroupoid} 
 In equivariant generalization of \cref{ASimpleCechGroupoid}:
 \\
 {\bf Top row:}
 The equivariant {\v C}ech groupoid (\cref{EquivariantCechGroupoids}) of a (good, \cref{GoodOpenCover}) equivariant open cover of the open disk $\mathbb{D}^2_{1 + \epsilon}$ equipped with rigid $\mathbb{Z}_3$-rotation action. The space of objects is shown in light gray, the space of non-identity morphisms is shown in darker shades of gray.
 \\
 {\bf Bottom row:}
 The $\mathbb{Z}_3$-action groupoid \cref{TheActionGroupoid} of the disk, receiving the canonical projection functor \cref{EquivariantCechGroupoidEquivalentToOriginalGSpace} from the equivariant {\v C}ech groupoid.
}

\centering

\adjustbox{
  rndfbox=5pt,
  scale=.9
}{
\hspace{-.8cm}
\begin{tikzpicture}

\begin{scope}
\clip
 (-.8,.8) rectangle 
 (12,-4.5);

\foreach \n [count=\np] in {0,1,2,3} {

\begin{scope}[
 shift={(10-\n*4.2,0)}
]

\filldraw[
    gray!68,
    line width=10pt,
    line join=round,
    shift={(-.16,-1.2)}
  ]
    (-1.6,0) to[bend left=30]
    (+1.6,0) --
    (0,-3) -- cycle;
\filldraw[
    gray!54,
    line width=10pt,
    line join=round,
    shift={(-.08,-.6)}
  ]
    (-1.6,0) to[bend left=30]
    (+1.6,0) --
    (0,-3) -- cycle;

\filldraw[
    gray!40,
    line width=10pt,
    line join=round
  ]
    (-1.6,0) to[bend left=30]
    (+1.6,0) --
    (0,-3) -- cycle;
  \node[
    scale=.7
  ] at (0,-2.95) {$0$};
  \draw[
    gray!80,
    line width=1pt
  ]
    (-1.6,0) to[bend left=30]
    (+1.6,0) --
    (0,-3) -- cycle;
  \draw[
    shift={(0,-2.5)},
    gray!80,
    line width=.4,
    Stealth-Stealth
  ]
    (90-32.6:2.5) arc 
    (90-32.6:90+32.6:2.5);
  \node[scale=.7]
    at (0,-.4) {
      $
        \tfrac{\np}{3}
        \!\geq\!
        \frac{
          \alpha
        }{6\pi}
        \!\geq\!
        \tfrac{\n}{3}
      $
    };

\end{scope}

\foreach \n [count=\np] in {-1,0,1,2,3} {
\begin{scope}[
 shift={(10-\n*4.2,0)}
]
\begin{scope}[shift={(-2.1,0)}]
 \filldraw[
     black!50,
     line width=5pt,
     line join=round
   ]
     (0,.03) --
     (0,-3.12)
      -- cycle;
 
 \foreach \k in {0,...,10} {
   \draw[
     {<[length=4pt,width=3pt]}-{>[length=4pt,width=3pt]},
     black!50,
     line width=.2
   ]
     ({-\k*.164-.42},{-\k*.31}) --
     ({+\k*.164+.42},{-\k*.31});
 }
\end{scope}
\end{scope}

}}

\foreach \n [count=\np] in {-1,0,1,2} {
\begin{scope}[
 shift={(10-\n*4.2,0)}
]
\foreach \k in {3,...,5} {
   \draw[
     shift={(-\k*.17,\k*.32)},
     -{>[length=4pt,width=3pt]},
     purple!80,
     line width=.8
   ]
  (-.14,-3.32) to[bend left=20] 
  (-4.3,-2.7);
   \draw[
     shift={(-\k*.17,\k*.32)},
     -{>[length=4pt,width=3pt]},
     purple!80,
     line width=.8
   ]
  (-.14,-3.32) to[bend left=10] 
  (-.2,-2.8);
}
\end{scope}
}

\end{scope}

\draw[
   shift={(5.6,-3.166)},
   black!50,
   line width=.4
]
  (90+70:6.8 and .2) arc
  (90+70:360+90-70:6.8 and .2);

\begin{scope}
\clip
  (-1.3,.2)  rectangle
  (-.3,-3.4);

\foreach \n in {0,...,9} {
\draw[
   shift={(5.6,-2.855+\n*.31)},
   black!50,
   line width=.2,
   dashed
]
  (90+70:6.8 and .2) arc
  (90+70:360+90-70:6.8 and .2);
}
\end{scope}

\begin{scope}
\clip[
  shift={(13.1,0)}
]
  (-1.3,.2)  rectangle
  (-.3,-3.4);

\foreach \n in {0,...,9} {
\draw[
   shift={(5.6,-2.855+\n*.31)},
   black!50,
   line width=.2,
   dashed
]
  (90+70:6.8 and .2) arc
  (90+70:360+90-70:6.8 and .2);
}
\end{scope}

\begin{scope}[
  shift={(5.8,-6-.9)}
]

\filldraw[
    gray!68,
    line width=10pt,
    line join=round,
    shift={(-.6,-2)}
  ]
 (0,0) circle (2);

\filldraw[
    gray!54,
    line width=10pt,
    line join=round,
    shift={(-.3,-1)}
  ]
 (0,0) circle (2);

\filldraw[
    gray!40,
    line width=10pt,
    line join=round
  ]
 (0,0) circle (2);

\draw[
    gray!80,
    line width=1.2pt,
]
  (0,0) circle (2);

\draw[
    gray!80,
    line width=1.2pt,
]
 (0,0) -- (30:2);
\draw[
    gray!80,
    line width=1.2pt,
]
 (0,0) -- (30+120:2);
\draw[
    gray!80,
    line width=1.2pt,
]
 (0,0) -- (30+240:1.95);

\foreach \k in {0,2} {
  \draw[
    {<[length=4pt,width=3pt]}-{>[length=4pt,width=3pt]},
    purple,
    line width=.9
  ]
    (115+\k*120+5:1.2) arc 
    (115+\k*120+5:115+120+\k*120-5:1.2);
}

  \draw[
    {<[length=4pt,width=3pt]}-{>[length=4pt,width=3pt]},
    purple,
    line width=.9
  ]
    (115+120+5:1.2) to[bend right=6]
    (-.8,-2.3) to[bend right=10]
    (115+120+120-5:1.2);

\end{scope}

\draw[
  |-{>[length=6pt,width=4pt]},
  line width=.7pt
]
  (5.5,-1) -- 
  (5.3,-5-.8);
\draw[
  |-{>[length=6pt,width=4pt]},
  line width=.7pt
]
  (1,-1) -- 
  (5.15,-7.1-.75);
\draw[
  |-{>[length=6pt,width=4pt]},
  line width=.7pt
]
  (9.6,-1) -- 
  (7.02,-6.1-.9);

\end{tikzpicture}
}

\end{figure}

In equivariant generalization of \cref{CechGroupoids}, we have:
\begin{example}[Equivariant {\v C}ech groupoids]
\label[example]{EquivariantCechGroupoids}
For $G \acts X$ a topological group action (\cref{TopologicalGroupAction}), consider an open cover 
$\mathcal{U} := \big\{\begin{tikzcd}[sep=15pt]U_i \ar[r, hook, "{ \iota_i }"] & X\end{tikzcd}\big\}_{i \in I}$ \cref{AnOpenCover} which is \emph{equivariant} in that the disjoint union of its patches is equipped with an action 
$G \acts \, \bigsqcup_i U_i$,  and compatibly so in that the covering map is equivariant \cref{Equivariance}:
\begin{equation}
  \label{EquivariantOpenCover}
  \begin{tikzcd}
    \underset{i}{\bigsqcup}
    \, 
    U_i
    \ar[
      out=180-60,
      in=60,
      looseness=4,
      "{
        \,\mathclap{G}\,
      }"{description}
    ]
    \ar[
      rr,
      ->>,
      "{
        (\iota_i)_{i \in i}
      }"
    ]
    &&
    X
    \ar[
      out=180-60,
      in=60,
      looseness=4,
      "{
        \,\mathclap{G}\,
      }"{description}
    ]
  \end{tikzcd}
\end{equation}
then the corresponding \emph{equivariant {\v C}ech groupoid} is
\begin{equation}
  \label{TheEquivariantChechGorupoid}
  (G \backsslash X)_{\mathcal{U}}
  \!=\!
  \left\{
  \begin{tikzcd}[
    row sep=24pt,
    column sep=35pt
  ]
    & 
    (
      g_1 \cdot x,\, 
      g_1 \cdot j
    )
    \ar[
      dr,
      "{
        \left(
          (
            g_1 \cdot x,\, 
            g_1\cdot j,\, 
            g_1 \cdot k
          ),\,
          g_2
        \right)
      }"{sloped, pos=.55        }
    ]
    \\
    (x,i)
    \ar[
      rr,
      "{
        \left(
          (x,i,j),\, 
          g_2 \cdot g_1
        \right)
      }"
    ]
    \ar[
      ur,
      "{
        \left(
          (x,i,j)
          ,\,
          g_1
        \right)
      }"{sloped, pos=.4}
    ]
    &&
    (
      g_2 \cdot g_1\cdot x,\,
      g_2 \cdot g_1 \cdot k
    )
  \end{tikzcd}
  \right\},
\end{equation}
where the general morphisms may usefully be thought of as composites of pure {\v C}ech morphisms and pure gauge morphisms, in either order:
\begin{equation}
  \begin{tikzcd}[
    row sep=5pt,
    column sep=60pt
  ]
    &[+5pt]
    (x,j)
    \ar[
      dr,
      "{
        (
          (x,j,j),
          g
        )
      }"{sloped}
    ]
    \\
    (x,i)
    \ar[
     dr,
     "{
       (
         (x,i,i),
         g 
       )
     }"{swap, sloped}
    ]
    \ar[
      ur,
      "{ 
        ((x,i,j), \mathrm{e}) 
      }"{sloped}
    ]
    \ar[
      rr,
      "{
        \big(
          (x,i,j), g
        \big)
      }"{description}
    ]
    &&
    (g \cdot x,\, g \cdot j)   
    \mathrlap{\,.}
    \\
    &
    (g\cdot x, g \cdot i)
    \ar[
      ur,
      "{ 
        (
        (g\cdot x, g\cdot i, g \cdot j),
        \mathrm{e}
        )
      }"{swap, sloped}
    ]
  \end{tikzcd}
\end{equation}
In equivariant generalization of \cref{CechGroupoidEquivalentToOriginalSpace}, the canonical functor from the equivariant {\v C}ech groupoid to the original action groupoid is an equivalence (\cref{EquivalenceOfTopGroupoids}, cf. \cite[Ex. 4.3.45]{SS25-Bun})
\begin{equation}
  \label{EquivariantCechGroupoidEquivalentToOriginalGSpace}
  \begin{tikzcd}[
    row sep=-1pt,
    column sep=15pt
  ]
    (G \backsslash X)_{\mathcal{U}}
    \ar[
      rr,
      "{ \sim }"
    ]
    &&
    G \backsslash X
    \\[-1pt]
    (x, i)
    \ar[
      dd,
      "{
        \left(
          (x,i,j)
          ,
          g
        \right)
      }"
    ]
    &&
    x
    \ar[
      dd,
      "{
        g
      }"{swap}
    ]
    \\
    & \longmapsto
    \\
    \big(
      (g\cdot x,\, g \cdot j)
    \big)
    &&
    g \cdot x
    \mathrlap{\,,}
  \end{tikzcd}
\end{equation}
so that, in generalization of \cref{CechCocycleAsFunctorOnCechGroupoid},
a Morita map now is a cocycle in nonabelian \emph{equivariant {\v C}ech cohomology} (cf. \cite[Rem. 4.3.46]{SS25-Bun}): 
\begin{equation}
  \label{EquivariantCechCocycle}
  \begin{tikzcd}[
    row sep={between origins,35pt},
    column sep={between origins,43pt}
  ]
    &&
    x
    \ar[
      dr,
      equals
    ]
    \ar[
      rrr,
      "{ g }"{description}
    ]
    &[20pt] 
    &[-20pt]
    &
    g \cdot x
    \ar[
      dr,
      equals
    ]
    &[20pt]
    \\
    G \backsslash X
    \ar[
      dddd,
      <-,
      "{ \sim }"{sloped}
    ]
    &&& 
    x
    \ar[
      rrr,
      "{ g }"{description, pos=.6}
    ]
    &&&
    g \cdot x
    \\
     &
     x
     \ar[
       uur,
       equals
     ]
     \ar[
       urr,
       equals
     ]
     \ar[
       rrr,
       "{ g }"{description}
     ]
     &&& 
    g \cdot x
    \ar[
      uur,
      crossing over,
      equals
    ]
    \ar[
      urr,
      equals
    ]
    \\[-10pt]
    &&&
    \rotatebox[origin=c]{90}{$\mapsto$}
    \\[-10pt]
    &
    & 
    (x,j)
    \ar[
      dr,
      shorten=-2pt,
      "{
        ((x,j,k),\mathrm{e})
      }"{swap, sloped, scale=.9}
    ]
    \ar[
      rrr,
      "{
        (
        (x,j,j),
        g
        )
      }"{description}
    ]
    &&&
    (g \cdot x, g \cdot j)
    \ar[
      dr,
      shorten >=-3pt,
      "{
        (
        (g\cdot x, g \cdot j, g \cdot k),
        \mathrm{e}
        )
      }"{sloped, scale=.75, pos=.7, yshift=1pt}
    ]
    \\
    (G \backsslash X)_{\mathcal{U}}
    \ar[
      dddd,
      "{
        \gamma
      }"
    ]
    &&&
    (x,k)
    \ar[
      rrr,
      "{
        ((x,k,k),g)
      }"{description, pos=.67}
    ]
    &&&
    (g \cdot x, g\cdot k)
    \\
    &
    (x,i)
    \ar[
      uur,
      "{
        ((x,i,j), \mathrm{e})
      }"{sloped}
    ]
    \ar[
      urr,
      "{
        ((x,i,k), \mathrm{e})
      }"{description, sloped}
    ]
    \ar[
      rrr,
      "{
        (
        (x,i,i),
        g
        )
      }"{description}
    ]
    &&&
    (g\cdot x, g \cdot i)
    \ar[
      uur,
      crossing over,
      "{
        \colorbox{white}{$
        (
        (g\cdot x, g\cdot i, g\cdot j),
        \mathrm{e}
        )
        $}
      }"{sloped, scale=.85, pos=.45}
    ]
    \ar[
      urr,
      "{
        (
        (g\cdot x, g\cdot i, g\cdot k),
        \mathrm{e}
        )
      }"{swap, sloped}
    ]
    \\[-15pt]
    &&&
    \rotatebox[origin=c]{-90}{$\mapsto$}
    \\[-13pt]
    & &
    \bullet
    \ar[
      dr,
      "{
        \gamma_{jk}(x)
      }"{description, sloped}
    ]
    \ar[
      rrr,
      "{
        \rho^g_j(x)
      }"{description}
    ]
    &&&
    \bullet
    \ar[
      dr,
      "{
        \gamma_{g\cdot j\, g\cdot k}(g \cdot x)
      }"{sloped, description}
    ]
    \\
    \mathbf{B}\Gamma
    &&&
    \bullet
    \ar[
      rrr,
      "{
        \rho^g_k(x)
      }"{description, pos=.58}
    ]
    &&&
    \bullet
    \\
    &
    \bullet
    \ar[
      uur,
      "{
        \gamma_{ij}(x)
      }"{description, sloped}
    ]
    \ar[
      urr,
      "{
        \gamma_{ik}(x)
      }"{description, sloped}
    ]
    \ar[
      rrr,
      "{
        \rho^g_i(x)
      }"{description}
    ]
    &&&
    \bullet
    \ar[
      uur,
      crossing over,
      "{
        \gamma_{g\cdot i\, g\cdot j}(g\cdot x)
      }"{sloped, description}
    ]
    \ar[
      urr,
      "{
        \gamma_{g\cdot i\, g\cdot k}(g\cdot x)
      }"{description, sloped}
    ]
    \\[-20pt]
  \end{tikzcd}
\end{equation}
\end{example}

In equivariant generalization of \cref{ExistenceOfGoodOpenCovers}, we have:
\begin{lemma}[{\cite{Yang2014}} 
\footnote{
In fact, \cite{Yang2014} proves something much stronger than \cref{ExistenceOfEquivariantGoodOpenCovers}, namely that good equivariant \emph{regular} covers exist (cf. also \cite[Def. 2.1.24]{SS25-Bun}). This is needed when expressing equivariance in terms of systems of $H$-fixed loci for subgroups $H \subset G$, cf. \cite[\S 6.2.1]{SS25-Bun}.}
]
\label[lemma]
  {ExistenceOfEquivariantGoodOpenCovers}
For $G$ a finite group, $X$ a smooth manifold and $G \acts X$ smooth action \textup{(\cref{TopologicalGroupAction})} there exists a good open cover of $X$ \textup{(\cref{GoodOpenCover})} which is equivariant \cref{EquivariantOpenCover}.
\end{lemma}

In equivariant generalization of \cref{MappingStackFromManifoldToDeloopingGroupoid}, we have:
\begin{definition}
  \label[definition]
    {MappingStackFromfHomotopyQuotToBGamma}
  For $G \acts X$ a topological $G$-space (\cref{TopologicalGroupAction}) that admits a good open cover $\mathcal{U}$ (\cref{GoodOpenCover}) which is $G$-equivariant \cref{EquivariantOpenCover} (such as is the case for $G$ a finite group acting smoothly on a smooth manifold $G$, by \cref{ExistenceOfEquivariantGoodOpenCovers}), and for $\Gamma$ a topological group, we say that
  \begin{equation}
    \label{MappingStackOutOfEquivariantCechGroupoid}
    \begin{aligned}
      \mathrm{Map}\big(
        G \backsslash X
        ,\,
        \mathbf{B}\Gamma
      \big)
      &
      :=
      \mathrm{Func}\big(
        (G \sslash X)_{\mathcal{U}}
        ,\,
        \mathbf{B}\Gamma
      \big)
      \\
      & 
      \defneq
      \big\{
        \begin{tikzcd}[sep=15pt]
          G \backsslash X 
          \ar[r, dashed]
          &
          \mathbf{B}\Gamma
        \end{tikzcd}
      \big\}
    \end{aligned}
  \end{equation}
  is the (up to equivalence) \emph{mapping stack} out of the homotopy quotient $G \backsslash X$ \cref{TheActionGroupoid} into the delooping $\mathbf{B}\Gamma$ (\cref{DeloopingGroupoid}), where $(G \backsslash X)_{\mathcal{U}}$ is the equivariant {\v C}ech groupoid from \cref{EquivariantCechGroupoids}.

  The underlying topologically discrete groupoid \cref{UnderlyingTopologicallyDiscreteGroupoid} we denote again by boldface $\mathbf{H}(-,-)$ \cref{MappingGroupoidOutOfGoodOpenCoverIntoBG}:
  \begin{equation}
      \mathbf{H}\big(
        G \backsslash X
        ,\,
        \mathbf{B}\Gamma
      \big)
      :=
      \flat\, 
      \mathrm{Map}\big(
        G \backsslash X
        ,\,
        \mathbf{B}\Gamma
      \big)
      \,,
  \end{equation}
  since this is now the \emph{cocycle groupoid} of ordinary nonabelian \emph{equivariant {\v C}ech cohomology}, see \cref{EquivariantNonabelianCechCohomology}.
\end{definition}

In generalization of \cref{OrdinaryNonabelianCechCohomology} we have:
\begin{example}[Equivariant nonabelian cohomology and equivariant bundles {\cite{SS25-Bun}}]
  \label[example]{EquivariantNonabelianCechCohomology}
  In the situation of \cref{MappingStackFromfHomotopyQuotToBGamma},
  we have that 
  \begin{equation}
    G \mathrm{Equ}\Gamma\mathrm{PrnBdl}(X)
    \simeq
    \mathbf{H}\big(
      G \backsslash X
      ,\,
      \mathbf{B}\Gamma
    \big)
  \end{equation}
  is equivalently the groupoid of \emph{$G$-equivariant $\Gamma$-principal bundles} on $X$ (see \cref{PrincipalBundleAsPullbackAlongCechCocycle} for the construction), and its connected components (hence equivalently the isomorphism classes of these bundles) is the \emph{equivariant nonabelian {\v C}ech cohomology} of $X$:
  \begin{equation}
    H^1_G\big(
      X
      ;\,
      \Gamma
    \big)
    \simeq
    \big[
    \mathbf{H}(
      G \backsslash X
      ,\,
      \mathbf{B}\Gamma
    )
    \big]_0
  \end{equation}
\end{example}

These constructions of mapping stacks \cref{MappingStackOutOfGoodOpenCoverIntoBG,MappingStackOutOfEquivariantCechGroupoid} do not actually depend, up to equivalence, on the use of good {\v C}ech groupoids, these are just a particularly nice choice (when they exist) of general \emph{cofibrant resolutions}, of which a larger class is the following:

\begin{definition}[{cf. \cite[Prop. 4.2.37]{SS25-Bun}}]
  \label[definition]
   {CofibResolutionOfTopologicalGroupoid}
  A Dugger-\emph{cofibrant resolution} of a topological groupoid $\mathcal{X}$ is an equivalence (\cref{EquivalenceOfTopGroupoids})
  \begin{equation}
    \begin{tikzcd}[sep=15pt]
      \widehat{\mathcal{X}}
      \ar[r, "{\sim}"]
      &
      \mathcal{X}
    \end{tikzcd}
  \end{equation}
  with a topological groupoid $\widehat{\mathcal{X}}$ for which:
  \begin{enumerate}
    \item all three component spaces $\mathrm{Obj}(\widehat{\mathcal{X}})$, $\mathrm{Mor}(\widehat{\mathcal{X}})$, $\mathrm{Mor}(\widehat{\mathcal{X}}) \tensor[_s]{\times}{_t} \mathrm{Mor}(\widehat{\mathcal{X}})$ are homeomorphic to disjoint unions of Cartesian spaces \cref{CartesianSpaceHomeomorphicToOpenBall},
    \item
    the identity-morphisms including maps maps like $\begin{tikzcd}\mathrm{Obj}(\widehat{\mathcal{X}}) \ar[r, hook,  "{ \mathrm{e} }"] & \mathrm{Mor}(\widehat{\mathcal{X}})\end{tikzcd}$ 
    and $\begin{tikzcd}[column sep=32pt]\mathrm{Mor}(\widehat{\mathcal{X}}) \tensor[_s]{\times}{_t} \mathrm{Obj}(\widehat{\mathcal{X}}) \ar[r, "{ \mathrm{id} \tensor[_s]{\times}{_t} \mathrm{e} }"] & \mathrm{Mor}(\widehat{\mathcal{X}}) \tensor[_s]{\times}{_t} \mathrm{Mor}(\widehat{\mathcal{X}}) \end{tikzcd}$, etc., are the inclusions of disjoint summands.
  \end{enumerate}
\end{definition}

Therefore in further generalization of \cref{MappingStackFromfHomotopyQuotToBGamma} we set:
\begin{definition}
  \label[definition]
   {MappingStackFromDuggerCofibrantToBGamma}
  For $\begin{tikzcd}[sep=small] \widehat{\mathcal{X}} \ar[r, "{ \sim }"] & \mathcal{X} \end{tikzcd}$ a Dugger-cofibrant resolution (\cref{CofibResolutionOfTopologicalGroupoid}) of a topological groupoid and $\Gamma$ a topological group, then the \emph{mapping stack} from $\mathcal{X}$ to the delooping $\mathbf{B}\Gamma$ is
  \begin{equation}
    \begin{aligned}
    \mathrm{Map}\big(
      \mathcal{X}
      ,\,
      \mathbf{B}\Gamma
    \big)
    & 
    :=
    \mathrm{Func}\big(
      \widehat{\mathcal{X}}
      ,\,
      \mathbf{B}\Gamma
    \big)
    \\
    & \defneq
    \big\{
    \begin{tikzcd}[sep=15pt]
      \mathcal{X}
      \ar[r, dashed]
      &
      \mathbf{B}\Gamma
    \end{tikzcd}
    \big\}.
    \end{aligned}
  \end{equation}
\end{definition}
\begin{lemma}
  \label[lemma]
  {MoritaEquivalenceOfMappingStack}
  The mapping stack $\mathrm{Map}(\mathcal{X}, \mathbf{B}\Gamma)$ in \cref{MappingStackFromDuggerCofibrantToBGamma} is well-defined in that it depends, up to Morita equivalence \textup{(\cref{MoritaMap})}, only on the equivalence class of $\mathcal{X}$:
  \begin{equation}
    \label{WitnessingMoritaInvarianceOfMappingSStacl}
    \begin{tikzcd}[sep=18pt]
      \mathcal{X}'
      \ar[
        r,
        "{ F }",
        "{ \sim }"{swap}
      ]
      &
      \mathcal{X}
    \end{tikzcd}
    \;\;\;\;\;
    \Rightarrow
    \;\;\;\;\;
    \begin{tikzcd}[sep=18pt]
      \mathrm{Map}\big(
        \mathcal{X}
        ,
        \mathbf{B}\Gamma
      \big)
      \ar[
        r,
        "{ F^\ast }",
        "{ \sim }"{swap}
      ]
      &
      \mathrm{Map}\big(
        \mathcal{X}'
        ,
        \mathbf{B}\Gamma
      \big).
    \end{tikzcd}
  \end{equation}
\end{lemma}
\begin{proof}
  Discussed below in \cref{ComputingMappingStacksInSmplPSh} of \cref{SmoothInfinityGroupoids}.
\end{proof}

\subsubsection{Slice mapping stacks and Twisted cohomology}
\label{SlicingAndTwisting}

\begin{definition}
  \label[definition]{GlobalFibrations}
  Say that a continuous functor $F$ (\cref{TopologicalFunctor}) between topological groupoids 
  is a \emph{global fibration}, to be denoted as on the left here:
  \begin{equation}
    \label{GlobalFibration}
    \begin{tikzcd}[
      row sep=3pt
    ]
      \mathcal{X}
      \ar[
        dd,
        "{ F }"{swap},
        "{ \in \mathrm{Fib} }"
      ]
      &&
      \forall \, x
      \ar[
        r,
        dashed,
        "{ \exists \, \widehat{f} }"
      ]
      \ar[
        dd,
        shorten=3pt,
        |->
      ]
      &
      t(\widehat{f})
      \ar[
        dd,
        shorten=3pt,
        |->
      ]
      \\
      & \Leftrightarrow
      \\
      \mathcal{Y}
      &&
      s(f) 
      \ar[
        r, 
        "{ \forall\, f }"
      ]
      &
      t(f)
    \end{tikzcd}
  \end{equation}
  if, as indicated on the right, for all
  \begin{enumerate}
    \item $n \in \mathbb{N}$, $f : \mathbb{R}^n \to \mathrm{Mor}(\mathcal{Y})$, 
    a plot of morphisms in $\mathbb{Y}$,
    
    \item $x : \mathbb{R}^n \to \mathrm{Obj}(\mathcal{Y})$,
    a plot of objects of $\mathcal{Y}$
    such that $F(x) = s(f)$,
  \end{enumerate}
  there exists 
  $\widehat{f} : \begin{tikzcd}[sep=small] \mathbb{R}^n \ar[r] & \mathcal{X}\end{tikzcd}$ such that
  $F(\widehat{f}) = f$ and $s(\widehat{f}) = x$ (a ``lift'').
\end{definition}

\begin{example}
\label[example]
  {HomotopyQuotientFibration}
For an action $\Gamma \acts \, X$ \textup{(\cref{TopologicalGroupAction})}, its homotopy quotient \textup{(\cref{ActionGroupoid})} comes with a canonical functor to $\mathbf{B}\Gamma$ \cref{DeloopingGroupoidIsPointQuotient}, which is a global fibration \cref{GlobalFibration}:
\begin{equation}
  \label{FibrationFromActionGroupoidToBG}
  \begin{tikzcd}[
    column sep=12pt
  ]
    \Gamma
      \backsslash 
    X
    \ar[
      d,
      "{
        p
          ^{\scalebox{.7}{$\Gamma \acts X$}}
          _{\mathrm{univ}}
      }"{swap},
      "{
        \in 
        \mathrm{Fib}
      }"
    ]
    &
    x 
      \ar[
        rr, 
        "{ \gamma }"{description}
      ]
    &
    {}
    \ar[
      d,
      |->,
      shorten <=2pt,
      shorten >=6pt
    ]
    &
    \gamma \cdot x
    \\
    \mathbf{B}\Gamma
    &
    \bullet 
    \ar[
      rr, 
      "{ \gamma }"{description,pos=.43}
    ] 
    &{}& 
    \bullet
    \mathrlap{\,.}
  \end{tikzcd}
\end{equation}
Here the subscript notation indicates that this simple fibration is in fact is the (stacky) \emph{universal $\Gamma$-associated $X$-fiber bundle} and even the universal \emph{globally equivariant} such: see \cref{PrincipalBundleAsPullbackAlongCechCocycle}.
\end{example}

\begin{definition}[Homotopy fiber product of topological groupoids]
  \label[definition]
  {HomotopyFiberProductOfTopGroupoids}
  Given a pair of coincident continuous functors, $\begin{tikzcd}[sep=small] \mathcal{X} \ar[r, "{ F_1 }"] & \mathcal{B} \ar[r, <-, "{ F_2 }"] & \mathcal{Y}\end{tikzcd}$, the corresponding fiber product (\cref{FiberProduct}) of their morphisms spaces canonically inherits the structure of a topological groupoid, the \emph{fiber product of topological groupoids}: 
  \begin{equation}
    \label{FiberProductOfTopologicalGroupoids}
    \begin{tikzcd}[column sep=large, row sep=small]
      \mathcal{X}
      \underset{\mathcal{B}}{\times}
      \mathcal{Y}
      \ar[
        dr,
        phantom,
        "{ \lrcorner }"{pos=.05}
      ]
      \ar[d]
      \ar[r]
      &
      \mathcal{Y}
      \ar[d, "{ F_2 }"]
      \\
      \mathcal{X}
      \ar[r, "{ F_1 }"{swap}]
      &
      \mathcal{B}
      \mathrlap{\,,}
    \end{tikzcd}
  \end{equation}
  hence
  \begin{equation}
    \begin{aligned}
    \mathrm{Mor}\big(
      \mathcal{X}
      \times_{\mathcal{B}}
      \mathcal{Y}
    \big)
    & =
      \mathrm{Mor}(\widehat{\mathcal{X}})
      \;\;
      \underset{
        \mathclap{
          \mathrm{Mor}(\widehat{\mathcal{B}})
        }
      }{\times}
      \;\;
      \mathrm{Mor}(\widehat{\mathcal{Y}})
      \mathrlap{\,,}
    \\
    \mathrm{Obj}\big(
      \mathcal{X}
      \times_{\mathcal{B}}
      \mathcal{Y}
    \big)
    & =
      \mathrm{Obj}(\widehat{\mathcal{X}})
      \;\;
      \underset{
        \mathclap{
          \mathrm{Obj}(\widehat{\mathcal{B}})
        }
      }{\times}
      \;\;
      \mathrm{Obj}(\widehat{\mathcal{Y}}) 
      \mathrlap{\,,}
    \end{aligned}
  \end{equation}
  etc.,
  and the structure maps on $\mathcal{X} \times_{\mathcal{B}} \mathcal{Y}$ are induced entry-wise from the given ones on $\mathcal{X}$, $\mathcal{Y}$ and $\mathcal{B}$.
  
  If here $F_1$ or $F_2$ are global fibrations (\cref{GlobalFibrations}) then this is a representation of the (up to equivalence) \emph{homotopy fiber product} of the corresponding topological stacks, or the \emph{homotopy pullback} of one either along the other. 
  \footnote{
    We are tacitly using here that stackification preserves finite homotopy limits, so that a representative for the homotopy fiber product in topological stacks is given already by a homotopy fiber product of topological groupoids, hence already by a fiber product with a global fibration, not necessarily a local fibration.
  }
\end{definition}

\begin{example}
  \label[example]
  {FibrationOfHomotopyQuotientOverBG}
  The homotopy fiber (\cref{HomotopyFiberProductOfTopGroupoids}) of $\begin{tikzcd}[sep=small] G \backsslash X \ar[r] & \mathbf{B}G\end{tikzcd}$ \cref{FibrationFromActionGroupoidToBG} is $X$ (via \cref{SpaceAsGroupoid}):
  \begin{equation}
  \label{HomotopyFiberOfHoQuotientProjection}
  \begin{tikzcd}[column sep=large]
    X 
    \ar[r]
    \ar[d]
    \ar[
      dr,
      phantom,
      "{ 
        \lrcorner
      }"{pos=.2}
    ]
    &
    G \backsslash X
    \ar[
      d,
      "{
        \in \mathrm{Fib}
      }"
    ]
    \\
    \ast
    \ar[r, "{ \exists ! }"]
    &
    \mathbf{B}G
  \end{tikzcd}
  \end{equation}
  Conversely (\cite[Prop. 4.2.77]{SS25-Bun}), every fibration over $\mathbf{B}G$ with fiber $X$ exhibits a $G$-action on $X$ this way!
  More generally (cf. \cite[Prop. 4.2.79]{SS25-Bun}), for $\phi : \begin{tikzcd}[sep=small] G' \ar[r] & G\end{tikzcd}$ a homomorphism of topological groups, and noting that this induces from the $G$ action on $X$ also $G'$-action on $X$, via
  \begin{equation}
    \begin{tikzcd}[sep=0pt]
      G' \times X
      \ar[rr]
      &&
      X
      \\[-2pt]
      (g',x)
      &\longmapsto&
      \phi(g'\,)\cdot x
      \mathrlap{\,,}
    \end{tikzcd}
  \end{equation}
  then the homotopy quotient by $G'$ is the following homotopy pullback of that by $G$:
  \begin{equation}
    \label{PullbackAction}
    \begin{tikzcd}[column sep=large]
      G' \backsslash X
      \ar[r]
      \ar[d]
      \ar[
        dr,
        phantom,
        "{ 
          \lrcorner 
        }"{pos=.1}
      ]
      &
      G \backsslash X
      \ar[
        d,
        "{
          p
        }"{swap},
        "{
         \in \mathrm{Fib}
        }"
      ]
      \\
      \mathbf{B}G'
      \ar[
        r,
        "{
          \mathbf{B}\phi
        }"
      ]
      &
      \mathbf{B}G
        \mathrlap{\,.}
    \end{tikzcd}
  \end{equation}
\end{example}

\begin{example}[Stacky universal fiber bundles]
  \label[example]
    {PrincipalBundleAsPullbackAlongCechCocycle}
  For $G \acts X$ a finite group action on a smooth manifold, and for $\begin{tikzcd}[sep=small](G \backsslash X)_{\mathcal{U}} \ar[r, "{\sim}"] & G \backsslash G\end{tikzcd}$ the {\v C}ech groupoid resolution (\cref{EquivariantCechGroupoids}) of an equivariant good open cover (\cref{ExistenceOfEquivariantGoodOpenCovers}), and given a map $\begin{tikzcd}[sep=small] G \backsslash X \ar[r, "{ \gamma }"] & \mathbf{B}\Gamma\end{tikzcd}$ to the delooping of a topological group $\Gamma$, hence an equivariant {\v C}ech cocycle \cref{EquivariantCechCocycle}, then the homotopy pullback (\cref{HomotopyFiberProductOfTopGroupoids}) along $\gamma$ of the homotopy quotient $\Gamma \backsslash \Gamma$ \cref{FibrationFromActionGroupoidToBG}
  is, up to equivalence (homotopy $G$-quotient of), the $G$-equivariant $\Gamma$-principal bundle
  \footnote{
    With the conventions used here, the principal bundle
    $\mathcal{P}$ in \cref{PullbackOfUniversalStackyPrincipalBundleAlongCechCocycle} carries a \emph{right} $\Gamma$-action.
  }
  $\begin{tikzcd}[sep=small]\mathcal{P} \ar[r, "{ p }" ] & X\end{tikzcd}$ corresponding to $\gamma$ under the {\v C}ech theory of \cref{OrdinaryNonabelianCechCohomology,EquivariantNonabelianCechCohomology}:
  \begin{equation}
    \label{PullbackOfUniversalStackyPrincipalBundleAlongCechCocycle}
    \begin{tikzcd}[
      column sep=25pt
    ]
      G \backsslash \mathcal{P}
      \ar[
        d,
        "{
          G \backsslash p
        }"
      ]
      \ar[
        r,
        <-,
        "{ \sim }"
      ]
      &
      G \backsslash 
      \widehat{\mathcal{P}}
      \ar[r]
      \ar[d]
      \ar[
        dr,
        phantom,
        "{ \lrcorner }"{pos=.1}
      ]
      &
      \Gamma \backsslash \Gamma
      \ar[
        d,
        "{ 
          p
            ^{\scalebox{.6}{$\Gamma \acts \Gamma$}}
            _{\mathrm{univ}} 
        }"{swap},
        "{
          \in \mathrm{Fib}
        }"
      ]
      \\
      G \backsslash X
      \ar[
        r,
        <-,
        "{ \sim }"
      ]
      &
      (G \backsslash X)_{\mathcal{U}}
      \ar[r, "{ \gamma }"]
      &
      \mathbf{B}\Gamma
      \mathrlap{\,.}
    \end{tikzcd}
  \end{equation}
  It is straightforward but instructive to unwind the definitions to see that this diagram reduces in components exactly to the traditional constructions relating principal bundles to their {\v C}ech data. A comprehensive discussion is in \cite[\S 5.1]{SS25-Bun}.

  In fact this situation is universal in the following way:
  Given a $\Gamma$-principal bundle $\mathcal{P}$ and any action $\Gamma \acts F$ then there is classically the \emph{$\mathcal{P}$-associated $F$-fiber bundle}
  \begin{equation}
    \mathcal{P} 
    \!\otimes_\Gamma\!
    F
    :=
    \big\{
      (\pi,\phi)
      \in
      \mathcal{P}
      \times
      F
    \big\}
    \big/
    \big(
      \forall_{\gamma \in \Gamma}
      \;
      (\pi \cdot \gamma, \phi)
      \sim 
      (\pi, \gamma \cdot \phi)
    \big)
    \mathrlap{\,,}
  \end{equation}
  and this is equivalently just the pullback of $\Gamma \backsslash F$ \cref{FibrationFromActionGroupoidToBG} along the cocycle $\gamma$ for $\mathcal{P}$ (cf. \cite[\S 5.2.2]{SS25-Bun}):
  \begin{equation}
    \label{PullbackOfUniversalStackyFiberBundleAlongCechCocycle}
    \begin{tikzcd}[column sep=huge]
      G \backsslash 
      \big(
        \mathcal{P}
        \!\otimes_\Gamma\!
        F
      \big)
      \ar[
        d,
        "{
          G \backsslash p
        }"{swap},
        "{
          \in \mathrm{Fib}
        }"
      ]
      \ar[
        r,
        <-,
        "{ \sim }"
      ]
      &
      G \backsslash 
      \widehat{
      \big(
        \mathcal{P}
        \!\otimes_\Gamma\!
        F
      \big)
      }
      \ar[r]
      \ar[d]
      \ar[
        dr,
        phantom,
        "{ \lrcorner }"{pos=.1}
      ]
      &
      \Gamma \backsslash F
      \ar[
        d,
        "{ 
          p
            ^{\scalebox{.6}{$\Gamma \acts F$}}
            _{\mathrm{univ}} 
        }"{swap},
        "{
          \in \mathrm{Fib}
        }"
      ]
      \\
      G \backsslash X
      \ar[
        r,
        <-,
        "{ \sim }"
      ]
      &
      (G \backsslash X)_{\mathcal{U}}
      \ar[r, "{ \gamma }"]
      &
      \mathbf{B}\Gamma
      \mathrlap{\,.}
    \end{tikzcd}
  \end{equation}
\end{example}

In variation of \cref{MappingStackFromManifoldToDeloopingGroupoid}, we set:
\begin{definition}[Slice mapping stack,
  {cf \cite[4.2.66]{SS25-Bun}}]
\label[definition]
  {SliceMappingStack}
Given 
\begin{enumerate}
\item $\Gamma \acts \, Y$ a topological group action (\cref{TopologicalGroupAction}),
\item $\begin{tikzcd}[sep=15pt]\widehat{\mathcal{X}} \ar[r, "{\sim}"] & \mathcal{X}\end{tikzcd}$ a cofibrant resolution (\cref{CofibResolutionOfTopologicalGroupoid}),

\item
$\begin{tikzcd}[sep=15pt]\widehat{\mathcal{X}} \ar[r, "{ \tau }"] & \mathbf{B}\Gamma\end{tikzcd}$
a continuous functor (the \emph{twist}),
\end{enumerate}
then the \emph{slice mapping stack} from $\mathcal{X}$ to the homotopy quotient $\Gamma \backsslash Y$ over $\mathbf{B}\Gamma$ \cref{FibrationFromActionGroupoidToBG} is this fiber product \cref{FiberProductOfTopologicalGroupoids}:
\begin{equation}
  \label{SliceMappingStackAsFiberProduct}
  \hspace{1.5cm}
  \begin{tikzcd}[
    column sep=25pt
  ]
  &&
  \mathllap{
  \mathrm{Map}\big(
    \mathcal{X}
    ,
    \Gamma \backsslash Y
  \big)_{\!\mathbf{B}\Gamma}
  :=
  \;
  }
  \mathrm{Func}\big(
    \widehat{\mathcal{X}}
    ,
    \Gamma \backsslash Y
  \big)
  \tensor
    [_{p_\ast\!\!}]
    {\times}
    {_{\widetilde{\tau}} }
  \ast
    \ar[d]
    \ar[r]
    \ar[
      dr,
      phantom,
      "{ 
        \lrcorner 
      }"{pos=.2}
    ]
    &
    \mathrm{Func}\big(
      \widehat{\mathcal{X}}
      ,
      \Gamma \backsslash Y
    \big)
    \ar[
      d,
      "{
        p_\ast
      }"
    ]
    \\
  &&  \ast
    \ar[
      r,
      "{ \widetilde{\tau} }"
    ]
    &
    \mathrm{Func}\big(
      \widehat{\mathcal{X}}
      ,
      \mathbf{B}\Gamma
    \big)
    \mathrlap{\,.}
  \end{tikzcd}
\end{equation}
Following \cref{MappingStackOutOfGoodOpenCoverIntoBG}, we will more suggestively denote this also as follows:
\begin{equation}
  \mathrm{Map}\big(
    \mathcal{X},
    \Gamma \backsslash Y
  \big)_{\mathbf{B}\Gamma}
  \defneq
  \left\{
  \begin{tikzcd}[
    row sep=13pt, 
    column sep=37pt
  ]
    &
    \Gamma \backsslash Y
    \ar[d, "{p}"]
    \\
    \mathcal{X}
    \ar[
      ur,
      dashed,
    ]
    \ar[
      r,
      "{ \tau }"
    ]
    &
    \mathbf{B}\Gamma
  \end{tikzcd}
  \right\}.
\end{equation}
\end{definition}

\begin{remark}
  \label[remark]
  {SliceManningSpaceIntoQuotientStackOverBG}
  A slice mapping stack of the form \cref{SliceMappingStack} is just a topological space, in that it is a topological groupoid with only identity morphisms (\cref{SpaceAsGroupoid}). This is because the fiber product \cref{SliceMappingStackAsFiberProduct} enforces the horizontal composition \cref{HorizontalWhiskeringOfTransformations} of the transformations between slice maps with the projection map to be the identity on $\tau$:
  \begin{equation}
    \begin{tikzcd}
      & 
      \Gamma \backsslash Y
      \ar[
        d,
        "{ p }"
      ]
      \\
      \widehat{\mathcal{X}}
      \ar[
        ur, 
        bend left=35,
        dashed,
        "{f}"
        {description, name=s}
      ]
      \ar[
        ur, 
        bend right=30,
        dashed,
        "{f'}"{description, name=t}
      ]
      \ar[
        from=s,
        to=t,
        shorten=-2pt,
        Rightarrow,
        "{\eta}"
      ]
      &
      \mathbf{B}\Gamma
    \end{tikzcd}
    \adjustbox{raise=2pt}{$
    \quad = \quad
    $}
     \adjustbox{raise=12pt}{$
    \begin{tikzcd}[
      column sep=25
    ]
      {}
      \\
      \widehat{\mathcal{X}}
      \ar[
        r,
        bend left=25,
        "{ \tau }"{description, name=s}
      ]
      \ar[
        r,
        bend right=25,
        "{ \tau }"{description, name=t}
      ]
      \ar[
        from=s,
        to=t,
        equals
      ]
      &
      \!\!\mathbf{B}\Gamma
      \,.
    \end{tikzcd}
    $}
  \end{equation}
  This means that the component morphisms of $p \cdot \eta$ all have to be the identity 
  $\begin{tikzcd}[sep=small] \bullet \ar[r, "{\mathrm{e}}"] & \bullet\end{tikzcd}$, which here means \cref{FibrationFromActionGroupoidToBG} that already the components maps of $\eta$ have to be identities, 
  $\begin{tikzcd}[sep=small] x \ar[r, "{\mathrm{e}}"] & x\end{tikzcd}$.
\end{remark}

\begin{example}
  Given a $G$-equivariant $\Gamma$-associated $F$-fiber bundle according to \cref{PullbackOfUniversalStackyFiberBundleAlongCechCocycle} in  \cref{PrincipalBundleAsPullbackAlongCechCocycle}, then the slice mapping stack into it is equivalently its space $\Gamma_{\!X}(-)^{G}$ of equivariant \emph{sections} of this bundle:
  \begin{equation}
    \label{SectionsAsSliceMapping}
    \begin{aligned}
      \Gamma_{\!X}\big(
        \mathcal{P}
        \!\otimes_\Gamma\!
        F
      \big)^{\! G}
      &
      \simeq
      \left\{
      \begin{tikzcd}[
        column sep=15pt
      ]
        &
        G \backsslash
        \big(
          \mathcal{P}
          \!\otimes_\Gamma\!
          F
        \big)
        \ar[
          d,
          "{
            G \backsslash p
          }"
        ]
        \\
        G \backsslash X
        \ar[
          ur,
          dashed
        ]
        \ar[
          r,
          equals
        ]
        &
        G \backsslash X
      \end{tikzcd}
      \right\}
      \\
      & \simeq
      \left\{
      \begin{tikzcd}[
        column sep=30pt
      ]
        &
        \Gamma \backsslash F
        \ar[
          d,
          "{
            p
              ^{\scalebox{.6}{$\Gamma \acts F$}}
              _{\mathrm{univ}}
          }"
        ]
        \\
        G \backsslash X
        \ar[
          r,
          "{ \gamma }"
        ]
        \ar[
          ur,
          dashed
        ]
        &
        \mathbf{B}\Gamma
      \end{tikzcd}
     \, \right\}.
    \end{aligned}
  \end{equation}
  Here the first equivalence follows because, by the slicing,  dashed functors out of any ({\v C}ech) resolution actually have to factor through $G \backsslash X$, while the second follows either by direct inspection or abstractly by the universal property of the homotopy pullback in \cref{PullbackOfUniversalStackyFiberBundleAlongCechCocycle}.
\end{example}

\begin{example}[{cf. \cite[Prop. 4.2.77]{SS25-Bun}}]
\label[example]
{EquivariantMapsAsSliceMapsOfHomotopyQuotients}
For a pair of $G$-spaces $G \acts X$ and $G \acts Y$, the slice mapping stack (\cref{SliceMappingStack}) between their homotopy quotients over $\mathbf{B}G$ (via \cref{FibrationOfHomotopyQuotientOverBG}), which a topological space by \cref{SliceManningSpaceIntoQuotientStackOverBG}, is naturally homeomorphic to the equivariant mapping space \cref{EquivariantMappingSpace} between $X$ and $Y$:
\begin{equation}
  \label{GEquivariantMapsAsSliceMapsOverBG}
  {
  \renewcommand{\arraystretch}{1.3}  
  \setlength{\arraycolsep}{-5pt}
  \begin{array}{ccc}
    \mathrm{Map}(X,Y)^G
    &
    \begin{tikzcd}
    \ar[
      rr,
      "{ \sim }"
    ]
    &&
    {}
    \end{tikzcd}
    &
    \;
    \mathrm{Map}\big(
      G \backsslash X
      ,
      G \backsslash Y
    \big)_{\mathbf{B}G}
  \\
  \bigg(
  \adjustbox{raise=-6pt}{
  \begin{tikzcd}
    X 
    \ar[
      in=60,
      out=180-60,
      looseness=4,
      "{ 
        \,\mathclap{G}\, 
      }"{description}
    ]
    \ar[
      r, 
      dashed, 
      "{ f }"
    ]
    &
    Y
    \ar[
      in=60,
      out=180-60,
      looseness=4,
      "{ 
        \,\mathclap{G}\, 
      }"{description}
    ]
  \end{tikzcd}
  }
  \Bigg)
  &
  \longmapsto
  &
  \Bigg(
  \begin{tikzcd}[
    column sep=10pt,
    row sep=0pt
  ]
    G \backsslash X
    \ar[dr]
    \ar[
      rr,
      dashed, 
      "{
        G \backsslash f
      }"
    ]
    &&
    G \backsslash Y
    \ar[dl]
    \\
    & 
    \mathbf{B}G
  \end{tikzcd}
  \Bigg)
  \mathrlap{.}
  \end{array}
  }
\end{equation}
This follows because the slicing over $\mathbf{B}G$ entails that all {\v C}ech morphism in any resolution $\widehat{G \backsslash X}$ of $G \backsslash X$ have to map to identity morphisms in $G \backsslash Y$, whence all Morita maps come from plain functors 
$\begin{tikzcd}[sep=small]G \backsslash X \ar[r] & G \backsslash Y\end{tikzcd}$. 

In particular: 
\begin{enumerate}
\item 
an \emph{equivariant homotopy} \cref{EquivariantHomotopy} is equivalently a homotopy taking values in the slice mapping stack on the right of \cref{GEquivariantMapsAsSliceMapsOverBG}.
\begin{equation}
  \begin{tikzcd}[row sep=small,
    column sep=13pt
  ]
    \{0\}
    \ar[d, hook]
    \ar[
      drr,
      "{ 
        \widetilde{f} 
      }" 
    ]
    \\
    {[0,1]}
    \ar[
      rr,
      dashed,
      "{
        \widetilde{\eta}
      }"{description, pos=.4}
    ]
    &&
    \mathrm{Map}(X,Y)^G
    \\
    \{1\}
    \ar[u]
    \ar[
      urr,
      "{ \widetilde{g} }"{swap} 
    ]
  \end{tikzcd}
  \hspace{.4cm}
  \Leftrightarrow
  \hspace{.6cm}
  \begin{tikzcd}[row sep=small,
    column sep=13pt
  ]
    \{0\}
    \ar[d, hook]
    \ar[
      drr,
      "{ 
        \widetilde{
          G \backsslash f
        } 
      }"{pos=.3, yshift=-2.2} 
    ]
    \\
    {[0,1]}
    \ar[
      rr,
      dashed,
      "{
        \widetilde{\eta}
      }"{description, pos=.4}
    ]
    &&
    \mathrm{Map}\big(
      G \backsslash X
      ,
      G \backsslash Y
    \big)_{\mathbf{B}G}
    \mathrlap{\,.}
    \\
    \{1\}
    \ar[u]
    \ar[
      urr,
      "{ 
        \widetilde{
          G \backsslash g
        } 
      }"{swap, pos=.3, yshift=1.8}
    ]
  \end{tikzcd}
\end{equation}

\item 
for
\begin{equation}
  \mbox{$G \acts X$ a trivial action}
  \;\;\;\;
  \Leftrightarrow
  \;\;\;\;
  G \backsslash X
  =
  X \times \mathbf{B}G
  \mathrlap{\,,}
\end{equation}
so that
$G$-equivariant maps from $X$ to $Y$ are equivalently (cf. \cite[Ex. 1.8]{SS25-Bun}) maps from $X$ into the fixed locus $Y^G \subset Y$ \cref{FixedSubspace}, it follows that
\begin{equation}
 \label{LiftsOfBGThroughQuotientProjectionAreFixedPoints}
  \left\{
  \begin{tikzcd}[row sep=small, column sep=large]
      &
      G \backsslash Y
      \ar[d]
      \\
      X \times \mathbf{B}G
      \ar[
        r,
        "{ \mathrm{pr}_2 }"
      ]
      \ar[
        ur,
        dashed
      ]
      &
      \mathbf{B}G
  \end{tikzcd}
  \right\}
  \simeq
  \big\{
    \begin{tikzcd}
      X 
      \ar[
        r, 
        dashed
      ]
      &
      Y^G
    \end{tikzcd}
  \big\}
  \mathrlap{\,.}
\end{equation}
\end{enumerate}
\end{example}

In mild but crucial generalization of \cref{EquivariantMapsAsSliceMapsOfHomotopyQuotients}, we also have:
\begin{example}
  \label[example]
    {QuotientSliceMapsOverBPhi}
  Given a homomorphism of topological groups $\phi : \begin{tikzcd}[sep=small]G' \ar[r] & G\end{tikzcd}$,
  and actions $G' \acts \, X$ and $G \acts \, Y$, then the slice mapping stack (\cref{SliceMappingStack}) between their homotopy quotients sliced over $\mathbf{B}G$ (as in \cref{FibrationOfHomotopyQuotientOverBG}, for $G' \backsslash X$ via $\mathbf{B}\phi$) is a topological space naturally homeomorphic to the $G'$-equivariant maps from $X$ to $Y$, the latter with its $G'$-action induced via $\phi$ \cref{PullbackAction}:
  \begin{equation}
    \label{FactoringQuotientSliceMapsOverBPhi}
    \left\{
    \begin{tikzcd}[
      column sep=20pt
    ]
      G' \backsslash X
      \ar[d]
      \ar[
        rr,
        dashed
      ]
      &&
      G \backsslash X
      \ar[d]
      \\
      \mathbf{B}G'
      \ar[
        rr,
        "{
          \mathbf{B}\phi
        }"{description}
      ]
      &&
      \mathbf{B}G
    \end{tikzcd}
    \right\}
    \simeq
    \left\{
    \begin{tikzcd}
      G' \backsslash X
      \ar[
        rr,
        dashed
      ]
      \ar[dr]
      &[-40pt]&[-20pt]
      G' \backsslash Y
      \ar[
        dl
      ]
      \ar[r]
      \ar[
        dr,
        phantom,
        "{ 
          \lrcorner 
        }"{pos=.15}
      ]
      &[-10pt]
      G \backsslash X
      \ar[d]
      \\
      &
      \mathbf{B}G'
      \ar[
        rr,
        "{
          \mathbf{B}\phi
        }"{description}
      ]
      &&
      \mathbf{B}G
    \end{tikzcd}
    \right\}.
  \end{equation}
  As indicated on the right, this follows from \cref{EquivariantMapsAsSliceMapsOfHomotopyQuotients} by the universal property \cref{PullbackSquare} of the pullback. 
\end{example}

In view of \cref{NonabelianCohomologyRepresentedByBGAndbfBG}, we thereby obtain a natural re-formulation of \emph{equivariant cohomology} as a form of twisted cohomology of topological stacks (and it is this reformulation which naturally generalizes to a notion of orbifold cohomology, see \cref{PropertiesOfOrbifoldCohomology} below in \cref{TwistedOrbifoldCohomology}):
\begin{proposition}
  \label[proposition]
    {GEquivariantCohomologyAsStackyMaps}
  For a pair of topological $G$-actions $G \acts X$ and
  $G \acts A$, the \emph{$G$-equivariant cohomology} of $X$ with coefficients in $A$ --- in the style as on the left hand side of \cref{NonabelianCohomologyRepresentedByBGAndbfBG} ---, is:
  \begin{equation}
    H_G(X,A)
    :=
    \pi_0
    \left\{
    \begin{tikzcd}[row sep=small]
      & 
      G \backsslash A
      \ar[d]
      \\
      G \backsslash X
      \ar[r]
      \ar[
        ur,
        dashed
      ]
      &
      \mathbf{B}G
    \end{tikzcd}
    \right\}
    \simeq
    \pi_0
    \left\{\!
    \adjustbox{raise=-7pt}{
    \begin{tikzcd}
      X
      \ar[
        in=60,
        out=180-60,
        looseness=4,
        "{
          \,\mathclap{G}\,
        }"{description}
      ]
      \ar[r, dashed]
      &
      A
      \ar[
        in=60,
        out=180-60,
        looseness=4,
        "{
          \,\mathclap{G}\,
        }"{description}
      ]
    \end{tikzcd}
    }
   \! \right\}
    \mathrlap{.}
  \end{equation}
\end{proposition}
\begin{proof}
  By \cref{EquivariantMapsAsSliceMapsOfHomotopyQuotients}.
\end{proof}

More generally, in twisted generalization of \cref{NonabelianCohomologyRepresentedByBGAndbfBG} in \cref{OrdinaryNonabelianCechCohomology}, we have:
\begin{definition}
  \label[definition]
    {BGTwistedCohomologyOfTopGrpd}
  In the situation of \cref{SliceMappingStack}, we say that the connected components of the space which is, by \cref{SliceManningSpaceIntoQuotientStackOverBG}, the slice mapping stack of \cref{SliceMappingStack}, is the \emph{$\tau$-twisted cohomology} of $\mathcal{X}$ with coefficients in $Y$:
  \begin{equation}
    \label{bfBGammaTwistedCohomology}
    H^\tau\big(
      \mathcal{X}
      ;\,
      Y
    \big)
    :=
    \Big[
      \shape
      \mathrm{Map}\big(
        \mathcal{X},
        \Gamma \backsslash Y
      \big)_{\mathbf{B}\Gamma}
    \Big]_0
    \underset{
      \scalebox{.7}{\cref{IsoClassesOfFundamentalGroupoid}}
    }{=}
    \pi_0
  \left\{
  \begin{tikzcd}[
    row sep=13pt, 
    column sep=27pt
  ]
    &
    \Gamma \backsslash Y
    \ar[d, "{\,p}"]
    \\
    \mathcal{X}
    \ar[
      ur,
      dashed,
    ]
    \ar[
      r,
      "{ \tau }"
    ]
    &
    \mathbf{B}\Gamma
  \end{tikzcd}
 \right\}
  \mathrlap{.}
  \end{equation}
\end{definition}
\begin{remark}
  When the twist $\tau$ in \eqref{bfBGammaTwistedCohomology} is trivial, in that it factors through the point, then this twisted cohomology \cref{bfBGammaTwistedCohomology} reduces, via \cref{FactoringQuotientSliceMapsOverBPhi,HomotopyFiberOfHoQuotientProjection}, to the ordinary nonabelian cohomology as on the left of \cref{NonabelianCohomologyRepresentedByBGAndbfBG}:
  \begin{equation}
    \left\{
    \begin{tikzcd}[
    row sep=13pt, 
    column sep=27pt
  ]
     &[-10pt]&[-10pt]
     \Gamma \backsslash Y
     \ar[
       d,
       "{\, p }"
     ]
      \\
      \mathcal{X}
      \ar[
        urr,
        dashed
      ]
      \ar[r]
      \ar[
        rr,
        downhorup,
        "{ \tau }"{description}
      ]
      &
      \ast
      \ar[r]
      &
      \mathbf{B}\Gamma
    \end{tikzcd}
    \right\}
    \simeq
    \left\{
    \begin{tikzcd}[
    row sep=13pt, 
    column sep=27pt
  ]
     &[-5pt]
     Y
     \ar[r]
     \ar[d]
     \ar[
       dr,
       phantom,
       "{ 
         \lrcorner 
       }"{pos=.2}
     ]
     &[-10pt]
     \Gamma \backsslash Y
     \ar[
       d,
       "{ \, p }"
     ]
      \\
      \mathcal{X}
      \ar[
        ur,
        dashed
      ]
      \ar[r]
      \ar[
        rr,
        downhorup,
        "{ \tau }"{description}
      ]
      &
      \ast
      \ar[r]
      &
      \mathbf{B}\Gamma
    \end{tikzcd}
    \right\}
    \simeq
    \big\{
    \begin{tikzcd}[sep=small]
      \mathcal{X}
      \ar[r, dashed]
      &
      Y
    \end{tikzcd}
    \big\}
    \mathrlap{\,.}
  \end{equation}
  In contrast, to obtain the twisted generalization of the right hand side of \cref{NonabelianCohomologyRepresentedByBGAndbfBG} requires considerably more structure, namely passage, via a \emph{twisted Elmendorf theorem} (\cite[Thm. 6.2.3]{SS25-Bun}), to a perspective (cf. \cite[\S 6.2.1]{SS25-Bun}) where topological groupoids are probed not just by Cartesian plots (as indicated in \cref{ProbingTopologicalGroupoidByRns}) but also by deloopings $\mathbf{B}K$ of all finite groups $K$. Here we shall not further dwell on this phenomenon, but see \parencites{SS25-Bun}{SS26-Orb}.
\end{remark}

\begin{lemma}
  \label[lemma]
    {InvarianceOfTwistedGroupoidCohomology}
  A (Morita) equivalence \cref{AMoritaEquivalence} of topological groupoids, induces 
  \begin{enumerate}
  \item
  a bijection between their equivalence classes of twists \textup{(\cref{SliceMappingStack})},
  \item
  an isomorphism in the correspondingly  twisted twisted cohomology \textup{(\cref{BGTwistedCohomologyOfTopGrpd})}:
  \begin{equation}
    \begin{tikzcd}[
      column sep=13pt,
      row sep=2pt
    ] 
      \mathcal{X}' 
       \ar[
         dr,
         "{
           \tau'
         }"{swap}
       ]
       \ar[
         rr, 
         "{ F }",
         "{ \sim }"{swap}
       ] 
       &&
       \mathcal{X}
       \ar[
         dl,
         "{ \tau }"
       ]
       \\
       & 
       \mathbf{B}\Gamma
    \end{tikzcd}
    \;\;\;\;\;
    \Rightarrow
    \;\;\;\;\;
    \begin{tikzcd}[sep=small] 
      H^{\tau'}\big(
        \mathcal{X}'
        ;\,
        Y
      \big)
       \ar[r, "{\sim}"]
       & 
      H^{\tau}\big(
        \mathcal{X}
        ;\,
        Y
      \big)
      \mathrlap{\,.}
    \end{tikzcd}
  \end{equation}
  \end{enumerate}
\end{lemma}
\begin{proof}
The first statement follows by \cref{MoritaEquivalenceOfMappingStack}.
With this, the second follows on general abstract grounds discussed in \cref{SmoothInfinityGroupoids}, along the following lines: In this commuting diagram,
\begin{equation}
  \begin{tikzcd}
    \mathrm{Map}\big(
      \mathcal{X},
      \Gamma \backsslash Y
    \big)_{\!\mathbf{B}\Gamma}
    \ar[
      r
    ]
    \ar[
      d
    ]
    \ar[
      dr,
      phantom,
      "{ 
        \lrcorner 
      }"{pos=.15}
    ]
    &
    \mathrm{Map}\big(
      \mathcal{X}
      ,
      \Gamma \backsslash Y
    \big)
    \ar[
      d,
      "{ p_\ast }"
    ]
    \ar[
      r, 
      "{ \sim }",
      "{ F^\ast }"{swap}
    ]
    &
    \mathrm{Map}\big(
      \mathcal{X}'
      ,
      \Gamma \backsslash Y
    \big)
    \ar[
      d,
      "{ p_\ast }"
    ]
    \\
    \ast
    \ar[
      r, 
      "{ 
        \widetilde{\tau} 
      }"{description}
    ]
    \ar[
      rr,
      downhorup,
      "{ 
        \widetilde{\tau'} 
      }"{description}
    ]
    &
    \mathrm{Map}\big(
      \mathcal{X}, 
      \mathbf{B}\Gamma
    \big)
    \ar[
      r,
      "{ \sim }",
      "{ F^\ast }"{swap}
    ]
    &
    \mathrm{Map}\big(
      \mathcal{X}', 
      \mathbf{B}\Gamma
    \big)
    \mathrlap{\,,}
  \end{tikzcd}
\end{equation}
the left square is a homotopy pullback by definition \cref{SliceMappingStackAsFiberProduct}, and the right square is because both its horizontal maps are equivalences, by \cref{MoritaEquivalenceOfMappingStack}. Therefore the total rectangle is a homotopy pullback by the homotopy pasting law. This implies an equivalence between 
$  \mathrm{Map}\big(
      \mathcal{X},
      \Gamma \backsslash Y
    \big)_{\!\mathbf{B}\Gamma}
$ 
and 
$  \mathrm{Map}\big(
      \mathcal{X}',
      \Gamma \backsslash Y
    \big)_{\!\mathbf{B}\Gamma},
$
and that implies the claim by \cref{EquivalentGroupoidsHaveHomeomorphismIsomorphismClasses}.
\end{proof}

\subsection{Orbifold Cohomology}
\label{OrbifoldCohomology}

We have seen in \cref{SlicingAndTwisting} that --- when seen from the perspective of the geometric homotopy theory of topological groupoids --- $G$-\emph{equivariant cohomology} (\cref{GEquivariantCohomologyAsStackyMaps}) is a form of \emph{twisted cohomology} (\cref{BGTwistedCohomologyOfTopGrpd}), namely with coefficient fibration a homotopy quotient projection $\begin{tikzcd}[sep=small]G \backsslash Y \ar[r] & \mathbf{B}G
\end{tikzcd}$ and twist \emph{also} such a projection, $\begin{tikzcd}[sep=small]G \backsslash X \ar[r] & \mathbf{B}G
\end{tikzcd}$. 
If we here just allow the twist to be a more general map of topological groupoids 
$\begin{tikzcd}[sep=small]\widehat{\mathcal{X}} \ar[r, "{ \tau }"] & \mathbf{B}G\end{tikzcd}$, then we immediately have a notion of twisted cohomology on topological groupoids, which restricts (by \cref{QuotientSliceMapsOverBPhi}) to $G'$-equivariant cohomology on all subgroupoids of $\mathcal{X}$ that look 
like global homotopy quotients!

In particular, this immediately gives a good notion of twisted cohomology of \emph{orbifolds} (see \cref{Orbifolds}), namely of topological groupoids that admit an open cover by Cartesian homotopy quotients $G' \sslash \mathbb{R}^n$, for varying finite groups $G'$.

\subsubsection{Orbifolds as Groupoids}
\label{OrbifoldsAsGroupoids}

With the language of topological and Lie groupoids (\cref{TopologicalGroupoid,LieGroupoidsAndSmoothGroupoids}) in hand (\cref{TopologicalGroupoidsAndStacks}), there is a beautifully transparent and powerful definition of \emph{orbifolds}: These are the \emph{proper {\'e}tale} groupoids (cf. \cref{Orbifolds}), meaning essentially that orbifolds are those groupoids that \emph{locally} look like homotopy quotients $G \backsslash \mathbb{D}^n$ (\cref{ActionGroupoid}) of finite groups $G$ acting continuously/smoothly on open disks $\mathbb{D}^n$, in generalization of how a manifold locally looks like a plain open disk $\mathbb{D}^n$.
This observation is due to \parencites[pp. 15]{MoerdijkPronk1997}[\S 4]{MoerdijkPronk1999}{Moerdijk2002} with further details spelled out in \parencites{Lerman2010}{Amenta2012}{Coufal2015}.

\begin{figure}[htb]
\caption{
\label{VisualizationOfOrbifoldCharts}
Like manifolds are spaces locally modeled on open balls, so \emph{orbifolds} (\cref{Orbifolds}) are ``higher spaces'' (groupoids) locally modeled on (homotopy) quotients of open balls by finite group actions. In the simple case of rigid rotation actions, these local (homotopy) quotients are \emph{cones} with a single orbi-singular point at their tip.
}
  
\centering

\adjustbox{
  rndfbox=5pt,
  scale=.9
}{
\begin{tikzpicture}

\node at (0,1.5) {\phantom{x}};

\begin{scope}[
  shift={(-4.2,0)}
]
\draw[
  line width=1.6pt,
  draw=black,
  fill=blue!20
]
  (0,0) circle
  (1.5);

\node at (0,-2.2) {
  \adjustbox{}{
    \def\arraystretch{.9}
    \begin{tabular}{c}
      local chart
      \\
      of a manifold
    \end{tabular}
  }
};

\end{scope}

\begin{scope}[
  shift={(-0,0)}
]
\draw[
  line width=1.6pt,
  draw=black,
  fill=blue!20
]
  (0,0) circle
  (1.5);

\draw[
  line width=.9,
  fill=red
] (0,0) circle (.14);

\draw[
  |-Latex,
  dashed
] 
  (30:1.2) arc (30:30+180:1.2);

\draw[
  |-Latex,
  dashed
] 
  (90:.7) arc 
  (90:90+180:.7);

\draw[
  |-Latex,
  dashed
]
  (+.18, -.1) .. controls
  (+.7, -.5) and
  (+.7, +.5) ..
  (+.18, +.1);

\node at (0,-2.2) {
  \adjustbox{}{
    \def\arraystretch{.9}
    \begin{tabular}{c}
      group action
      \\
      on the chart
    \end{tabular}
  }
};

\end{scope}

\begin{scope}[
  shift={(+4.2,0)}
]
\node at (0,0) {
  \includegraphics[width=4cm]{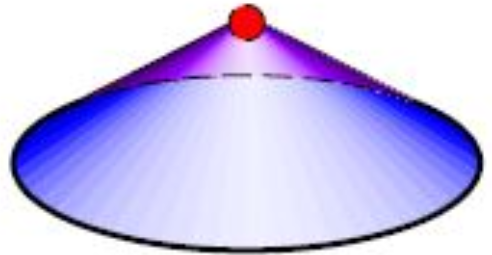}
};

\draw[
  |-Latex,
  dashed
] 
  (+.15, +1) .. controls
  (+.5,+1.5) and 
  (-.5,+1.5) ..
  (-.15, +1);

\node at (0,-2.2) {
  \adjustbox{}{
    \def\arraystretch{.9}
    \begin{tabular}{c}
      quotient chart
      \\
      of an orbifold
    \end{tabular}
  }
};
\end{scope} 
\end{tikzpicture}
}

\end{figure}
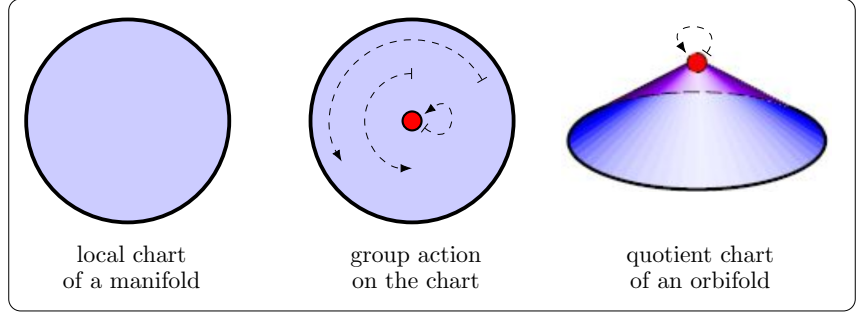

\newpage 
\begin{definition}
\label[definition]{Orbifolds}
A (topological/smooth) \emph{orbifold} is a topological/Lie groupoid $\mathcal{X}$ (\cref{TopologicalGroupoid,LieGroupoidsAndSmoothGroupoids}) such that:
\begin{enumerate}
\item
\emph{{\'E}tale}:
\begin{enumerate}
  \item
  the spaces of Objects and morphisms are equipped with the structure of topological/smooth manifolds,
  \item 
  the source or target map, and with it necessarily all the structure maps \cref{StructureOfATopologicalGroupoid}, are local homeomorphisms/diffeomorphisms.
\end{enumerate}
\item
\emph{Proper}:
the combined source/target map \cref{CombinedSourceAndTargetMap}
is \emph{proper}, which with the previous item means that this map preserves closed subsets and has compact fibers.
\end{enumerate}
The space $[\mathcal{X}]_0$ of isomorphism classes \cref{IsomorphismClassesOfTopGroupoid} of the groupoid $\mathcal{X}$ is called the underlying \emph{coarse space} of the orbifold. Conversely, given a (paracompact Hausdorff) topological space $X$ then an \emph{orbifold structure} on $X$ is a choice of proper {\'e}tale groupoid $\mathcal{X}$ and of a homeomorphism $\begin{tikzcd}[sep=small][\mathcal{X}]_0 \ar[r, "{\sim}"] & X \end{tikzcd}$. Proper {\'e}tale groupoids which are equivalent (\cref{EquivalenceOfTopGroupoids}) have homeomorphic coarse spaces \cref{EquivalentGroupoidsHaveHomeomorphismIsomorphismClasses} and are regarded as presenting the same orbifold structure.
\end{definition}

\begin{remark}
  Given an orbifold $\mathcal{X}$ (\cref{Orbifolds}) then for all $x \in \mathrm{Obj}(\mathcal{X})$ the isotropy groups $\mathcal{X}_x$\cref{IsotropyGroupOfTopGrpd} are discrete by {\'e}taleness and compact by properness of $(s,t)$, hence are \emph{finite groups}.
  Furthermore, properness of $(s,t)$ together with the manifoldness of $\mathrm{Obj}(\mathcal{X})$ implies that the space of isomorphism classes $[\mathcal{X}]_0$ \cref{IsomorphismClassesOfTopGroupoid} is paracompact and Hausdorff.
\end{remark}

\begin{example}[Global homotopy quotient orbifolds]
 \label[example]{GlobalQuotientOrbifolds}
 Consider $X$ a topological/smooth manifold, $G$ a Lie group and $G \acts X$ a continuous/smooth action (\cref{TopologicalGroupAction}). Then the homotopy quotient $G \backsslash X$ \cref{TheActionGroupoid} is an orbifold (\cref{Orbifolds}) under the following conditions of increasing generality:
 \begin{enumerate}
   \item 
     $G$ is finite:

     in this case $G \backsslash X$ is called a \emph{very good orbifold},
   \item
     $G$ is discrete acting properly discontinuously with finite stabilizers:

     in this case $G \backsslash X$ is called a \emph{good orbifold},
   \item
     $G$ is compact Lie, acting properly with finite stabilizers:

     $G \backsslash X$ of this form but not of the previous forms are called \emph{bad orbifolds}.
 \end{enumerate}
\end{example}

\begin{example}[Spindle orbifold]
  \label[example]
    {SpindleOrbifold}
  For $n_+, n_1 \in \mathbb{N}_{\geq 1}$ a pair of positive integers which are coprime, $\mathrm{gcd}(n_+, n_-) = 1$, then the \emph{$(n_+, n_-)$-spindle orbifold} (cf. \cref{SpindleOrbifold}) is the result of gluing the conical global quotients $\mathbb{Z}_{n_{\pm}} \backsslash \mathbb{D}^2_{1}$ (\cref{GlobalQuotientOrbifolds}), with respect to  $n_\pm$-fold rotation action on the disc, along a joint open annulus at their smooth ends. When $n_{\mp} = 1$ but $n_{\pm} > 1$ then the spindle reduces to what is called the \emph{teardrop orbifold}. When $n_{+} = 1 = n_i$ then the spindle orbifold reduces to the 2-sphere manifold.

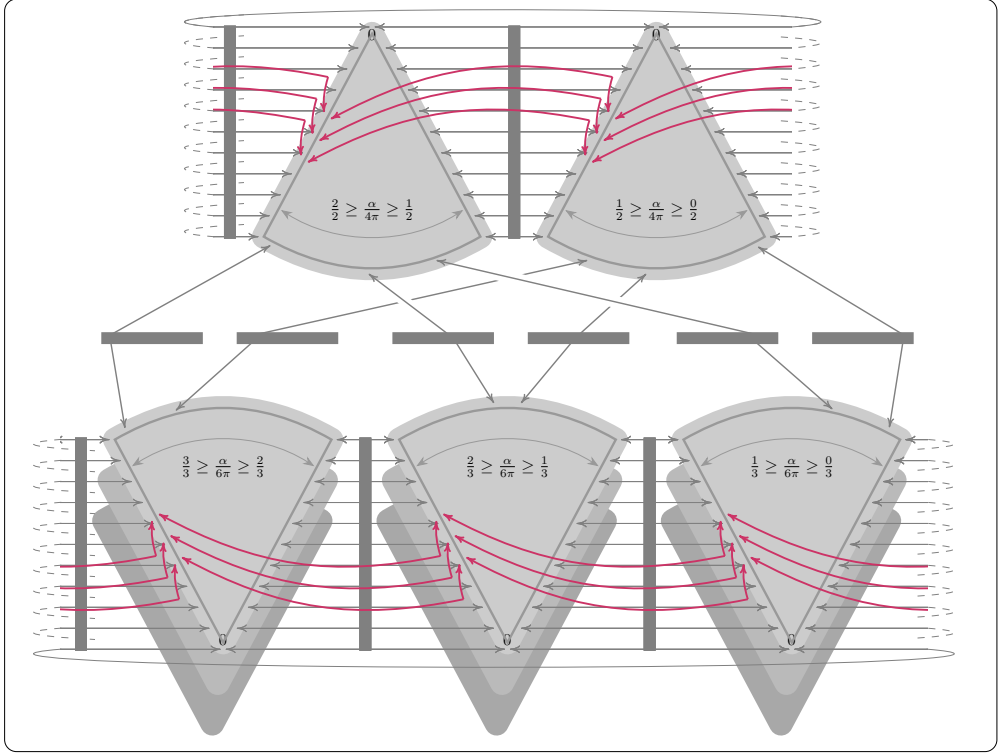
\begin{figure}[htb]
\caption{
  \label{GoodGroupoidFor23Spindle}
  Indicated is a proper {\'e}tale Lie groupoid which represents the $(2,3)$-spindle orbifold (\cref{SpindleOrbifold}) and is Dugger-cofibrant (\cref{CofibResolutionOfTopologicalGroupoid}).
  Shown in light gray is the space of objects, while the darker shades of gray are summands of the space of non-invertible morphisms.
  \\
  On the bottom we have the equivariant {\v C}ech groupoid resolution (\cref{EquivariantCechGroupoids}) of the action groupoid $\mathbb{Z}_3 \backsslash \mathbb{D}^2_{1+\epsilon}$ (from \cref{ASimpleEquivariantCechGroupoid}), and on the top the analogous equivariant {\v C}ech groupoid for $\mathbb{Z}_2 \backsslash \mathbb{D}^2_{1 + \epsilon}$. Indicated in the middle are (the component spaces of) further gluing morphisms which identify all the outer strips with each other and thereby glue the bottom conical orbifold of order 3 to the top conical orbifold of order 2, to form the $(2,3)$-spindle. To note that the composition of any two composable such gluing morphisms gives a translation morphism by either group action.
}

\centering

\adjustbox{
  rndfbox=5pt,
  scale=.9
}{
\hspace{-.8cm}
\begin{tikzpicture}

\begin{scope}[
  shift={(-2,+1.5)},
  yscale=-1
]

\begin{scope}
\clip
 (3.45,.8) rectangle 
 (12,-3.3);

\foreach \n [count=\np] in {0,1,2} {

\begin{scope}[
 shift={(10-\n*4.2,0)}
]
\filldraw[
    gray!40,
    line width=10pt,
    line join=round
  ]
    (-1.6,0) to[bend left=30]
    (+1.6,0) --
    (0,-3) -- cycle;
  \node[
    scale=.7
  ] at (0,-2.98) {$0$};
  \draw[
    gray!80,
    line width=1pt
  ]
    (-1.6,0) to[bend left=30]
    (+1.6,0) --
    (0,-3) -- cycle;
  \draw[
    shift={(0,-2.5)},
    gray!80,
    line width=.4,
    Stealth-Stealth
  ]
    (90-32.6:2.5) arc 
    (90-32.6:90+32.6:2.5);
  \node[scale=.7]
    at (0,-.4) {
      $
        \tfrac{\np}{2}
        \!\geq\!
        \frac{
          \alpha
        }{4\pi}
        \!\geq\!
        \tfrac{\n}{2}
      $
    };

\end{scope}

\foreach \n [count=\np] in {-1,0,1,2,3} {
\begin{scope}[
 shift={(10-\n*4.2,0)}
]
\begin{scope}[shift={(-2.1,0)}]
 \filldraw[
     black!50,
     line width=5pt,
     line join=round
   ]
     (0,.03) --
     (0,-3.12)
      -- cycle;
 
 \foreach \k in {0,...,10} {
   \draw[
     {<[length=4pt,width=3pt]}-{>[length=4pt,width=3pt]},
     black!50,
     line width=.2
   ]
     ({-\k*.164-.42},{-\k*.31}) --
     ({+\k*.164+.42},{-\k*.31});
 }
\end{scope}
\end{scope}

}}

\foreach \n in {-1,0,1,2} {
\begin{scope}[
 shift={(10-\n*4.2,0)}
]
\foreach \k in {3,...,5} {
   \draw[
     shift={(-\k*.17,\k*.32)},
     -{>[length=4pt,width=3pt]},
     purple!80,
     line width=.8
   ]
  (-.14,-3.32) to[bend left=20] 
  (-4.3,-2.7);
   \draw[
     shift={(-\k*.17,\k*.32)},
     -{>[length=4pt,width=3pt]},
     purple!80,
     line width=.8
   ]
  (-.14,-3.32) to[bend left=10] 
  (-.2,-2.8);
}
\end{scope}
}

\end{scope}

\draw[
   shift={(7.73,-3.175)},
   black!50,
   line width=.4,
]
  (90+65:4.7 and .18) arc
  (90+65:360+90-65:4.7 and .18);

\begin{scope}
\clip
  (2.9,.2)  rectangle
  (3.9,-3.4);

\foreach \n in {0,...,10} {
\draw[
   shift={(7.73,-3.175+\n*.31)},
   black!50,
   line width=.4,
   dashed
]
  (90+65:4.7 and .18) arc
  (90+65:360+90-65:4.7 and .18);
}
\end{scope}

\begin{scope}
\clip
  (11.8,.2)  rectangle
  (12.5,-3.4);

\foreach \n in {0,...,10} {
\draw[
   shift={(7.73,-3.175+\n*.31)},
   black!50,
   line width=.4,
   dashed
]
  (90+65:4.7 and .18) arc
  (90+65:360+90-65:4.7 and .18);
}
\end{scope}

\end{scope}


\begin{scope}[
  shift={(0,-1.5)}
]

\begin{scope}
\clip
 (-.8,.8) rectangle 
 (12,-4.5);

\foreach \n [count=\np] in {0,1,2,3} {

\begin{scope}[
 shift={(10-\n*4.2,0)}
]

\filldraw[
    gray!68,
    line width=10pt,
    line join=round,
    shift={(-.16,-1.2)}
  ]
    (-1.6,0) to[bend left=30]
    (+1.6,0) --
    (0,-3) -- cycle;
\filldraw[
    gray!54,
    line width=10pt,
    line join=round,
    shift={(-.08,-.6)}
  ]
    (-1.6,0) to[bend left=30]
    (+1.6,0) --
    (0,-3) -- cycle;

\filldraw[
    gray!40,
    line width=10pt,
    line join=round
  ]
    (-1.6,0) to[bend left=30]
    (+1.6,0) --
    (0,-3) -- cycle;
  \node[
    scale=.7
  ] at (0,-2.95) {$0$};
  \draw[
    gray!80,
    line width=1pt
  ]
    (-1.6,0) to[bend left=30]
    (+1.6,0) --
    (0,-3) -- cycle;
  \draw[
    shift={(0,-2.5)},
    gray!80,
    line width=.4,
    Stealth-Stealth
  ]
    (90-32.6:2.5) arc 
    (90-32.6:90+32.6:2.5);
  \node[scale=.7]
    at (0,-.4) {
      $
        \tfrac{\np}{3}
        \!\geq\!
        \frac{
          \alpha
        }{6\pi}
        \!\geq\!
        \tfrac{\n}{3}
      $
    };

\end{scope}

\foreach \n [count=\np] in {-1,0,1,2,3} {
\begin{scope}[
 shift={(10-\n*4.2,0)}
]
\begin{scope}[shift={(-2.1,0)}]
 \filldraw[
     black!50,
     line width=5pt,
     line join=round
   ]
     (0,.03) --
     (0,-3.12)
      -- cycle;
 
 \foreach \k in {0,...,10} {
   \draw[
     {<[length=4pt,width=3pt]}-{>[length=4pt,width=3pt]},
     black!50,
     line width=.2
   ]
     ({-\k*.164-.42},{-\k*.31}) --
     ({+\k*.164+.42},{-\k*.31});
 }
\end{scope}
\end{scope}

}}

\foreach \n [count=\np] in {-1,0,1,2} {
\begin{scope}[
 shift={(10-\n*4.2,0)}
]
\foreach \k in {3,...,5} {
   \draw[
     shift={(-\k*.17,\k*.32)},
     -{>[length=4pt,width=3pt]},
     purple!80,
     line width=.8
   ]
  (-.14,-3.32) to[bend left=20] 
  (-4.3,-2.7);
   \draw[
     shift={(-\k*.17,\k*.32)},
     -{>[length=4pt,width=3pt]},
     purple!80,
     line width=.8
   ]
  (-.14,-3.32) to[bend left=10] 
  (-.2,-2.8);
}
\end{scope}
}

\end{scope}

\draw[
   shift={(5.6,-3.166)},
   black!50,
   line width=.4
]
  (90+70:6.8 and .2) arc
  (90+70:360+90-70:6.8 and .2);

\begin{scope}
\clip
  (-1.3,.2)  rectangle
  (-.3,-3.4);

\foreach \n in {0,...,9} {
\draw[
   shift={(5.6,-2.855+\n*.31)},
   black!50,
   line width=.2,
   dashed
]
  (90+70:6.8 and .2) arc
  (90+70:360+90-70:6.8 and .2);
}
\end{scope}

\begin{scope}
\clip[
  shift={(13.1,0)}
]
  (-1.3,.2)  rectangle
  (-.3,-3.4);

\foreach \n in {0,...,9} {
\draw[
   shift={(5.6,-2.855+\n*.31)},
   black!50,
   line width=.2,
   dashed
]
  (90+70:6.8 and .2) arc
  (90+70:360+90-70:6.8 and .2);
}
\end{scope}

\end{scope}

\begin{scope}[
  shift={(-4.2,0)}
]
 \filldraw[
     black!50,
     line width=5pt,
     line join=round
   ]
     (4,0) --
     (5.5,0)
      -- cycle;
 \filldraw[
     black!50,
     line width=5pt,
     line join=round
   ]
     (6,0) --
     (7.5,0)
      -- cycle;
\end{scope}

\begin{scope}[
  shift={(.1,0)}
]
 \filldraw[
     black!50,
     line width=5pt,
     line join=round
   ]
     (4,0) --
     (5.5,0)
      -- cycle;
 \filldraw[
     black!50,
     line width=5pt,
     line join=round
   ]
     (6,0) --
     (7.5,0)
      -- cycle;
\end{scope}

\begin{scope}[
  shift={(+4.3,0)}
]
 \filldraw[
     black!50,
     line width=5pt,
     line join=round
   ]
     (4,0) --
     (5.5,0)
      -- cycle;
 \filldraw[
     black!50,
     line width=5pt,
     line join=round
   ]
     (6,0) --
     (7.5,0)
      -- cycle;
\end{scope}

  \draw[
    -{>[length=4pt,width=3pt]},
     black!50,
     line width=.6
   ]
   (-.15,0) --
   (2.3,1.37);
  \draw[
    -{>[length=4pt,width=3pt]},
     black!50,
     line width=.6
   ]
   (-.07,0) --
   (.15,-1.33);

  \draw[
    -{>[length=4pt,width=3pt]},
     black!50,
     line width=.6
   ]
   (1.85,0) --
   (6.98,1.15);
  \draw[
    -{>[length=4pt,width=3pt]},
     black!50,
     line width=.6
   ]
   (2.1,0) --
   (.9,-1.07);

  \begin{scope}
  \clip
    (4,.4) rectangle 
    (4.5,.7);
  \draw[
     line width=4.5,
     white
   ]
   (5,0) --
   (3.75,.95);
  \end{scope}
  \draw[
    -{>[length=4pt,width=3pt]},
     black!50,
     line width=.6
   ]
   (5,0) --
   (3.75,.95);  
  \draw[
    -{>[length=4pt,width=3pt]},
     black!50,
     line width=.6
   ]
   (5,0) --
   (5.6,-.96);

  \draw[
    -{>[length=4pt,width=3pt]},
     black!50,
     line width=.6
   ]
   (6.8,0) --
   (7.85,.95);
  \draw[
    -{>[length=4pt,width=3pt]},
     black!50,
     line width=.6
   ]
   (6.8,0) --
   (6,-.96);

  \begin{scope}
  \clip
    (5.6,.3) rectangle
    (8,1.4);
  \draw[
     white,
     line width=5
   ]
   (9.7,0) --
   (4.75,1.15);
  \end{scope}
  \draw[
    -{>[length=4pt,width=3pt]},
     black!50,
     line width=.6
   ]
   (9.7,0) --
   (4.75,1.15);
  \draw[
    -{>[length=4pt,width=3pt]},
     black!50,
     line width=.6
   ]
   (9.4,0) --
   (10.6,-1.05);

  \draw[
    -{>[length=4pt,width=3pt]},
     black!50,
     line width=.6
   ]
   (11.73,0) --
   (9.5,1.35);
  \draw[
    -{>[length=4pt,width=3pt]},
     black!50,
     line width=.6
   ]
   (11.65,0) --
   (11.44,-1.32);

\end{tikzpicture}
}

\end{figure}

  We indicate two explicit proper {\'e}tale groupoids representing the spindle orbifold:
  \begin{enumerate}
  \item
  A minimal model for the $(n_+, n_-)$-spindle has 
  \begin{equation}
    \begin{aligned}
      \mathrm{Obj}
       := &
      \mathbb{C}_+ 
      \sqcup 
      \mathbb{C}_-
      \\
      \mathrm{Mor}
      := &
      \big(
        \mathbb{C}_+ \times \mathbb{Z}_{n_+}
        \;\sqcup\;
        \mathbb{C}_+ \times \mathbb{Z}_{n_+}
      \big)
      \\
      &
       \sqcup\;
      \Big(
        \big(
          \mathbb{C}^\times_{+-}
          \times \mathbb{Z}_{n_+}
          \times \mathbb{Z}_{n_-}
        \big)
        \sqcup
        \big(
          \mathbb{C}^\times_{-+}
          \times \mathbb{Z}_{n_-}
          \times \mathbb{Z}_{n_+}
        \big)
      \big)
    \end{aligned}
  \end{equation}
  (where $\mathbb{C}^\times := \mathbb{C} \setminus \{0\}$, and the subscripts on $\mathbb{C}$ are just to index the disjoint summands) and, with the abbreviation
  \begin{equation}
    q_{\pm}
    :=
    e^{2 \pi \mathrm{i}/ n_{\pm}}
  \end{equation}
  the source/target maps for the internal and for the weightless gluing morphisms are
  \begin{equation}
    \begin{tikzcd}[row sep=-3pt, column sep=0pt]
      \mathrm{Mor}
      \ar[
        rr,
        "{ (s,t) }"
      ]
      &&
      \mathrm{Obj}^2
      \\
      \big(
        (z,\pm), [k_{\pm}]
      \big)
       &\longmapsto&
       \big(
         (z, \pm),
         (z q_{\pm}^{k_{\pm}}, \pm)
       \big)
       \\
       \big(
         (v, \pm \mp),
         [0], [0]
       \big)
       &\longmapsto&
       \big(
         (v^{n_{\pm}}, \pm),
         (v^{-n_{\mp}}, \mp)
       \big)\,.
    \end{tikzcd}
  \end{equation}
  Finally, the general gluing morphism is the unique composite
  \begin{equation}
    \begin{aligned}
      &
      \big(
        (v,\pm\mp), [k_\pm], [k_\mp]
      \big)
      \\
      &
      \;\equiv\;
      \big(
        (v^{-n_\mp},\mp), [k_\mp])
      \big)
      \circ
      \big(
        (v,\pm,\mp), [0], [0]
      \big)
      \circ 
      \big(
        (v^{n_\pm} \cdot q^{-k_{\pm}}, \pm), [k_\pm]
      \big)
      \,,
    \end{aligned}
  \end{equation}    
    which defines either side by the other (meaning that the gluing morphisms are a groupoid torsor from either side over these action groupoids) and its source/target map is:
    \begin{equation}
      (s,t)
      \;\colon\;
      \big(
        (v,\pm\mp), [k_\pm], [k_\mp]
      \big)
      \mapsto
      \big(
        (v^{-n_{\mp}} \cdot q^{+ k_{\mp}}, \pm)
        ,
        (v^{n_{\pm}} \cdot q^{- k_{\pm}}, \mp)
      \big)
      \mathrlap{\,.}
    \end{equation}

  In the special case of $n_\pm = 1$ this definition reduces to the {\v C}ech groupoid (\cref{CechGroupoids}) corresponding to the standard open cover of the Riemann sphere by two copies $\mathbb{C}_{\pm}$ with double overlap $\mathbb{C}^\times \subset \mathbb{C}_+$ identified inside the second copy via $z \mapsto z^{-1}$.

  \item
  For the purpose of computing mapping stacks (\cref{MappingStacksAndNonabelianCohomology}) out of the spindle orbifold, we want a  presentation by a Dugger-cofibrant groupoid (\cref{CofibResolutionOfTopologicalGroupoid}). But the above minimal model fails to be Dugger-cofibrant since the factor $\mathbb{C}^\times$ of the space of gluing morphisms is (not contractible and hence) not diffeomorphic to $\mathbb{R}^2$.

  One obtains a Dugger-cofibrant spindle by considering the disjoint union of the Dugger-cofibrant {\v C}ech groupoids for the two cones $\mathbb{Z}_{n_{\pm}} \backsslash \mathbb{D}^2_{1+\epsilon}$ from \cref{ASimpleEquivariantCechGroupoid} and glueing the collars of all their segments to each other by a summand-wise contractible space of gluing morphisms. 
  
  This construction is indicated in \cref{GoodGroupoidFor23Spindle} for the case $(n_+, n_-) = (3,2)$.
  \end{enumerate}
\end{example}

\subsubsection{Twisted Orbifold Cohomology}
\label{OnTwistedOrbifoldCohomology}

We saw in \cref{SlicingAndTwisting} a fairly general definition of (nonabelian) twisted cohomology of topological groupoids, and in \cref{OrbifoldsAsGroupoids} that orbifolds are a special case of topological groupoids. Therefore we immediately obtain a notion of \emph{twisted orbifold cohomology}, inheriting all the good properties that we saw hold in general for twisted nonabelian cohomolgy formulated via mapping stacks.

Therefore, in discussing twisted orbifold cohomology here we are essentially reduced to summarizing previous statements on the twisted nonabelian cohomolgy of general topological groupoids --- which may serve as a concise review of the above discussion and to highlight how phenomena discussed in the literature by other means are elegantly reproduced by our general approach to cohomology via mapping stacks.

This also means that, at this level of plain (twisted nonabelian) cohomology, the actual \emph{geometry} of orbifolds --- their smooth structure and possibly further geometric structure such as Riemannian, holomorphic, conformal structure, ... --- plays no role. The refinement to \emph{geometric orbifold cohomology} which \emph{is} sensitive to further geometric structure is discussed in \cite{SS26-Orb}.

\medskip

First, in specialization of \cref{BGTwistedCohomologyOfTopGrpd} we have:
\begin{definition}[Twisted orbifold cohomology]
  \label[definition]
    {TwistedOrbifoldCohomology}
  For $\Gamma \acts Y$ a topological group action (\cref{TopologicalGroupAction}) on a topological space $Y$ (to play the role of the \emph{classifying space} for the orbifold cohomology theory), then the \emph{twisted orbifold cohomology} with coeffcients in $Y$, 
  of an orbifold $\mathcal{X}$ equipped with a map (the twist) $\begin{tikzcd}[sep=small] \mathcal{X} \ar[r, "{\tau}"] & \mathbf{B}\Gamma\end{tikzcd}$, is the $\tau$-twisted cohomology of the orbifold groupoid $\mathcal{X}$ according to \cref{BGTwistedCohomologyOfTopGrpd}, hence is
  the connected components of the slice mapping stack (\cref{SliceMappingStack}, which is just a space, by \cref{SliceManningSpaceIntoQuotientStackOverBG}):
  \begin{equation}
    \label{TwistedOrbifoldCohomologyViaMaps}
    H^\tau(\mathcal{X};Y)
    :=
    \Big[
      \shape\,
      \mathrm{Map}\big(
        \widehat{X}
        ,\,
        \Gamma \backsslash Y
      \big)_{\!\mathbf{B}\Gamma}
    \Big]_0
    =
    \pi_0
    \left\{
    \begin{tikzcd}[row sep=small]
      & 
      \Gamma \backsslash Y
      \ar[
        d,
        "{\, p }"
      ]
      \\
      \mathcal{X}
      \ar[
        ur,
        dashed
      ]
      \ar[
        r,
        "{ \tau }"
      ]
      &
      \mathbf{B}\Gamma
    \end{tikzcd}
    \right\}
    \mathrlap{.}
  \end{equation}
\end{definition}

\begin{theorem}[Properties]
  \label[theorem]
  {PropertiesOfOrbifoldCohomology}
Twisted orbifold cohomology \textup{(\cref{TwistedOrbifoldCohomology})} has the following properties:

\begin{enumerate}
  \item
  It is an invariant of the (Morita) equivalence class of the orbifold.

  In particular, if $G \backsslash X$ and $G' \backsslash X'$ are two global quotient presentations \textup{(\cref{GlobalQuotientOrbifolds})} of the same orbifold class, then their orbifold cohomology agrees, and their twisted cohomology agrees for corresponding twists.

  \item If on a global quotient presentation \textup{(\cref{GlobalQuotientOrbifolds})} the twist is just the canonical projection $\begin{tikzcd}[sep=small] G \backsslash X \ar[r, "{ p }"] & \mathbf{B}\Gamma\end{tikzcd}$, then twisted orbifold cohomology reduces to the $G$-equivariant cohomology of $X$. 
\end{enumerate}
\end{theorem}
\begin{proof}
  The first statement is \cref{InvarianceOfTwistedGroupoidCohomology}, the second is \cref{GEquivariantCohomologyAsStackyMaps}.
\end{proof}

\begin{remark}
  The basic properties in \cref{PropertiesOfOrbifoldCohomology} have a somewhat more convoluted history in the literature on orbifold cohomology by other means (reviewed in \cite[\S 1.1]{SS26-Orb}) than the mapping stack approach used here, for more discussion of of which see also \parencites[Rem. 6.2.4]{SS25-Bun}{Sc25-ItaCa}.
\end{remark}

\subsection{%
  \texorpdfstring%
  {Embedding into Smooth $\infty$-Groupoids}
  {Embedding into Smooth infinity-Groupoids}
}
\label{SmoothInfinityGroupoids}

We indicate how the above homotopy theory of topological groupoids is a fragment of  the cohesive homotopy theory of \emph{smooth $\infty$-groupoids} (``$\infty$-stacks over smooth manifolds''), which also serves to neatly justify/prove key properties of topological groupoids/stacks that we invoked above.

We will be very brief here, providing a digest of what is spelled out in detail in \parencites[\S 1]{FSS23-Char}[\S 4]{SS25-Bun}{SS26-Orb}, going back to \parencites{SSS12}{Sc13-dcct}. For this purpose we now switch to  category-theoretic language (for introduction in our context see \cite{Sc18-ToposLectures} and for exposition see \cite{Schreiber2025}).

We write
\begin{equation}
  \label{CategoryOfCartesianSpaces}
  \mathrm{CartSp}
  =
  \Bigg\{\,
  \begin{tikzcd}[row sep=5pt,
    column sep=15pt
  ]
    & 
    \mathbb{R}^{n_1}
    \ar[
      dr,
      "{ g }"
    ]
    \\
    \mathbb{R}^{n_0}
    \ar[
      ur,
      "{ f }"
    ]
    \ar[
      rr,
      "{ g \circ f }"
    ]
    &&
    \mathbb{R}^{n_2}
  \end{tikzcd}
  \Bigg\}
\end{equation}
for the category whose objects are the Cartesian spaces $\mathbb{R}^n$, $n \in \mathbb{N}$ \cref{CartesianSpaceHomeomorphicToOpenBall}, and whose morphisms are the \emph{smooth} functions between these. We regard this as a site with respect to the coverage (Grothendieck pre-topology) of differentially good open covers (\cref{GoodOpenCover}). The sheaves on this site we call \emph{smooth sets} (\parencites[\S 1.2.2]{Sc13-dcct}[Def. 2.1]{KS17-PDEs}[Ntn. 4.3.15]{SS25-Bun}); these faithfully subsume (cf. \parencites[Prop. 4.3.19]{SS25-Bun}) smooth manifolds, D-topological spaces (\cref{DTopologicalSpace}), and diffeological spaces:
\begin{equation}
  \begin{tikzcd}[
    row sep=-6pt
  ]
    \mathrm{DTopSp}
    \ar[dr, hook]
    \\
    &
    \mathrm{DiflSp}
    \ar[r, hook]
    &
    \mathrm{SmthSet}
    :=
    \mathrm{Sh}(\mathrm{CartSp})
    \mathrlap{\,.}
    \\
    \mathrm{SmthMfd}
    \ar[ur, hook']
  \end{tikzcd}
\end{equation}

The higher groupoidal/homotopy theoretic version of this sheaf topos is the (hypercomplete) $\infty$-sheaf topos over $\mathrm{CartSp}$, whose objects we call \emph{smooth $\infty$-groupoids} (\parencites[\S 4.4]{Sc13-dcct}[Ntn. 4.3.27]{SS25-Bun}[Ex. 4.1.9]{SS26-Orb}, or \emph{smooth $\infty$-stacks}) and which faithfully subsume the Morita (stack) theory of D-topological groupoids (\cref{TopologicalGroupAction}), Lie groupoids and diffeological groupoids (\cref{LieGroupoidsAndSmoothGroupoids}):
\begin{equation}
  \begin{tikzcd}[
    row sep=-6pt
  ]
    \mathrm{DTopGrpd}
    \ar[dr, hook]
    \\
    &
    \mathrm{DiflGrpd}
    \ar[r, hook]
    &
    \mathrm{SmthGrpd}_\infty
    :=
    \mathrm{Sh}_\infty(\mathrm{CartSp})
    \mathrlap{\,.}
    \\
    \mathrm{LieGrpd}
    \ar[ur, hook']
  \end{tikzcd}
\end{equation}

More concretely, $\mathrm{SmthGrpd}_\infty$ is equivalently the $\infty$-category presented by the projective model structure on simplicial presheaves \cite{nLab:ModelStrucOnSimpPresheaves} over $\mathrm{CartSp}$ \cref{CategoryOfCartesianSpaces}, left Bousfield-localized at the class $W$ of stalkwise simplicial weak equivalences (cf. \parencites[Ex. 1.20]{FSS23-Char}):
\begin{equation}
  \label{SmoothGrpdInftyAsSimpLocalization}
  \mathrm{SmthGrpd}_\infty
  \simeq
  L^W
  \mathrm{PSh}(
    \mathrm{CartSp}
    \times 
    \Delta
  )
  \mathrlap{\,.}
\end{equation}

A D-topological groupoid  is represented here by (cf. \parencites[Ntn. 2.2.24]{SS25-Bun}) the simplicial presheaf which to $\mathbb{R}^n$ assigns the \emph{simplicial nerve} of its underlying discrete groupoid $\flat \mathrm{Map}(\mathbb{R}^n, \mathcal{X})$ of $\mathbb{R}^n$-plots (\cref{ProbingTopologicalGroupoidByRns,GeometricallyDiscreteUnderlyingGroupoid}), hence to $\mathbb{R}^n \times \Delta^k$ the set of $\mathbb{R}^n$-plots of its space of $k$-tuples of sequentially composable morphisms:
\begin{equation}
  \label{TopGrpdAsSimpPresheaf}
  \begin{tikzcd}[row sep=-3pt, 
    column sep=0pt
  ]
    \mathrm{DTopGrpd}
    \ar[
      rr,
      hook,
      "{ N }"
    ]
    &&
    \mathrm{PSh}(
      \mathrm{CartSp}\times\Delta
    )
    \\
    \mathcal{X}
    &\longmapsto&
    \Big(
      \mathbb{R}^n
        \times
      \Delta^k
      \,\mapsto\,
      \flat \mathrm{Map}\big(
        \mathbb{R}^n
        ,
        \mathrm{Mor}(\mathcal{X})^{
          \tensor[_s]{{\times^{\mathrlap{k}}}}{_t}
        }
      \big)
    \Big),
  \end{tikzcd}
\end{equation}
and this constitutes a fully faithful embedding of 1-categories with naive (not Morita) morphisms on both sides.
Under this embedding \cref{TopGrpdAsSimpPresheaf}:
\begin{itemize}
\item The internal hom is given by the topological functor groupoid \cref{FunctorGroupoid}.
\item The weak equivalences in $W$ \cref{SmoothGrpdInftyAsSimpLocalization} are just the equivalences of topological groupoids according to \cref{EquivalenceOfTopGroupoids}.
\item The global fibrations are just the global fibrations of topological groupoids according to \cref{GlobalFibrations}.
\end{itemize}

Moreover, the key point for our purpose then is that:

\newpage 
\begin{fact}
\label[fact]
 {ComputingMappingStacksInSmplPSh}
With topological groupoids regarded among simplicial presheaves via \cref{TopGrpdAsSimpPresheaf}:
\begin{enumerate}
\item Their (correctly ``derived'') mapping stack is computed (by general model category theory, cf. \parencites[Ex. 1.10]{FSS23-Char}) as the functor groupoid
\begin{enumerate}
\item
out of a global projective cofibrant resolution,
\item
into a local projective fibrant resolution.
\end{enumerate}
\item 
The nerves of 
\begin{enumerate}
  \item
  Dugger-cofibrant topological groupoids according to \cref{CofibResolutionOfTopologicalGroupoid} are global projective cofibrant (cf. \parencites[Cor. 9.4]{Dugger2001}[Prop. 1.23]{FSS23-Char}),
  \item
  delooping groupoids $\mathbf{B}\Gamma$ and with them the action groupoids $G \backsslash \Gamma$ (\cref{ActionGroupoid}) are local projective fibrant (by \parencites[Prop. 4.13]{Pavlov2022}[Lem. 4.3.30]{SS25-Bun})
\end{enumerate}
over $\mathrm{CartSp}$.
\end{enumerate}
\end{fact}

Together, these imply that the (sliced) mapping stacks according to 
\cref{MappingStackFromDuggerCofibrantToBGamma,SliceMappingStack} represent, under \cref{TopGrpdAsSimpPresheaf}, the derived internal hom --- which proves \cref{MoritaEquivalenceOfMappingStack,InvarianceOfTwistedGroupoidCohomology} and thereby ultimately the final \cref{PropertiesOfOrbifoldCohomology} about Morita invariance of twisted orbifold cohomology.



\section{Twisted Orbi K-Theory and Unstable Orientation}
\label{OrientationsInOrbiKTheory}

We describe in \cref{OrbifoldKTheory} an elegant model of \emph{twisted orbifold K-theory} based on the general geometric homotopy formulation of twisted orbifold cohomology (from \cref{TwistedOrbifoldCohomology}) via mapping stacks between topological groupoids (\cref{SomeCohesiveHomotopyTheory}).

Then we explicitly construct (\cref{TheEquivariantOrientation}) in this model the low-dimensional equivariant orientation character map, in specialization of the general discussion in \cref{OrientationsMeasuringRelativeCharges}. Here for orbifold K-theory this turns out to be   orified incarnation of the equivariant trivialization of the pullback of the tautological line bundle along the Hopf fibration, which we review in \cref{OnTwistedOrbifoldKTheory}.

\subsection{Twisted Orbifold K-theory via Geometric Homotopy}
\label{OrbifoldKTheory}

We recall and expand on the elegant model of \emph{twisted orbifold K-theory} from \parencites[\S 2.2]{SS22-Ord}[Ex. 6.2.5]{SS25-Bun}, complementing traditionally more component-based definitions in the literature (cf. \parencites{AdemRuan2003}{LupercioUribe2004}{TuXuLG2004}{Jarvis2006}{AdemLeidaRuan2007}{FreedHopkinsTeleman2011}{Gomi2017}). Our model is based on the picture (surveyed in \cref{OverviewChargesInCohomology,TableOfNotionsOfCohomology}, reviewed in \cref{MappingStacksAndNonabelianCohomology,SlicingAndTwisting}) of generalized cohomology as being about \emph{homotopy classes of classifying maps}: in the present case between topological groupoids/stacks (reviewed in \cref{TopologicalGroupoidsAndStacks}) from the domain orbifold $\mathcal{X}$ to the universal $\PUH$-structured Fredholm space bundle that classifies twisted equivariant K-theory according to \cite{AtiyahSegal2004}, but: 
\begin{enumerate}
\item
with the domain $\mathcal{X}$ regarded as the {\'e}tale groupoid incarnation of orbifolds, cf. \cref{Orbifolds}, 

\item
with the classifying fibration promoted \cref{TheUniversalQSAssociatedFredBundle}
from a topological bundle over the ordinary \emph{classifying space} $B \PUH$ to a stacky universal fiber bundle (\cref{PrincipalBundleAsPullbackAlongCechCocycle}) over the topological \emph{moduli stack} $\mathbf{B} \PUH \simeq \ast \sslash \PUH$ (cf. \cref{DeloopingGroupoid,EquivariantCechGroupoids}), in fact over its further extension by PCT symmetries (\cref{TheQSGroup}) to the moduli stack $\mathbf{B}\QuantumSymmetries$ of \emph{quantum symmetries} (\cref{QuantumSymmetries}).
\end{enumerate}
This construction utilizes that, as homotopy types, the space of Fredholm operators is a classifying space of complex K-theory (\cref{AtiyahJaenichTheorem})
but, as an actual topological space, it moreover carries a strict topological group action of quantum symmetries $\begin{tikzcd}[sep=small]\QuantumSymmetries \ar[r, ->>] & \PUH\end{tikzcd}$,  
$\QuantumSymmetries \, \acts \, \GradedFredholmOperators$,
which implements (\cref{The10PCTFixedLoci}, following \cite[\S 2.2]{SS22-Ord}  inspired by \cite{FreedMoore2013}) the grading and the flavors of topological K-theory ($\mathrm{KU}$, $\mathrm{KO}$, $\mathrm{KR}$).  

We see in \cref{TheEquivariantOrientation} how the model lends itself neatly to the construction of the four/ten-dimensional equivariant $\mathbb{C}/\mathbb{H}$-orientations of $\mathrm{KU}$ (\cref{MeasuringRelativeChargesInQCohom}), where we obtain the relevant Fredholm operators and unitary operators essentially ``tautologically'' from the $\mathbb{C}/\mathbb{H}$-operator algebraic construction of the $\mathbb{C}/\mathbb{H}$-Hopf fibration (discussed in \cref{TheEquivariantLineBundle}).

\subsubsection{Hilbert Space and PCT}
\label{TheHilbertSpaceAndPCT}

\begin{notation}[The Hilbert space]
In all of the following, we consider:
\begin{enumerate}
\item 
A countably infinite-dimensional complex Hilbert space
\begin{equation}
  \label{TheHilbertSpace}
  \HilbertSpace
  \in 
  \HilbertSpaces_{\mathbb{C}}
  \mathrlap{\,,}
\end{equation}
as such unique up to isomorphism; in particular isomorphic to the direct sum, as Hilbert spaces, of countably many copies of the standard $k$-dimensional complex vector space:
\begin{equation}
  \label{HilbertSpaceAsDirectSumOfComplexVectorSpaces}
  \forall{k \in \mathbb{N}_{\geq 0}}
  \;:
  \;\;
  \HilbertSpace
  \simeq
  \bigoplus_{n \in \mathbb{N}}
  \mathbb{C}^k
  \,.
\end{equation}

\item
A \emph{real structure} on $\HilbertSpace$ \cref{TheHilbertSpace}, hence a \emph{complex anti-linear} involution and as such self-adjoint (cf. \cite[Def. 3.1]{Uhlmann2016})
\begin{equation}
  \label{TheRealStructure}
  \begin{tikzcd}[row sep=-3pt, column sep=0pt]
    \HilbertSpace
    \ar[
      rr,
      "{
        \tOperator
      }"
    ]
    &&
    \HilbertSpace
    \mathrlap{\,,}
    \\
    \psi 
      &\mapsto&
    \tOperator\psi
    \\
    \mathrm{i}\psi
      &\mapsto&
    - \mathrm{i}\tOperator{\psi}
  \end{tikzcd}
  \hspace{1cm}
  \begin{aligned}
    &
    \tOperator^2
     = 
    \mathrm{id}
    \\
    &
    \langle 
      -
      ,\, 
      \tOperator 
      -
    \rangle
    =
    \overline{
    \langle 
      \tOperator 
      -
      ,\,
      -
    \rangle
    }
    \mathrlap{\,,}
  \end{aligned}
\end{equation}
which under any of the above identifications \eqref{HilbertSpaceAsDirectSumOfComplexVectorSpaces} we may take to be given by ordinary complex conjugation on the $\mathbb{C}$ summands.

\item
A \emph{quaternionic structure} on $\HilbertSpace$ \cref{TheHilbertSpace}, hence a \emph{complex anti-linear} endomorphism, and as such self-adjoint (cf. \cite[Def. 3.1]{Uhlmann2016}) and squaring to minus the identity:
\begin{equation}
  \label{TheQuaternionicStructure}
  \begin{tikzcd}[row sep=-3pt, column sep=0pt]
    \HilbertSpace
    \ar[
      rr,
      "{
        J
      }"
    ]
    &&
    \HilbertSpace
    \mathrlap{\,,}
    \\
    \psi 
      &\mapsto&
    J\psi
    \\
    \mathrm{i}\psi
      &\mapsto&
    - \mathrm{i}J{\psi}
  \end{tikzcd}
  \hspace{1cm}
  \begin{aligned}
    &
    J^2
     = 
    -\mathrm{id}
    \\
    &
    \langle 
      -
      ,\, 
      J
      -
    \rangle
    =
    -
    \overline{
    \langle 
      J
      -
      ,\,
      -
    \rangle
    }
    \mathrlap{\,,}
  \end{aligned}
\end{equation}
which under any of the above identifications \eqref{HilbertSpaceAsDirectSumOfComplexVectorSpaces} for even $k$ we may take to be given on the $\mathbb{C}^2$ summands by
\begin{equation}
  \begin{tikzcd}[row sep=-3pt, column sep=0pt]
    \mathbb{C}^2
    \ar[rr]
    &&
    \mathbb{C}^2
    \\
    \left(
    \begin{matrix}
      z_1
      \\
      z_2
    \end{matrix}
    \right)
    &\mapsto&
    \left(
    \begin{matrix}
      - \ComplexConjugation{z_2}
      \\
      \phantom{-} \ComplexConjugation{z_1}
    \end{matrix}
    \right)
    \mathrlap{.}
  \end{tikzcd}
\end{equation}

\item 
  The group generated by $\tOperator$ \cref{TheRealStructure}, to be denoted
  \begin{equation}
    \label{tSymmetryGroup}
    \mathbb{Z}_2^{\tSymmetry}
    :=
    \big\{
      \mathrm{id},
      \tOperator
    \big\}
    \mathrlap{.}
  \end{equation}

\item
  The fixed locus of this involution  \cref{TheRealStructure,tSymmetryGroup}, which is the countably infinite-dimensional \emph{real} Hilbert space, to be denoted
  \begin{equation}
    \label{TheRealHilbertSpace}
    \HilbertSpace
      ^{ \mathbb{Z}_2^{\tSymmetry} }
    \in
    \mathrm{Hilb}_{\mathbb{R}}
    \mathrlap{\,.}
  \end{equation}

\item
The \emph{$\mathbb{Z}_2$-graded Hilbert space} whose homogeneous summands are both \cref{TheHilbertSpace}, to be denoted
\begin{equation}
  \label{GradedHilbertSpace}
  \GradedHilbertSpace
  :=
  \HilbertSpace \ominus \HilbertSpace
  \mathrlap{\,,}
\end{equation}
which inherits a real structure via \cref{TheRealStructure}.

\item 
The \emph{grading involution}
on $\GradedHilbertSpace$ \cref{GradedHilbertSpace}, to be denoted:
\begin{equation}
  \label{GradingInvolutionOnHgr}
  \begin{tikzcd}[
ampersand replacement=\&, 
    row sep=-3pt,
   column sep=0pt
  ]
    \GradedHilbertSpace
    \ar[
      rr,
      "{
        \pOperator
      }"
    ]
    \&\&
    \GradedHilbertSpace
    \\
    \left(
    \begin{matrix}
      \psi_+
      \\
      \psi_-
    \end{matrix}
    \right)
      \&\mapsto\&
    \left(
    \begin{matrix}
      \psi_-
      \\
      \psi_+
    \end{matrix}
    \right)
    ,
  \end{tikzcd}
\end{equation}
which satisfies
\begin{equation}
  \begin{aligned}
    \pOperator^2 = \mathrm{id}, \quad 
    \pOperator^\dagger = \pOperator, 
    \quad
    \pOperator \circ \tOperator
    =
    \tOperator \circ \pOperator
     \mathrlap{\,.}
  \end{aligned}
\end{equation}
\item
The group generated by 
\begin{equation}
  \cOperator
  :=
  \tOperator\pOperator
  \mathrlap{\,,}
\end{equation}
to be denoted
\begin{equation}
  \mathbb{Z}_2^{\cSymmetry}
  :=
  \big\{
    \mathrm{id}
    ,\,
    \cOperator
  \big\}
  \mathrlap{\,.}
\end{equation}
\item
The diagonal subgroup
\begin{equation}
  \begin{tikzcd}[row sep=-3pt, column sep=0pt]
    \mathbb{Z}_2^{\pSymmetry}
    &:=&
    \big\{
      \mathrm{id},
      \pOperator
    \big\}
    \ar[
      rr,
      hook
    ]
    &&
    \mathbb{Z}_2^{\tOperator}
    \times
    \mathbb{Z}_2^{\cOperator}
    \\
    &&
    \pOperator
    &\mapsto&
    \tOperator\cOperator
    \mathrlap{\,,}
  \end{tikzcd}
\end{equation}
\item
so that we have these two isomorphic incarnations of the \emph{PCT-group}:
\begin{equation}
  \label{IncarnationsOfPCTGroup}
  \begin{tikzcd}[row sep=-3pt, column sep=0pt]
    \mathbb{Z}_2^{\tSymmetry}
    \times
    \mathbb{Z}_2^{\pSymmetry}
    \ar[
      rr,
      "{ \sim }"
    ]
    &&
    \mathbb{Z}_2^{\tSymmetry}
    \times
    \mathbb{Z}_2^{\pSymmetry}    
    \\
    \tOperator &\mapsto& \tOperator
    \\
    \pOperator 
      &\mapsto&
    \cOperator\tOperator
    \mathrlap{\,.}
  \end{tikzcd}
\end{equation}
\end{enumerate}

\end{notation}

\begin{definition}[Projective unitary groups]  \label[definition]{ProjectiveUnitarGroupOnAHilbertSpace} $\,$
\begin{enumerate}
  \item
  We write
  \begin{equation}
    \label{TheGroupUH}
    \UH
    \,\in\, 
    \mathrm{Grp}(\mathrm{Top})
  \end{equation}
  for the topological group of unitary linear operators on $\HilbertSpace$ \cref{TheHilbertSpace},
  equipped with the norm topology.\footnote{
    \label{StrongTopologyOnUH}
    More ambitiously (cf. \cite{AtiyahSegal2004}), one equips $\UH$ 
    \cref{TheGroupUH} with the strong/weak operator topology or compact-open topology, which are all equal here but strictly coarser than the norm topology, see \cite[Ex. 2.3.19]{SS25-Bun} for further pointers.
  }
  This is famously contractible, by Kuiper's theorem (cf. \cite{nLab:KuiperTheorem}):
  \begin{equation}
    \label{KuiperTheorem}
    \UH \underset{\mathrm{hmtpy}}{\simeq} \ast
    \,.
  \end{equation}

  \item
  We write
  \begin{equation}
    \label{TheGroupPUH}
    \PUH
    \coloneqq
    \UH/\mathrm{U}(1)
    \;\;\;
    \in
    \;
    \mathrm{Grp}(\mathrm{Top})
  \end{equation}
  for the topological quotient group of \cref{TheGroupUH}
  by its subgroup of operators acting by multiplication with a complex number of unit norm.
  The quotient coprojection is a locally trivial $\mathrm{U}(1)$-principal bundle
  \cite[Thm. 1]{Simms1970}:
  \begin{equation}
    \label{PUHFiberSequence}
    \begin{tikzcd}[row sep=12pt]
      \mathrm{U}(1)
      \ar[r, hook]
      &
      \UH
      \ar[
        d,
        ->>
      ]
      \\
      &
      \PUH
      \mathrlap{\,.}
    \end{tikzcd}
  \end{equation}

\item 
The group $\mathbb{Z}^{\tSymmetry}_2$ \cref{tSymmetryGroup} acts compatibly on these groups:
  \begin{equation}
  \label{ComplexConjugationActionOnProjectiveUnitaryGroup}
  \begin{tikzcd}
    \mathrm{U}(1)
    \ar[r, hook]
    \ar[
      out=180-58, 
      in=59, 
      looseness=4, 
      "\scalebox{.9}{$\;\mathclap{
      \mathbb{Z}
        ^{\mathrlap{\tSymmetry}}
        _2
      }\;$}"{description},shift right=1
    ]
    &
    \UH
    \ar[r, ->>]
    \ar[
      out=180-58, 
      in=59, 
      looseness=4.3, 
      "\scalebox{.9}{$\;\mathclap{
      \mathbb{Z}
        ^{\mathrlap{\tSymmetry}}
        _2
      }\;$}"{description},shift right=1
    ]
    &
    \PUH
    \mathrlap{\,.}
    \ar[
      out=180-58, 
      in=59, 
      looseness=4.3, 
      "\scalebox{.9}{$\;\mathclap{
      \mathbb{Z}
        ^{\mathrlap{\tSymmetry}}
        _2
      }\;$}"{description},shift right=1
    ]
    \end{tikzcd}    
\end{equation}

\item
  The fixed locus of this action is the group of orthogonal operators, equipped with its operator topology, on the
  {\it real} Hilbert space $\HilbertSpace^{\mathbb{Z}^{\tSymmetry}_2}$ \cref{TheRealHilbertSpace}:
  \begin{equation}
    \label{TheGroupOH}
    \OH
    \simeq
    \big(%
      \UH%
    \big)^{\mathbb{Z}^{\tSymmetry}_2}.
  \end{equation}

\item
  In turn, the quotient of the latter
  by the subgroup
  of operators acting by multiplication with real units
  is the infinite
  {\it projective orthogonal group}
  \parencites[\S 3]{Rosenberg1989, MathaiMurrayStevenson2003}:
  \begin{equation}
    \label{TheGroupPOH}
    \POH
    =
    \big(%
      \PUH%
    \big)^{\mathbb{Z}^{\tSymmetry}_2}
    =
    \OH/\{\pm 1\}
    \;\;\;
    \in
    \;
    \mathrm{Grp}(\mathrm{Top})
    \,.
  \end{equation}
  \end{enumerate}
\end{definition}

\subsubsection{Quantum Symmetries}
\label{QuantumSymmetries}

\begin{remark}
  We have an isomorphism
  \begin{equation}
    \label{TheUnitaryAntiunitaryGroup}
    \begin{tikzcd}[row sep=-3pt, column sep=0pt]
      \UH
      \rtimes
      \mathbb{Z}_2^{\tSymmetry}
      \ar[
        rr,
        "{ \sim }"
      ]
      &&
      \UH
      \sqcup 
      \antiUH
      \\
      \big(
        U, \mathrm{id}
      \big)
      &\mapsto&
      U
      \\
      \big(
        U, \tOperator
      \big)
      &\mapsto&
      U \circ \tOperator
    \end{tikzcd}
  \end{equation}
  between the semidirect product of the unitary group \cref{TheGroupUH} with $\mathbb{Z}_2^{\tSymmetry}$ \cref{tSymmetryGroup} and the group of \emph{unitary or anti-unitary} maps (cf. \cite[p. 29-30]{Uhlmann2016}); and analogously for the projective group \cref{TheGroupPUH}:
  \begin{equation}
    \begin{tikzcd}[row sep=-3pt, column sep=0pt]
      \PUH
      \rtimes
      \mathbb{Z}_2^{\tSymmetry}
      \ar[
        rr,
        "{ \sim }"
      ]
      &&
      \PUH
      \sqcup 
      \antiPUH
      \\
      \big(
        [U], \mathrm{id}
      \big)
      &\mapsto&
      {[U]}
      \\
      \big(
        [U], \tOperator
      \big)
      &\mapsto&
      {[U \circ \tOperator]}
      \mathrlap{\,.}
    \end{tikzcd}
  \end{equation}  
  This is the group of \emph{quantum symmetries} according to Wigner's theorem (cf. \cite{nLab:WignerTheorem}), in which context the operator $\tOperator$ \cref{TheRealStructure} is interpreted as \emph{time-reversal symmetry}.
\end{remark}
We next enlarge this group by what one may think of as \emph{particle/anti-particle} symmetry (cf. \cite[Fact 2.3]{SS22-Ord}) or ``charge conjugation''.

\begin{definition}[Graded projective unitary group] $\,$
\begin{enumerate}
\item
We write
\begin{equation}
  \label{GradedUnitaryGroup}
  \begin{aligned}
  &\GradedUH
  :=
  \big(
    \UH \times \UH
  \big) 
    \rtimes
  \mathbb{Z}_2^{\pSymmetry}
  \\
  &
  =
  \Bigg\{
    \left(
    \begin{matrix}
      \UnitaryOperator_{{}_{++}}
      &
      0
      \\
      0
      &
      \UnitaryOperator_{{}_{--}}
    \end{matrix}
    \right)
    ,\,
    \left(
    \begin{matrix}
      0
      &
      \UnitaryOperator_{{}_{+-}}
      \\
      \UnitaryOperator_{{}_{-+}}
      &
      0
    \end{matrix}
    \right)
    \Bigg\vert\,
    U_{{}_{\bullet, \bullet}}
    \in 
    \UnitaryGroup(\HilbertSpace)
  \Bigg\}
  \subset
  \UnitaryGroup
  \big(
    \GradedHilbertSpace
  \big)
  \end{aligned}
\end{equation}
for the \emph{graded unitary group} acting in homogeneous degrees on the graded Hilbert space $\GradedHilbertSpace$ \cref{GradedHilbertSpace}.

\item
The \emph{graded projective unitary group} 
$\GradedPUH$ (cf. \parencites[Prop. 2.2]{Parker1988}[p. 5]{CareyWang2007})
is the quotient of this graded unitary group \cref{GradedUnitaryGroup} by the
diagonal subgroup $\mathrm{U}(1) \xhookrightarrow{\;} \UH \xhookrightarrow{\;} \GradedUH$, and the $\mathbb{Z}^{\tSymmetry}_2$-action \cref{ComplexConjugationActionOnProjectiveUnitaryGroup}
evidently extends to an automorphism action on all these graded groups:
\begin{equation}
  \label{TheGroupGradedPUH}
  \begin{tikzcd}
    \mathrm{U}(1)
    \ar[
      out=180-58, 
      in=59, 
      looseness=4, 
      "\scalebox{.9}{$\;\mathclap{
      \mathbb{Z}
        ^{\mathrlap{\tSymmetry}}
        _2
      }\;$}"{description},shift right=1
    ]
    \ar[r, hook]
      &
    \GradedUH
    \ar[
      out=180-58, 
      in=59, 
      looseness=4, 
      "\scalebox{.9}{$\;\mathclap{
      \mathbb{Z}
        ^{\mathrlap{\tSymmetry}}
        _2
      }\;$}"{description},shift right=1
    ]
    \ar[r, ->>]
      &
    \GradedPUH
    \ar[
      out=180-58, 
      in=59, 
      looseness=4, 
      "\scalebox{.9}{$\;\mathclap{
      \mathbb{Z}
        ^{\mathrlap{\tSymmetry}}
        _2
      }\;$}"{description},shift right=1
    ]
    \mathrlap{\,.}
  \end{tikzcd}
\end{equation}
\end{enumerate}
\end{definition}

\begin{notation}[Graded quantum symmetries]
\label[notation]{TheQSGroup}
The further semidirect product 
$
  \GradedPUH
  \rtimes
  \mathbb{Z}_2^{\tSymmetry}
$
of the graded projective unitary group \cref{TheGroupGradedPUH} 
with $\mathbb{Z}^{\tSymmetry}_2$ \cref{tSymmetryGroup}
plays a key role in the following, to be called the \emph{group of quantum symmetries} (cf. \parencites{nLab:WignerTheorem}[(15)]{SS22-Ord}):
\begin{equation}
  \label{GroupOfQuantumSymmetries}
  \QuantumSymmetries
  :=
  \frac{
    \UH^2
  }{
    \mathrm{U}(1)
  }
  \rtimes
  \Big(
    \mathbb{Z}_2^{\tSymmetry}
    \times
    \mathbb{Z}_2^{\cSymmetry}
  \Big)
  \mathrlap{\,,}
\end{equation}
where on the right we used \cref{IncarnationsOfPCTGroup}.
\end{notation}
\begin{remark}
The group $\QuantumSymmetries$ \cref{GroupOfQuantumSymmetries} sits in this system of (split) short exact sequences of topological groups:
\begin{equation}
  \label{ProjectiveGradedExtensionOfZTwoTimesZTwo}
  \begin{tikzcd}[row sep=13pt, 
    column  sep=10pt
  ]
    \mathrm{U}(1)
    \ar[rr,-,shift left=1pt]
    \ar[rr,-,shift right=1pt]
    \ar[d, hook]
    &&
    \mathrm{U}(1)
    \ar[rr, ->>]
    \ar[d, hook]
    &&
    1
    \ar[d, hook]
    \\
    \UH^2
    \ar[rr, hook]
    \ar[d, ->>]
    &&
    \GradedUH 
      \rtimes 
    \mathbb{Z}^{\tSymmetry}_2
    \ar[rr, ->>]
    \ar[d, ->>]
    &&
    \mathbb{Z}^{\cSymmetry}_2 
      \times 
    \mathbb{Z}^{\tSymmetry}_2
    \ar[d,-,shift left=1pt]
    \ar[d,-,shift right=1pt]
    \\
    \UH^2\!/\mathrm{U}(1)
    \ar[rr, hook]
    &&
    \QuantumSymmetries
    \ar[
      rr, 
      ->>,
      "{
        p_{\mathrm{ct}}
      }"
    ]
    &&
    \mathbb{Z}^{\cSymmetry}_2 
      \times 
    \mathbb{Z}^{\tSymmetry}_2
    \ar[
      ll, 
      hook',
      bend left=15,
      "{
        \iota_{\mathrm{ct}}
      }"
    ]
    \,.
  \end{tikzcd}
\end{equation}
With the bottom row we also have the diagonal inclusion of the ordinary projective group: 
\begin{equation}
  \label{IncludionOfPUintoQuantumSymmetries}
  \begin{tikzcd}[sep=-2pt]
    \PUH
    \ar[
      rrrr,
      uphordown,
      "{
        i_{\mathrm{pu}}
      }"{description}
    ]
    &
    \defneq
    &
    \tfrac{
      \UH
    }{
      \mathrm{U}(1)
    }
    \ar[
      r,
      hook,
      "{
        \tfrac
          {\mathrm{diag}}
          {\mathrm{U}(1)}
      }"{description}
    ]
    &[60pt]
    \tfrac{
      \UH^2
    }{
      \mathrm{U}(1)
    }
    \ar[r, hook]
    &[20pt]
    \QuantumSymmetries
  \end{tikzcd}
\end{equation}
\end{remark}

\begin{definition}
  \label[definition]{PCTQuantumSymmetry}
  A \emph{PCT quantum symmetry} is
  \begin{enumerate}
    \item
      a subgroup
      $$
        G \subset 
        \mathbb{Z}_2^{\tSymmetry}
        \times
        \mathbb{Z}_2^{\cSymmetry}
      $$
      of the PCT group \cref{IncarnationsOfPCTGroup},
    \item
      equipped with a lift $\big[\widehat{(-)}\big]$ to the group $\QuantumSymmetries$ of graded quantum symmetries (\cref{TheQSGroup}), hence with a dashed homomorphism making this diagram commute:
      \begin{equation}
        \label{SomePCTQuantumSymmetry}
        \begin{tikzcd}[
          column sep=25pt,
          row sep=-1pt
        ]
          G
          \ar[
            dr,
            hook'
          ]
          \ar[
            rr,
            dashed,
            "{
              \big[\widehat{(-)}\big]
            }"
          ]
          &&
          \QuantumSymmetries
          \mathrlap{\,.}
          \ar[
            dl,
            ->>
          ]
          \\
          &
          \mathbb{Z}_2^{\tSymmetry}
          \times
          \mathbb{Z}_2^{\cSymmetry}
        \end{tikzcd}
      \end{equation}
  \end{enumerate}
  In other words, this is a graded or ungraded and unitary or anti-unitary projective representation of $G$ where the grading and anti-unitarity is fixed by the embedding of $G$ into the PCT group. 

Here the notation for the homomorphism \cref{SomePCTQuantumSymmetry} means that $[-]$ is the $\mathrm{U}(1)$-equivalence class of a representative $\widehat{(-)}$:
\begin{equation}
  \label{ComponentsOfPCTQuantumSymmetry}
  \begin{tikzcd}[row sep=-3pt, column sep=5pt]
    G
    \ar[
      rr,
      "{
        \big[
        \widehat{
          (-)
        }
        \big]
      }"
    ]
    &&
    \QuantumSymmetries
    \simeq
    \Big(
    \UH^2
    \rtimes
    \big(
      \mathbb{Z}_2^{\tSymmetry}
      \times
      \mathbb{Z}_2^{\cSymmetry}
    \big)
    \Big)/\mathrm{U}(1)
    \\
    g 
      &\longmapsto&
    {[\,\widehat{g}\,]}
    \;\;
      \mbox{for 
        $\widehat{g} \in \UH^2 \rtimes
        \big(
          \mathbb{Z}_2^{\tSymmetry}
          \times
          \mathbb{Z}_2^{\cSymmetry}
        \big)$.
      }    
  \end{tikzcd}
\end{equation}
\end{definition}

\begin{table}[htb]
\caption{
  \label{TableOfPCTQuantumSymmetries}
  \textbf{Rows 1-3}: The possible cases of PCT quantum symmetries (\cref{PCTQuantumSymmetry}) according to \cref{ClassifyingPCTQuantumSymmetries}.
  \\
  \textbf{Row 4}: The spaces of Fredholm operators commuting with the given PCT quantum symmetries are classifying spaces for the 10 flavors of topological K-theory (cf. \cref{FredholmOperatorsAndKTheory}).
  \\
  \textbf{Row 5}: Traditional labels of these 10 cases in the context of the \emph{10-fold way} classification of topological phases of matter (going back to \parencites[Tbl. 1]{SchnyderRyuFurusakiLudwig2008}, review in \parencites[Tbl. 1]{ChiuTeoSchnyderRyu2016}).
}

\vspace{-3mm}
\hspace{-2.5mm}
\adjustbox{scale=.88}{
$
{
\setlength\arraycolsep{3pt}
\def\arraystretch{2.1}
\begin{array}{|cr||c|c|c|c|c|c|c|c|c|c|}
  \hline
  \scalebox{.85}{\bf PCT subgroup}
  &
  G = 
  &
  \{\mathrm{e}\}
  &
  \{\mathrm{e}, \pOperator\}
  &
  \multicolumn{2}{c|}{
    \mathbb{Z}_2^{\tSymmetry}
    \!\!\!:=\!\!
    \{\mathrm{e}, \tOperator\}
  }
  &
  \multicolumn{2}{c|}{
    \mathbb{Z}_2^{\cSymmetry}
    \!\!\!:=\!\!
    \{\mathrm{e}, \cOperator\}
  }
  &
  \multicolumn{4}{c|}{
  \big\{
    \mathrm{e}, 
    \tOperator, 
    \cOperator,
    \pOperator 
      \!=\! 
    \cOperator\tOperator
  \big\}
  }
  \\  
  \hline
  \hline
  \multirow{2}{*}{$
  \substack{
    \scalebox{.8}{\bf Quantum symmetry}
    \\
    \adjustbox{scale=.85}{
      \begin{tikzcd}[
        ampersand replacement=\&,
        column sep=-3pt,
        row sep=4pt
      ]
        G
        \ar[
          dr, 
          hook',
          shorten=-2pt
        ]
        \ar[
          rr,
          dashed,
          "{ \scalebox{0.7}{$
            \big[\widehat{(-)}\big]
            $}
          }"{description, pos=.45}
        ]
        \&\&
        \QuantumSymmetries
        \ar[
          dl,
          ->>,
          shorten=-2pt
        ]
        \\
        \&
        \mathbb{Z}_2^{\tSymmetry}
        \times
        \mathbb{Z}_2^{\cSymmetry}
      \end{tikzcd}
    }
  }
  $}
  &
  \widehat{\tOperator}^2 =
  &
  &
  &
  +1
  &
  -1
  &
  & 
  & 
  +1
  &
  +1
  & 
  -1
  &
  -1
  \\
  \hhline{~-----------}
  & 
  \widehat{\cOperator}^2 =
  &&&&&
  +1 
  &
  -1
  &
  +1 
  &
  -1
  &
  +1 
  &
  -1
  \\
  \hline
  \hspace{-15pt}
  \substack{
    \scalebox{.7}{\bf Fred. operators}
    \\
    \scalebox{.7}{\bf 
    commuting with $\widehat{G}$}
  }
  &
  \mathllap{
  \GradedFredholmOperators
    ^{\widehat{G}}
  }
  \simeq
  &
  K \mathrm{U}^0
  &
  K \mathrm{U}^1
  &
  K \mathrm{O}^0
  &
  K \mathrm{O}^4
  &
  K \mathrm{O}^2
  &
  K \mathrm{O}^6
  &
  K \mathrm{O}^1
  &
  K \mathrm{O}^7
  &
  K \mathrm{O}^5
  &
  K \mathrm{O}^3
  \\
  \hline
  &&&&&&&&&&&
  \\[-26pt]
  \substack{
    \scalebox{.8}{\bf 10-fold way}
    \\
    \scalebox{.8}{\bf class label} 
  }
  &
  &
  \mbox{A}
  &
  \mbox{AIII}
  &
  \mbox{AI}
  &
  \mbox{AII}
  &
  \mbox{D}
  &
  \mbox{C}
  &
  \mbox{BDI}
  &
  \mbox{CI}
  &
  \mbox{CII}
  &
  \mbox{DIII}
  \\
  \hline
\end{array}
}
$
}
\end{table}

\begin{proposition}
  \label[proposition]
   {ClassifyingPCTQuantumSymmetries}
  The PCT quantum symmetries \textup{(\cref{PCTQuantumSymmetry})} satisfy:
\begin{enumerate}
\item
If $G = \{\mathrm{id}, \pOperator \!=\! \tOperator\cOperator\}$, then $[\widehat{\pOperator}]$ has a  representative $\widehat{\pOperator}$ \cref{ComponentsOfPCTQuantumSymmetry} satisfying
\begin{equation}
  \label{HatPOperatorSquare}
  \widehat{\pOperator}^2 = \mathrm{id}
  \mathrlap{\,.}
\end{equation}

\item 
If $\tOperator \in G$, then $[\widehat{\tOperator}]$ has a representative $\widehat{\tOperator}$ \cref{ComponentsOfPCTQuantumSymmetry} satisfying
\begin{equation}
  \label{HatTOperatorSquare}
  \widehat{\tOperator}^2
  \in
  \{
    \pm 
    \mathrm{id}
  \}
  \mathrlap{\,.}
\end{equation}

\item 
If $\cOperator \in G$, then $[\widehat{\cOperator}]$ has a representative $\widehat{\cOperator}$ \cref{ComponentsOfPCTQuantumSymmetry} satisfying
\begin{equation}
  \label{HatCOperatorSquare}
  \widehat{\cOperator}^2
  \in
  \{
    \pm 
    \mathrm{id}
  \}
  \mathrlap{\,,}
\end{equation}
\end{enumerate}
and all these cases occur. In consequence, PCT quantum symmetries fall into 10 classes as shown in \cref{TableOfPCTQuantumSymmetries}.
\end{proposition}
\begin{proof} 
\begin{enumerate}
\item
The condition that $\big[\widehat{(-)}\big]$ be a group homorphism is equivalently that
\begin{equation}
  \begin{aligned}
    \big[\widehat{P}\,\big]^2 
    = [\mathrm{id}]
    \;\;
    \Leftrightarrow
    \;\;
    \widehat{P}^2 
    = \omega \, \mathrm{id}
    \,,
    \;
    \mbox{for some $\omega \in \mathrm{U}(1)$.}
  \end{aligned}
\end{equation}
But then for any choice of square root $\sqrt{\omega} \in \mathrm{U}(1)$ we have
\begin{equation}
  \label{RescalingThePOperator}
  \big[\widehat{\pOperator}\,\big] 
    = 
  \Big[
    \tfrac{1}{\sqrt{\omega}}
    \widehat{\pOperator}
  \,\Big]
  \mathrlap{\,,}
\end{equation}
and since $\widehat{P}$ is unitary \cref{TheUnitaryAntiunitaryGroup} and hence in particular complex-linear, this rescaled operator has the claimed property:
\begin{equation}
  \label{HowTheRescaledPOperatorSquaresToUnity}
  \Big(
    \tfrac{1}{\sqrt{\omega}}
    \widehat{\pOperator}
  \,\Big)^2
  =
  \tfrac{1}{\sqrt{\omega}}
  \widehat{\pOperator}
  \tfrac{1}{\sqrt{\omega}}
  \widehat{\pOperator}
  =
  \tfrac{1}{\sqrt{\omega}}
  \tfrac{1}{\sqrt{\omega}}
  \widehat{\pOperator}
  \widehat{\pOperator}
  = 
  1
  \mathrlap{\,.}
\end{equation}

\item
As in the previous case, homomorphy requires that
$$
  \begin{aligned}
    \big[\widehat{T}\,\big]^2 
    = [\mathrm{id}]
    \;\; \Leftrightarrow
    \;\;
    \widehat{T}^2 
    = \omega \, \mathrm{id}
    \,,
    \;
    \mbox{for some $\omega \in \mathrm{U}(1)$},
  \end{aligned}
$$
but since now $\widehat{\tOperator}$ is anti-unitary \cref{TheUnitaryAntiunitaryGroup} and hence complex anti-linear, there is a constraint forcing $\omega$ to be real:
$$
  \begin{tikzcd}[row sep=-5pt, column sep=0pt]
    \widehat{\tOperator}
    \widehat{\tOperator}^2
    &=&
    \widehat{\tOperator}^2
    \widehat{\tOperator}
    \\
    \rotatebox[origin=c]{90}{$=$}
    &&
    \rotatebox[origin=c]{90}{$=$}
    \\
    \ComplexConjugation{\omega}
    \widehat{\tOperator}
    &=&
    \omega
    \,
    \widehat{\tOperator}
  \end{tikzcd}
  \hspace{.6cm}
  \Leftrightarrow
  \hspace{.6cm}
  \omega 
  \in 
    \mathrm{U}(1) \cap \mathbb{R}
  =
  \{ \pm 1 \}
  \mathrlap{\,.}
$$
Moreover, the one remaining non-trivial value, $\omega = -1$, can \emph{not} be scaled away as in \cref{RescalingThePOperator,HowTheRescaledPOperatorSquaresToUnity}, since anti-linearity of $\widehat{\tOperator}$ implies that
$$
  \big(\pm\mathrm{i}\widehat{\tOperator}\,\big)^2
  =
  (\pm\mathrm{i})\widehat{\tOperator}
  (\pm\mathrm{i})\widehat{\tOperator}
  =
  (\pm\mathrm{i})(\mp\mathrm{i})
  \widehat{\tOperator}
  \widehat{\tOperator}
  =
  \widehat{\tOperator}^2
.
$$
\item
This case works verbatim like the  previous case. \qedhere
\end{enumerate}
\end{proof}

\subsubsection{Fredholm Operators}

Our construction of twisted orbifold K-theory (in \cref{OnTwistedOrbifoldKTheory}) is based on the fact that \emph{classifying spaces} for topological K-theory are given by space of \emph{Fredholm operators} (the \emph{Atiyah-J{\"a}nich theorem}, \cref{AtiyahJaenichTheorem} below). Here we compile some basics on Fredholm operators that we will refer to for this discussion.

For standard review of Fredholm operators see \parencites[\S 1.4]{Murphy1990}[\S 3.3]{Arveson2002}[\S 3]{DSBW2023}. As has become common, we will consider the space of \emph{graded self-adjoint} Fredholm operators in the following \cref{SpaceOfFredholmOperators}, since this is what lends itself naturally to the classification of the different ``quantum symmetry protected'' flavors of K-theory (in \cref{FredholmOperatorsAndKTheory}) and to the
construction of twisted equivariant K-theory (as observed by \cite{AtiyahSegal2004}), and therefore eventually to the construction of twisted orbifold K-theory (in \cref{OnTwistedOrbifoldKTheory}).

\begin{definition}
  \label[definition]
  {SpaceOfFredholmOperators}
  \begin{enumerate}

  \item
  A \emph{Fredholm operator} on $\HilbertSpace$ \cref{TheHilbertSpace} is a bounded operator whose kernel and cokernel are of finite dimension: 
  \begin{equation}
   \label{TheSpaceOfFredholmOperators}
   \FredholmOperators(\HilbertSpace)
   :=
   \left\{
   f \in \BoundedOperators(\HilbertSpace) 
   \;
   \middle\vert
   \;
   \begin{aligned}
      \mathrm{dim}\big(\mathrm{ker}(f)\big)
      & < \infty
      \\[-2pt]
      \mathrm{dim}\big(\mathrm{coker}(f)\big)
      & < \infty
    \end{aligned}
    \right\}.
  \end{equation}
  We consider this space of Fredholm operators as equipped with the operator norm topology. 
  \footnote{
     More ambitiously, one modifies the space of graded Fredholm operators in \cref{TheSpaceOfFredholmOperators} up to homotopy equivalence (\cite[\S 4]{AtiyahSegal2004}) such that the $\PUH$-action \cref{ConjugationActionOfProjectiveGradedUnitaryGroup} on it becomes continuous also in the latter's strong/weak operator topology. Here we disregard this, for simplicity. 
  }

\item 
Homeomorphically, we regard this space as that of \emph{odd graded} but self-adjoint Fredholm operators on $\GradedHilbertSpace$ \cref{GradedHilbertSpace}:
\begin{equation}
  \label{TheSpaceOfGradedFredholmOperators}
  \begin{tikzcd}[
    column sep=62pt
  ]
  \GradedFredholmOperators(\GradedHilbertSpace)
  \coloneqq
  \left\{
    \FredholmOperator
    =
    \left(
    \begin{matrix}
      0 & f
      \\
      f^\dagger & 0 
    \end{matrix}
    \right)
    \in
    \BoundedOperators\big(
      \GradedHilbertSpace
    \big)
    \,\middle\vert\,
    f \in \FredholmOperators(\HilbertSpace)
  \right\}.
  \end{tikzcd}
\end{equation}

\item 
We take this space to be pointed by:
\begin{equation}
  \label{ZeroFredholmOperator}
  F_0
  :=
  \left(
  \begin{matrix}
    0 & \mathrm{id}
    \\
    \mathrm{id} & 0
  \end{matrix}
  \right).
\end{equation}

\item
The kernel of a graded Fredholm operatorator $F$ \cref{TheSpaceOfGradedFredholmOperators} is hence a finite-dimensional, 
virtual $\mathbb{C}$-vector space:
\begin{equation}
  \label{VirtualKernel}
    \mathrm{ker}(F)
    =
    \begin{tikzcd}[sep=-3pt]
      \mathrm{ker}(f)
      \\
      \ominus
      \\
      \mathrm{ker}(f^\dagger)
    \end{tikzcd}
    =
    \begin{tikzcd}[
      sep=-3pt
    ]
      \mathrm{ker}(f)
      \\
      \ominus
      \\
      \mathrm{coker}(f)
      \mathrlap{\,,}
    \end{tikzcd}    
\end{equation}
whose virtual dimension is the \emph{Fredholm index} of $f$:
\begin{equation}
  \label{FredholmIndex}
  \mathrm{dim}\big(
    \mathrm{ker}(F)
  \big)
  =
  \mathrm{ind}(f)
  :=
    \mathrm{dim}_{\mathbb{C}}\big(
      \mathrm{ker}(f)
    \big)
    -
    \mathrm{dim}_{\mathbb{C}}\big(
      \mathrm{coker}(f)
    \big)
    \mathrlap{\,.}
\end{equation}

\item 
In this graded form \cref{TheSpaceOfGradedFredholmOperators}, 
the graded quantum symmetries 
(\cref{TheQSGroup}) 
\begin{equation}
  \left(
    \left[
      \begin{matrix}
      U_{{}_{++}}
      &
      U_{{}_{+-}}
      \\
      U_{{}_{-+}}
      &
      U_{{}_{--}}
      \end{matrix}
    \right]
    ,
    \sigma
  \right)
  \in
  \GradedPUH
  \,\rtimes\,
  \mathbb{Z}^{\tSymmetry}_2
  \defneq
  \QuantumSymmetries(\HilbertSpace)
\end{equation}
have a canonical topological group action \cref{TopologicalGroupAction} by conjugation:
\begin{equation}
\label{ConjugationActionOfProjectiveGradedUnitaryGroup}
  \begin{tikzcd}[
    row sep=-2pt, 
    column sep=5pt,
    ampersand replacement=\&
  ]
    \GradedFredholmOperators
    \ar[rr]
    \&\&
    \GradedFredholmOperators
    \\
    \left(
      \begin{matrix}
      0
      &
      f_{{}_{+-}}
       \\
      f^\dagger_{{}_{+-}}
      &
      0
      \end{matrix}
    \right)
    \&\longmapsto\&
    \left(
      \begin{matrix}
        U_{{}_{++}}
        &
        U_{{}_{+-}}
        \\
        U_{{}_{-+}}
        &
        U_{{}_{--}}
    \end{matrix}
    \right)
    \circ
    \left(
    \begin{matrix}
      0
      &
      f^\sigma_{{}_{+-}}
      \\
      f^{\dagger \sigma}_{{}_{+-}}
      &
      0
    \end{matrix}
    \right)
    \circ
    \left(
      \begin{matrix}
      U_{{}_{++}}
      &
      U_{{}_{+-}}
      \\
      U_{{}_{-+}}
      &
      U_{{}_{--}}
      \end{matrix}
    \right)^{\mathrlap{-1}}
    \,.
  \end{tikzcd}
\end{equation}
(Here either
$\UnitaryOperator_{{}_{+-}}, \UnitaryOperator_{{}_{-+}} = 0$
or
$\UnitaryOperator_{{}_{++}}, \UnitaryOperator_{{}_{--}} = 0$;
the square bracket denotes the $\mathrm{diag}(\mathrm{U}(1))$-equivalence class
of the matrix;
and $\FredholmOperator^{\sigma}$ equals $\FredholmOperator$ when $\sigma = \mathrm{e}$
and equals its complex conjugate \cref{ComplexConjugationActionOnProjectiveUnitaryGroup} otherwise, cf. \cite[\S 5.B]{Matumoto71}.)

\end{enumerate}

\end{definition}

\begin{proposition}[Atkinson's theorem {(cf. \cite[Thm. 1.4.16]{Murphy1990}})]
  \label[proposition]
   {AtkinsonTheorem} 
  A bounded operator is Fredholm \cref{TheSpaceOfFredholmOperators} iff there exists another bounded operator (a \emph{parametrix}) such that their composites differ from the identity by a compact operator:
  \begin{equation}
    \label{FredholmViaParametrix}
    \mbox{\textup{$f \in \BoundedOperators(\HilbertSpace)$
    is Fredholm}}
    \;\;\;\;
    \Leftrightarrow
    \;\;\;\;
    \exists
    \tilde f \in \BoundedOperators(\HilbertSpace)
    :
    \left\{
    \begin{aligned}
      f \circ \tilde f - \mathrm{id}
      \in 
      \CompactOperators(\HilbertSpace)
      \mathrlap{\,,}
      \\[-2pt]
      \tilde f \circ f - \mathrm{id}
      \in 
      \CompactOperators(\HilbertSpace)
      \mathrlap{\,.}
    \end{aligned}
    \right.
  \end{equation}
\end{proposition}

\begin{remark}[Arithmetic on Fredholm operators]
  Given a graded Fredholm operator $F$ as in \cref{TheSpaceOfGradedFredholmOperators}, its \emph{charge reversal} 
  \begin{equation}
    \label{ChargeReversalOfGradedFredholm}
    \ominus F
    :=
    \left(
    \begin{matrix}
      0 & f^\dagger
      \\
      f & 0
    \end{matrix}
    \right)
    \in
    \GradedFredholmOperators
  \end{equation}
  is also graded Fredholm, with negative virtual kernel:
  \begin{equation}
    \mathrm{ker}\big(\ominus\!F_P\big)
    =
    \ominus \, \mathrm{ker}\big(F_P\big)
    \mathrlap{\,.}
  \end{equation}
  Given a pair $F_1, F_2 \in \GradedFredholmOperators$, their \emph{direct sum} is
  \begin{equation}
    F_1 \oplus F_2
    =
    \left(
    \begin{matrix}
      0 & f_1 \oplus f_2
      \\
      f_1^\dagger \oplus f_2^\dagger & 0 
    \end{matrix}
    \right)
  \end{equation}
  acting on $\HilbertSpace \oplus \HilbertSpace \simeq \HilbertSpace$, with virtual kernel the direct sum of those of the summands:
  \begin{equation}
    \mathrm{ker}\big(
      F_1 \oplus F_2
    \big)
    =
    \mathrm{ker}(F_1) 
    \oplus
    \mathrm{ker}(F_2)
    \mathrlap{\,.}
  \end{equation}
  In particular, the \emph{virtual difference} between a pair of Fredholm operators is the direct sum of the first with the charge reversal of the second:
  \begin{equation}
    \label{VirtualDifferenceOfGradedFredholmOps}
    F_1 \ominus F_2
    :=
    F_1 \oplus \big( \ominus F_2\big)
    \mathrlap{\,.}
  \end{equation}
\end{remark}

\begin{proposition}[{\cite[Lem 4]{Janich1965}}]
\label[proposition]{VirtualKernelBundles}
For $X$ a topological space and $F_{(-)}:X \to \GradedFredholmOperators$ a continuous map, this Fredholm index \cref{FredholmIndex} is continuous and hence a locally constant function $\mathrm{ind}\big(F_{(-)}\big) :  X \to \mathbb{Z}$. The graded contributions $\mathrm{dim}_{\mathbb{C}}\big(\mathrm{ker}(f_X)\big)$ and $\mathrm{dim}_{\mathbb{C}}\big(\mathrm{coker}(f_X)\big)$ need not be locally constant, but if they are then the \emph{virtual kernel bundle} is a virtual topological vector bundle over $X$: 
\begin{equation}
  \label{VirtualVectorBundleFromFredholmMap}
  \left.
  \begin{aligned}
  F_{X}
  & :
  \begin{tikzcd}
    X \ar[r, "{ \mathrm{cntns} }"]
    &
    \GradedFredholmOperators
  \end{tikzcd}
  \\[-2pt]
  \mathrm{dim}_{\mathbb{C}}\big(
    \mathrm{ker}(f_X)
  \big)
  & :
  \begin{tikzcd}[]
    X 
    \ar[r, "\mathrm{cntns}" ]
    &
    \mathbb{N}
  \end{tikzcd}
  \end{aligned}
  \right\}
  \;\;
  \Rightarrow
  \;\;
  \mathrm{ker}\big(
    F_{X}
  \big)
  \in
  \GradedVectorSpaces_{X}
  \mathrlap{\,.}
\end{equation}
\end{proposition}

\begin{example}
  \label[example]{TheLeftShiftOperator}
  For $k \in \mathbb{N}$ and identifying $\HilbertSpace \simeq \bigoplus_{\mathbb{N}} \mathbb{C}^k$, the following graded Fredholm operator (a ``left shift operator'')
  \begin{equation}
    \label{ShiftOperator}
    F^{\mathbb{C}^k} \! 
    :=
    \left(
    \begin{matrix}
      0 & f^{\mathbb{C}^k}
      \\
      (f^{\mathbb{C}^k})^\dagger & 0
    \end{matrix}
    \right)
    , \;
    f^{\mathbb{C}^k}\!:=
    \left(
    \begin{matrix}
      0 & \mathrm{id} & 0 & 0 & \cdots
      \\
      0 & 0 & \mathrm{id} & 0 & \cdots
      \\
      0 & 0 & 0 & \mathrm{id} & \cdots 
      \\[-4pt]
      \vdots & \vdots & \vdots & \vdots
      & \ddots
    \end{matrix}
    \right)
    ,
    \;\;
    \begin{tikzcd}[
      row sep=-4pt,
      column sep=20pt
    ]
      \HilbertSpace
      \ar[
        rr,
        "{ 
          f^{\mathbb{C}^k} 
        }"
      ]
      &&
      \HilbertSpace
      \\
      \rotatebox[origin=c]{-90}{$\simeq$}
      &&
      \rotatebox[origin=c]{-90}{$\simeq$}
      \\
      \mathbb{C}^k
      &&
      \mathbb{C}^k
      \\
      \oplus && \oplus
      \\
      \mathbb{C}^k 
      \ar[
        uurr,
        "{ 
          \mathrm{id} 
        }"{sloped, description}
      ]
      &&
      \mathbb{C}^k
      \\
      \oplus && \oplus
      \\
      \mathbb{C}^k 
      \ar[
        uurr,
        "{ 
          \mathrm{id} 
        }"{sloped, description}
      ]
      &&
      \mathbb{C}^k
      \\
      \oplus && \oplus
      \\[-4pt]
      \vdots
      \ar[
        uurr,
        dotted
      ]
      && 
      \vdots
     \mathrlap{\;,}
    \end{tikzcd}
  \end{equation}
  has virtual kernel 
  \cref{VirtualKernel} isomorphic to $\mathbb{C}^k$:
  \begin{equation}
    \mathrm{ker}\big(
      F^{\mathbb{C}^k}
    \big)
    =
    \mathbb{C}^k \ominus 0
    \mathrlap{\,.}
  \end{equation}
  Hence a map constant on $F^{\mathbb{C}^k}$ represents, under \cref{VirtualVectorBundleFromFredholmMap} the trivial rank=$k$ vector bundle:
  \begin{equation}
    \label{FRedholmOperatorForTrivialBundle}
    \Big(
    \begin{tikzcd}[sep=small]
      X
      \ar[r]
      &
      \ast
      \ar[
        rr,
        "{ 
          F^{\mathbb{C}^k} 
        }"
      ]
      &&
      \GradedFredholmOperators
    \end{tikzcd}
    \Big)
    \;\;
    \begin{tikzcd}
      {}
      \ar[
        r,
        |->,
        "{ \mathrm{ker}_X }"
      ]
      &
      {}
    \end{tikzcd}
    \;\;
    \mathbb{C}^k_X
    \ominus 
    0
    \in
    \mathbb{C}\mathrm{Vec}_X^{\mathrm{gr}}
    \mathrlap{\,.}
  \end{equation}
\end{example}

We also observe the following example, which will be crucial in \cref{TheEquivariantOrientation} (\cref{TautologicalHP1IndexedFredholm}).
\begin{example}
  \label[example]{FredholmFromProjector}
  For $k \in \mathbb{N}$ and 
  \begin{equation}
    P : 
    \begin{tikzcd}[sep=small]
      \mathbb{C}^k
      \ar[r]
      &
      \mathbb{C}^k
      \mathrlap\,,
    \end{tikzcd}
    \,\;\;
    \begin{aligned}
       P^\dagger = P\,, \quad 
       P \circ P = P\,,
    \end{aligned}
\end{equation}
a hermitian projection operator, then we obtain a graded Fredholm operator \cref{TheSpaceOfFredholmOperators}
  \begin{equation}
  \label{FredholmOperatorFromProjector}
    F_P
    :=
    \left(
    \begin{matrix}
      0 & f_P
      \\
      f^\dagger_P & 0
    \end{matrix}
    \right)
    \in
    \GradedFredholmOperators
  \end{equation}
  by setting
  \begin{equation}
    \label{FredholmOperatorFromProjection}
    f_P
    :=
    \left(
    \begin{matrix}
      \;P\; & 1\!-\!P & 0 & 0 & 0 & \cdots 
      \\
      0 & \;P\; & 1\!-\!P & 0 & 0 & \cdots 
      \\
      0 & 0 & \;P\; & 1\!-\!P & 0 & \cdots
      \\[-4pt]
      \vdots & \vdots & \vdots & \vdots & \vdots
      & \ddots
    \end{matrix}
    \right)
    \,,
    \hspace{.7cm}
    \begin{tikzcd}[
      row sep=-4pt
    ]
      \HilbertSpace
      \ar[
        rr,
        "{
          f_P
        }"
      ]
      &&
      \HilbertSpace
      \\
      \rotatebox[origin=c]{-90}{$\simeq$}
      &&
      \rotatebox[origin=c]{-90}{$\simeq$}
      \\
      \mathbb{C}^k
      \ar[
         rr,
         "{ P }"{description}
      ]
      &&
      \mathbb{C}^k
      \\
      \oplus
      &&
      \oplus
      \\
      \mathbb{C}^k
      \ar[
        uurr,
        "{
          1 - P
        }"{description, sloped}
      ]
      \ar[
        rr,
        "{ P }"{description}
      ]
      &&
      \mathbb{C}^k
      \\
      \oplus
      &&
      \oplus
      \\
      \mathbb{C}^k
      \ar[
        uurr,
        "{
          1 - P
        }"{description, sloped}
      ]
      \ar[
        rr,
        "{ P }"{description}
      ]
      &&
      \mathbb{C}^k
      \\
      \oplus
      &&
      \oplus
      \\[-5pt]
      \vdots
      \ar[
        uurr,
        dotted,
      ]
      &&
      \vdots
      \mathrlap{\;,}
    \end{tikzcd}
  \end{equation}
  whose virtual kernel \eqref{VirtualKernel} reproduces the kernel of $P$:
  \begin{equation}
    \label{VirtualKernelOfGradedFredOfProjector}
    \mathrm{ker}\big(
      F_P
    \big)
    =
    \mathrm{ker}(P)
    \ominus
    0
    \mathrlap{\,.}
  \end{equation}
  
  This construction constitutes a continuous map from hermitian projectors to graded Fredholm operators
  \begin{equation}
    \label{MapFromProjectorsToFredholmOperators}
    F_{(-)}
    :
    \begin{tikzcd}
      \Big\{
        P \in \BoundedOperators(\mathbb{C}^k)
        \;\Big\vert\;
        \substack{
          P^\dagger = P
          \\
          P \circ P = P
        }
      \Big\}
      \ar[r]
      &
      \GradedFredholmOperators
      \mathrlap{\,.}
    \end{tikzcd}
  \end{equation}
\end{example}
\begin{proof}
  That the virtual kernel is as claimed in \cref{VirtualKernelOfGradedFredOfProjector} follows by direct inspection, crucially using that $P$ is a projector, so that $P$ and $1-P$ have complementary kernels and images in $\mathbb{C}^k$.

  Incidentally, $f^\dagger_P$ is a parametrix \cref{FredholmViaParametrix} for $f_P$ \cref{FredholmOperatorFromProjection}: From
  \begin{equation}
    \label{AdjointFredholmOperatorFromProjection}
    f_P^\dagger
    =
    \left(
    \begin{matrix}
      \;P\; & 0 & 0 & 0 &  \cdots 
      \\
      1\!-\!P & \;P\; & 0 & 0 &  \cdots 
      \\
      0 & 1\!-\!P & \;P\; & 0 &  \cdots
      \\
      0 & 0 & 1\!-\!P & \;P\; &   \cdots
      \\[-4pt]
      \vdots & \vdots & \vdots & \vdots 
      & \ddots
    \end{matrix}
    \right)
    \,,
    \hspace{.7cm}
    \begin{tikzcd}[
      row sep=-4pt
    ]
      \HilbertSpace
      \ar[
        rr,
        "{
          f_P^\dagger
        }"
      ]
      &&
      \HilbertSpace
      \\
      \rotatebox[origin=c]{-90}{$\simeq$}
      &&
      \rotatebox[origin=c]{-90}{$\simeq$}
      \\
      \mathbb{C}^k
      \ar[
        ddrr,
        "{
          1 - P
        }"{description, sloped}
      ]
      \ar[
         rr,
         "{ P }"{description}
      ]
      &&
      \mathbb{C}^k
      \\
      \oplus
      &&
      \oplus
      \\
      \mathbb{C}^k
      \ar[
        rr,
        "{ P }"{description}
      ]
      \ar[
        ddrr,
        "{
          1 - P
        }"{description, sloped}
      ]
      &&
      \mathbb{C}^k
      \\
      \oplus
      &&
      \oplus
      \\
      \mathbb{C}^k
      \ar[
        rr,
        "{ P }"{description}
      ]
      \ar[
        ddrr,
        dotted,
      ]
      &&
      \mathbb{C}^k
      \\
      \oplus
      &&
      \oplus
      \\[-5pt]
      \vdots
      &&
      \vdots
      \mathrlap{\;,}
    \end{tikzcd}
  \end{equation}
  we find
\begin{equation}
  \left.
  \begin{aligned}
  f_P \circ f_P^\dagger - \mathrm{id}
  & =
  0
  \\
  f_P^\dagger \circ f_P - \mathrm{id}
  & = 
  \left(
  \begin{matrix}
    P\!-\!1 & 0 & 0 &  \cdots
    \\
    0 & 0 & 0 & \cdots
    \\
    0 & 0 & 0 & \cdots
    \\[-4pt]
    \vdots & \vdots & \vdots & \ddots
  \end{matrix}
  \right)
  \end{aligned}
  \right\}
  \in
  \CompactOperators(\HilbertSpace)
  \mathrlap{\,.}
\end{equation}
This completes the proof.
\end{proof}

\begin{example}
  In further specialization of \cref{FredholmFromProjector}, consider the projector
  \begin{equation}
    P_0 
    :=
    \left(
    \begin{matrix}
      0 & 0 
      \\
      0 & \mathrm{id}_b
    \end{matrix}
    \right)
    :
    \begin{tikzcd}
      \mathbb{C}_a^2 \oplus \mathbb{C}_b^2
      \ar[r]
      &
      \mathbb{C}_a^2 \oplus \mathbb{C}_b^2
      \mathrlap{\,,}
    \end{tikzcd}
  \end{equation}
  where the subscripts are just to name these two subspaces of $\mathbb{C}^4$, both isomorphic to $\mathbb{C}^2$.
  Then, under the unitary transformation which  ``unshuffles'' the copies of these two subspaces in the infinite direct sum Hilbert space
  \begin{equation}
    \label{UnshuffleUnitaryTransformation}
      \HilbertSpace
      :=
      \bigoplus_{\mathbb{N}}
      \big(
        \mathbb{C}^2_a
        \oplus 
        \mathbb{C}^2_b
      \big)
      \xrightarrow[\sim]{\quad U \quad}
      \grayunderbrace{
      \Big(
        \textstyle{\bigoplus_{\mathbb{N}}}
        \,
        \mathbb{C}^2_a
      \Big)
      }{%
        \mathcolor{black}{%
          \HilbertSpace_a%
        }%
      }
      \oplus
      \grayunderbrace{
      \Big(
        \textstyle{\bigoplus_{\mathbb{N}}}
        \,
        \mathbb{C}^2_b
      \Big)
      }{%
        \mathcolor{black}{%
          \HilbertSpace_b%
        }%
      }
      \mathrlap{\,,}
  \end{equation}
  the corresponding Fredholm operator 
  \cref{FredholmOperatorFromProjection}   \begin{equation}
    f_{P_0}
    =
    \left(
    \begin{matrix}
      0 & 0 & \mathrm{id}_a & 0 & 0 & 0 & \cdots
      \\
      0 & \mathrm{id}_b & 0 & 0 & 0 & 0 & \cdots
      \\
      0 & 0 & 0 & 0 & \mathrm{id}_a & 0 & \cdots
      \\
      0 & 0 & 0 & \mathrm{id}_b & 0 & 0 & \cdots
      \\[-4pt]
      \vdots & \vdots & \vdots & \vdots 
      & \vdots & \vdots & \ddots
    \end{matrix}
    \right)
    \,,
    \;\;\;\;\;
    \begin{tikzcd}[row sep=-4pt]
      \HilbertSpace
      \ar[
        rr,
        "{ 
          f_{P_0} 
        }"{description}
      ]
      &&
      \HilbertSpace
      \\
      \rotatebox[origin=c]{-90}{$\simeq$}
      &&
      \rotatebox[origin=c]{-90}{$\simeq$}
      \\
      \mathbb{C}^4
      \ar[
        rr,
        "{ P_0 }"{description}
      ]
      &&
      \mathbb{C}^4
      \\
      \oplus && \oplus
      \\
      \mathbb{C}^4
      \ar[
        rr,
        "{ P_0 }"{description}
      ]
      \ar[
        uurr,
        "{ 1 - P_0 }"{description, sloped}
      ]
      &&
      \mathbb{C}^4
      \\
      \oplus && \oplus
      \\
      \mathbb{C}^4
      \ar[
        rr,
        "{ P_0 }"{description}
      ]
      \ar[
        uurr,
        "{ 1 - P_0 }"{description, sloped}
      ]
      &&
      \mathbb{C}^4
      \\
      \oplus
      &&
      \oplus
      \\[-4pt]
      \vdots
      \ar[
        uurr,
        dotted
      ]
      && \vdots
    \end{tikzcd}
  \end{equation}
  transforms into the direct sum of the left shift operator \cref{ShiftOperator} and the identity: 
  \begin{equation}
    \label{TransformedFredholmUnderUnshuffle}
    \begin{aligned}
    U \circ f_{P_0} \circ U^{-1}
    & =
    \left(
    \begin{matrix}
      0 & \mathrm{id}_a & 0 & 0 & 0 & \cdots 
      \\
      0 & 0 & \mathrm{id}_a & 0 & 0 & \cdots
      \\
      0 & 0 & 0 & \mathrm{id}_a & 0 & \cdots
      \\
      0 & 0 & 0 & 0 & \mathrm{id}_a & \cdots
      \\[-4pt]
      \vdots & \vdots & \vdots & \vdots 
    \end{matrix}
    \right)
    \oplus
    \left(
    \begin{matrix}
      \mathrm{id}_b & 0 & 0 & 0 & 0 & \cdots 
      \\
      0 & \mathrm{id}_b & 0 & 0 & 0 & \cdots
      \\
      0 & 0 & \mathrm{id}_b & 0 & 0 & \cdots
      \\
      0 & 0 & 0 & \mathrm{id}_b & 0 & \cdots
      \\[-4pt]
      \vdots & \vdots & \vdots & \vdots & \ddots
    \end{matrix}
    \right)
    \\
    & =
    f^{\mathbb{C}^2} \oplus \mathrm{id}
    :
    \begin{tikzcd}
      \HilbertSpace_a \oplus \HilbertSpace_b
      \ar[r]
      &
      \HilbertSpace_a \oplus \HilbertSpace_b      
      \,.
    \end{tikzcd}
    \end{aligned}
  \end{equation}
  This simple transformation is going to be useful in 
  the proof of \cref{PullbackOfTautologicalFredholmOperatorAlongQuaternionicHopfFibrationTrivializes}, via \cref{UnshufflingTrivialAndDefiningReps}.
\end{example}

\subsubsection{Classifying K-Theory}
\label{FredholmOperatorsAndKTheory}

The space of Fredholm operators \cref{TheSpaceOfFredholmOperators} and its subspaces compatible with PCT quantum symmetries (\cref{ClassifyingPCTQuantumSymmetries}) turn out to the classifying spaces (cf. \cref{TableOfNotionsOfCohomology}) for the flavors of the generalized cohomology theory called \emph{topological K-theory} (cf. \cite{nLab:TopologicalKTheory}). As such, one may take the following statements to be the \emph{definition} of topological K-theory in its various flavors and degrees.

\begin{proposition}[Atiyah-J{\"a}nich theorem  {\parencites{Janich1965}[Thm. A1]{Atiyah1967}}]
  \label[proposition]
  {AtiyahJaenichTheorem}
  The space of Fredholm operators \cref{TheSpaceOfFredholmOperators} is a classifying space for \emph{complex K-theory} in degree=0,
  \begin{equation}
    \pi_0
    \big\{
      \begin{tikzcd}[sep=small]
        X
        \ar[r, dashed]
        &
        \FredholmOperators
      \end{tikzcd}
    \big\}
    \simeq
    \mathrm{KU}^0(X)
    \mathrlap{\,.}
  \end{equation}
  Over compact Hausdorff spaces $X$, where 
  \begin{equation}
    \mathrm{KU}^0(X)
    \simeq
    K\big[
      \mathbb{C}\mathrm{Vec}_X
    \big]
  \end{equation}
  is the Grothendieck group of complex vector bundles, this equivalence is given by passing to virtual kernel bundles \cref{VirtualVectorBundleFromFredholmMap} of good representative maps $F^{\mathrm{gd}}_X$ in the homotopy class \textup{(meaning that $\mathrm{dim}\big(\mathrm{ker}(f^{\mathrm{gd}}_X)\big)$ is locally constant, see \cref{VirtualKernelBundles})}:
  \begin{equation}
    \label{JaenichIndexMap}
    \begin{tikzcd}[row sep=-4pt, column sep=0pt]
    \pi_0\big\{
      X 
      \dashrightarrow
      \FredholmOperators
    \big\}
    \ar[
      rr
    ]
    &&
    K\big[
      \mathbb{C}\mathrm{Vec}_X
    \big]    
    \\
    {[F_X]} 
      &\longmapsto&
    \big[
      \mathrm{ker}\big(F^{\mathrm{gd}}_X\big)
    \big]
    \mathrlap{\,.}
    \end{tikzcd}
  \end{equation}
  In particular, the connected components of $\FredholmOperators$,
  \begin{equation}
    \pi_0 \, \FredholmOperators
    \simeq
    \pi_0
    \big\{
      \begin{tikzcd}[sep=small]
      \ast 
      \ar[r, dashed]
      &
      \GradedFredholmOperators
      \end{tikzcd}
    \big\}
    \mathrlap{\,,}
  \end{equation}
  are indexed by the integer Fredholm index \cref{FredholmIndex}:
  \begin{equation}
    \label{ConnectedComponentsOfSpaceOfFredholmOperators}
    \begin{tikzcd}[
      ampersand replacement=\&, row sep=-3pt,
      column sep=0pt
    ]
      \pi_0
      \,
      \GradedFredholmOperators
      \ar[
        rr,
        "{ \sim }"
      ]
      \&\&
      \mathbb{Z}
      \\
      \left[
        F = 
        \left(
        \begin{matrix}
          0 & f
          \\
          f^\dagger & 0
        \end{matrix}
        \right)
      \right]
      \&\longmapsto\&
      \begin{aligned}
      &
      \mathrm{dim}\big(
        \mathrm{ker}(F)
      \big)
      \\
      &
      =
      \mathrm{ind}(f)
      \mathrlap{\,.}
     \end{aligned}
    \end{tikzcd}
  \end{equation}
\end{proposition}

\begin{notation}[Subspaces of Fredholm operators]
\label[notation]{SubspacesOfFredholmOperators}
    For
    \begin{enumerate}
    \item
    $
      \mathbb{K} \in \big\{
        \mathbb{R}, 
        \mathbb{C}, 
        \overline{\mathbb{C}}, 
        \mathbb{H}
      \big\}
    $ 
    \item 
    $\sigma \in \{\pm 1\}$,
    \end{enumerate} 
    consider the following subspaces of that of Fredholm operators (\cref{SpaceOfFredholmOperators}):
    \begin{equation}
      \begin{tikzcd}[
        row sep=-4pt, column sep=0pt
      ]
        \FredholmOperators
          ^\sigma
          _{\mathbb{K}}
        \ar[
          rr,
          hook
        ]
        &&
        \FredholmOperators
          _{\mathbb{K}}
        \ar[
          r,
          hook
        ]
        &
        \FredholmOperators(\HilbertSpace)
        \\
        \rotatebox[origin=c]{-90}{$:=$}
        &&
        \rotatebox[origin=c]{-90}{$:=$}
        \\
        \big\{
          f \in \FredholmOperators_{\mathbb{K}}
          \;\big\vert\;
          f^\dagger = \sigma \cdot f
        \big\}
        &\subset&
        \big\{
          f \in \FredholmOperators(\HilbertSpace)
          \;\big\vert\;
          \scalebox{.9}{$f$ is $\mathbb{K}$-linear}
        \big\}
        \mathrlap{\,,}
      \end{tikzcd}
    \end{equation}
    where $f$ being $\mathbb{R}$-linear/$\mathbb{H}$-linear means that it commutes with a fixed real structure \cref{TheRealStructure} or quaternionic structure \cref{TheQuaternionicStructure} on $\HilbertSpace$, respectively, and where by it being ``$\overline{\mathbb{C}}$-linear'' we mean that it is $\mathbb{C}$-\emph{anti}-linear (while it being $\mathbb{C}$-linear is no additional condition, of course). 
\end{notation}

\begin{proposition}[{\parencites{AtiyahSinger1969}[Thm. 1.16]{Karoubi1970}}]
\label[proposition]
 {SubspacesOfFredAsClassifyingSpacesForK}
The subspaces of Fredholm operators from \cref{SubspacesOfFredholmOperators} are classifying spaces for flavors of topological K-theory as follows:
\begin{subequations}
  \label{MatchingKTheoryToFredSubspaces}
  \begin{align}
    \pi_0\big\{
        X 
          \dashrightarrow
        \FredholmOperators
          ^{\phantom{+}}
          _{\mathbb{C}}
    \big\}
    &
    \simeq
    \mathrm{KU}^0(X)
    \\
    \pi_0\big\{
        X 
          \dashrightarrow
        \FredholmOperators
          ^{+}
          _{\mathbb{C}}
    \big\}
    &
    \simeq
    \mathrm{KU}^1(X)
    \sqcup \{+,-\}
    \\
    \pi_0\big\{
        X 
          \dashrightarrow
        \FredholmOperators
          ^{\phantom{\pm}}
          _{\mathbb{R}}
    \big\}
    &
    \simeq
    \mathrm{KO}^0(X)
    \\
    \pi_0\big\{
        X 
          \dashrightarrow
        \FredholmOperators
          ^{+}
          _{\mathbb{R}}
    \big\}
    &
    \simeq
    \mathrm{KO}^1(X)
    \sqcup \{+,-\}
    \\
    \pi_0\big\{
        X 
          \dashrightarrow
        \FredholmOperators
          ^{-}
          _{\mathbb{R}}
    \big\}
    &
    \simeq
    \mathrm{KO}^7(X)
    \\
    \pi_0\big\{
        X 
          \dashrightarrow
        \FredholmOperators
          ^{+}
          _{\overline{\mathbb{C}}}
    \big\}
    &
    \simeq
    \mathrm{KO}^2(X)
    \\
    \pi_0\big\{
        X 
          \dashrightarrow
        \FredholmOperators
          ^{-}
          _{\overline{\mathbb{C}}}
    \big\}
    &
    \simeq
    \mathrm{KO}^6(X)
    \\
    \pi_0\big\{
        X 
          \dashrightarrow
        \FredholmOperators
          ^{\phantom{+}}
          _{\mathbb{H}}
    \big\}
    &
    \simeq
    \mathrm{KO}^4(X)
    \\
    \pi_0\big\{
        X 
          \dashrightarrow
        \FredholmOperators
          ^{+}
          _{\mathbb{H}}
    \big\}
    &
    \simeq
    \mathrm{KO}^5(X)
    \sqcup \{+,-\}
    \\
    \pi_0\big\{
        X 
          \dashrightarrow
        \FredholmOperators
          ^{-}
          _{\mathbb{H}}
    \big\}
    &
    \simeq
    \mathrm{KO}^3(X)
    \mathrlap{\,.}
  \end{align}
\end{subequations}
\end{proposition}

\begin{remark}
  \label[remark]{DisjointComponentsOfSAFred}
  The disjoint sets $\{+,-\}$ in \cref{MatchingKTheoryToFredSubspaces} witness  pairs of contractible components present in each of the spaces of self-adjoint $\mathbb{R}/\mathbb{C}/\mathbb{H}$-linear Fredholm operators, corresponding to such operators  whose essential spectrum (\cref{SepctrumOfSAOperators}) is purely positive or purely negative, respectively \parencites[Thm. B]{AtiyahSinger1969}[\S I]{Karoubi1970}.

  When the self-adjoint $\mathbb{C}$-linear Fredholm operators are interpreted as tachyon field values of open strings stretching between coincident $\mathrm{D9}/\overline{D9}$-branes (cf. \parencites[\S 3.1]{Szabo2002}[\S 8]{Gao2010}) then those with a mixed essential spectrum are interpreted as witnessing solitonic pair annihilation in a background where both $\mathrm{D9}/\overline{D9}$-branes appear with infinite multiplicity (cf. \parencites[(3.11)]{Horava1998}[p. 7]{Witten2001} and \cref{RevisitingBraneCharges}). Accordingly, the components $\pm \in \{+,-\}$ would have to be interpreted as backgrounds where an infinite number of $\mathrm{D}9$-branes is compensated only by a finite number of $\overline{D9}$-branes or vice versa, respectively.
  
  For the purpose of K-theory classification these components are to be disregarded, but if we regard the full spaces of Fredholm operators as classifying a (slightly) non-abelian cohomology theory (cf. \cref{TableOfNotionsOfCohomology}) variant in its own right, then $\pm \in \{+,-\}$ are valid classes in the cohomology set of any space $X$ which reflect a pair of ``collapsed'' cases that one may keep track of.

\end{remark}

We now reformulate the classical result \cref{SubspacesOfFredAsClassifyingSpacesForK} in the guise of the ``10-fold way'' (cf. \cite{nLab:TenFoldWay}) of topological phases protected by PCT quantum symmetry (\cref{PCTQuantumSymmetry}), following \cite[Fact 2.12]{SS22-Ord}, as shown in \cref{TableOfPCTQuantumSymmetries}:
\begin{proposition}
  \label[proposition]
  {The10PCTFixedLoci}
  The 10 subspaces of Fredholm operators \textup{(\cref{SubspacesOfFredholmOperators})} appearing in \cref{MatchingKTheoryToFredSubspaces}, when regarded as subspaces of $\GradedFredholmOperators(\GradedHilbertSpace)$ \cref{TheSpaceOfGradedFredholmOperators} 
  are homeomorphic to the fixed subspaces \cref{FixedSubspace}
  \begin{equation}
    \GradedFredholmOperators^{\widehat{G}}
    \subset
    \GradedFredholmOperators
      (\GradedHilbertSpace)
  \end{equation}
  for the conjugation action \cref{ConjugationActionOfProjectiveGradedUnitaryGroup} of the 10 PCT quantum symmetries \textup{(\cref{ClassifyingPCTQuantumSymmetries,TableOfPCTQuantumSymmetries})}, respectively:
  \begin{subequations}
  \begin{align}
     \GradedFredholmOperators
     & \simeq
     \FredholmOperators_{\mathbb{C}}
    \\
    \label{pFixedFredholmIsComplexSelfAdjoint}
     \GradedFredholmOperators
       ^{ 
         \langle\widehat{\pOperator}\rangle
       }
     & \simeq
     \FredholmOperators^{\pm}_{\mathbb{C}}
     \;\;\;\;
     \mbox{\rm for 
       $\widehat{\pOperator}{}^2 = \pm \mathrm{id}$
     }
     \\
     \label{t+FixedFredholmIsReal}
     \GradedFredholmOperators
       ^{ 
         \langle\widehat{\tOperator}\rangle
       }
     & \simeq
     \FredholmOperators_{\mathbb{R}}
     \;\;\;\;
     \mbox{\rm for 
       $\widehat{\tOperator}{}^2 = + \mathrm{id}$
     }
     \\
     \label{t+c+FixedFredholmIsRealSelfAdjoint}
     \GradedFredholmOperators
       ^{ 
         \langle
           \widehat{\tOperator},
           \widehat{\cOperator}
         \rangle
       }
     & \simeq
    \FredholmOperators
      ^{+}
      _{\mathbb{R}}
     \;\;\;\;
     \mbox{\rm for 
       $\widehat{\tOperator}{}^2 = + \mathrm{id}$,
       $\widehat{\cOperator}{}^2 = + \mathrm{id}$,
     }
     \\
     \label{t+c-FixedFredholmIsRealAntiSelfAdjoint}
     \GradedFredholmOperators
       ^{ 
         \langle
           \widehat{\tOperator},
           \widehat{\cOperator}
         \rangle
       }
     & \simeq
    \FredholmOperators
      ^{-}
      _{\mathbb{R}}
     \;\;\;\;
     \mbox{\rm for 
       $\widehat{\tOperator}{}^2 = + \mathrm{id}$,
       $\widehat{\cOperator}{}^2 = - \mathrm{id}$,
     }
     \\
     \label{c+FixedFredholmIsAntilinearSelfAdjoint}
     \GradedFredholmOperators
       ^{ 
         \langle
           \widehat{\cOperator}
         \rangle
       }
     & \simeq
    \FredholmOperators
      ^{+}
      _{\overline{\mathbb{C}}}
     \;\;\;\;
     \mbox{\rm for 
       $\widehat{\cOperator}{}^2 = + \mathrm{id}$,
     }
     \\
     \label{c-FixedFredholmIsAntilinearAntiSelfAdjoint}
     \GradedFredholmOperators
       ^{ 
         \langle
           \widehat{\cOperator}
         \rangle
       }
     & \simeq
    \FredholmOperators
      ^{-}
      _{\overline{\mathbb{C}}}
     \;\;\;\;
     \mbox{\rm for 
       $\widehat{\cOperator}{}^2 = - \mathrm{id}$,
     }
     \\
     \label{t-FixedFredholmIsQuaternionic}
     \GradedFredholmOperators
       ^{ 
         \langle
           \widehat{\tOperator}
         \rangle
       }
     & \simeq
    \FredholmOperators
      _{\mathbb{H}}
     \;\;\;\;
     \mbox{\rm for 
       $\widehat{\tOperator}{}^2 = - \mathrm{id}$,
     }
     \\
     \label{t-c+FixedFredholmIsQuaternSelfAdjoint}
     \GradedFredholmOperators
       ^{ 
         \langle
           \widehat{\tOperator},
           \widehat{\cOperator}
         \rangle
       }
     & \simeq
    \FredholmOperators
      ^{+}
      _{\mathbb{H}}
     \;\;\;\;
     \mbox{\rm for 
       $\widehat{\tOperator}{}^2 = - \mathrm{id}$,
       $\widehat{\cOperator}{}^2 = + \mathrm{id}$,
     }
     \\
     \label{t-c-FixedFredholmIsQuaternAntiSelfAdjoint}
     \GradedFredholmOperators
       ^{ 
         \langle
           \widehat{\tOperator},
           \widehat{\cOperator}
         \rangle
       }
     & \simeq
    \FredholmOperators
      ^{-}
      _{\mathbb{H}}
     \;\;\;\;
     \mbox{\rm for 
       $\widehat{\tOperator}{}^2 = - \mathrm{id}$,
       $\widehat{\cOperator}{}^2 = - \mathrm{id}$.
     }
  \end{align}
  \end{subequations}
\end{proposition}
\begin{proof}
  The quantum symmetry $\widehat{\pOperator}$ with $\widehat{\pOperator}{}^2 = \pm \mathrm{id}$ \cref{HatPOperatorSquare} is represented by swapping graded summands \cref{GradingInvolutionOnHgr}, possibly up to a sign, so that a graded Fredholm operator $F$ of the form \cref{TheSpaceOfGradedFredholmOperators} commutes with $\widehat{\pOperator}$ iff its component Fredholm operator $f$ is (anti-)self-adjoint:
  \begin{equation}
    \left(
    \begin{matrix}
      0 & f
      \\
      f^\dagger & 0
    \end{matrix}
    \right)
    \cdot
    \left(
    \begin{matrix}
      0 & \pm 1
      \\
      1 & 0
    \end{matrix}
    \right)
    =
    \left(
    \begin{matrix}
      0 & \pm 1
      \\
      1 & 0
    \end{matrix}
    \right)
    \cdot
    \left(
    \begin{matrix}
      0 & f
      \\
      f^\dagger & 0
    \end{matrix}
    \right)
    \;\;\;\;
    \Leftrightarrow
    \;\;\;\;
    f^\dagger = \pm f
    \,.
  \end{equation}
   This gives the homeomorphy \cref{pFixedFredholmIsComplexSelfAdjoint}.

  The quantum symmetry $\widehat{\tOperator}$ with $\widehat{\tOperator}{}^2 = \pm \mathrm{id}$ \cref{HatTOperatorSquare} is manifestly a real/quaternionic structure on $\HilbertSpace$, \cref{TheRealStructure,TheQuaternionicStructure}, whence operators commuting with it are manifestly $\mathbb{R}$/$\mathbb{C}$-linear. This gives the homeomorphisms \cref{t+FixedFredholmIsReal,t-FixedFredholmIsQuaternionic}.

  Combining these two cases, if the quantum symmetries involve $\widehat{\tOperator}$ and $\widehat{\cOperator} = \widehat{\tOperator} \widehat{\pOperator}$ then the fixed locus is equivalently that of $\widehat{\tOperator}$ and $\widehat{\pOperator}$ and is hence given by $\mathbb{R}$/$\mathbb{C}$-linear (anti-)self-adjoint operators. This gives the homeomorphisms \cref{t+c+FixedFredholmIsRealSelfAdjoint,t+c-FixedFredholmIsRealAntiSelfAdjoint,t-c+FixedFredholmIsQuaternSelfAdjoint,t-c-FixedFredholmIsQuaternAntiSelfAdjoint}.

  Finally, if there is only $\widehat{\cOperator}$ then the condition on a Fredholm operator $F$ to be fixed is that 
  \begin{equation}
    \label{ConditionOnGradedFredholmFixedByHatC}
    \left(
    \begin{matrix}
      0 & f
      \\
      f^\dagger & 0
    \end{matrix}
    \right)
    \cdot
    \left(
    \begin{matrix}
      0 & \widehat{\tOperator}
      \\
      \widehat{\tOperator} & 0
    \end{matrix}
    \right)
    =
    \left(
    \begin{matrix}
      0 & \widehat{\tOperator}
      \\
      \widehat{\tOperator} & 0
    \end{matrix}
    \right)
    \cdot
    \left(
    \begin{matrix}
      0 & f
      \\
      f^\dagger & 0
    \end{matrix}
    \right)
    \;\;\;\;
    \Leftrightarrow
    \;\;\;\;
    f \circ \widehat{\tOperator}
    =
    \widehat{\tOperator} \circ f^\dagger
    \,.
  \end{equation}
  Now to observe that $f$ satisfying this condition is equivalent to the anti-linear Fredholm operator $f \circ \widehat{\tOperator}$ being (anti-)self-adjoint
  \begin{equation}
    \big(
      f
      \circ
      \widehat{\tOperator}
    \big)^\dagger
    =
    \widehat{\tOperator}^{\dagger}    
      \circ 
    f^\dagger
    =
    (\pm)
    \widehat{\tOperator}    
      \circ 
    f^\dagger
    \underset{
      \scalebox{.7}{
        \cref{ConditionOnGradedFredholmFixedByHatC}
      }
    }{=}
    \pm
    f \circ \widehat{\tOperator}
    \,,
  \end{equation}
  where in the second step we used that real structures are self-adjoint \cref{TheRealStructure} while quaternionic structures are anti-self-adjoint \cref{TheQuaternionicStructure}. Therefore the remaining homeomorphisms \cref{c+FixedFredholmIsAntilinearSelfAdjoint,c-FixedFredholmIsAntilinearAntiSelfAdjoint} are given by $f \mapsto f \circ \widehat{\tOperator}$.
\end{proof}

\subsubsection{Twisted Orbifold K-Theory}
\label{OnTwistedOrbifoldKTheory}

We have in hand now:
\begin{enumerate}
\item
a general definition of (generalized nonabelian) twisted orbifold cohomology (\cref{OnTwistedOrbifoldCohomology})
with coefficients in any space $Y$ equipped with a topological group action $\Gamma \acts \, Y$,
\item 
a classifying space $Y = \GradedFredholmOperators(\GradedHilbertSpace)$ for topological K-theory (\cref{FredholmOperatorsAndKTheory}) equipped with a topological action $\QuantumSymmetries \acts \, \GradedFredholmOperators$ 
\cref{ConjugationActionOfProjectiveGradedUnitaryGroup}
by ``graded quantum symmetries'' $\QuantumSymmetries$ (\cref{TheQSGroup}), combining the graded projective unitary group with ``PCT symmetries'': $\QuantumSymmetries = \big(\UH^2/\mathrm{U}(1)\big) \rtimes \big(\mathbb{Z}_2^{\tSymmetry} \times \mathbb{Z}_2^{\cSymmetry}\big)$, and hence the corresponding stacky incarnation of the universal globally equivariant $\QuantumSymmetries$-associated $\GradedFredholmOperators$-fiber bundle (according to \cref{PrincipalBundleAsPullbackAlongCechCocycle}):
\begin{equation}
  \label{TheUniversalQSAssociatedFredBundle}
  \begin{tikzcd}
    \QuantumSymmetries
    \backsslash
    \GradedFredholmOperators
    \ar[
      d,
      "{
        p
          ^{\scalebox{.6}{$\QuantumSymmetries \acts \GradedFredholmOperators$}}
          _{\mathrm{univ}}
      }"{swap},
      "{
        \in \mathrm{Fib}
      }"
    ]
    \\
    \mathbf{B}\QuantumSymmetries
  \end{tikzcd}
\end{equation}
\end{enumerate}

Therefore we now immediately obtain a good definition of twisted orbifold K-theory (\cref{TwistedOrbifoldKTheory} below).

In fact, we obtain a highly structured version of such: Since the fixed loci by the PCT quantum symmetries inside the classifying space $\GradedFredholmOperators$ are the classifying spaces of topological K-theory in \emph{all its flavors and degrees} (\cref{The10PCTFixedLoci,TableOfPCTQuantumSymmetries}), the orbifolding of the cohomology theory represented by $\GradedFredholmOperators$ not only imposes equivariance conditions on any given  flavor of topological K-theory, but it also serves to ``dial'' which flavor and degree of K-theory is in effect where on the orbifold. Since for the $\tOperator$- and $\cOperator$-symmetries one wants (in applications to quantum physics) the corresponding orbifolding to be by orientation-reversing isotropy actions, one also speaks of \emph{orbi-orientifolds} (or just: \emph{orientifold}, for short):

\begin{definition}
  \label[definition]
   {OrbiOrientifoldWithKTheoryTwist}
\begin{enumerate}
  \item 
  An \emph{orbi-orientifold} (or \emph{Real parity orbifold}, with capital ``R'') is an orbifold $\mathcal{X}$ (\cref{Orbifolds}) equipped with a map to the delooping (\cref{DeloopingGroupoid}) of the PCT group \cref{IncarnationsOfPCTGroup}:
  \begin{equation}
    \label{GradingMapOnOrbifold}
    \begin{tikzcd}
      \sigma
      : 
      \mathcal{X}
      \ar[r]
      &
      \mathbf{B}\Big(
        \mathbb{Z}_2^{\tSymmetry}
        \times
        \mathbb{Z}_2^{\cSymmetry}
      \Big)
      \mathrlap{.}
    \end{tikzcd}
  \end{equation}
  \item
  A \emph{geometric K-theory twist}, or just \emph{twist} for short, on such an orbi-orientifold is map $\tau$ to the delooping of the group of graded quantum symmetries (\cref{TheQSGroup}) which lifts the PCT grading \cref{GradingMapOnOrbifold}, hence a dashed map making this diagram commute:
  \begin{equation}
    \begin{tikzcd}[
      column sep=15pt,
      row sep=4pt
    ]
      \mathllap{
        \tau
        :
        \;
      }
      \mathcal{X}
      \ar[
        dr,
        shorten >=-9pt,
        "{ \sigma }"{swap}
      ]
      \ar[
        rr,
        dashed,
        "{
        }"
      ]
      &&
      \mathbf{B}\QuantumSymmetries 
      \mathrlap{\,.}
      \ar[
        dl,
        shorten >=-5pt,
        "{
          \mathbf{B}P_{\mathrm{ct}}
        }"
      ]
      \\
      & 
      \mathbf{B}\Big(
        \mathbb{Z}_2^{\tSymmetry}
        \times
        \mathbb{Z}_2^{\cSymmetry}
      \Big)
    \end{tikzcd}
  \end{equation}
  \end{enumerate}
  Of course, such $\tau$ fully determines $\sigma$.
\end{definition}

Now recall that for a topological group $\Gamma$ acting continuously on a topological space $X$, we obtain the topological \emph{homotopy quotient} groupoid $\Gamma \backsslash X$ (\cref{ActionGroupoid}), canonically $X$-fibered over $\mathbf{B}\Gamma \defneq \ast \sslash G$ \cref{HomotopyFiberOfHoQuotientProjection}. Using this for the action $\QuantumSymmetries \acts \, \GradedFredholmOperators$ from \cref{SpaceOfFredholmOperators},  we obtain the following concise and transparent definition, in specialization of \cref{TwistedOrbifoldCohomology}:

\begin{definition}[Twisted orbi-orientifold K-theory {\parencites[Ex. 6.2.5]{SS25-Bun}[\S 2.2]{SS22-Ord}}]
  \label[definition]
    {TwistedOrbifoldKTheory}  
Given an orbi-orientifold $(\mathcal{X}, \sigma)$ with a K-theory twist $\tau$ (\cref{OrbiOrientifoldWithKTheoryTwist}) then its \emph{$\tau$-twisted K-cohomology} is the connected components of dashed slice maps (\cref{SliceMappingStack}) from $\tau$ into the universal stacky $\QuantumSymmetries$-associated $\GradedFredholmOperators$-fiber bundle \cref{TheUniversalQSAssociatedFredBundle}, as shown on the right here:
  \begin{equation}
    \label{DiagramForTwistedOrbifoldKTheory}
    \mathrm{K}^\tau\big(
      \mathcal{X}
    \big)
    :=
    \pi_0
    \left\{
    \begin{tikzcd}[
      row sep = 7pt,
      column sep=15pt
    ]
      && 
      \QuantumSymmetries
      \backsslash 
      \GradedFredholmOperators
      \ar[
        dd,
        ->>
      ]
      \\
      \\
      \mathcal{X}
      \ar[
        rr,
        "{ \tau }"{description}
      ]
      \ar[
        uurr,
        dashed
      ]
      &&
      \mathbf{B} \QuantumSymmetries
    \end{tikzcd}
    \right\}
    \mathrlap{.}
  \end{equation}
\end{definition}

\begin{example}[Reduction to equivariant $\mathrm{KU}$]
  \label[example]
   {ReductionToEquivariantKU}
  \Cref{TwistedOrbifoldKTheory} is a ``globally equivariant'' orbifold cohomology theory (in specialization of \cref{PropertiesOfOrbifoldCohomology}) in that the definition does not refer to any particular equivariance/isotropy group, and yet the construction automatically reduces to equivariant $\mathrm{KU}$-theory $\mathrm{KU}_G(-)$ \parencites[p 10]{AtiyahSegal1969}[\S 6]{AtiyahSegal2004} --- when:
  \begin{enumerate}
  \item 
  the domain 
  $\mathcal{X}$ is presented as a global quotient $\mathcal{X} \simeq G \backsslash X$ (\cref{GlobalQuotientOrbifolds}), 
  \item the twist $\tau$ factors as a (stable, cf. \cite[Lem. 5.1.45]{SS25-Bun}) projective $G$-representation $\rho$
  \begin{equation}
    \begin{tikzcd}[sep=15pt]
      G
      \ar[
        r,
        "{
          \rho
        }"
      ]
      &
      \PUH
      \ar[
        rr,
        hook,
        "{
          i_{\mathrm{pu}}
        }",
        "{
          \scalebox{.7}{%
            \cref{IncludionOfPUintoQuantumSymmetries}%
          }%
        }"{swap}
      ]
      &&
      \QuantumSymmetries
    \end{tikzcd}
  \end{equation}
  in that
  \begin{equation}
    \label{TwistExpressingGRep}
    \begin{tikzcd}[
      row sep=10pt, 
      column sep=35pt
    ]
      G \backsslash X
      \ar[
        rrrr,
        uphordown,
        "{
          \tau_{\rho}
        }"
      ]
      \ar[
        r,
        ->>
      ]
      \ar[d]
      \ar[
        drr,
        "{ \sigma }"{description}
      ]
      &
      \mathbf{B}G
      \ar[
        rr,
        "{
          \mathbf{B}\rho
        }"
      ]
      &[-45pt]&[-45pt]
      \mathbf{B}\PUH
      \ar[
        r,
        "{
          \mathbf{B}i_{\mathrm{pu}}
        }"
      ]
      &
      \mathbf{B}\QuantumSymmetries
      \ar[
        dll,
        "{
          \mathbf{B}p_{\mathrm{ct}}
        }"
      ]
      \\
      \ast
      \ar[rr]
      &
      &
      \mathbf{B}
      \Big(
        \mathbb{Z}_2^{\tSymmetry}
        \times
        \mathbb{Z}_2^{\cSymmetry}
      \Big)
      \mathrlap{\,,}
    \end{tikzcd}
  \end{equation}
  \end{enumerate}
  then
  \begin{equation}
    \label{OrbiKReducingToEquivariantK}
    \mathrm{K}^{\tau_{\rho}}(G \backsslash X)
    \simeq
    \mathrm{KU}_G(X)
    \mathrlap{\,.}
  \end{equation}
  In particular this shows that the orbifold K-theory is independent of the choice of global quotient presentation, a key property that is less manifest otherwise (cf. \cite[Prop. 4.1]{PronkScull2010}).
\end{example}
\begin{proof}
Using the general results from \cref{SomeCohesiveHomotopyTheory} we straightforwardly obtain that:
\begin{equation}
  \begin{aligned}
  \mathrm{K}^{\tau_\rho}(
    G \backsslash X
  )
  &
    \defneq
    \pi_0
    \left\{
    \begin{tikzcd}[
      column sep=22pt,
      row sep=10pt
    ]
      &&
      \QuantumSymmetries
      \backsslash
      \GradedFredholmOperators
      \ar[
        dd,
        ->>
      ]
      \\
      \\
      G \backsslash X
      \ar[
        r,
        ->>
      ]
      \ar[
        uurr,
        dashed
      ]
      &
      \mathbf{B}G
      \ar[
        r,
        "{ 
          \mathbf{B}\rho 
        }"
      ]
      &
      \mathbf{B} \QuantumSymmetries
    \end{tikzcd}
    \right\}
  \\
  &
  \underset{\mathclap{
    \scalebox{.7}{\cref{PullbackAction}}
  }}{\simeq}
    \pi_0
    \left\{
    \begin{tikzcd}[
      column sep=16.5pt,
      row sep=10pt
    ]
      &[-10pt] 
      G
      \backsslash^{\!\rho}
      \FredholmOperators
      \ar[
        ddr,
        phantom,
        "{ \lrcorner }"{pos=.15}
      ]
      \ar[
        r
      ]
      \ar[
        dd,
        ->>
      ]
      &
      \QuantumSymmetries
      \backsslash
      \GradedFredholmOperators
      \ar[
        dd,
        ->>
      ]
      \\
      \\
      G \backsslash X
      \ar[
        r,
        ->>
      ]
      \ar[
        uur,
        dashed
      ]
      &
      \mathbf{B}G
      \ar[
        r,
        "{ 
          \mathbf{B}\rho 
        }"
      ]
      &
      \mathbf{B} \QuantumSymmetries
    \end{tikzcd}
    \right\}
    \\
    & 
    \underset{\mathclap{
      \scalebox{.7}{\cref{GEquivariantMapsAsSliceMapsOverBG}}
    }}{\simeq}
    \pi_0
    \left\{
      \adjustbox{raise=-7pt}{
      \begin{tikzcd}[
        column sep=23pt
      ]
        X
        \ar[
          in=60,
          out=180-60,
          looseness=4,
          shift right=1pt,
          "{
            \,\mathclap{G}\,
          }"{description}
        ]
        \ar[
          r,
          dashed
        ]
        &
        \FredholmOperators
        \ar[
          in=60,
          out=180-60,
          looseness=4,
          shift right=1pt,
          "{
            \,\mathclap{G^{\mathrlap{\rho}}}\,
          }"{description}
        ]
      \end{tikzcd}
      }
    \right\}
    \mathrlap{.}
  \end{aligned}
\end{equation}
But the last line is the traditional definition of $\mathrm{KU}_G$ (cf. \parencites[p 10]{AtiyahSegal1969}).
\end{proof}

We saw around \cref{GEquivariantCohomologyAsStackyMaps} that equivariance is just a special case of twisting, and indeed, in direct generalization of \cref{ReductionToEquivariantKU}:
\begin{example}[Reduction to twisted equivariant $\mathrm{KU}$]
When:
\begin{enumerate}
  \item 
  the domain 
  $\mathcal{X}$ is presented as a global quotient $\mathcal{X} \simeq G \backsslash X$ (\cref{GlobalQuotientOrbifolds}), 
  \item 
  the twist $\tau$ classifies a $G$-equivariant $\PUH$-principal bundle $\begin{tikzcd}[sep=small]\mathcal{P} \ar[r, "{p}"] & X \end{tikzcd}$
  in that it factors through a corresponding equivariant {\v C}ech cocycle (cf.  \cref{PrincipalBundleAsPullbackAlongCechCocycle}, the arrows below denote maps of stacks, hence notationally suppressing the inverse equivalences seen on the left of \cref{PullbackOfUniversalStackyPrincipalBundleAlongCechCocycle})
  \begin{equation}
    \begin{tikzcd}
     G \backsslash
      \mathcal{P}
      \ar[
        r
      ]
      \ar[
        d
      ]
      \ar[
        dr,
        phantom,
        "{ \lrcorner }"{pos=.1}
      ]
      &
      \PUH \backsslash \PUH
      \ar[
        d,
        ->>,
        "{\,
          p_{\mathrm{univ}}
        }"
      ]
      \\
      G \backsslash X
      \ar[
        r,
        "{ \gamma }"
      ]
      &
      \mathbf{B}\PUH
    \end{tikzcd}
  \end{equation}
  as
  \begin{equation}
    \label{TwistEncodingAPUHBundle}
    \begin{tikzcd}[
      row sep=17pt, 
      column sep=35pt
    ]
      G \backsslash X
      \ar[
        rrr,
        uphordown,
        "{
          \tau_{{}_{\mathcal{P}}}
        }"
      ]
      \ar[
        d
      ]
      \ar[
        rr,
        "{
          \gamma
        }"
      ]
      \ar[
        dr,
        "{ \sigma }"
      ]
      &[-15pt]&[-65pt]
      \mathbf{B}\PUH
      \ar[
        r,
        "{
          \mathbf{B}
          i_{\mathrm{pu}}
        }"
      ]
      &
      \mathbf{B}
      \QuantumSymmetries
      \ar[
        dll,
        "{
          \mathbf{B}p_{\mathrm{ct}}
        }"
      ]
      \\
      \ast
      \ar[r]
      &
      \mathbf{B}\big(
        \mathbb{Z}_2^{\tSymmetry}
        \times
        \mathbb{Z}_2^{\cSymmetry}
      \big)
    \end{tikzcd}
  \end{equation}
\end{enumerate}
then the twisted orbi-K-theory (\cref{TwistedOrbifoldKTheory}) reduces to the corresponding $\mathcal{P}$-twisted equivariant K-theory $\mathrm{KU}_G^{[\mathcal{P}]}(X)$ of \cite[\S 6]{AtiyahSegal2004}:
\begin{equation}
  \mathrm{K}^{\tau_{\!{}_{\mathcal{P}}}}\big(
    G \backsslash X
  \big)
  \simeq
  \mathrm{KU}^{[P]}_G(X)
  \,.
\end{equation}
\end{example}
\begin{proof}
  As before, by the universal property of the pullback:
  \begin{equation}
    \begin{aligned}
      \mathrm{K}^{\tau_{\!{}_{\mathcal{P}}}}
      (G \backsslash X)
      & 
      \defneq
      \pi_0
      \left\{
      \hspace{.75cm}
      \begin{tikzcd}[
        column sep=25pt,
        row sep=22pt
      ]
      &[+6pt]&
      \QuantumSymmetries
       \backsslash 
      \GradedFredholmOperators
      \ar[
        d
      ]
      \\
      G \backsslash X
      \ar[
        r,
        "{ \gamma }"{description}
      ]
      \ar[
        urr,
        dashed
      ]
      &
      \mathbf{B}\PUH
      \ar[
        r,
        "{
          \mathbf{B}
          i_{\mathrm{pu}}
        }"
      ]
      &
      \mathbf{B}
      \QuantumSymmetries
      \end{tikzcd}
      \right\}
      \\
      &
      \underset{\mathclap{
        \scalebox{.7}{%
          \cref{PullbackAction}%
        }
      }}{
        \simeq
      }
      \pi_0
      \left\{
      \hspace{.75cm}
      \begin{tikzcd}[
        column sep=10pt,
        row sep=22pt
      ]
      &[+10pt]
      \PUH 
      \backsslash 
      \GradedFredholmOperators
      \ar[r]
      \ar[d]
      \ar[
        dr,
        phantom,
        "{ \lrcorner }"{pos=.1}
      ]
      &
      \QuantumSymmetries
       \backsslash 
      \GradedFredholmOperators
      \ar[
        d
      ]
      \\
      G \backsslash X
      \ar[
        r,
        "{ \gamma }"{description}
      ]
      \ar[
        ur,
        dashed
      ]
      &
      \mathbf{B}\PUH
      \ar[
        r,
        "{
          \mathbf{B}
          i_{\mathrm{pu}}
        }"
      ]
      &
      \mathbf{B}
      \QuantumSymmetries
      \end{tikzcd}
      \right\}
      \\
      &
      \underset{\mathclap{
        \scalebox{.7}{%
          \cref{PullbackOfUniversalStackyFiberBundleAlongCechCocycle}
        }
      }}{\simeq}
      \pi_0
      \left\{
      \begin{tikzcd}[
        column sep=8pt,
        row sep=22pt
      ]
      G\backsslash(
      \mathcal{P}
      \times_\Gamma
      \FredholmOperators
      )
      \ar[
        r,
        shorten=-2pt
      ]
      \ar[d]
      \ar[
        dr,
        phantom,
        "{ \lrcorner }"{pos=.1}
      ]
      &
      \PUH 
      \backsslash 
      \FredholmOperators
      \ar[
        r,
        shorten=-2pt
      ]
      \ar[d]
      \ar[
        dr,
        phantom,
        "{ \lrcorner }"{pos=.1}
      ]
      &[0pt]
      \QuantumSymmetries
       \backsslash 
      \GradedFredholmOperators
      \ar[
        d
      ]
      \\
      G \backsslash X
      \ar[
        r,
        "{ \gamma }"{description}
      ]
      \ar[
        u,
        bend left=30,
        dashed,
        shift left=3pt
      ]
      &
      \mathbf{B}\PUH
      \ar[
        r,
        "{
          \mathbf{B}
          i_{\mathrm{pu}}
        }"
      ]
      &
      \mathbf{B}
      \QuantumSymmetries
      \end{tikzcd}
      \right\}
      \\
      & 
      \underset{\mathclap{%
        \scalebox{.7}{\cref{SectionsAsSliceMapping}}
      }}{\simeq}
      \pi_0
      \,
      \Gamma_{\!{}_{X}}\big(
        \mathcal{P}
        \!\otimes_\Gamma\!
        \FredholmOperators
      \big)^{\!G}
      .
    \end{aligned}
  \end{equation}
  But the last expression is just the traditional definition (\cite[Def. 6.1]{AtiyahSegal2004}) of twisted equivariant $\mathrm{KU}^0$. (The equivariant local triviality condition on $\mathcal{P}$ which is required by \cite{AtiyahSegal2004} is actually implied by our stacky construction, see \cite[\S 5.2 \& Ex. 6.1.2]{SS25-Bun}.)
\end{proof}

\begin{example}[Reduction to $\mathrm{KR}$ in any degree]
  \label[example]
   {ReductionOfOrbiKToAnyKRDegree}
  When:
  \begin{enumerate}
    \item 
    the grading $\sigma$ is \emph{globally constant} over $\mathcal{X}$, in that there is a PCT-subgroup $G$
    of which $\mathcal{X}$ is a global quotient orbifold 
    (\cref{GlobalQuotientOrbifolds})
    of a manifold $X$ fixed by the action:     
    \begin{equation}
      \begin{tikzcd}
      G
      \ar[
        r, 
        hook,
        "{ \iota }"
      ]
      &
      \mathbb{Z}_2^{\tSymmetry}
      \times
      \mathbb{Z}_2^{\cSymmetry}
      \end{tikzcd}
      \,\;\;\;
      \mathcal{X}
      \simeq
      G
      \backsslash
      X
      \simeq
      X
      \times
      \mathbf{B}G
      \mathrlap{\,,}
    \end{equation}
    \item and there is no further twist besides a PCT quantum symmetry $\big[\widehat{(-)}\big]$
    $$
      \begin{tikzcd}[
        column sep=50pt
      ]
        \mathcal{X}'
        \times
        \mathbf{B}G
        \ar[
          dr,
          "{
            \sigma
          }"{description}
        ]
        \ar[
          r,
          "{
            \mathrm{pr}_2
          }"{description}
        ]
        \ar[
          rr,
          uphordown,
          "{
            \tau_{\widehat{G}}
          }"
        ]
        &
        \mathbf{B}
        G
        \ar[
          d,
          "{
            \mathbf{B}\iota
          }"
        ]
        \ar[
          r,
          "{
            \big[
              \widehat{(-)}
            \big]
          }"{description, pos=.45}
        ]
        &
        \mathbf{B}
        \QuantumSymmetries
        \ar[
          dl,
          "{
            \mathbf{B}
            p_{\mathrm{ct}}
          }"
        ]
        \\
        &
        \mathbf{B}
        \big(
          \mathbb{Z}_2^{\tOperator}
          \times
          \mathbb{Z}_2^{\cOperator}
        \big)
        \mathrlap{\,,}
      \end{tikzcd}
    $$
  \end{enumerate}
then the $\tau_{\widehat{G}}$-twisted orbi K-cohomology according to \cref{TwistedOrbifoldKTheory} reduces to the flavor and degree of K-theory corresponding to the PCT quantum symmetry as shown in \cref{TableOfPCTQuantumSymmetries}:
$$
  \mathrm{K}
    ^{\tau_{\widehat{G}}}%
  (X \times \mathrm{B}G)
  \simeq
  \left\{
  {
  \def\arraystretch{1.3}
  \begin{array}{lcl}
    \mathrm{KU}^0(X)
    &\big\vert&
    \widehat{G}
    =
    \mbox{A}
    \\
    \mathrm{KU}^1(X)
    \sqcup \{+,-\}
    &\big\vert&
    \widehat{G}
    =
    \mbox{AIII}
    \\
    \mathrm{KO}^0(X)
    &\big\vert&
    \widehat{G}
    =
    \mbox{AI}
    \\
    \mathrm{KO}^1(X)
    \sqcup \{+,-\}
    &\big\vert&
    \widehat{G}
    =
    \mbox{BDI}
    \\
    \mathrm{KO}^2(X)
    &\big\vert&
    \widehat{G}
    =
    \mbox{D}
    \\
    \mathrm{KO}^3(X)
    &\big\vert&
    \widehat{G}
    =
    \mbox{BDI}
    \\
    \mathrm{KO}^4(X)
    &\big\vert&
    \widehat{G}
    =
    \mbox{AII}
    \\
    \mathrm{KO}^5(X)
    \sqcup \{+,-\}
    &\big\vert&
    \widehat{G}
    =
    \mbox{CII}
    \\
    \mathrm{KO}^6(X)
    &\big\vert&
    \widehat{G}
    =
    \mbox{C}
    \\
    \mathrm{KO}^7(X)
    &\big\vert&
    \widehat{G}
    =
    \mbox{CI}
    \mathrlap{\,.}
  \end{array}
  }
  \right.
$$
\end{example}
\begin{proof}
Using the general results from \cref{SomeCohesiveHomotopyTheory} we obtain straightforwardly:
\begin{equation}
 \begin{aligned}
  \mathrm{K}
    ^{\tau_{\widehat{G}}}%
  (X \times \mathbf{B}G)
  & 
  \defneq
  \pi_0
  \left\{
  \begin{tikzcd}[
    column sep=20.5pt
  ]
    &&
    \QuantumSymmetries 
    \backsslash
    \GradedFredholmOperators
    \ar[
      d
    ]
    \\
    X
    \times
    \mathbf{B}G
    \ar[
      urr,
      dashed
    ]
    \ar[
      r
    ]
    &
    \mathbf{B}
    G
    \ar[
      r,
      "{
         [\widehat{-}]
      }"{pos=.45}
    ]
    &
    \mathbf{B}\QuantumSymmetries
  \end{tikzcd}
  \right\}
  \\
  & 
  \underset{\mathclap{
    \scalebox{.7}{\cref{PullbackAction}}
  }}{\simeq}
  \pi_0
  \left\{
  \begin{tikzcd}[
    column sep=8pt
  ]
    &
    \widehat{G}
    \backsslash
    \GradedFredholmOperators
    \ar[
      r,
      shorten <=-3pt,
    ]
    \ar[
      d
    ]
    &
    \QuantumSymmetries 
    \backsslash
    \GradedFredholmOperators
    \ar[
      d
    ]
    \\
    X
    \times
    \mathbf{B}G
    \ar[
      ur,
      dashed
    ]
    \ar[
      r
    ]
    &
    \mathbf{B}
    G
    \ar[
      r,
      "{
        [\widehat{-}]
      }"
    ]
    &
    \mathbf{B}\QuantumSymmetries
  \end{tikzcd}
  \right\}
  \\
  & 
  \underset{\mathclap{
    \scalebox{.7}{\cref{LiftsOfBGThroughQuotientProjectionAreFixedPoints}}
  }}{\simeq}
  \pi_0
  \left\{
  \begin{tikzcd}[
    column sep=25pt
  ]
    X
    \ar[
      r,
      dashed
    ]
    &
    \GradedFredholmOperators
      ^{ \widehat{G} }
  \end{tikzcd}
  \right\}
  \mathrlap{\,.}
  \end{aligned}
\end{equation}
But the last line implies the claim by \cref{The10PCTFixedLoci}.
\end{proof}

These examples show that the twisted orbi-orientifold K-theory of  \cref{TwistedOrbifoldKTheory} unifies and mixes these aspects of K-theory: $\mathbb{R}$/$\mathbb{C}$/$\mathbb{H}$-flavors, degrees, twisting and equivariance. Along the same lines the reader can find the sectors of twisted orbi $\mathrm{KU}^1$ and $\mathrm{KO}^n$, etc.
 
In the following subsections we focus on the $\mathrm{KU}^0$-sector, twisted/orbifolded by $\Gamma \in \{ \mathrm{U}(2), \mathrm{Sp}(2) \}$.

\subsection
  {The Equivariant Tautological Line Bundle}
\label{TheEquivariantLineBundle}

In preparation of constructing the four/ten-dimensional $\mathrm{U}(2)$/$\mathrm{Sp}(2)$-equivariant orientation of orbi-$\mathrm{KU}^0$, below \cref{TheEquivariantOrientation},
here  we discuss the canonical $\mathrm{U}(2)$/$\mathrm{Sp}(2)$-equivariant $\mathbb{C}$/$\mathbb{H}$-line bundle on $\mathbb{C}P^1$/$\mathbb{H}P^1$, and the equivariant trivialization of its pullback along the $\mathbb{C}$/$\mathbb{H}$-Hopf fibration.
We use quaternionic 2-component spinor calculus (following \cite[\S 2.1]{FSS22-Twistorial}) in a way that lends itself to constructing the required Fredholm operator families.

This subsection uses well-known constructions in quaternionic algebra, and the result will not be surprising to experts. But since it does not seem properly citable, and in order to establish notation needed in \cref{TheEquivariantOrientation},  we spell out the pleasant construction explicitly.

\subsubsection{Quaternion algebra}

\begin{definition}[{cf. \parencites{Zhang1997}[\S 1]{Morais2013}}]
  \label[definition]{Quaternions}
  We write:
  
 \begin{itemize}

  \item $\mathbb{H}$ for the real vector space of \emph{quaternions}, 

  \item with their real associative product $(-)\cdot(-)$,
  
  \item which is a real star-algebra via conjugation $(-)^\ast$, 
  
  \item and equipped with a norm ${\vert-\vert}$ given by ${\vert q \vert}^2 = q q^\ast$.

  \item $\mathbb{H}_{\mathrm{in}} \subset \mathbb{H}$ for the real subspace of \emph{imaginary quaternions}, $q^\ast \,=\, - q$,

\item for which we choose, as usual, an orthonormal linear basis $\mathbf{i}, \mathbf{j}, \mathbf{k} \in \mathbb{H}_{\mathrm{im}}$ such that 
\begin{equation}
  \mathbf{i} \cdot \mathbf{j} \,=\, \mathbf{k}
  \mathrlap{\,,}
\end{equation}

  \item  which in fact generates $\mathbb{H}$ subject to the further relations 
  \begin{equation}
    \begin{aligned}
    \mathbf{i}\, \mathbf{j} 
    =  - \mathbf{j} \, \mathbf{i}
    , \quad 
    \mathbf{i}^2 
     = \mathbf{j}^2 \,=\, -1
    \mathrlap{\,,}
    \end{aligned}
  \end{equation}

  \item and exhibits a star-algebra inclusion of the complex numbers:
  \begin{equation}
  \label{ComplexNumbersInsideQuaternions}
  \begin{tikzcd}[sep=-3pt]
    \mathbb{C} 
    \ar[rr, hook]
    &&
    \mathbb{H}
    \\
    \mathrm{i} &\longmapsto& \mathbf{i}
    \mathrlap{\,.}
  \end{tikzcd}
  \end{equation}
  \item 
  $\BoundedOperators(\mathbb{H}^2)$
  for the real algebra of $2 \times 2$ matrices with quaternion entries,
  \item 
  with star-operation
  \begin{equation}
    \label{MatrixConjugation}
    \left(
    \begin{matrix}
      a & b
      \\
      c & d
    \end{matrix}
    \right)^\dagger
    :=
    \left(
    \begin{matrix}
      a^\ast & c^\ast
      \\
      b^\ast & d^\ast 
    \end{matrix}
    \right)
    \mathrlap{.}
  \end{equation}
  \item
  $\mathrm{Sp}(1) = S(\mathbb{H})$ for the unit norm quaternions, with their group structure under quaternion multiplication.
\end{itemize} 
\end{definition}

\begin{example}[{cf. \cite[\S 1.27]{Morais2013}}]
  \label[example]
  {PauliMatricesAsStarRepresentation}
  The {\it Pauli matrices} define a homomorphism of real star-algebras from the quaternions (\cref{Quaternions}) to the linear operators on $\mathbb{C}^2$:
  \begin{equation}
    \label{MatrixRepresentationOfQuaternions}
    \begin{tikzcd}[
      sep=9pt
    ]
      \mathbb{H}
      \ar[
        rr,
        "{ \CliffordElement }"
      ]
      &&
      \BoundedOperators(\mathbb{C}^2)
      \mathrlap{\,,}
    \end{tikzcd}
  \end{equation}
given by 
$$
    \begin{tikzcd}[
      sep=-2pt,
      ampersand replacement=\&
    ]
        1 
        \&\longmapsto\&
      \left[
      \begin{array}{@{\hspace{-2pt}}r@{\hspace{3pt}}r@{}}
        \phantom{+}1 & \phantom{+}0 
        \\
        \phantom{+}0 & \phantom{+}1
      \end{array}
      \right]
     , \quad 
      \mathbf{i}
        \&\longmapsto\&
      \mathrm{i}
      \left[
      \begin{array}{@{\hspace{-2pt}}r@{\hspace{3pt}}r@{}}
        \phantom{+}1 & \phantom{+}0 
        \\
        \phantom{+}0 & \scalebox{.94}{$-$}\hspace{-.8pt}1
      \end{array}
      \right]
      ,\quad 
      \mathbf{j}
        \&\longmapsto\&
      \mathrm{i}
      \left[
      \begin{array}{@{\hspace{-2pt}}r@{\hspace{3pt}}r@{}}
        \phantom{+}0 & \phantom{+}1 
        \\
        \phantom{+}1 & \phantom{+}0
      \end{array}
      \right]
      ,\quad 
      \mathbf{k}
        \&\longmapsto\&
      \mathrm{i}
      \left[
      \begin{array}{@{\hspace{-2pt}}r@{\hspace{3pt}}r@{}}
        \phantom{+}0 & \phantom{+}\mathrm{i} 
        \\
        - \mathrm{i} & \phantom{+}0
      \end{array}
      \right]
      \mathrlap{.}
    \end{tikzcd}
$$
  One immediately checks that for a general quaternion $x := x_0  + x_1 \mathbf{i} + x_2 \mathbf{j} + x_3 \mathbf{k} \in \mathbb{H}$, with $(x_i \in \mathbb{R})_{i=0}^3$, one has 
  $
    \mathrm{det}(
      \CliffordElement_{x}
    )
    =
    \textstyle{\sum_{i=0}^3} (x_i)^2
    =
    \CliffordElement_x 
    \CliffordElement_x^\dagger
  $.
  This implies that under the matrix representation \cref{MatrixRepresentationOfQuaternions} the unit norm quaternions are identified with the special unitary $2 \times 2$ matrices:
  \begin{equation}
    \label{UnitQuaternionsAsSU2}
    \begin{tikzcd}
      S(\mathbb{H})
      \;=\;
      \mathrm{Sp}(1)
      \ar[
        r,
        "{ \CliffordElement }",
        "{ \sim }"{swap}
      ]
      &
      \mathrm{SU}(2)
      \,\subset\,
      \BoundedOperators(\mathbb{C}^2)
      \,.
    \end{tikzcd}
  \end{equation}
\end{example}

\begin{definition}[cf. {\parencites[p 11]{Cohen2000}[(2.4)]{VenancioBatista2021}}]
\label[definition]{GroupOfUnimodularHMatrices}
  The \emph{group of unimodular quaternionic $2 \times 2$ matrices} is
  \begin{equation}
    \label{SL2H}
    \mathrm{SL}(\mathbb{H}^2)
    :=
    \Big\{
      A 
        \in 
      \BoundedOperators(\mathbb{H}^2)
      \;\big\vert\;
      \mathrm{det}_D(A)
      =
      1
    \Big\}
    \mathrlap{\,,}
  \end{equation}
  where
  \begin{equation}
    \label{QuaternionicDeterminant}
    \mathrm{det}_D
    \left(
    \begin{matrix}
      a & b 
      \\
      c & d
    \end{matrix}
    \right)
    =
    \left\{
    {
    \renewcommand{\arraystretch}{1.3}
    \begin{array}{ccl}
      0 
        & \mathrm{if} & 
      b = c = d = 0
      \\
      \big\vert
        a d - a c a^{-1} b
      \big\vert
        & \mathrm{if} & 
      a \neq 0
      \\
      \big\vert
        d a - d b d^{-1} c
      \big\vert
        & \mathrm{if} & 
      d \neq 0
      \\
      \big\vert
        b d b^{-1} a - b c
      \big\vert
        & \mathrm{if} & 
      b \neq 0
      \\
      \big\vert
        c a c^{-1} d - c b
      \big\vert
        & \mathrm{if} & 
      c \neq 0
      \mathrlap{\,.}
    \end{array}
    }
    \right.
  \end{equation}
\end{definition}
\begin{definition}[cf. {\cite[p. 28]{Zhang1997}}]
  The \emph{group of unitary quaternionic $2 \times 2$ matrices} (also: \emph{compact symplectic group}) is:
  \begin{equation}
    \label{SP2}
    \mathrm{Sp}(2)
    \equiv
    \mathrm{U}(\mathbb{H}^2)
    :=
    \Big\{
      U \in 
      \BoundedOperators(\mathbb{H}^2)
      \;\big\vert\;
      U \cdot U^\dagger = 1
    \Big\}
    \mathrlap{\,.}
  \end{equation}
\end{definition}
\begin{remark}
  The condition in \cref{SP2} is indeed sufficient, 
  because for $B \in \BoundedOperators(\mathbb{H}^2)$ we have (cf. \cite[Prop. 4.1]{Zhang1997}):
    \begin{equation}
      B \cdot B^\dagger = 1
      \hspace{.5cm}
      \Leftrightarrow
      \hspace{.5cm}
      B^\dagger \cdot B = 1
      \mathrlap{\,.}
    \end{equation}
\end{remark}
\begin{remark}
  \label[remark]{StarAlgebraHomFromPauliMatrices}
  Applying the star-algebra homomorphism $\CliffordElement$ \cref{MatrixRepresentationOfQuaternions} entrywise yields a star-algebra homomorphism
  \begin{equation}
    \label{QuaternionicMatricesAsComplexMatrices}
    \begin{tikzcd}[
      ampersand replacement=\&,
      row sep=-2pt, column sep=0pt
    ]
      \BoundedOperators(\mathbb{H}^2)
      \ar[
        rr,
        "{ \gamma }"
      ]
      \&\&
      \BoundedOperators(\mathbb{C}^4)
      \\
      \left(
      \begin{matrix}
        a & b
        \\
        c & d
      \end{matrix}
      \right)
      \&\longmapsto\&
      \left(
      \begin{matrix}
        \gamma_a & \gamma_b
        \\
        \gamma_c & \gamma_d
      \end{matrix}
      \right)      
      \mathrlap{\,,}
    \end{tikzcd}
  \end{equation}
  where on the right the 2x2 matrix of complex 2x2 matrices is canonically regarded as a 4x4 matrix.
  In particular,  \eqref{QuaternionicMatricesAsComplexMatrices} gives a 
  complex-unitary representation of the quaternionic unitary group \cref{SP2}:
  \begin{equation}
    \label{TheComplexSp2Representation}
    \begin{tikzcd}
      \mathbb{C}^4
      \ar[
        out=180-60,
        in=60,
        looseness=4,
        shift right=2.3pt,
        "{
          \hspace{4pt}\mathclap{\mathrm{Sp}(2)}\hspace{5pt}
        }"{description}
      ]
    \end{tikzcd}
    :
    \begin{tikzcd}[
      ampersand replacement=\&,
    ]
      \mathrm{Sp}(2)
      \ar[
        r,
        "{
          \gamma
        }"
      ]
      \&
      \mathrm{SU}(4)
      \mathrlap{\,.}
    \end{tikzcd}
  \end{equation}
\end{remark}
\begin{lemma}[cf. {\cite[Cor. 6.4]{Cohen2000}}]
  \label[lemma]{Sp2InSL2H}
  Every unitary matrix according to \cref{SP2} is unimodular according to \cref{SL2H}:
  \begin{equation}
    \label{Sp2AsSubgroupOfSL2H}
    \mathrm{Sp}(2)
    \subset
    \mathrm{SL}(\mathbb{H}^2)
    \mathrlap{\,.}
  \end{equation}
\end{lemma}
\begin{proof}
  The characteristic properties of $\mathrm{det}_D$ \cref{QuaternionicDeterminant} are (cf. \cite[Thm. 5.1,  \& Cor. 6.4]{Cohen2000}), for all $A,B, \in \mathrm{Mat}_{2 \times 2}(\mathbb{H})$
  \begin{equation}
    \label{PropertiesOfQuaternionicDeterminant}
    \begin{aligned}
      \mathrm{det}_D(A) 
      & \in \mathbb{R}_{\geq 0}
      \\
      \mathrm{det}_D(1)
      & 
      = 
      1
      \\
      \mathrm{det}_D(A \cdot B)
      &
      =
      \mathrm{det}_D(A)
      \cdot
      \mathrm{det}_D(B)
      \\
      \mathrm{det}_D\big(A^\dagger\big)
      & = 
      \mathrm{det}_D(A)
      \mathrlap{\,.}
    \end{aligned}
  \end{equation}
  Thereby, $U \cdot U^\dagger = 1$ implies
  \begin{equation}
    \begin{aligned}
      1 & = 
      \mathrm{det}_D\big(  
        U \cdot U^\dagger
      \big)
      \\
      & =
      \mathrm{det}_D(U) 
      \cdot
      \mathrm{det}_D\big(U^\dagger\big)
      \\
      & =
      \big(\mathrm{det}_D(U)\big)^2
      \,,
    \end{aligned}
  \end{equation}
  and the only solution to that in $\mathbb{R}_{\geq 0}$ is $\mathrm{det}_D(U) = 1$.
\end{proof}
\begin{remark}
  \Cref{Sp2InSL2H} entails that the quaternionic unitary group \cref{SP2} plays the role of a \emph{special} unitary group, certainly in its following appearance in \cref{6DSpacetimeViaQuternionMatrices}. In view of this, the notation ``$\mathrm{Sp}(2)$'' is more suggestive than ``$\mathrm{U}(2,\mathbb{H})$'', with its (coincidental but fortunate) alliteration to ``$\mathrm{SU}(2)$''.
\end{remark}

\subsubsection{The Tautological Line Bundle}

What drives the following construction is this fact:

\begin{proposition}[{\cite{Kugo1983}, streamlined review in \parencites{BaezHuerta2010, VenancioBatista2021}[\S 3.2]{FSS21-Emergence}}]
\label[proposition]{6DSpacetimeViaQuternionMatrices}
The following isomorphism of quadratic real vector spaces, from 6D Minkowski spacetime with signature $\eta := \mathrm{diag}(-1,+1, \cdots, +1)$, to the space of hermitian $2 \times 2$ quaternionic matrices with quadratic form being minus the ordinary determinant, intertwines the canonical $\mathrm{Spin}(1,5) \xrightarrow{\;} \mathrm{SO}(1,5)$-action on the left with the conjugation action 
\begin{equation}
  \label{ActionOfSL2HOnHermitianMatrices}
  \begin{tikzcd}[row sep=-4pt, column sep=-1pt]
    \mathrm{SL}(\mathbb{H}^2)
    \ar[
      rr,
      "{ \mathrm{conj} }"
    ]
    &&
    \mathrm{Aut}_{\mathbb{R}}\Big(
      \big\{
        A \in \BoundedOperators({\mathbb{H}^2})
        \,\big\vert\,
        A^\dagger = A
      \big\}
    \Big)
    \\
    A &\longmapsto&
    A \cdot (-) \cdot A^\dagger
  \end{tikzcd}
\end{equation}
of $\mathrm{SL}(\mathbb{H}^2)$ \textup{(\cref{GroupOfUnimodularHMatrices})} on the right:
\begin{equation}
  \label{6DMinkowskiSpacetimeAsQuaternionMatrices}
  \begin{tikzcd}[
    column sep=2pt,
   row sep=-3pt
  ]
  \mathrm{Spin}(1,5)
  \acts_{\mathrm{can}}
  \big(
    \mathbb{R}^{1,5},
    \eta
  \big)
  \ar[
    rr,
    "{
      \sim
    }"
  ]
  &[-20pt]&
  \mathrm{SL}(\mathbb{H}^2)
  \acts_{\scalebox{.6}{\cref{ActionOfSL2HOnHermitianMatrices}}}
  \Big(
   \big\{
     A 
     \in
     \BoundedOperators(\mathbb{H}^2)
     \,\vert\,
     A^\dagger = A
   \big\},
   - \mathrm{det}
  \Big)
  \\
  \phantom{---}
  \left[
  \begin{matrix}
    x^0
    \\
    x^1
    \\
    \vdots
    \\
    x^5
  \end{matrix}
  \right]
  &\longmapsto&
  \left(
  {
  \setlength{\arraycolsep}{1pt}
  \begin{matrix}
    x^0 - x^1
    &
    x^2 
    + \mathbf{i} x^3 
    + \mathbf{j} x^4
    + \mathbf{k} x^5
    \\
    x^2 
    - \mathbf{i} x^3 
    - \mathbf{j} x^4
    - \mathbf{k} x^5
    & x^0 + x^1
  \end{matrix}
  }
  \right),
  \end{tikzcd}
\end{equation}
and analogously so for 4D Minkowski spacetime and complex matrices, by restriction along the inclusion $\mathbb{C} \hookrightarrow \mathbb{H}$ \cref{ComplexNumbersInsideQuaternions}:
\begin{equation}
  \begin{tikzcd}[
   row sep=-3pt, column sep=0pt
  ]
  \mathrm{Spin}(1,3)
  \acts_{\mathrm{can}}
  \big(
    \mathbb{R}^{1,3},
    \eta
  \big)
  \ar[
    rr,
    "{
      \sim
    }"
  ]
  &&
  \mathrm{SL}(\mathbb{C}^2)
  \acts \,
  \Big(
    \big\{
      A \in 
      \BoundedOperators(\mathbb{C}^2)
      \,\big\vert\,
      A^\dagger = A
    \big\}
    ,
      - \mathrm{det}
  \Big)
  \\
  \left[
  \begin{matrix}
    x^0
    \\
    x^1
    \\
    \vdots
    \\
    x^3
  \end{matrix}
  \right]
  &\longmapsto&
  \left(
  {
  \setlength{\arraycolsep}{2pt}
  \begin{matrix}
    x^0 - x^1
    &
    x^2 
    + \mathrm{i} x^3 
    \\
    x^2 
    - \mathrm{i} x^3 
    & x^0 + x^1
  \end{matrix}
  }
  \right).
  \end{tikzcd}
\end{equation}
\end{proposition}

\begin{corollary}
\label[corollary]{IdentifyingThe4SphereViaMatrices}
Under the identification of \cref{6DSpacetimeViaQuternionMatrices},
the subspace $\mathbb{R}^5 \subset \mathbb{R}^{1,5}$ with its Euclidean norm $g$, and furthermore (pointed) $S^4 \subset \mathbb{R}^5$, are identified with those Hermitian matrices that are traceless and traceless unitary, respectively, and the corresponding subgroups $\mathrm{Spin}(4) \subset \mathrm{Spin}(5) \subset \mathrm{Spin}(1,5)$ are identified with $\mathrm{Sp}(1) \times \mathrm{Sp}(1) \subset \mathrm{Sp}(2) \subset \mathrm{SL}(\mathbb{H}^2)$ \cref{Sp2AsSubgroupOfSL2H},
as follows:
\begin{equation}
  \label{4SphereViaMatrices}
  \begin{tikzcd}[
    column sep=15pt,
    row sep=9pt
  ]
  \mathrm{Spin}(1,5)
  \acts_{\mathrm{can}}
  \big(
    \mathbb{R}^{1,5},
    \eta
  \big)
  \ar[
    r,
    "{
      \sim
    }"
  ]
  &
  \mathrm{SL}(\mathbb{H}^2)
  \acts_{\mathrm{conj}}
  \Big(
   \big\{
     A 
     \in
     \BoundedOperators(\mathbb{H}^2)
     \,\vert\,
     \substack{
       A^\dagger = A
      }
   \big\},
   - \mathrm{det}
  \Big)
    \\
    \mathrm{Spin}(5)
    \acts_{\mathrm{can}}
    \big(
      \mathbb{R}^{5},
      g
    \big)
    \ar[
      r,
      "{ \sim }"
    ]
    \ar[
      u,
      hook,
      shift right=9
    ]
    &
    \mathrm{Sp}(2)
    \acts_{\mathrm{conj}}
    \bigg(
      \Big\{
        A 
          \in 
        \BoundedOperators(\mathbb{H}^2)
        \,\Big\vert\,
        \substack{
          A^\dagger = A
          \\
          \mathrm{tr}(A) = 0
        }
      \Big\},
      -\mathrm{det}
    \bigg)
    \ar[
      u,
      hook,
      shorten >=-3pt,
      shorten <=-4pt
    ]
    \\
    \mathrm{Spin}(5)
    \acts_{\mathrm{can}}
    S^4,
    \ar[
      r,
      "{ \sim }"
    ]
    \ar[
      u, 
      hook,
      shift right=9
    ]
    &
    \mathrm{Sp}(2)
    \acts_{\mathrm{conj}}
      \bigg\{
        A 
          \in 
        \BoundedOperators(\mathbb{H}^2)
        \,\bigg\vert\,
        \substack{
          A^\dagger = A
          \\
          \mathrm{tr}(A) = 0
          \\
          A \cdot A^\dagger = 1
        }
      \bigg\}
    \ar[
      u, 
      hook,
      shorten >=-5pt,
      shorten <=-7pt
    ]
    \\
    \mathrm{Spin}(4)
    \acts_{\mathrm{can}}
    S^4_{\mathrm{nth}},
    \ar[
      r,
      "{ \sim }"
    ]
    \ar[
      u, 
      shift right=9
    ]
    &
    \mathrm{Sp}(1)^2
    \acts_{\mathrm{conj}}
      \bigg\{
        A 
          \in 
        \BoundedOperators(\mathbb{H}^2)
        \,\bigg\vert\,
        \substack{
          A^\dagger = A
          \\
          \mathrm{tr}(A) = 0
          \\
          A \cdot A^\dagger = 1
        }
      \bigg\}_{\!\!
        \mathrlap{
          \left(
          \substack{%
            1 \; \phantom{-}0%
            \\%
            0 \; -1%
          }
          \right)
          \mathrlap{\,.}
        }
      }
    \ar[
      u, 
      shorten >=-5pt,
      shorten <=-7pt
    ]
  \end{tikzcd}
\end{equation}
\end{corollary}
\begin{proof}
  From the component expression in \cref{6DMinkowskiSpacetimeAsQuaternionMatrices} it is manifest that the condition $x^0 = 0$, characterizing the subset $\mathbb{R}^5 \subset \mathbb{R}^{1,5}$, corresponds to vanishing trace on the Hermitian matrices.
    To see that its is precisely $\mathrm{Sp}(2) = \mathrm{U}(\mathbb{H}^2) \subset \mathrm{SL}(\mathbb{H}^2)$ which preserves this tracelessness condition, we note that this is equivalent to this subgroup preserving the orthogonal complement of the traceless Hermitian matrices. Again by the component expression in  \cref{6DMinkowskiSpacetimeAsQuaternionMatrices}, this orthogonal complement is spanned by the identity matrix $1$. Therefore $G \in \mathrm{SL}(\mathbb{H}^2)$ preserves the traceless matrices iff $G \cdot 1 \cdot G^\dagger = 1$, hence iff $G$ is unitary \cref{SP2}.

    Further from the component expression in \cref{6DMinkowskiSpacetimeAsQuaternionMatrices} one sees that  
    \begin{equation}
      \left.
      \substack{
        A^\dagger = A
        \\[2pt]
        \mathrm{tr}(A) = 0
      }
      \right\}
      \;\;\;
      \Rightarrow
      \;\;\;
      A A^\dagger = -\mathrm{det}(A) 
      \,,
    \end{equation}
    which shows that the unitary traceless Hermitian matrices form the unit sphere inside $\mathbb{R}^5$.
    Finally, it is readily seen that the subgroup of $\mathrm{Sp}(2)$ whose conjugation action fixes the base point 
    $\scalebox{0.7}{$\left(\begin{matrix}1 & \phantom{-}0 \\ 0  & -1\end{matrix}\right)$}$ among these matrices is the diagonal 
    $$
      \mathrm{Sp}(1) \times \mathrm{Sp}(1) \simeq \left(\begin{matrix}\mathrm{Sp}(1) & 0 \\ 0 & \mathrm{Sp}(1)\end{matrix}\right).
    $$
  This concludes the proof.
\end{proof}

\begin{remark}
  The direct analogue of \cref{IdentifyingThe4SphereViaMatrices} over the complex numbers gives
  \begin{equation}
    \label{SU2ActionOnS2}
    \begin{tikzcd}
      \mathrm{Spin}(3)
      \acts_{\mathrm{can}}
      S^2
      \ar[
        rr,
        "{ \sim }"
      ]
      &&
      \mathrm{SU}(2)
      \acts_{\mathrm{conj}}
      \left\{
        A \in
        \BoundedOperators(\mathbb{C}^2)
      \,\middle\vert\,
        \substack{
          A^\dagger = A
          \\
          \mathrm{tr}(A) = 0
          \\
          A \cdot A^\dagger = 1
        }
      \right\}
      \mathrlap{.}
    \end{tikzcd}
  \end{equation}
  However --- in line with the fact that $\mathrm{Sp}(2) = \mathrm{U}(\mathbb{H}^2)$ is actually the quaternionic unitary group \cref{SP2} and since unity of the determinant no longer plays a role when we conjugating matrices $A$ constrained to have $\mathrm{det}(A) = -1$ --- the $\mathrm{SU}(2)$-action on the right of \cref{SU2ActionOnS2} evidently extends to a $\mathrm{U}(2) \simeq \mathrm{Spin}^c(3)$-action:
  \begin{equation}
    \label{U2ActionOnS2}
    \begin{tikzcd}
      \mathrm{Spin}^c(3)
      \acts_{\mathrm{can}}
      S^2
      \ar[
        rr,
        "{ \sim }"
      ]
      &&
      \mathrm{U}(2)
      \acts_{\mathrm{conj}}
      \left\{
        A \in
        \BoundedOperators(\mathbb{C}^2)
      \,\middle\vert\,
        \substack{
          A^\dagger = A
          \\
          \mathrm{tr}(A) = 0
          \\
          A \cdot A^\dagger = 1
        }
      \right\}
      \mathrlap{.}
    \end{tikzcd}
  \end{equation}
\end{remark}

\begin{definition}[{cf. \cite[\S 5.3]{Brown1968}}]
\label[definition]{ProjectiveSpace}
For $\mathbb{K} \in \{\mathbb{R}, \mathbb{C}, \mathbb{H}\}$,
The \emph{$\mathbb{K}$-projective space} of dimension $n \in \mathbb{N}$ is the space of \emph{right} $\mathbb{K}$-lines in $\mathbb{K}^{n+1}$,
\begin{equation}
  \label{HP1}
  \begin{aligned}
  \mathbb{K}P^{n}
  &=
  \big\{
    v \cdot \mathbb{K}
    \subset 
    \mathbb{K}^{n+1}
    \;\big\vert\;
    v \in \mathbb{K}^{n+1} \setminus \{0\}
  \big\}
  \mathrlap{\,,}
  \end{aligned}
\end{equation}
traditionally denoted
\begin{equation}
  \label{TradNotationForLines}
  \big[
    v_1 : v_2 : \cdots : v_{n+1}
  \big]
  :=
  v \cdot \mathbb{K}
  \,,
  \;\;
  \mbox{
    for 
    $v = (v_1, \cdots, v_{n+1}) \in \mathbb{K}^{n+1} \setminus \{0\}$
  }
  \mathrlap{,}
\end{equation}
with standard injections
\begin{equation}
  \label{StandardInclusionOfProjectiveSpaces}
  \begin{tikzcd}[
    sep=0pt
  ]
    \mathbb{K}P^n
    \ar[
      rr, 
      hook
    ]
    &&
    \mathbb{K}P^{n+1}
    \\
    \big[
      v_1 : \cdots : v_{n+1}
    \big]
    &\mapsto&
    \big[
      v_1 : \cdots : v_{n+1} : 0
    \big]
    \mathrlap{\,,}
  \end{tikzcd}
\end{equation}
canonical fibrations
\begin{equation}
  \begin{tikzcd}[
    column sep=3pt,
    row sep=0pt
  ]
    \mathbb{R}P^{2n+1}
    \ar[
      rr,
      ->>,
    ]
    &&
    \mathbb{C}P^n
    \\
    v \cdot \mathbb{R}
    &\mapsto&
    v \cdot \mathbb{C}
    \mathrlap{\,,}
  \end{tikzcd}
  \hspace{1cm}
  \begin{tikzcd}[
    column sep=3pt,
    row sep=0pt
  ]
    \mathbb{C}P^{2n+1}
    \ar[
      rr,
      ->>,
    ]
    &&
    \mathbb{H}P^n
    \\
    v \cdot \mathbb{C}
    &\mapsto&
    v \cdot \mathbb{H}
    \mathrlap{\,,}
  \end{tikzcd}
\end{equation}
and their composites
\begin{equation}
  \label{MappingRP2nToCPn}
  \begin{tikzcd}
    \mathbb{R}P^{2n}
    \ar[
      r,
      hook
    ]
    \ar[
      rr,
      uphordown,
      "{ 
        f^{\mathbb{R}}_{\mathbb{C}} 
      }"{description}
    ]
    &
    \mathbb{R}P^{2n+1}
    \ar[
      r,
      ->>,
    ]
    &
    \mathbb{C}P^{n}
    \mathrlap{\,,}
  \end{tikzcd}
  \;\;\;\;\;
  \begin{tikzcd}
    \mathbb{C}P^{2n}
    \ar[
      r,
      hook
    ]
    \ar[
      rr,
      uphordown,
      "{ 
        f^{\mathbb{C}}_{\mathbb{H}} 
      }"{description}
    ]
    &
    \mathbb{C}P^{2n+1}
    \ar[
      r,
      ->>,
    ]
    &
    \mathbb{H}P^{n}
    \mathrlap{.}
  \end{tikzcd}
\end{equation}
These spaces
carry a canonical action of the unitary group $\mathrm{U}(n+1, \mathbb{K})$ \cref{SP2} by left multiplication:
\begin{equation}
  \label{Sp2ActionOnHP1}
  \begin{tikzcd}[
    sep=0
  ]
    \mathrm{O}(n+1) 
    \times
    \mathbb{R}P^n
    \ar[
      rr,
      "{ \mathrm{mult} }"
    ]
    &&
    \mathbb{R}P^{n}
    \\
    \mathrm{U}(n+1) 
    \times
    \mathbb{C}P^n
    \ar[
      rr,
      "{ \mathrm{mult} }"
    ]
    &&
    \mathbb{C}P^{n}
    \\
    \mathrm{Sp}(n+1) 
    \times
    \mathbb{H}P^n
    \ar[
      rr,
      "{ \mathrm{mult} }"
    ]
    &&
    \mathbb{H}P^{n}
    \\
    \big(
      G,\, v \cdot \mathbb{K}
    \big)
    &\longmapsto&
    G \cdot v \cdot \mathbb{K}
    \mathrlap{\,.}
  \end{tikzcd}
\end{equation}
\end{definition}

\begin{corollary}
      We have an identification of the two/four-sphere with the $\mathbb{C}/\mathbb{H}$-projective line \cref{HP1} which intertwines the canonical $\mathrm{Spin}^c(3)$/\,$\mathrm{Spin}(5)$-action on $S^2/S^4$ with the multiplication action \cref{Sp2ActionOnHP1} of $\mathrm{U}(2)$/\,$\mathrm{Sp}(2)$ on 
  $\mathbb{C}P^1$/\,$\mathbb{H}P^1$ \cref{HP1}:
  \begin{equation}
    \label{4SphereAsHP1}
    \begin{tikzcd}[
      row sep=0pt, column sep=15pt
    ]
      \mathrm{Spin}^c(3) \acts_{\mathrm{can}} S^2
      \ar[
        rr,
        "{ \sim }"
      ]
      &&
      \mathrm{U}(2) 
        \acts_{\mathrm{mult}}
      \mathbb{C}P^1,
      \\
      \mathrm{Spin}(5) \acts_{\mathrm{can}} S^4
      \ar[
        rr,
        "{ \sim }"
      ]
      &&
      \mathrm{Sp}(2) \acts_{\mathrm{mult}}
      \mathbb{H}P^1.
    \end{tikzcd}
  \end{equation}
\end{corollary}
\begin{proof}
We spell out the argument over $\mathbb{H}$:
By \cref{IdentifyingThe4SphereViaMatrices}, the 4-sphere is identified with traceless unitary hermitian matrices. These  are furthermore identified, first with rank=1 projection operators on $\mathbb{H}^2$, and then with lines in $\mathbb{H}^2$, like this:
\begin{equation}
  \label{IdentifyingS4WithHP1}
  \begin{tikzcd}[
    row sep=-2pt,
    column sep=-2pt
  ]
    S^4
    \ar[
      rr,
      <->,
      "{ \sim }",
      "{
        \scalebox{.6}{%
          \cref{4SphereViaMatrices}%
        }
      }"{swap, yshift=-1pt}
    ]
    &&
    \bigg\{
      A \in \BoundedOperators(\mathbb{H}^2)
      \,\bigg\vert\,
      \substack{
        A^\dagger = A
        \\
        \mathrm{tr}(A) = 0
        \\
        A\cdot A = 1
      }
    \bigg\}
    \ar[
      rr,
      <->,
      "{ \sim }"
    ]
    &&
    \bigg\{
      P \in \BoundedOperators(\mathbb{H}^2)
      \,\bigg\vert\,
      \substack{
        P^\dagger = P
        \\
        \mathrm{tr}(P) = 1
        \\
        P \cdot P = P
        \\
      }
    \bigg\}
    \ar[
      rr,
      <->,
      "{ \sim }"
    ]
    &&
    \mathbb{H}P^1
    \\
    &\hspace{12pt}&
    A 
      &\mapsto& 
    P_A := \tfrac{1}{2}(1 - A)
      &\mapsto&
    \mathrm{ker}(P_A)
    \\
    &&
    A_P := 
    (1 - 2 P) 
      &\mapsfrom& 
    P,\,
    P_v
    :=
    1
      -
    \tfrac
      {v \cdot v^\dagger}
      {\rule{0pt}{5.8pt}\vert v \vert^2} 
      &\mapsfrom&
    v \cdot \mathbb{H}
    \mathrlap{\,.}
  \end{tikzcd}
\end{equation}
Now, \cref{IdentifyingThe4SphereViaMatrices} also shows that the $\mathrm{Spin}(5)$-action on $S^4$ becomes the $\mathrm{Sp}(2)$-conjugation action \cref{ActionOfSL2HOnHermitianMatrices} on the matrices $A$. Under the above bijection \cref{IdentifyingS4WithHP1}, this translates first to the same kind of conjugation action on projection operators $P$, and then to left multiplication on $\mathbb{H}$-lines, as claimed, since:
\begin{equation}
  \label{KernelOfConjugatedProjectionOperator}
  \mathrm{ker}\big(
    G \cdot P \cdot G^\dagger
  \big)
  \simeq
  G \cdot \mathrm{ker}(P)
  \,.
  \qedhere
\end{equation}

The argument over $\mathbb{C}$ is verbatim the same, now using the $\mathrm{U}(2)$-action from \cref{U2ActionOnS2}.
\end{proof}

\begin{definition}
  \label[definition]{TheTautologicalHLineBundle}
  The kernel bundle of the family of the $\mathbb{C}P^1/\mathbb{H}P^1$-parameterized projection operators
  \cref{IdentifyingS4WithHP1}
  is the \emph{tautological $\mathbb{C}/\mathbb{H}$-line bundle} over $\mathbb{C}P^1/\mathbb{H}P^1$:
  \begin{equation}
    \label{TautologicalLineBundleAsSubbundle}
    \begin{tikzcd}[
      column sep=18pt,
      row sep=-1pt
    ]
    &
      \mathllap{
        \mathcal{L}^{\mathrm{taut}}_{\mathbb{K}P^1}
        :=
        \;
      }
      \Big\{
        (P,v)
        \,\Big\vert\,
        \substack{
          P \in \mathbb{K}P^1
          \\
          v \in \mathrm{ker}(P)
        }
      \Big\}
      \ar[
        rr,
        hook
      ]
      \ar[
        dr
      ]
      &&
      \mathbb{K}P^1
      \times
      \mathbb{K}^2
      \mathrlap{\,.}
      \ar[
        dl
      ]
      \\
&       & \mathbb{H}P^1
    \end{tikzcd}
  \end{equation}
  
\end{definition}

Next we are after the canonical trivialization of this tautological line bundle after pullback along the \emph{Hopf fibration}.

\subsubsection{The Hopf Fibration}

\begin{definition}
\label[definition]{HopfFibration}
  The $\mathbb{R}/\mathbb{C}/\mathbb{H}$-\emph{Hopf fibration} is the map which sends unit norm elements in $\mathbb{R}^2/\mathbb{C}^2/\mathbb{H}^2$ (\cref{Quaternions}) to the $\mathbb{R}/\mathbb{C}/\mathbb{H}$-lines which they span \cref{HP1}:
  \begin{equation}
    \label{TheHopfFibration}
    \begin{tikzcd}[
      column sep=-5pt,
      row sep=20pt
    ]
      S^1
      \ar[
        d,
        "{
          \RealHopfFibration
        }"
      ]
      &\simeq&
      S\big(
        \mathbb{R}^2
      \big)
      \ar[
        d
      ]
      &\ni&
      v
      \ar[
        d,
        |->,
        shorten=6pt
      ]
      \\
      S^1
      &\simeq&
      \mathbb{R}P^1
      &\ni&
      v \cdot \mathbb{R}
      \mathrlap{\,,}
    \end{tikzcd}
    \hspace{.8cm}
    \begin{tikzcd}[
      column sep=-5pt,
      row sep=20pt
    ]
      S^3
      \ar[
        d,
        "{
          \ComplexHopfFibration
        }"
      ]
      &\simeq&
      S\big(
        \mathbb{C}^2
      \big)
      \ar[
        d
      ]
      &\ni&
      v
      \ar[
        d,
        |->,
        shorten=6pt
      ]
      \\
      S^2
      &\simeq&
      \mathbb{C}P^1
      &\ni&
      v \cdot \mathbb{C}
      \mathrlap{\,,}
    \end{tikzcd}
    \hspace{.8cm}
    \begin{tikzcd}[
      column sep=-5pt,
      row sep=20pt
    ]
      S^7
      \ar[
        d,
        "{
          \QuaternionicHopfFibration
        }"
      ]
      &\simeq&
      S\big(
        \mathbb{H}^2
      \big)
      \ar[
        d
      ]
      &\ni&
      v
      \ar[
        d,
        |->,
        shorten=6pt
      ]
      \\
      S^4
      &
        \underset{%
          \mathclap{
          \scalebox{.6}{
            \cref{4SphereAsHP1}
          }        
          }%
        }{
          \simeq
        }
      &
      \mathbb{H}P^1
      &\ni&
      v \cdot \mathbb{H}
      \mathrlap{\,.}
    \end{tikzcd}
  \end{equation}
\end{definition}
\begin{remark}
  \label[remark]{FactorizationOfHopfFibration}
  By incrementally forming lines along the sequence of inclusions $\mathbb{R} \subset \mathbb{C} \subset \mathbb{H}$, the $\mathbb{C}/\mathbb{H}$-Hopf fibration \cref{TheHopfFibration} canonically factors
  through further relevant fibrations: 
  \begin{equation}
    \label{FactoringTheCHopfFibration}
    \begin{tikzcd}[row sep=-1pt, column sep=small]
      S(\mathbb{C}^2)
      \ar[
        rrrr,
        uphordown,
        ->>,
        "{
          \ComplexHopfFibration
        }"
      ]
      \ar[
        rr,
        ->>
      ]
      &&
      \mathbb{R}P^3
      \ar[
        rr,
        ->>,
        "{
          \FactoredCHopfFibration
        }"
      ]
      &&
      \mathbb{C}P^1
      \\
      v 
        &\longmapsto&
      v \cdot \mathbb{R}
        &\longmapsto&
      v \cdot \mathbb{C}
      \mathrlap{\,,}
    \end{tikzcd}
  \end{equation}
  and
  \begin{equation}
    \label{FactoringTheHHopfFibration}
    \begin{tikzcd}[row sep=-1pt, column sep=small]
      S(\mathbb{H}^2)
      \ar[
        rrrrrr,
        uphordown,
        ->>,
        "{
          \QuaternionicHopfFibration
        }"
      ]
      \ar[
        rr,
        ->>
      ]
      &&
      \mathbb{R}P^7
      \ar[
        rr,
        ->>
      ]
      &&
      \mathbb{C}P^3
      \ar[
        rr,
        ->>,
        "{
          \TwistorFibration
        }"
      ]
      &&
      \mathbb{H}P^1
      \\
      v 
        &\longmapsto&
      v \cdot \mathbb{R}
        &\longmapsto&
      v \cdot \mathbb{C}
        &\longmapsto&
      v \cdot \mathbb{H}
      \mathrlap{\,.}
    \end{tikzcd}
  \end{equation}
  The map $\FactoredCHopfFibration :\begin{tikzcd}[sep=small]\mathbb{R}P^3 \ar[r, ->>] & \mathbb{C}P^1\end{tikzcd}$ \cref{FactoringTheCHopfFibration} plays a key role below in \cref{ClassifyingFibrationForMassTerms}, while
  $\TwistorFibration : \begin{tikzcd}[sep=small]\mathbb{C}P^3 \ar[r, ->>] & \mathbb{H}P^1\end{tikzcd}$ \cref{FactoringTheHHopfFibration}, also known as the \emph{twistor fibration} (cf. \cite[\S 2]{FSS22-Twistorial}), plays a key role in more refined variants of flux quantization on M5-branes, discussed elsewhere (cf. \parencites{SS25-Seifert}{SS25-Srni}[\S 12]{FSS23-Char}{SS25-TEC}).
\end{remark}
\begin{lemma}
  \label[lemma]{EquivarianceOfHopfFibration}
  The $\mathbb{C}/\mathbb{H}$-Hopf fibration \cref{TheHopfFibration} is equivariant with respect to
  \begin{enumerate}
  \item
  the canonical $\mathrm{U}(2)$/\,$\mathrm{Sp}(2)$-action on $S(\mathbb{C}^2)/S(\mathbb{H}^2)$;

  \item the multiplication action \cref{Sp2ActionOnHP1} on $\mathbb{C}P^1$/\,$\mathbb{H}P^1$:
  \end{enumerate}
  \begin{equation}
    \label{TheEquivariantOfTheHopfFibration}
    \begin{tikzcd}[
      column sep=-2pt,
      row sep=33pt
    ]
      S^3 
      \ar[
        out=31+180,
        in=-31+180,
        looseness=5,
        shift right=1,
        "{
          \mathrm{Spin}^c(3)
        }"{description}
      ]
      \ar[
        d,
        "{
          \ComplexHopfFibration
        }"
      ]
      &\simeq&
      S(\mathbb{C}^2)
      \ar[
        out=-22,
        in=+22,
        looseness=3.4,
        shift left=1,
        "{
          \mathrm{U}(2)
        }"{description}
      ]
      \ar[
        d,
        "{
          \substack{
            v
            \\
            \rotatebox[origin=c]{-90}{$\mapsto$}
            \\
            v \cdot \mathbb{C}
          }
        }"
      ]
      \\
      S^2
      \ar[
        out=31+180,
        in=-31+180,
        looseness=5,
        shift right=1,
        "{
          \mathrm{Spin}^c(3)
        }"{description}
      ]
      &\simeq&
      \mathbb{C}P^1
      \ar[
        out=-25,
        in=+25,
        looseness=4,
        shift left=1.5,
        "{
          \mathrm{U}(2)
        }"{description}
      ]
    \end{tikzcd}
    \hspace{1.5cm}
    \begin{tikzcd}[
      column sep=-2pt,
      row sep=33pt
    ]
      S^7 
      \ar[
        out=31+180,
        in=-31+180,
        looseness=5,
        shift right=1,
        "{
          \mathrm{Spin}(5)
        }"{description}
      ]
      \ar[
        d,
        "{
          \QuaternionicHopfFibration
        }"
      ]
      &\simeq&
      S(\mathbb{H}^2)
      \ar[
        out=-22,
        in=+22,
        looseness=3.4,
        shift left=1,
        "{
          \mathrm{Sp}(2)
        }"{description}
      ]
      \ar[
        d,
        "{
          \substack{
            v
            \\
            \rotatebox[origin=c]{-90}{$\mapsto$}
            \\
            v \cdot \mathbb{H}
          }
        }"
      ]
      \\
      S^4
      \ar[
        out=31+180,
        in=-31+180,
        looseness=5,
        shift right=1,
        "{
          \mathrm{Spin}(5)
        }"{description}
      ]
      &\simeq&
      \mathbb{H}P^1
      \ar[
        out=-25,
        in=+25,
        looseness=4,
        shift left=1.5,
        "{
          \mathrm{Sp}(2)
          \mathrlap{\,.}
        }"{description}
      ]
    \end{tikzcd}
  \end{equation}
  This equivariance extends canonically to the factorizations of \cref{FactorizationOfHopfFibration}.
\end{lemma}
\begin{proof}
  Evidently, for $v \in \mathbb{H}^2$ and $G \in \mathrm{Sp}(2)$, the assignments
  \begin{equation}
    \begin{tikzcd}[
      column sep=20pt,
      row sep=15pt
    ]
      v 
        \ar[
          r,
          shorten=3pt,
          |->
        ]
        \ar[
          d,
          shorten=3pt,
          |->
        ]
        &
      G \cdot v
        \ar[
          d,
          shorten=3pt,
          |->
        ]
      \\
      v \cdot \mathbb{H}
        \ar[
          r,
          shorten=3pt,
          |->
        ]
      &
      G \cdot v \cdot \mathbb{H}
    \end{tikzcd}
  \end{equation}
  commute. Analogously so over $\mathbb{C}$.
\end{proof}

\begin{lemma}
  \label[lemma]{CofibersOfHopfFibration}
  The cell attachments \textup{(\cref{CellAttachment})} along the Hopf fibrations \textup{(\cref{HopfFibration})} and their factorizations \textup{(\cref{FactorizationOfHopfFibration})} are projective spaces \textup{(\cref{ProjectiveSpace})} as follows:%
  \def\ColumnSep{15pt}
  \begin{equation}
    \begin{tikzcd}[
      column sep= 10pt,
      row sep=37pt
    ]
      S\big(\mathbb{R}^2\big)
      \ar[
        dd,
        leftvertright,
        ->>,
        "{
          \RealHopfFibration
        }"{description}
      ]
      \ar[
        r,
        hook
      ]
      \ar[
        dd,
        ->>
      ]
      &
      D\big(\mathbb{R}^2\big)
      \ar[
        dd,
        ->
      ]
      \\
      {}
      \ar[
        dr,
        phantom,
        "{
          \mathrm{(po)}
        }"{pos=.75, scale=.6}
      ]
      \\
      \mathbb{R}P^1
      \ar[
        r,
        hook
      ]
      &
      \mathbb{R}P^2
      \mathrlap{\,,}
    \end{tikzcd}
    \;\;\;\;
    \begin{tikzcd}[
      column sep=\ColumnSep,
      row sep=34pt
    ]
      S\big(\mathbb{C}^2\big)
      \ar[
        dd,
        leftvertright,
        ->>,
        "{
          \ComplexHopfFibration
        }"{description}
      ]
      \ar[
        r,
        hook
      ]
      \ar[
        d,
        ->>
      ]
      \ar[
        dr,
        phantom,
        "{
          \mathrm{(po)}
        }"{pos=.75, scale=.6}
      ]
      &
      D\big(\mathbb{C}^2\big)
      \ar[
        d,
        ->
      ]
      \\
      \mathbb{R}P^3
      \ar[
        r,
        hook
      ]
      \ar[
        d,
        ->>
      ]
      \ar[
        dr,
        phantom,
        "{
          \mathrm{(po)}
        }"{pos=.75, scale=.6}
      ]
      &
      \mathbb{R}P^4
      \ar[
        d,
        "{ 
          f^{\mathbb{R}}_{\mathbb{C}} 
        }"{description}
      ]
      \\
      \mathbb{C}P^1
      \ar[
        r,
        hook
      ]
      &
      \mathbb{C}P^2
      \mathrlap{\,,}
    \end{tikzcd}
    \;\;\;\;
    \begin{tikzcd}[
      column sep=\ColumnSep
    ]
      S\big(\mathbb{H}^2\big)
      \ar[
        ddd,
        leftvertright,
        ->>,
        "{
          \QuaternionicHopfFibration
        }"{description}
      ]
      \ar[
        r,
        hook
      ]
      \ar[
        d,
        ->>
      ]
      \ar[
        dr,
        phantom,
        "{
          \mathrm{(po)}
        }"{pos=.75, scale=.6}
      ]
      &
      D\big(\mathbb{H}^2\big)
      \ar[
        d,
        ->
      ]
      \\
      \mathbb{R}P^7
      \ar[
        r,
        hook
      ]
      \ar[
        d,
        ->>
      ]
      \ar[
        dr,
        phantom,
        "{
          \mathrm{(po)}
        }"{pos=.75, scale=.6}
      ]
      &
      \mathbb{R}P^8
      \ar[
        d,
        "{
          f^{\mathbb{R}}_{\mathbb{C}}
        }"
      ]
      \\
      \mathbb{C}P^3
      \ar[
        d,
        ->>
      ]
      \ar[
        r,
        hook
      ]
      \ar[
        dr,
        phantom,
        "{
          \mathrm{(po)}
        }"{pos=.75, scale=.6}
      ]
      &
      \mathbb{C}P^4
      \ar[
        d,
        "{
          f^{\mathbb{C}}_{\mathbb{H}}
        }"
      ]
      \\
      \mathbb{H}P^1
      \ar[
        r,
        hook
      ]
      &
      \mathbb{H}P^2
      \mathrlap{\,,}
    \end{tikzcd}
  \end{equation}
  where the top right maps are given by
  \begin{equation}
    \label{TopCellInclusionIntoRPn}
    \begin{tikzcd}[row sep=-3pt,column sep=0pt]
      D(\mathbb{R}^{n})
      \ar[rr]
      &&
      \mathbb{R}P^{n}
      \\
      (v_1, \cdots, v_n)
      &\mapsto&
      \Big[
        v_1 
          : 
        \cdots 
          : 
        v_n 
          : 
        \sqrt{
          1 - \vert \vec v \vert^2
        }
      \,\Big]
      \mathrlap{\,.}
    \end{tikzcd}
  \end{equation}
\end{lemma}
\begin{proof}
  By the pasting law (\cref{PastingLaw}), the claim follows from the general and standard statement (cf. \cite[\S IX.3.6]{Massey1991}) that the following is a pushout for all $n \in \mathbb{N}$ and $\mathbb{K} \in \{\mathbb{R}, \mathbb{C}, \mathbb{H}\}$:
  \begin{equation}
    \begin{tikzcd}[row sep=15pt, column sep=0pt]
      S(\mathbb{K}^{n+1})
      \ar[r, hook]
      \ar[d, ->>]
      \ar[
        dr,
        phantom,
        "{ 
          \mathrm{(po)} 
        }"{pos=.7, scale=.6}
      ]
      &[20pt]
      D(\mathbb{K}^{n+1})
      \ar[
        d
      ]
      &
      {\big[
        v_1 : \cdots : v_{n+1}
      \big]}
      \ar[
        d,
        |->,
        shorten=3pt
      ]
      \\
      \mathbb{K}P^n
      \ar[
        r,
        hook
      ]
      &
      \mathbb{K}P^{n+1}
      &
      {\Big[v_1
        : 
        \cdots 
        : 
        v_{n+1} 
        : 
        \sqrt{%
          1 - \vert \vec v \vert^2 
        } 
        \,
      \Big]}
      \mathrlap{\,.}
    \end{tikzcd}
  \end{equation}
  Namely this expresses the fact that an element $[v_1 : \cdots : v_{n+1} : v_{n+2}] \in \mathbb{K}P^{n+1}$:
  \begin{itemize}
  \item either has $v_{n+2} = 0$, in which case it is in the image of $\begin{tikzcd}[sep=small] \mathbb{K}P^n \ar[r, hook] & \mathbb{K}P^{n+1}\end{tikzcd}$ and equivalently in the image of $\begin{tikzcd}[sep=small]S(\mathbb{K}^{n+1}) \ar[r, hook] & D(\mathbb{K}^{n+1}) \ar[r] & \mathbb{K}P^{n+1}\end{tikzcd}$, 
  \item or else it equals, 
  (setting $v'_n := v_n/v_{n+2}$ for $n \in \{1, \cdots, n+1\}$):
  $$
    \Big[
      v'_1 : \cdots : v'_{n+1} : 1 
    \Big]
    =
    \left[
      \tfrac{
        \rule[-3pt]{0pt}{2pt}%
        v'_1
      }{ 
        \rule{0pt}{9pt}%
        \sqrt{
          \vert \vec v' \vert^2 + 1  
        }
      } 
      : \cdots : 
      \tfrac{
        \rule[-3pt]{0pt}{2pt}%
        v'_{n+1}
      }{ 
        \rule{0pt}{9pt}%
        \sqrt{
          \vert \vec v' \vert^2 + 1  
        }
      } 
      :
      \tfrac{
        \rule[-2pt]{0pt}{2pt}%
        1
      }{ 
        \rule{0pt}{9pt}%
        \sqrt{
          \vert \vec v' \vert^2 + 1  
        }
      }
      =
      \scalebox{.88}{$
        \sqrt{
          1 -
          \left\vert
            \frac{
              \rule[-2pt]{0pt}{2pt}%
              \vec v'
            }{
              \rule{0pt}{9pt}%
              \sqrt{
                \vert \vec v' \vert^2
                +1
              }
            }
          \right\vert^2
        }
      $}
      \,
    \right]
  $$
  and hence is in the image of $\begin{tikzcd}[sep=small] D(\mathbb{K}^{n+1}) \ar[r] & \mathbb{K}P^{n+1}\end{tikzcd}$.
  \qedhere
  \end{itemize}
\end{proof}

\begin{lemma}
  \label[lemma]{CosetRealizationOfHopfFibration}
  The $\mathbb{R}/\mathbb{C}/\mathbb{H}$-Hopf fibrations \textup{(\cref{HopfFibration})} are equivalently the following coset coprojections:
  \begin{subequations}
    \begin{align}
    \label{RHopfFibrationAsCosetCoprojection}
    &
    \begin{tikzcd}[
      ampersand replacement=\&,
      column sep=-1pt
    ]
       S^1
       \ar[
         d,
         "{
           \RealHopfFibration
         }"
       ]
       \&
       \simeq
       \&
       S(\mathbb{R}^2)
       \ar[
         d,
         "{
           \substack{
             v 
             \\
             \rotatebox[origin=c]{-90}{$\mapsto$}
             \\
             v \cdot \mathbb{R}
           }
         }"
       ]
       \&\simeq\&
       \frac{
         \rule[10pt]{0pt}{0pt}%
         \mathrm{O}(2)
       }{
         \rule{0pt}{6pt}%
         \mathrm{O}(1)_{\mathrm{stb}}
       }
       \ar[
         d,
         ->>,
         shorten >=-4pt
       ]
       \&
       :=
       \&
       \smash{
       \mathrm{O}(2)
       \big/
       \left(
       \begin{matrix}
         \;\,1\;\, & 0
         \\
         0 & \mathrm{O}(1)
       \end{matrix}
       \right)
       }
       \\
       S^1
       \&
       \simeq
       \&
       \mathbb{R}P^1
       \&
       \simeq
       \&
       \frac{
         \rule[10pt]{0pt}{0pt}%
         \mathrm{O}(2)
       }{
         \rule{0pt}{6pt}%
         \mathrm{O}(1)^2
       }
       \&:=\&
       \mathrm{O}(2) 
       \big/
       \left(
       \begin{matrix}
         \mathrm{O}(1) & 0
         \\
         0 & \mathrm{O}(1)
       \end{matrix}
       \right)
       \mathrlap{,}
    \end{tikzcd}
    \\
    \label{CHopfFibrationAsCosetCoprojection}
    &
    \begin{tikzcd}[
      ampersand replacement=\&,
      column sep=-1pt
    ]
       S^3
       \ar[
         d,
         "{
           \ComplexHopfFibration
         }"
       ]
       \&
       \simeq
       \&
       S(\mathbb{C}^2)
       \ar[
         d,
         "{
           \substack{
             v 
             \\
             \rotatebox[origin=c]{-90}{$\mapsto$}
             \\
             v \cdot \mathbb{C}
           }
         }"
       ]
       \&\simeq\&
       \frac{
         \rule[10pt]{0pt}{0pt}%
         \mathrm{U}(2)
       }{
         \rule{0pt}{6pt}%
         \mathrm{U}(1)_{\mathrm{stb}}
       }
       \ar[
         d,
         ->>,
         shorten >=-4pt
       ]
       \&
       :=
       \&
       \smash{
       \mathrm{U}(2)
       \big/
       \left(
       \begin{matrix}
         \;\,1\;\, & 0
         \\
         0 & \mathrm{U}(1)
       \end{matrix}
       \right)
       }
       \\
       S^2
       \&
       \simeq
       \&
       \mathbb{C}P^1
       \&
       \simeq
       \&
       \frac{
         \rule[10pt]{0pt}{0pt}%
         \mathrm{Sp}(2)
       }{
         \rule{0pt}{6pt}%
         \mathrm{U}(1)^2
       }
       \&:=\&
       \mathrm{U}(2) 
       \big/
       \left(
       \begin{matrix}
         \mathrm{U}(1) & 0
         \\
         0 & \mathrm{U}(1)
       \end{matrix}
       \right)
       \mathrlap{,}
    \end{tikzcd}
    \\ 
    \label{HHopfFibrationAsCosetCoprojection}
    &
    \begin{tikzcd}[
      ampersand replacement=\&,
      column sep=-2pt
    ]
       S^7
       \ar[
         d,
         "{
           \QuaternionicHopfFibration
         }"
       ]
       \&
       \simeq
       \&
       S(\mathbb{H}^2)
       \ar[
         d,
         "{
           \substack{
             v 
             \\
             \rotatebox[origin=c]{-90}{$\mapsto$}
             \\
             v \cdot \mathbb{H}
           }
         }"
       ]
       \&\simeq\&
       \frac{
         \rule[10pt]{0pt}{0pt}%
         \mathrm{Sp}(2)
       }{
         \rule{0pt}{6pt}%
         \mathrm{Sp}(1)_{\mathrm{stb}}
       }
       \ar[
         d,
         ->>,
         shorten >=-4pt
       ]
       \&
       :=
       \&
       \smash{
       \mathrm{Sp}(2)
       \big/
       \left(
       \begin{matrix}
         \;\,1\;\, & 0
         \\
         0 & \mathrm{Sp}(1)
       \end{matrix}
       \right)
       }
       \\
       S^4
       \&
       \simeq
       \&
       \mathbb{H}P^1
       \&
       \simeq
       \&
       \frac{
         \rule[10pt]{0pt}{0pt}%
         \mathrm{Sp}(2)
       }{
         \rule{0pt}{6pt}%
         \mathrm{Sp}(1)^2
       }
       \&:=\&
       \mathrm{Sp}(2) 
       \big/
       \left(
       \begin{matrix}
         \mathrm{Sp}(1) & 0
         \\
         0 & \mathrm{Sp}(1)
       \end{matrix}
       \right)
       \mathrlap{.}
    \end{tikzcd}
    \end{align}
  \end{subequations}
\end{lemma}
\begin{proof}
  We indicate the argument over $\mathbb{H}$, the other two cases are directly analogous:
  By \cref{EquivarianceOfHopfFibration}, we have a transitive $\mathrm{Sp}(2)$-action, exhibiting the 7-sphere as the $\mathrm{Sp}(2)$-translations of its base point
  \begin{equation}
    \mathrm{Sp}(2) 
      \cdot
    \begin{pmatrix}
      1
      \\
      0
    \end{pmatrix}
    \simeq
    S(\mathbb{H}^2)
    \mathrlap{\,.}
  \end{equation}
  Now the stabilizer group of the base point is manifestly
  \begin{equation}
    \label{StabilizerSubgroup}
    \mathrm{Stab}_{\mathrm{Sp}(2)}
    \left(
    \begin{matrix}
      1
      \\
      0
    \end{matrix}
    \right)
    =
    \left(
    \begin{matrix}
      0 & 0
      \\
      0 & \mathrm{Sp}(1)
    \end{matrix}
    \right)
    =:
    \mathrm{Sp}(1)_{\mathrm{stb}}
    \subset
    \mathrm{Sp}(2)
    \,,
  \end{equation}
  while the fiber of the base point is manifestly 
  \begin{equation}
    \mathrm{fib}_{\QuaternionicHopfFibration}
    =
    \left(
    \begin{matrix}
      \mathrm{Sp}(1)
      \\
      0
    \end{matrix}
    \right)
    =
    \left(
    \begin{matrix}
      1 
      \\
      0
    \end{matrix}
    \right)
    \cdot
    \left(
    \begin{matrix}
      \mathrm{Sp}(1) & 0
      \\
      0 & 1
    \end{matrix}
    \right).
  \end{equation}
  This implies the claim by the \emph{orbit-stabilizer theorem} for Lie group actions (cf. \parencites[Thm. 3.62]{Warner1983}[Thm. 21.18]{Lee2012}). 
\end{proof}
Over $\mathbb{C}$ we obtain from this yet another incarnation of the Hopf fibration, which is useful to make explicit (cf. \cite[(3)]{Lyons2003}):
\begin{corollary}
The $\mathbb{C}$-Hopf fibration is equivalently given by the conjugation action of $\mathrm{SU}(2)$ on the matrix representing the basepoint of $S^2$ \cref{4SphereViaMatrices}:
\begin{equation}
  \begin{tikzcd}[
    ampersand replacement=\&,
    row sep=40pt
  ]
    S^3
    \ar[
      r,
      <->,
      "{ \sim }"
    ]
    \ar[
      d,
      "{
        \ComplexHopfFibration
      }"
    ]
    \&
    \mathrm{SU}(2)
    \ar[
      d,
      "{
        \substack{
          g
          \\
          \rotatebox[origin=c]{-90}{$\mapsto$}
          \\
          g 
          \cdot
          \left(%
          \renewcommand{\arraystretch}{.9}%
          \setlength{\arraycolsep}{1pt}%
          \begin{matrix}%
            +1 & \;0
            \\
            \, 0 & -1
          \end{matrix}%
          \right)
          \cdot
          g^{-1}
        }
      }"
    ]
    \\
    S^2
    \ar[
      r,
      <->,
      "{ \sim }"
    ]
    \&
    \left\{
    A \in \BoundedOperators(\mathbb{C}^2)
    \,\middle\vert\,
    \substack{
      A^\dagger = A
      \\
      \mathrm{tr}(A) = 0
      \\
      A \cdot A = 1
    }
    \right\}.
  \end{tikzcd}
\end{equation}
\end{corollary}
\begin{proof}
First to observe that we have a homeomorphism of underlying spaces:
\begin{equation}
  \begin{tikzcd}[
    row sep=-2pt, column sep=10pt
  ]
    \frac{
      \mathrm{U}(2)
    }{
      \mathrm{U}(1)_{\mathrm{stb}}
    }
    \ar[
      rr,
      <-,
      "{ \sim }"
    ]
    &&
    \mathrm{SU}(2)
    \\
    g \cdot \mathrm{U}(1)_{\mathrm{stb}}
      &\longmapsfrom& 
    g
    \mathrlap{\,,}
  \end{tikzcd}
\end{equation}
conversely reflecting the fact that every coset in$\mathrm{U}(2)/\mathrm{U}(1)_{\mathrm{stb}}$ contains precisely one special unitary matrix.

Next to note that under the isomorphisms of \cref{IdentifyingS4WithHP1}, one incarnation of the $\mathbb{C}$-Hopf fibration is as shown on the left and middle of the following diagram:
\begin{equation}
  \begin{tikzcd}[
    ampersand replacement=\&,
    column sep=2pt, row sep=13pt
  ]
    S(\mathbb{C}^2)
    \ar[
      d
    ]
    \&
    U \cdot v
    \ar[
      d,
      |->,
      shorten=4pt
    ]
    \&[+5pt]
    g 
      \cdot 
    \left(
    \begin{matrix}
      1
      \\
      0
    \end{matrix}
    \right)
    \ar[
      d,
      |->,
      shorten=1pt
    ]
    \\
    \left\{
    A \in \BoundedOperators(\mathbb{C}^2)
    \,\middle\vert\,
    \substack{
      A^\dagger  = A
      \\
      \mathrm{tr}(A) = 0
      \\
      A \cdot A = 1
    }
    \right\}
    \&
    U 
      \cdot 
    \big(
      2 v \cdot v^\dagger - 1
    \big)
    \cdot
    U^\dagger
    \&
    g \cdot 
    \left(
    \begin{matrix}
      +1 & \; 0 
      \\
     \, 0 & -1
    \end{matrix}
    \right)
    \cdot
    g^{-1}
    \mathrlap{\,.}
  \end{tikzcd}
\end{equation}
Combining these two statements gives that $g \in \mathrm{SU}(2)$ maps as shown on the right, which is the claim to be proven.
\end{proof}

In \cref{TheEquivariantOrientation} we will crucially use the following curious incarnation of the Hopf fibration in geometric homotopy theory (\cref{SomeCohesiveHomotopyTheory}), which may be less widely appreciated:
\begin{corollary}
  \label[corollary]{HopfFibrationOnHomotopyQuotients}
  After passage to homotopy $\mathrm{Sp}(2)$-quotients, the quaternionic Hopf fibration is equivalent to the delooping of the inclusion of the stabilizer subgroup $\mathrm{Sp}(1)_{\mathrm{stb}}$ into its product with the active $\mathrm{Sp}(1)$:
  \begin{equation}
    \begin{tikzcd}[
      column sep=-2pt
    ]
      \mathrm{Spin}(5)
      \backsslash
      S^7  
      \ar[
        d,
        "{
          \mathrm{Spin}(5)
          \backsslash
          \QuaternionicHopfFibration
        }"{swap}
      ]
      &\simeq&
      \mathrm{Sp}(2)
      \backsslash
      S(\mathbb{H}^2) 
      \ar[
        d
      ]
      &\simeq&
       \mathrm{Sp}(2)
       \backsslash
       \smash{
       \frac{
         \rule[10pt]{0pt}{0pt}%
         \mathrm{Sp}(2)
       }{
         \rule{0pt}{6pt}%
         \mathrm{Sp}(1)_{\mathrm{stb}}
       }
       }
      \ar[
        d
      ]
      &\simeq&
      \mathrm{Sp}(1)_{\mathrm{stb}} 
      \backsslash 
      \ast
      \ar[
        d
      ]
      \\
      \mathrm{Spin}(5)
      \backsslash
      S^4 
      &\simeq&
      \mathrm{Sp}(2)
      \backsslash
      \mathbb{H}P^1
      &\simeq&
       \mathrm{Sp}(2)
       \backsslash
       \smash{
       \frac{
         \rule[10pt]{0pt}{0pt}%
         \mathrm{Sp}(2)
       }{
         \rule{0pt}{6pt}%
         \mathrm{Sp}(1)^2
       }
       }
      &\simeq&
      \mathrm{Sp}(1)^2
      \backsslash \ast
      \mathrlap{\,.}
    \end{tikzcd}
  \end{equation}
\end{corollary}
\begin{proof}
  By the general equivalence $\mathbf{B}H \simeq G \backsslash G/H$ \cref{HoQuotientOfGModHByG}, for subgroups $H \subset G$.
\end{proof}

\subsubsection{Trivialization along Hopf Fibration}

For completeness, we close this discussion by highlighting the traditional trivialization of the pullback of the tautological line bundle along the Hopf fibration (but below in \cref{TheEquivariantOrientation} we instead use \cref{HopfFibrationOnHomotopyQuotients} for a more powerful argument):

\begin{proposition}
  \label{TrivializationOfPullbackOfTauLongHopf}
  The unit sphere bundle of the tautological line bundle \cref{TautologicalLineBundleAsSubbundle} on $\mathbb{C}P^1\!$/\,$\mathbb{H}P^1$ is isomorphic to the Hopf fibration \cref{TheHopfFibration}:
  \begin{equation}
    \begin{tikzcd}[
      column sep=25pt, row sep=10pt
    ]
      S\big(
        \mathcal{L}^{\mathrm{taut}}_{\mathbb{H}P^1}
      \big)
      \ar[d]
      \ar[
        r,
        "{ \sim }"
      ]
      &
      S^7
      \ar[
        d,
        "{
          \QuaternionicHopfFibration
        }"
      ]
      \\
      \mathbb{H}P^1
      \ar[
        r,
        "{ \sim }"
      ]
      &
      S^4
      \mathrlap{.}
    \end{tikzcd}
  \end{equation}
\end{proposition}
\begin{proof}
Unwinding the definitions, this is again essentially a tautology:
\begin{equation}
  \begin{tikzcd}[
    column sep=0pt, 
    row sep=-5pt
  ]
    S\big(
      \mathcal{H}_{\mathbb{H}P^1}
    \big)
    \ar[
      rrr,
      "{ \sim }"
    ]
    \ar[
      ddd
    ]
    &&[15pt]&
    S(\mathbb{H}^2)
    \ar[
      ddd
    ]
    \\
    &
    (P,v)
    \ar[
      r,
      |->, 
      shorten=6pt,
      shift left=3pt
    ]
    \ar[
      d,
      |->, 
      shorten=3pt,
    ]
    \ar[
      r,
      <-|, 
      shorten=6pt,
      shift right=3pt
    ]
    &
    v
    \ar[
      d,
      |->, 
      shorten=4pt,
    ]
    \\[20pt]
    &
    P
    \ar[
      r,
      equals
    ]
    &
    1  -
    \frac{
      v v^{\smash{\dagger}}
    }{
      \rule{0pt}{6pt}\vert v \vert^2
    }
    \\
    \mathbb{H}P^1
    \ar[
      rrr,
      equals
    ]
    &&&
    \mathbb{H}P^1
    \mathrlap{,}
  \end{tikzcd}
\end{equation}
using here that every $v \in \mathbb{H}^2$ is in the kernel of a unique self-adjoint rank=1 projector \cref{IdentifyingS4WithHP1}, and observing that the condition $\vert v \vert = 1$ is the same on both sides.
\end{proof}

\begin{proposition}
  \label[proposition]{PullbackOfTautologicalLineBundleAlongHopfFibrationTrivializes}
  The pullback of the tautological $\mathbb{H}$-line bundle \cref{TautologicalLineBundleAsSubbundle}
  along the $\mathbb{H}$-Hopf fibration
  \cref{TheHopfFibration}
  trivializes $\mathrm{Sp}(2)$-equivariantly:
  \begin{equation}
    \begin{tikzcd}[
      column sep=25pt, row sep=-2pt
    ]
      S^7
      \times
      \mathbb{H}
      \ar[
        out=180-58, 
        in=58, 
        looseness=4, 
        shift right=1pt,
        "{
          \scalebox{.9}{$\;\,\mathclap{
            \mathrm{Sp}(2)
          }\;\,$}
        }"{description},
      ]
      \ar[
         rr,
         "{ \sim }"
      ]
      \ar[
        dr
      ]
      &&
      \QuaternionicHopfFibration^\ast
      \mathcal{L}^{\mathrm{taut}}_{\mathbb{H}P^1}
      \,.
      \ar[
        out=180-58, 
        in=60, 
        looseness=4, 
        shift right=1pt,
        "{
          \scalebox{.9}{$\;\;\mathclap{
            \mathrm{Sp}(2)
          }\;\;$}
        }"{description},
      ]
      \ar[
        dl
      ]
      \\
      & 
      S^7
    \end{tikzcd}
  \end{equation}
\end{proposition}
\begin{proof}
In components, the bundle isomorphism may be given as
\begin{equation}
  \begin{tikzcd}[
    column sep=3pt,
    row sep=-3pt
  ]
    \grayoverbrace
    {
      \bigg\{
        (u,q)
        \;\Big\vert\;
        \substack{
          u \in \mathbb{H}^2
          \\
          q \in \mathbb{H}
          \\
          \vert u \vert = 1
        }
      \bigg\}
    }
    {
      \mathcolor{black}{
        S(\mathbb{H}^2) 
          \times 
        \mathbb{H}
      }
    }
  \;  \ar[
      rr,
      <->,
      "{ \sim }"
    ]
    &&
    \;
    \grayoverbrace{
    \bigg\{
    (u, P,v)
    \;\bigg\vert\;
    \substack{
      P \in \mathbb{H}P^1
      \\
      u,v \in \mathrm{ker}(P)
      \\
      \vert u \vert = 1
    }
    \bigg\}
    }{
      \mathcolor{black}{
        \QuaternionicHopfFibration^\ast
        \mathfrak{L}_{\mathbb{H}P^1}    
      }
    }
    \\
    (u,q)
    &\longmapsto&
    \big(
      u
      ,
      1 
        - 
      {u u^\dagger}
      , 
      u \cdot q
    \big)
    \\
    \big(
      u, v/u
    \big)
    &\longmapsfrom&
    (u, P, v)
    \mathrlap{\,,}
  \end{tikzcd}
\end{equation}
whose equivariance is manifest.
\end{proof}

\begin{remark}
  \label[remark]{TrivializationOfPullbackOfTautologicalLineBundleAlongHopfFibrationInTermsOfClassifyingMaps}
  In terms of classifying maps, the trivialization of \cref{TrivializationOfPullbackOfTauLongHopf} is a homotopy of this form
  \begin{equation}
    \begin{tikzcd}[row sep=small, column sep=large]
      S^7
      \ar[
        dd,
        "{
          \ComplexHopfFibration
        }"
      ]
      \ar[
        rr
      ]
      &&
      \ast
      \ar[
        dd,
        "{
        }"
      ]
      \ar[
        ddll,
        shorten=15pt,
        Rightarrow,
        "{ \sim }"{sloped}
      ]
      \\
      \\
      \mathbb{H}P^1
      \ar[
        rr,
        "{
          \vdash 
          \mathcal{L}^{\smash{\mathbb{H}P^1}}
        }"
        {description}
      ]
      &&
      \mathbb{H}P^\infty
      \mathrlap{.}
    \end{tikzcd}
  \end{equation}
  This shows that postcomposition with a quaternion orientation of complex K-theory, 
  $\mathbb{H}P^\infty \to \mathrm{KU}$, gives a homotopy $\mathcolor{purple}{h_3^{\mathrm{KU}}}$ according to \cref{MeasuringRelativeChargesInQCohom}.
\end{remark}
Next we turn to constructing this in detail and equivariantly (cf. \cref{PullbackOfTautologicalFredholmOperatorAlongQuaternionicHopfFibrationTrivializes} below).

\subsection
  {The Unstable Equivariant K-Orientation}
\label{TheEquivariantOrientation}

Recall the construction of the $\mathrm{Sp}(2)$-equivariant tautological quaternionic line bundle $\mathcal{L}^{\mathrm{taut}}_{\mathbb{H}P^1}$ (\cref{TheTautologicalHLineBundle}) over $S^4 \simeq \mathbb{H}P^1$, from \cref{TheEquivariantLineBundle}. We now explicitly construct its incarnation in the orbifold K-theory of \cref{TwistedOrbifoldKTheory} for trivial PCT symmetry (hence in the $\mathrm{KU}^0$-sector, by \cref{ReductionOfOrbiKToAnyKRDegree}), and exhibit the trivialization of its pullback along the quaternionic 
Hopf fibration. 
This explicitly exhibits the four/ten-dimensional equivariant orientation of $\mathrm{KU}^0$ in the manner explained around \cref{MeasuringRelativeChargesInQCohom}.

\subsubsection{The equivariant unit}
\label{TheEquivariantUnit}

In fact, the tautological realization of $\mathcal{L}^{\mathrm{taut}}_{\mathbb{H}P^1}$ as the kernel bundle of the projectors which are the points of $\mathbb{H}P^1$ (\cref{TheTautologicalHLineBundle}) induces a similarly tautological incarnation of its reduced K-theory class as the virtual kernel of parameterized Fredholm operators \cref{JaenichIndexMap}, by use of \cref{FredholmFromProjector}:

\begin{lemma}
  \label[lemma]{TautologicalHP1IndexedFredholm}
  The complex vector bundle underlying the tautological $\mathbb{H}$-line bundle over $\mathbb{H}P^1$ \textup{(\cref{TheTautologicalHLineBundle})} is the virtual kernel bundle \cref{VirtualVectorBundleFromFredholmMap} of the following $\mathbb{H}P^1$-parameterized graded Fredholm operator:
  \begin{equation}
    \label{TheTautologicalParameterizedFredholmOperator}
    \begin{tikzcd}[
      row sep=0pt,
      column sep=-1pt
    ]
      \mathbb{H}P^1
      \ar[
        rrrrrr,
        uphordown,
        "{
           F
             ^{\mathrm{taut}}
             _{\mathbb{H}P^1}
        }"
      ]
      \ar[
        rr,
        "{
          \sim
        }",
        "{
          \scalebox{.7}{
            \cref{IdentifyingS4WithHP1}
          }
        }"{swap}
      ]
      &&
      \Big\{
        P \in 
        \BoundedOperators(\mathbb{H}^2)
        \,\Big\vert\,
        \substack{
          P^\dagger = P
          \\
          P \circ P = P
          \\
          \mathrm{tr}_{\mathbb{H}}(P) = 1
        }
      \Big\}
      \ar[
        rr,
        "{ \gamma }",
        "{
          \scalebox{.7}{
            \cref{QuaternionicMatricesAsComplexMatrices}
          }
        }"{swap}
      ]
      &&
      \Big\{
        P \in 
        \BoundedOperators(\mathbb{C}^4)
        \,\Big\vert\,
        \substack{
          P^\dagger = P
          \\
          P \circ P = P
        }
      \Big\}
      \ar[
        rr,
        "{ F_{(-)} }",
        "{
          \scalebox{.7}{
            \cref{MapFromProjectorsToFredholmOperators}
          }
        }"{swap}
      ]
      &&
      \GradedFredholmOperators
      \\
      P
        &\mapsto& 
      P 
        &\mapsto&
      \gamma(P)
        &\mapsto&
      F_{\gamma(P)}
      \mathrlap{\,,}
    \end{tikzcd}
  \end{equation}
  in that
  \begin{equation}
   \mathrm{ker}\big(
     F^{\mathrm{taut}}
       _{\mathbb{H}P^1}
  \big)
     =
     \gamma(
     \mathcal{L}^{\mathrm{taut}}_{\mathbb{H}P^1}
     )
     \ominus
     0
     \mathrlap{\,.}
  \end{equation}
\end{lemma}
\begin{proof}
  By \cref{FredholmFromProjector}, the virtual kernel of $F^{\mathrm{taut}}_P \defneq F_{\gamma(P)}$ is that of $\gamma(P)$, by \cref{StarAlgebraHomFromPauliMatrices} this is the underlying complex vector space of the kernel of $P$, 
  and that, by \cref{TheTautologicalHLineBundle}, is the fiber of the tautological $\mathbb{H}$-line bundle over $P$.
\end{proof}

\begin{example}
The virtual difference \cref{VirtualDifferenceOfGradedFredholmOps} of the above \cref{TheTautologicalParameterizedFredholmOperator} and the constant $\mathbb{H}P^1$-parametrized Fredholm operator $P \mapsto F_{\mathbb{C}^2}$ \cref{FRedholmOperatorForTrivialBundle}
\begin{equation}
  \label{ReducedTautologicalFredholmOperator}
  \begin{tikzcd}
    \mathbb{H}P^1
    \ar[
      rrr,
      "{ 
        F^{\mathrm{taut}}
          _{\mathbb{H}P^1} 
          \ominus 
        F^{\mathbb{C}^2}
         _{\mathbb{H}P^1}
      }"
    ]
    &&&
    \GradedFredholmOperators
  \end{tikzcd}
\end{equation}
has as virtual kernel bundle the ``reduced'' version of the (underlying complex vector bundle) of the tautological $\mathbb{H}$-line bundle,
\begin{equation}
  \mathrm{ker}\big(
    F^{\mathrm{taut}}
      _{\mathbb{H}P^1} 
      \ominus 
    F^{\mathbb{C}^2}
     _{\mathbb{H}P^1}
  \big)
  =
  \gamma(
    \mathcal{L}^{\mathrm{taut}}_{\mathbb{H}P^1}
  )
  \ominus 
  \mathbb{C}^2_{\mathbb{H}P^1},
\end{equation}
of vanishing virtual dimension,
$
  \mathrm{dim}\Big(
  \mathrm{ker}\big(
    F^{\mathrm{taut}}
      _{\mathbb{H}P^1} 
      \ominus 
    F^{\mathbb{C}^2}
     _{\mathbb{H}P^1}
  \big)
  \Big) = 0
  \mathrlap{\,.}
$
\end{example}

\begin{notation}
  \label[notation]{StableActionOnFred}
  Let 
  \begin{equation}
    \label{StableUnitarySp2Action}
    \begin{tikzcd}[
      column sep=15pt
    ]
      \mathrm{Sp}(2)
      \ar[
        rr,
        "{
          \displaystyle{
            \oplus_{\mathbb{N}}
          }
          \,
          \gamma
        }"
      ]
      &&
      \UH
      \ar[r, ->>]
      &
      \PUH
      \ar[
        r, 
        hook
      ]
      &
      \GradedPUH
      \ar[
        r,
        hook
      ]
      &
      \mathrm{Aut}\big(
       \GradedFredholmOperators
      \big)
    \end{tikzcd}
  \end{equation}
  be (the projective image of) $\mathbb{N}$ direct summands of the 4-dimensional irrep \cref{TheComplexSp2Representation}, realizing the Hilbert space $\HilbertSpace$ \cref{TheHilbertSpace} as the unitary $\mathrm{Sp}(2)$-representation which is the $\mathbb{N}$-indexed direct sum of these irreps \cref{TheComplexSp2Representation}:
  \footnote{
    Generally one would take $\HilbertSpace$ to be a \emph{stable representation} of $\mathrm{Sp}(2)$, namely the $\mathbb{N}$-fold direct sum of the direct sum of all the complex irreducible representations of $\mathrm{Sp}(2)$. But for the present purpose, and not to overburden the notation unnecessarily, we may stick with \cref{StableUnitarySp2Action}.
  }
  \begin{equation}
    \label{NDirectSumOfSp2Representations}
      \begin{tikzcd}[sep=-3pt]
      \HilbertSpace
      \ar[
        out=180-55,
        in=55,
        looseness=4,
        shift right=2.3pt,
        "{
          \hspace{4pt}\mathclap{\mathrm{Sp}(2)}\hspace{6pt}
        }"{description}
      ]
      &:=&
      \bigoplus_{\mathbb{N}}
      &
      \mathbb{C}^4
      \mathrlap{\,.}
      \ar[
        out=180-60,
        in=60,
        looseness=4,
        shift right=2.3pt,
        "{
          \hspace{4pt}\mathclap{\mathrm{Sp}(2)}\hspace{5pt}
        }"{description}
      ]
    \end{tikzcd}
  \end{equation}
\end{notation}
\begin{remark}
  \label[remark]{UnshufflingTrivialAndDefiningReps}
  Below we will be concerned with the restriction of this $\mathrm{Sp}(2)$-representation \cref{NDirectSumOfSp2Representations} to a representation of the stabilizer subgroup \cref{StabilizerSubgroup},
  \begin{equation}
    \begin{tikzcd}[sep=small]
      \mathrm{Sp}(1)_{\mathrm{stb}} 
      \ar[
        rr, 
        hook,
        "{ 
          \iota_{\mathrm{stb}} 
        }"
      ] 
      &&
      \mathrm{Sp}(2) 
      \mathrlap{\,,}
    \end{tikzcd}
  \end{equation}
  in which form it becomes equivariantly unitarily equivalent, as in \cref{UnshuffleUnitaryTransformation}, to the direct sum of an infinite-dimensional trivial representation with an $\mathbb{N}$-indexed sum of copies of the defining representation $\mathbb{C}^2_{\mathrm{def}}$ of $\mathrm{Sp}(1) \simeq \mathrm{SU}(2)$ \cref{UnitQuaternionsAsSU2}:
  \begin{equation}
    \label{RestrictedNDirectSumOfSp2Representations}
      \begin{aligned}
      \;\,
      \mathclap{
      \begin{tikzcd}[sep=-3pt]
      \HilbertSpace
      \ar[
        out=180-55,
        in=55,
        looseness=4,
        shift right=2.3pt,
        "{
          \hspace{4pt}\mathclap{\mathrm{Sp}(1)_{\mathrm{stb}}}\hspace{6pt}
        }"{description}
      ]
      \end{tikzcd}
      }
      \;\,
      & =
      \textstyle{\bigoplus_{\mathbb{N}}}
      \big(
          \;\;\,
          \mathclap{
          \begin{tikzcd}
          \mathbb{C}^4
          \ar[
            out=180-60,
            in=60,
            looseness=4,
            shift right=2.3pt,
            "{
              \hspace{4pt}\mathclap{\mathrm{Sp}(1)_{\mathrm{stb}}}\hspace{5pt}
            }"{description}
          ]
          \end{tikzcd}
          }
          \;\;\,
      \big)
      =
      \textstyle{\bigoplus_{\mathbb{N}}}
      \big(
        \mathbb{C}^2_{\mathrm{triv}}
        \oplus
          \;\;\,
          \mathclap{
          \begin{tikzcd}
          \mathbb{C}^2_{\mathrm{def}}
          \ar[
            out=180-60,
            in=60,
            looseness=4,
            shift right=2.3pt,
            "{
              \hspace{4pt}\mathclap{\mathrm{Sp}(1)_{\mathrm{stb}}}\hspace{5pt}
            }"{description}
          ]
          \end{tikzcd}
          }
          \;\;\,
      \big)
      \\
      & 
      \xrightarrow[\sim]{\quad U \quad}
      \grayunderbrace{
      \Big(
      \textstyle{\bigoplus_{\mathbb{N}}}
      \mathbb{C}^2_{\mathrm{triv}}
      \Big)
      }{%
        \mathcolor{black}{%
          \HilbertSpace_{\mathrm{triv}}%
        }%
      }
      \oplus
      \grayunderbrace{
      \Big(
      \textstyle{\bigoplus_{\mathbb{N}}}
      \;\;\,
      \mathclap{
      \begin{tikzcd}
      \mathbb{C}^2_{\mathrm{def}}
      \ar[
        out=180-60,
        in=60,
        looseness=4,
        shift right=2.3pt,
        "{
          \hspace{4pt}\mathclap{\mathrm{Sp}(1)_{\mathrm{stb}}}\hspace{5pt}
        }"{description}
      ]
      \end{tikzcd}
      }
      \;\;\,
      \Big)
      }{%
        \mathcolor{black}{%
          \HilbertSpace_{\mathrm{def}}%
        }%
      }%
      \mathrlap{\,.}
    \end{aligned}
  \end{equation}  
  Under this transformation, the Fredholm operator $F_{P_0}$ transforms into 
  \cref{TransformedFredholmUnderUnshuffle}
  \begin{equation}
    \label{TransformedFredholmUnderRepUnshuffle}
    U \circ F_{P_0} \circ U^{-1}
    =
    F^{\mathbb{C}^2_{\mathrm{triv}}}
    \oplus
    F^{0 \cdot \mathbb{C}^2_{\mathrm{def}}}
    \mathrlap{\,.}
  \end{equation}
\end{remark}

\begin{lemma}
  \label[lemma]{EquivarianceOfTautologicalFred}
  The tautological $\mathbb{H}P^1$-parameterized Fredholm operator \cref{TheTautologicalParameterizedFredholmOperator} is $\mathrm{Sp}(2)$-equivariant with respect to the canonical action on $\mathbb{H}P^1$ \cref{TheEquivariantOfTheHopfFibration} and the action on $\GradedFredholmOperators$ from \cref{StableUnitarySp2Action}:
  \begin{equation}
    \begin{tikzcd}
      \mathbb{H}P^1
      \ar[
        out=180-60,
        in=60,
        looseness=4,
        shift right=2.3pt,
        "{
          \hspace{4pt}\mathclap{\mathrm{Sp}(2)}\hspace{5pt}
        }"{description}
      ]
      \ar[
        rr,
        "{
          F^{\mathrm{taut}}
          _{\mathbb{H}P^1}
        }"
      ]
      &&
      \GradedFredholmOperators
      \mathrlap{\,,}
      \ar[
        out=180-60,
        in=60,
        looseness=4,
        shift right=2.3pt,
        "{
          \hspace{4pt}\mathclap{\mathrm{Sp}(2)}\hspace{5pt}
        }"{description}
      ]
    \end{tikzcd}
  \end{equation}
  and hence descends to a map of homotopy quotient groupoids \cref{GEquivariantMapsAsSliceMapsOverBG}:
  \begin{equation}
    \begin{tikzcd}[
      column sep=25pt,
      row sep=-1pt
    ]
      \mathrm{Sp}(2)
      \backsslash
      \mathbb{H}P^1
      \ar[
        rr,
        "{
          F^{\mathrm{taut}}
            _{\mathbb{H}P^1}
          \sslash
          \mathrm{Sp}(2)
        }"
      ]
      \ar[
        dr,
        shorten=-1pt
      ]
      &&
      \mathrm{Sp}(2)
      \backsslash
      \GradedFredholmOperators
      \mathrlap{\,.}
      \ar[
        dl,
        shorten=-1pt
      ]
      \\
      &
      \mathbf{B}\mathrm{Sp}(2)
    \end{tikzcd}
  \end{equation}
\end{lemma}
\begin{proof}
  The first two maps in \cref{TheTautologicalParameterizedFredholmOperator} are tautologically equivariant, as around \cref{KernelOfConjugatedProjectionOperator}. The third map in \cref{TheTautologicalParameterizedFredholmOperator} is clearly equivariant by the block matrix form \cref{FredholmOperatorFromProjection} of $F_P$ being compatible with the direct sum of representations \cref{NDirectSumOfSp2Representations}:
  \begin{equation}
    \begin{aligned}
    &
    \big(
    \oplus_{\mathbb{N}}
    \gamma(G)
    \big)
    \left(
    \begin{matrix}
      \;P\; & 0 & 0 &   \cdots 
      \\
      1\!-\!P & \;P\; & 0 &   \cdots 
      \\
      0 & 1\!-\!P & \;P\; &   \cdots
      \\
      0 & 0 & 1\!-\!P &    \cdots
      \\[-4pt]
      \vdots & \vdots  & \vdots 
      & \ddots
    \end{matrix}
    \right)
    \big(
    \oplus_{\mathbb{N}}
    \gamma(G)^{-1}
    \big)
    \\[7pt]
    & =
    \left(
    \begin{matrix}
      \;\gamma(G) P \, \gamma(G)^{-1}\; & 0 & 0 &   \cdots 
      \\
      1\!-\!\gamma(G) P\, \gamma(G)^{-1} & \;\gamma(G) P\, \gamma(G)^{-1}\; & 0 &   \cdots 
      \\
      0 & 1\!-\!\gamma(G) P\, \gamma(G)^{-1} & \;\gamma(G) P \gamma(G)^{-1}\; &  \cdots
      \\
      0 & 0 & 1\!-\!\gamma(G) P\, \gamma(G)^{-1}  &   \cdots
      \\[-4pt]
      \vdots & \vdots & \vdots  
      & \ddots
    \end{matrix}
    \right)
    \mathrlap{.}
    \end{aligned}
  \end{equation}
  This establishes the claim.
\end{proof}

\subsubsection{The orienting homotopy}
\label{TheOrientingHomotopy}

With this in hand, we have the following Fredholm operator analogue of \cref{PullbackOfTautologicalLineBundleAlongHopfFibrationTrivializes}, cf. \cref{TrivializationOfPullbackOfTautologicalLineBundleAlongHopfFibrationInTermsOfClassifyingMaps}, which realizes the phenomenon of \cref{MeasuringRelativeChargesInQCohom,TwistedEquivariantOrientation} for $E = \mathrm{KU}$:
\begin{proposition}
  \label[proposition]{PullbackOfTautologicalFredholmOperatorAlongQuaternionicHopfFibrationTrivializes}
  The pullback of the reduced version \cref{ReducedTautologicalFredholmOperator} of the tautological $\mathbb{C}P^1/\mathbb{H}P^1$-parameterized Fredholm operator \cref{TheTautologicalParameterizedFredholmOperator} along the $\mathbb{C}/\mathbb{H}$-Hopf fibration \cref{TheHopfFibration} trivializes $\mathrm{U}(2)/\mathrm{Sp}(2)$-equivariantly, in that we have an equivariant homotopy \cref{EquivariantHomotopy} of this form:
  \begin{equation}
    \label{TheHomotopy}
    \begin{tikzcd}[
      column sep=25pt,
      ampersand replacement=\&
    ]
      S(\mathbb{C}^2)
      \ar[
        in=70+90,
        out=180-70+90,
        looseness=3,
        shift right=4pt,
        "{
          \mathrm{U}(2)
        }"{description}
      ]
      \ar[
        rr
      ]
      \ar[
        dd,
        "{
          \substack{
            v
            \\
            \rotatebox[origin=c]{-90}{$\mapsto$}
            \\
            v \cdot \mathbb{C}
          }
        }"{swap}
      ]
      \&\&
      \ast
      \ar[
        out=51-90,
        in=180-51-90,
        looseness=4,
        shift left=2pt,
        "{
          \mathrm{U}(2)
        }"{description}
      ]
      \ar[
        dd,
        "{
          F_0
        }"{description}
      ]
      \ar[
        ddll,
        shorten=15pt,
        Rightarrow,
        dashed,
        "{
          h_1^{\mathrm{KU}}
        }"{description}
      ]
      \\
      \\
      \mathbb{C}P^1
      \ar[
        in=65+90,
        out=180-65+90,
        looseness=3,
        shift right=2pt,
        "{
          \mathrm{U}(2)
        }"{description}
      ]
      \ar[
        rr,
        "{
          \FredholmOperator
            ^{\mathrm{taut}}
            _{\mathbb{C}P^1}
          \ominus
          \FredholmOperator
            ^{\mathbb{C}^1_{\mathrm{triv}}}
            _{\mathbb{C}P^1}
        }"{swap}
      ]
      \&\&
      \GradedFredholmOperators
      \ar[
        out=70-90,
        in=180-70-90,
        looseness=3,
        shift left=2pt,
        "{
          \mathrm{U}(2)
          \mathrlap{\,,}
        }"{description}
      ]
    \end{tikzcd}
    \;\;\;\;
    \begin{tikzcd}[
      column sep=25pt,
      ampersand replacement=\&
    ]
      S(\mathbb{H}^2)
      \ar[
        in=70+90,
        out=180-70+90,
        looseness=3,
        shift right=4pt,
        "{
          \mathrm{Sp}(2)
        }"{description}
      ]
      \ar[
        rr
      ]
      \ar[
        dd,
        "{
          \substack{
            v
            \\
            \rotatebox[origin=c]{-90}{$\mapsto$}
            \\
            v \cdot \mathbb{H}
          }
        }"{swap}
      ]
      \&\&
      \ast
      \ar[
        out=51-90,
        in=180-51-90,
        looseness=4,
        shift left=2pt,
        "{
          \mathrm{Sp}(2)
        }"{description}
      ]
      \ar[
        dd,
        "{
          F_0
        }"{description}
      ]
      \ar[
        ddll,
        shorten=15pt,
        Rightarrow,
        dashed,
        "{
          h_3^{\mathrm{KU}}
        }"{description}
      ]
      \\
      \\
      \mathbb{H}P^1
      \ar[
        in=65+90,
        out=180-65+90,
        looseness=3,
        shift right=2pt,
        "{
          \mathrm{Sp}(2)
        }"{description}
      ]
      \ar[
        rr,
        "{
          \FredholmOperator
            ^{\mathrm{taut}}
            _{\mathbb{H}P^1}
          \ominus
          \FredholmOperator
            ^{\mathbb{C}^2_{\mathrm{triv}}}
            _{\mathbb{H}P^1}
        }"{swap}
      ]
      \&\&
      \GradedFredholmOperators
      \ar[
        out=70-90,
        in=180-70-90,
        looseness=3,
        shift left=2pt,
        "{
          \mathrm{Sp}(2)
          \mathrlap{\,.}
        }"{description}
      ]
    \end{tikzcd}
  \end{equation}
\end{proposition}
\begin{proof}
  We spell out the proof over $\mathbb{H}$; for $\mathbb{C}$ it is verbatim the same up to the evident substitution.
  
  After equivalent passage to $\mathrm{Sp}(2)$-homotopy quotients \cref{GEquivariantMapsAsSliceMapsOverBG},
  and under the equivalence of \cref{HopfFibrationOnHomotopyQuotients}, we are reduced to exhibiting $\mathbf{B}\mathrm{Sp}(2)$-sliced homotopy \cref{EquivariantHomotopy} 
  as indicated by the dashed arrow in this diagram:
  \begin{equation}
    \begin{tikzcd}[
      column sep=35pt,
      row sep=14pt
    ]
      \mathrm{Sp}(1)_{\mathrm{stb}}
      \backsslash \ast
      \ar[
        rr,
        equals
      ]
      \ar[
        dd,
        "{
          \iota_{\mathrm{stb}} 
            \backsslash 
          \ast
        }"{swap}
      ]
      &\phantom{-}&
      \mathrm{Sp}(1)_{\mathrm{stb}}
      \backsslash \ast
      \ar[
        dd,
        "{ 
          \mathrm{Sp}(1)_{\mathrm{stb}}
          \backsslash
          F_0 
        }"{description}
      ]
      \ar[
        ddll,
        Rightarrow,
        dashed,
        shorten=14pt,
        "{ \eta }"{swap},
      ]
      \ar[
        r,
        "{
          \iota_{\mathrm{stb}}
          \backsslash 
          \ast
        }"
      ]
      &
      \mathrm{Sp}(2)
      \backsslash \ast
      \ar[
        ddl,
        equals,
        shorten=10pt
      ]
      \ar[
        dd,
        "{ 
          \mathrm{Sp}(2)
          \backsslash
          F_0 
        }"
      ]      
      \\
      \\
      \mathrm{Sp}(1)^2 
      \backsslash
      \ast
      \ar[
        rr,
        "{
          \mathrm{Sp}(1)^2 
          \backsslash
          \big(
          \FredholmOperator_{P_0}
          \ominus
          \FredholmOperator
            ^{\mathbb{C}^2_{\mathrm{triv}}}
          \big)
        }"{swap}
      ]
      &&
      \mathrm{Sp}(1)_{\mathrm{stb}}
      \backsslash
      \GradedFredholmOperators
      \ar[
        r,
        "{
          \iota_{\mathrm{stb}}
          \backsslash
          \GradedFredholmOperators
        }"{swap, pos=.4}
      ]
      &
      \mathrm{Sp}(2)
      \backsslash
      \GradedFredholmOperators
      \mathrlap{\,.}
    \end{tikzcd}
  \end{equation}
  Here, in identifying the bottom map as shown, we have used that restriction along the equivalence $\begin{tikzcd}[sep=small]\mathrm{Sp}(1)^1 \backsslash \ast \ar[r, "{\sim}"] & \mathrm{Sp}(2) \backsslash \mathrm{Sp}(2)/\mathrm{Sp}(1)^1 \simeq \mathrm{Sp}(2) \backsslash \mathbb{H}P^1 \end{tikzcd}$ means \cref{HoQuotientOfGModHByG} to (restrict the isotropy action along $\begin{tikzcd}[sep=small] \mathrm{Sp}(1)^2 \ar[r, hook] & \mathrm{Sp}(2)\end{tikzcd}$ and) evaluate on the neutral coset $\mathrm{e}\cdot \mathrm{Sp}(1)^2 \in \mathrm{Sp}(2)/\mathrm{Sp}(1)^2$, hence equivalently (by \cref{CosetRealizationOfHopfFibration}) on the base point $P_0 \in \mathbb{H}P^1$.
  
  However, precisely because the left map comes from inclusion of the stabilizer subgroup, its composite with the bottom map produces the situation \cref{TransformedFredholmUnderRepUnshuffle}
  discussed in \cref{UnshufflingTrivialAndDefiningReps}, whence the above diagram is equivalently of this form:
  \begin{equation}
    \begin{tikzcd}[
      column sep=60pt,
      row sep=15pt
    ]
      \mathrm{Sp}(1)_{\mathrm{stb}}
      \backsslash \ast
      \ar[
        rr,
        equals
      ]
      \ar[
        dd
      ]
      \ar[
        ddrr,
        "{
          \mathrm{Sp}(1)
          \backsslash
          \big(
            F^{\mathbb{C}^2_{\mathrm{triv}}}
            \,\ominus\,
            F^{\mathbb{C}^2_{\mathrm{triv}}}
          \big)
        }"{description, sloped}
      ]
      &&
      \mathrm{Sp}(1)_{\mathrm{stb}}
      \backsslash \ast
      \ar[
        dd,
        "{
          \mathrm{Sp}(1)_{\mathrm{stb}}
          \backsslash
          F_0
        }"
      ]
      \ar[
        dl,
        Rightarrow,
        dashed,
        shorten=10pt
      ]
      \\
      &
      {}
      \ar[
        dl,
        Rightarrow,
        shorten=13pt
      ]
      \\
      \mathrm{Sp}(1)^2 
      \backsslash
      \ast
      \ar[
        rr,
        "{
          \mathrm{Sp}(1)^2 
          \backsslash
          \big(
          \FredholmOperator
            ^{\mathrm{taut}}
            _{\mathbb{H}P^1}
          \ominus
          \FredholmOperator
            ^{\mathbb{C}^2_{\mathrm{triv}}}
            _{\mathbb{H}P^1}
          \big)
        }"{swap}
      ]
      &&
      \mathrm{Sp}(1)_{\mathrm{stb}}
      \backsslash
      \GradedFredholmOperators
       \mathrlap{\,.}
    \end{tikzcd}
  \end{equation}
  But now the maps of the top right triangle both pick Fredholm operators that are nontrivial only on the first summand of $\HilbertSpace \simeq \HilbertSpace_{\mathrm{triv}} \oplus \HilbertSpace_{\mathrm{def}}$ \cref{RestrictedNDirectSumOfSp2Representations}, hence on which the group action is trivial. Therefore, the remaining dashed homotopy is obtained from any plain homotopy (no equivariance constraint) in
  \begin{equation}
    \begin{tikzcd}[
      row sep=8pt, 
      column sep=30pt]
      \ast
      \ar[
        rr, 
        equals
      ]
      \ar[
        ddrr,
        "{
          F^{\mathbb{C}^2}
          \!\ominus
          F^{\mathbb{C}^2}
        }"{description, sloped, name=t}
      ]
      && 
      \ast
      \ar[
        dd,
        "{ 
           F_0 
        }"{description, name=s}
      ]
      \ar[
        from=s, 
        to=t,
        Rightarrow,
        dashed,
        shorten=5pt
      ]
      \\
      & {}
      \\
      && 
      \GradedFredholmOperators
      \,.
    \end{tikzcd}
  \end{equation}
  This does exist, by \cref{ConnectedComponentsOfSpaceOfFredholmOperators}, since both these Fredholm operators have vanishing index.  
\end{proof}

\begin{remark}
In summary, by \cref{PullbackOfTautologicalFredholmOperatorAlongQuaternionicHopfFibrationTrivializes}, we have constructed a homotopy of topological groupoids:
  \begin{equation}
    \label{TheEquivariantHomotopy}
    \begin{tikzcd}[
      column sep=50pt, 
      row sep=10pt
    ]
      \mathrm{Sp}(2)
      \backsslash %
      S^7 %
      \ar[
        rr
      ]
      \ar[
        dd,
        "{
          \mathrm{Sp}(2)
          \sslash
          \QuaternionicHopfFibration
        }"{swap}
      ]
      &&
      \mathrm{Sp}(2)
      \sslash
      \ast 
      \ar[
        dd,
        "{
          \mathrm{Sp}(2)
          \backsslash 
          F_0
        }"
      ]
      \ar[
        ddll,
        shorten=15pt,
        Rightarrow,
        "{
          \mathrm{Sp}(2)
          \sslash 
          h_3^{\mathrm{KU}}
        }"{description}
      ]
      \\
      \\
      \mathrm{Sp}(2)
        \backsslash
      S^4
      \ar[
        rr,
        "{
          \mathrm{Sp}(2)
          \sslash
          \big(
          F
            ^{\mathrm{taut}}
            _{\mathbb{H}P^1}
            \ominus
         F
           ^{\mathbb{C}^2_{\mathrm{triv}}}
          _{S^4}
         \big)
        }"{swap}
      ]
      &&
      \mathrm{Sp}(2)
      \backsslash
      \GradedFredholmOperators %
       \mathrlap{\,.}
    \end{tikzcd}
  \end{equation}
This is the announced unstable equivariant quaternionic orientation in topological K-theory, according to \cref{TwistedEquivariantOrientation} (bottom row). 

The construction for the complex orientation is obtained essentially verbatim by restricting all of the above discussion along the inclusion $\mathbb{C} \hookrightarrow \mathbb{H}$.
\end{remark}



\section{Fragile Topological Phases and Microscopic Charges}
\label{FragilePhasesAndMicroscopicCharges}

We expand on the application of the construction in \cref{OrientationsInOrbiKTheory} to the understanding of fragile topological phases and of microscopic charges in physical systems.

\subsection{On Band Nodes and Branes}

First some general words on the mathematical reflection of \emph{band nodes} in crystalline quantum matter and of \emph{monopole branes} in (higher) gauge quantum systems, in their charged singular version as well as in their \emph{gapped} or \emph{probe} incarnation that we are concerned with here.

For the following, consider:
\begin{enumerate}
\item $X^d$ a manifold representing either the momentum space of a crystalline quantum material or the physical space hosting (higher) gauge fields.

\item $\phi : \Sigma^p \hookrightarrow  X^d$ a submanifold  representing the spatial part of the \emph{worldvolume} of a higher dimensional gauge monopole (``$p$-brane'').

\item $\hotype{A}$ the classifying space of gapped Bloch Hamiltonians or of gauge charges, respectively.

\item $G \acts \, (-)$ a smooth action of a Lie group on this data, being the crystalline point group symmetry or the orbifolding group, respectively.
\end{enumerate}

\subsubsection{Essential nodes and singular branes}

Then the topological charge of $\phi$ being:
\begin{itemize}
\item an essential \emph{band node}, where the Berry curvature would diverge
\end{itemize}
respectively: 
\begin{itemize}
\item a singular \emph{monopole $p$-brane}, where the bulk flux density would diverge, 
\end{itemize}
is detected/measured by the $G$-equivariant $\hotype{A}$-cohomology of the \emph{complement} space $X^d \setminus \Sigma^p \subset X^d$ (disregarding here any further twists, just not to notationally overburden the discussion at this point):
\begin{equation}
  H_G\big(
    X^d \setminus \Sigma^p
    ;\,
    \hotype{A}
  \big)
  =
  \pi_0
\Bigg\{\!
 \adjustbox{raise=-8pt}{
  \begin{tikzcd}
    X^d \setminus \Sigma^p
    \ar[
      out=60,
      in=180-60,
      looseness=4,
      "{ 
        \;\mathclap{G}\; 
      }"{description}
    ]
    \ar[
      r,
      dashed
    ]
    &
    \hotype{A}
    \ar[
      out=57,
      in=180-57,
      looseness=4,
      "{ 
        \;\mathclap{G}\; 
      }"{description}
    ]
  \end{tikzcd}
  }
 \! \Bigg\}
  \mathrlap{.}
\end{equation}

The historical and archetypical example is that of a Dirac monopole $\phi : \begin{tikzcd}[sep=small]\{0\} \ar[r, hook] & \mathbb{R}^3\end{tikzcd}$ whose magnetic charge is classified by $\hotype{A} \simeq B \mathrm{U}(1)$ as
\begin{equation}
  \begin{tikzcd}
    H\big(
      \mathbb{R}^3 
        \setminus
      \{ 0 \}
      ;
      B \mathrm{U}(1)
    \big)
    \simeq
    H^2\big(
      S^2
      ;\,
      \mathbb{Z}
    \big)
    \simeq
    \mathbb{Z}
    \mathrlap{\,,}
  \end{tikzcd}
\end{equation}
where 
\begin{equation}
  S^2 
  \;\underset{
    \mathclap{\mathrm{hmtp}}
  }{\simeq}\;
  S^2 \times \mathbb{R}_{> 0}
  \simeq
  \mathbb{R}^3 \setminus \{0\}
\end{equation}
arises as the 2-sphere around the monopole at the origin of space. The physical picture is hence that the cohomology of the complement $X^d \setminus \Sigma^p$ measures the charge reflected in the total field flux that emanates from the monopole brane (where its density diverges) and penetrates through the boundary of a tubular neighborhood.

Of course, of $X^d$ itself has nontrivial cohomology then this may contribute to the cohomology of the complement $X^d \setminus \Sigma^p$. 

\subsubsection{Gapped nodes and probe branes}
\label{GappedNodesAndProbeBranes}

Now, when the band node gets \emph{gapped} (by deforming the nature of the underlying quantum material), or respectively when the monopole brane is regarded in the \emph{probe limit} where the backreaction of its (small) charge onto the ambient space is negligible and hence neglected, then this must mean that the charge is well-defined (non-divergent) on all of $X^d$, classified by a map
$ X^d \xrightarrow{ \tau} \hotype{A}$, 
\emph{and} that as such it receives no contribution from the previous node/brane locus, hence that the composite map $\phi^\ast c$
\begin{equation}
  \begin{tikzcd}[row sep=12pt, column sep=large]
    \Sigma^p
    \ar[
      d,
      hook,
      "{ \phi }"{swap}
    ]
    \ar[
      dr,
      "{
        \phi^\ast c
      }"
    ]
    \\
    X^d
    \ar[
      r, 
      "{ \tau }"{swap}
    ]
    &
    \hotype{A}
  \end{tikzcd}
\end{equation}
has trivial class, in a suitable sense, namely that it can be equipped with whatever structure it is that reflects the undoing of the previous charge. For instance, for a band node its undoing is its \emph{gapping}, typically exhibited by a choice of \emph{mass term}.

Whatever the trivialization process is, it will itself have a topological class which should have some classifying space $\hotype{B}$, equipped with a fibration $\begin{tikzcd}[sep=small] \hotype{B} \ar[r, ->>, "{ p }"] & \hotype{A}\end{tikzcd}$ encoding which trivializations $b \in \hotype{B}$ concern which charges $p(b) \in \hotype{A}$.

In conclusion, the topological data of $\phi$ a \emph{gapped} band node or \emph{probe} brane, relative to a background charge $c$, should be classified by a map $\begin{tikzcd}[sep=small] \Sigma^p \ar[r, dashed] & \hotype{B} \end{tikzcd}$ making this diagram commute:
\begin{equation}
  \begin{tikzcd}[row sep=12pt, column sep=large]
    \Sigma^p
    \ar[
      d,
      hook,
      "{ \phi }"{swap}
    ]
    \ar[
      r,
      dashed
    ]
    &
    \hotype{B}
    \ar[
      d,
      ->>,
      "{\, p }"
    ]
    \\
    X^d 
    \ar[
      r,
      "{ \tau }"
    ]
    &
    \hotype{A}
  \end{tikzcd}
\end{equation}
The deformation classes (relative homotopy classes) of such maps form the \emph{twisted relative cohomology} of \cref{NotionsOfTwistedRelativeCohomology}.

Thus, a pair of topological phases on $X^d$ may superficially have the same topological class in $H_G(X^d;\, \hotype{A})$, but in reality arise from a pair of topologically distinct gapping procedures $m_1 \neq m_2 \in H^{\phi^\ast c}_G\big(\Sigma^p;\, \hotype{B}, \hotype{A}\big) $ (cf. \cref{NotionsOfTwistedRelativeCohomology}) of a given mother phase. These twisted cohomology classes $m$ on $\Sigma^p$ hence witness that the two phases are not actually deformable into each other, after all.

We discuss the example of 2-band Chern insulators, below in \cref{GappedNodalLinesIn2BandChernInsulators}.

\subsection{Revisiting Fragile Topological Phases}
\label{RevisitingFragileTopologicalPhases}

Applying the construction of \cref{TheEquivariantUnit} to the 
fragile crystalline topological phase \cite{nLab:UnstableTopologicalPhaseOfMatter} of 2-band Chern insulators \cite{nLab:ChernInsulator} recovers exactly the Bloch Hamiltonian (cf. \cite{SS25-FQAH, SS25-CrystallineChern}) and then its equivariant K-theory class (cf. \cite{FreedMoore2013, SS22-Ord}). At gapped nodal lines (cf. \cref{GappedNodesAndProbeBranes}), the construction in \cref{TheOrientingHomotopy} gives the corresponding \emph{relative} K-classification, a new prediction.

We proceed to say this in more detail.

\subsubsection{Bloch Hamiltonian Maps}

In solid state physics, \emph{Bloch's theorem} (cf. \parencites[\S XIII.16]{ReedSimon1978}[\S 5.1.3]{Sergeev2023}) entails that the Hamiltonian operator $H$ for single electrons propagating in a $d$-dimensional crystalline material is a direct integral 
\begin{equation}
  \label{BlochHamiltonian}
  \begin{aligned}
    \HilbertSpace
    & \simeq
    \int_{\widehat{T}{}^d}
    \HilbertSpace_{\mathrm{Blch}}
    \,
    \mathrm{d}\mu
    \\
      H 
        & =
      \int_{\widehat{T}{}^d}
      H_{[\vec k]}
      \,
      \mathrm{d}\mu
      \;:\;
      \HilbertSpace
      \longrightarrow
      \HilbertSpace
  \end{aligned}
\end{equation}
over crystal momenta $[\vec k]$ varying in the \emph{Brillouin torus} (cf. \parencites[p. 52]{FreedMoore2013}[\S 2.1]{Thiang2025})
\begin{equation}
  \label{BrillouinTorus}
  \widehat{T}{}^d 
    \simeq 
  \mathbb{R}^d/\mathbb{Z}^d
  \,,
\end{equation}
of a continuous family of \emph{Bloch Hamiltonians}
\begin{equation}
  \label{BlochHamiltonianMap}
  H_{(-)}
  :
  \begin{tikzcd}
    \widehat{T}{}^d
    \ar[
      r,
      dashed
    ]
    &
    \mathrm{End}\big(
      \HilbertSpace_{\mathrm{Blch}}
    \big)
    \,.
  \end{tikzcd}  
\end{equation}
acting on a fiber Hilbert space $\HilbertSpace_{\mathrm{Blch}}$.
These Bloch Hamiltonians have discrete real spectrum, and the graphs of eigenvalues of $H_{(-)}$ are called the \emph{energy bands} (cf. \parencites[Fig. XIII.13]{ReedSimon1978}[\S 2]{Seeger2004}).

In an \emph{insulator} ground state, electron states occupy the lowest $v \in \mathbb{N}$ of these bands below a given \emph{Fermi energy}: the \emph{valence bands}. Depending on external excitations, some number $c \in \mathbb{N}$ of further bands may be accessible to excited electrons, the \emph{conduction bands}. Hence in dependence on external parameters, the system's ground state and its accessible excitations are approximately described by finite-rank Bloch Hamiltonians, given by maps of this form (cf. \cite[Prop. D.13]{FreedMoore2013}):
\begin{equation}
  \label{FiniteRankBlochHamiltonianMap}
  H_{(-)}
  :
  \begin{tikzcd}[sep=18pt]
    \widehat{T}{}^d 
    \ar[
      r,
      dashed
    ]
    &
    \BoundedOperators\big(
      \mathbb{C}^{v+c}
    \big)
    \ar[r, hook]
    &
    \mathrm{End}(\HilbertSpace_{\mathrm{Blch}})
    \,,
  \end{tikzcd}
\end{equation}
where we use the notation $\BoundedOperators(-)$ (``bounded operators'') for notational brevity:
\begin{equation}
  \BoundedOperators(\mathbb{C}^n)
  \simeq
  \mathrm{Mat}_{n \times n}(\mathbb{C})
  \mathrlap{\,.}
\end{equation}

Now for \emph{gapped} ground states, hence with a positive \emph{energy gap} between the valence and the conduction bands, the Bloch Hamiltonian map factors further through the subspace shown on the right here:
\begin{equation}
  \label{SpaceOfGappedFiniteRankNlochHamiltonians}
  H_{(-)}
  :
  \begin{tikzcd}
  \widehat{T}{}^d
  \ar[
    r,
    dashed
  ]
  &
  \BoundedOperators\big(
    \mathbb{C}^{v + c}
  \big)_{\!\mathrm{gap}}
  :=
  \left\{
    H \in \BoundedOperators\big(
      \mathbb{C}^{v + c}
    \big)
    \,\middle\vert\,
    \substack{
      H^\dagger = H
      \\
      \mathrm{Eig}_{< 0}(H)
      \,\simeq\,
      \mathbb{C}^v
      \\
      \mathrm{Eig}_{> 0}(H)
      \,\simeq\,
      \mathbb{C}^c
    }
  \right\}
  ,
  \end{tikzcd}
\end{equation}
where we now made explicit that the Bloch Hamiltonians are hermitian and we have chosen the origin of the energy scale to be the Fermi energy, so that the $v$ valence bands are those of negative energy. Quantum materials whose ground state fills the valence bundle of Bloch Hamiltonians of the form \cref{SpaceOfGappedFiniteRankNlochHamiltonians} with a \emph{nontrivial homotopy class} $\big[H_{(-)}\big] \in \pi_0 \mathrm{Map}\big( \widehat{T}{}^d, \BoundedOperators(\mathbb{C}^{v+c})_{\mathrm{gap}} \big)$ are called \emph{topological insulators} \cite{nLab:TopologicalInsulator}: Insulators because of the gap to the conduction band, and ``topological'' because of the twist $\big[H_{(-)}\big]$ in the electron couplings which is locally trivial but globally non-trivial.

More generally, that a subgroup $G$ of the point group of the crystal's space group (cf. \parencites[\S 26]{Armstrong1988}[\S 2]{Hammond2015}) is respected by (common jargon: ``protects'') the Bloch Hamiltonians means that there is a unitary representation of the point group on the Bloch quantum states
\begin{equation}
  \label{UnitaryRepresentationOfPointGroup}
  U : 
  \begin{tikzcd}[sep=18pt]
    G
    \ar[r]
    &
    \mathrm{U}\big(
      \mathbb{C}^{v+c}
    \big),
  \end{tikzcd}
\end{equation}
such that (cf. \parencites{NeupertSchindler2018}[\S 5.2]{Stanescu2020})
\begin{equation}
  \label{SymmetryTransformationOfBlochHamiltonian}
  \forall_{
    \substack{
      [\vec k] \in \widehat{T}{}^d
      \\
      g \in G
    }
  }
  \;\;:\;\;
  H_{g \cdot [\vec k]}
  =
  U(g) 
    \circ 
  H_{[\vec k]}
    \circ
  U(g)^{-1}
  \mathrlap{.}
\end{equation}
We highlight that this says equivalently that the Bloch Hamiltonian map \cref{SpaceOfGappedFiniteRankNlochHamiltonians} is \emph{$G$-equivariant} \cref{Equivariance} with respect to the given point group action on the Brillouin torus \cref{BrillouinTorus} and the conjugation of action of \cref{UnitaryRepresentationOfPointGroup} on the Bloch Hamiltonians:
\begin{equation}
  \label{EquivariantBlochHamiltonianMap}
  \begin{tikzcd}
    \widehat{T}{}^d
    \ar[
      in=60,
      out=180-60,
      looseness=4,
      "{
        \;\mathclap{G}\;
      }"{description}
    ]
    \ar[
      rr,
      dashed,
      "{ H_{(-)} }"
    ]
    &&
    \BoundedOperators\big(
      \mathbb{C}^{v+c}
    \big)_{\!\mathrm{gap}}
    \mathrlap{\,.}
    \ar[
      in=57,
      out=180-57,
      looseness=4,
      "{
        \;\mathclap{G}\;
      }"{description}
    ]
  \end{tikzcd}
\end{equation}

Yet more generally, the crystallographic symmetries may be accompanied by \emph{time reversal symmetries} in $\mathbb{Z}_2$ (cf. \cite[(32)]{SS22-Ord}). This is exhibited by equipping the symmetry group with a homomorphism
\begin{equation}
  \label{TheGradingOnSymmetryGroup}
  \sigma : 
  \begin{tikzcd}    
    G \ar[r] & \mathbb{Z}_2
  \end{tikzcd}
\end{equation}
and generalizing \cref{SymmetryTransformationOfBlochHamiltonian} to

\begin{equation}
  \label{SymmetryTransformationOfBlochHamiltonianIncludingTimeReversal}
  \forall_{
    \substack{
      [\vec k] \in \widehat{T}{}^d
      \\
      g \in G
    }
  }
  \;\;:\;\;
  H_{g \cdot [\vec k]}
  =
  \left\{
  \begin{aligned}
  U(g) 
    \circ 
  H_{[\vec k]}
    \circ
  U(g)^{-1}
  & \;\; 
  \mbox{if} \; \sigma(g) = \mathrm{e}
  \\
  U(g) 
    \circ 
  \overline{H_{[\vec k]}}
    \circ
  U(g)^{-1}  
  & \;\; 
  \mbox{otherwise}
  \mathrlap{\,,}
  \end{aligned}
  \right.
\end{equation}
where $\overline{(-)}$ denotes component-wise complex conjugation.

In summary: 
\begin{standout}
\textit{The parameters/couplings of $d$-dimensional $G$-crystalline $(v,c)$-band quantum materials vary in the equivariant mapping space \cref{EquivariantMappingSpace} $\mathrm{Map}\big(\widehat{T}{}^d, \BoundedOperators(\mathbb{C}^{v+c})_{\mathrm{gap}}\big)^G$.}
\end{standout}

\subsubsection{Unstable topological phases of matter}

A continuous deformation of the crystalline material --- say by external tuning, heat or other noise --- changes the Bloch Hamiltonians \cref{BlochHamiltonian} continuously, hence is an \emph{equivariant homotopy} of the corresponding maps \cref{BlochHamiltonianMap,FiniteRankBlochHamiltonianMap}. That these deformations 
\begin{enumerate}
\item preserve the presence of an energy gap means that this homotopy, too, factors through the space \cref{SpaceOfGappedFiniteRankNlochHamiltonians} of gapped finite-rank Bloch Hamiltonians, 
\item preserve the $G$-symmetry
\cref{EquivariantBlochHamiltonianMap} (hence respect the symmetry protection) means that this is an equivariant homotopy \cref{EquivariantHomotopy}:
\end{enumerate}
\begin{equation}
  \begin{tikzcd}[
    column sep=30pt
  ]
    \widehat{T}{}^d
    \ar[
      in=60,
      out=180-60,
      looseness=4,
      "{
        \;\mathclap{G}\;
      }"{description}
    ]
    \ar[
      rr, 
      bend left=20,
      "{ H_{(-)} }"{description, name=s}
    ]
    \ar[
      rr, 
      bend right=20,
      "{ H'_{(-)} }"{description, name=t}
    ]
    \ar[
      from=s,
      to=t,
      dashed,
      Rightarrow
    ]
    &&
    \BoundedOperators\big(
      \mathbb{C}^{v+ c}
    \big)_{\!\mathrm{gap}}
    \ar[
      in=57,
      out=180-57,
      looseness=4,
      "{
        \;\mathclap{G}\;
      }"{description}
    ]
  \end{tikzcd}
  \;
  :\;\;
  \begin{tikzcd}[
   column sep=7pt
  ]
    \{0\}
    \ar[
      d,
      hook,
      "{ 0 }"
    ]
    \ar[
      drr,
      "{
        \widetilde{H_{(-)}}
      }"{sloped}
    ]
    \\
    {[0,1]}
    \ar[
      rr,
      dashed,
    ]
    &&
    \mathrm{Map}\big(
      \widehat{T}{}^d
      ,\,
      \BoundedOperators(
        \mathbb{C}^{v+c}
      )_{\mathrm{gap}}
    \big)^G
    \mathrlap{\,,}
    \\
    \{1\}
    \ar[
      u,
      hook',
      "{ 1 }"{swap}
    ]
    \ar[
      urr,
      "{
        \widetilde{H_{(-)}}
      }"{swap, sloped}
    ]
  \end{tikzcd}
\end{equation}
and hence that the \emph{deformation classes} of such gapped ground states are classified by the corresponding equivariant homotopy classes \cref{EquivariantHomotopy}: 
\begin{equation}
  \label{SpaceOfUnstablePhasesViaMaps}
  (v,c)\mathrm{Phases}(d)^G
  :=
  \pi_0 
  \,
  \mathrm{Map}\Big(
    \widehat{T}{}^d
    ,\,
    \BoundedOperators\big(
      \mathbb{C}^{v+c}
    \big)_{\!\mathrm{gap}}
  \Big)^{\!G}
  .
\end{equation}

To get a better handle on this classification, and since these equivariant homotopy classes of maps depend only on the equivariant homotopy type of the classifying space, it is useful to pass to a tighter model of the latter. The following is immediate, but worth making explicit:
\begin{lemma}
The classifying space \cref{SpaceOfGappedFiniteRankNlochHamiltonians} of $(v,c)$-gapped Bloch Hamiltonians is $\mathrm{U}(\mathbb{C}^{v+c}) \rtimes \mathbb{Z}_2$-equivariantly homotopy equivalent to the Grassmannian space $\mathrm{Gr}_{v}^{v+c}$ \textup{(cf. \cite{Bendokat2024})} of $v$-dimensional complex subspaces of $\mathbb{C}^{v+c}$: 
\begin{equation}
  \label{HomotopyEquivalenceBetweenGappedBlochHamiltoniansAndLinearSubspaces}
  \begin{tikzcd}
    \BoundedOperators\big(
      \mathbb{C}^{v+c}
    \big)_{\!\mathrm{gap}}
    \ar[
      in=60,
      out=180-60,
      looseness=4,
      "{
        \mathrm{U}(\mathbb{C}^{v+c})
        \rtimes 
        \mathbb{Z}_2
      }"
    ]
    \ar[
      rr,
      "{
        \sim
      }"
    ]
    &&    
    \mathrm{Gr}_{v}^{v+c}
    \mathrlap{\,,}
    \ar[
      in=60,
      out=180-60,
      looseness=4,
      "{
        \mathrm{U}(\mathbb{C}^{v+c})
        \rtimes
        \mathbb{Z}_2
      }"
    ]
  \end{tikzcd}
\end{equation}
{where $\mathrm{U}(\mathbb{C}^{v,c})$ acts by conjugation on the left and by left multiplication on the right, and $\mathbb{Z}_2$ acts by complex conjugation on both sides.}
\end{lemma}
\begin{proof}
  Consider the following homeomorphic subspaces of $\BoundedOperators_v(\mathbb{C}^{v+c})$:
  \begin{equation}
    \label{SpaceOfNormalizedBlochHamiltonians}
    \begin{tikzcd}[row sep=-3pt,
      column sep=0pt
    ]
    \left\{
      N \in 
      \BoundedOperators\big(
        \mathbb{C}^{v+c}
      \big)
      \,\middle\vert\,
      \substack{
        N^\dagger = N
        \\
        \mathrm{Eig}_{-1}
        \simeq
        \mathbb{C}^{v}
        \\
        \mathrm{Eig}_{+1}
        \simeq
        \mathbb{C}^{c}
      }
    \right\}
    \ar[
      rr,
      <->,
      "{ \sim }"
    ]
    &&
    \left\{
      P \in
      \BoundedOperators\big(
        \mathbb{C}^{v+c}
      \big)
      \,\middle\vert\,
      \substack{
        P^\dagger = P
        \\
        P \circ P = P
        \\
        \mathrm{ker}(P)
        \simeq
        \mathbb{C}^v
      }
    \right\}
    \\
    N &\longmapsto&
    \tfrac{1}{2}\big(
      1 + N
    \big)
    \\
    2P - 1
    &\longmapsfrom&
    P
    \mathrlap{\,,}
    \end{tikzcd}
  \end{equation}
  where the homeomorphism evidently respects the given group action on both sides. The space of projectors on the right is moreover homeomorphic to the Grassmannian
  \begin{equation}
    \begin{tikzcd}[row sep=-3pt, column sep=0pt]
    \left\{
      P \in
      \BoundedOperators\big(
        \mathbb{C}^{v+c}
      \big)
      \,\middle\vert\,
      \substack{
        P^\dagger = P
        \\
        P \circ P = P
        \\
        \mathrm{ker}(P)
        \simeq
        \mathbb{C}^v
      }
    \right\}
    \ar[
      rr,
      <->,
      "{ \sim }"
    ]
    &&
    \grayoverbrace{
    \Big\{
      V \subset \mathbb{C}^{v+c}
      \,\Big\vert\,
      \mathrm{dim}(V) = v
    \Big\}
    }{
      \mathrm{Gr}_{v}^{v+c}
    }
    \\
    P &\longmapsto&
    \mathrm{ker}(P)
    \mathrlap{\,,}
    \end{tikzcd}
  \end{equation}
  and, again, this homeomorphism is evidently equivariant for the given group actions.
  Therefore, we are reduced to showing that the space of gapped Bloch Hamiltonians is equivariantly homotopy equivalent to the space of normalized Bloch Hamiltonians on the left of \cref{SpaceOfNormalizedBlochHamiltonians}.
  To this end, consider these maps:
  \begin{equation}
    \begin{tikzcd}[
      row sep=-3pt,
      column sep=0pt
    ]
    &&
    N 
      \ar[
        rr, 
        <-|,
        shorten=10pt
      ]
      &\phantom{---}&
    N
    \\
    \BoundedOperators\big(
      \mathbb{C}^{v+c}
    \big)_{\!\mathrm{gap}}
    &\defneq&
    \left\{
      H \in 
      \BoundedOperators\big(
        \mathbb{C}^{v+c}
      \big)
      \,\middle\vert\,
      \substack{
        H^\dagger = H
        \\
        \mathrm{Eig}_{< 0}
        \simeq
        \mathbb{C}^{v}
        \\
        \mathrm{Eig}_{> 0 }
        \simeq
        \mathbb{C}^{c}
      }
    \right\}
    \ar[
      rr,
      <-,
      shift left=5pt
    ]
    \ar[
      rr,
      shift right=5pt
    ]
    &&    
    \left\{
      N \in 
      \BoundedOperators\big(
        \mathbb{C}^{v+c}
      \big)
      \,\middle\vert\,
      \substack{
        N^\dagger = N
        \\
        \mathrm{Eig}_{-1}
        \simeq
        \mathbb{C}^{v}
        \\
        \mathrm{Eig}_{+1}
        \simeq
        \mathbb{C}^{c}
      }
    \right\}
    \\
    &&
    H 
    \ar[
      rr, 
      |->,
      shorten=10pt
    ]
      &&
    H \circ \sqrt{H^2}^{-1}
    \mathrlap{\,,}
    \end{tikzcd}
  \end{equation}
  where, on the right, $\sqrt{-}$ denotes the unique positive definite operator square root, and $(-)^{-1}$ its inverse operator.

  Now, the map going right-left-right is the identity, so that we are reduced to showing that the map going left-right-left is equivariantly homotopic to the identity:
  \begin{equation}
    \eta
    :
    \Big(
      H \mapsto
      H \circ \sqrt{H^2}^{-1}
    \Big)
    \Rightarrow
    \big(
      H \mapsto H
    \big)
    \mathrlap{\,.}
  \end{equation}
  But such a homotopy is, for instance, given by
  \begin{equation}
    \begin{tikzcd}[row sep=-3pt,
      column sep=5pt
    ]
      \BoundedOperators\big(
        \mathbb{C}^{v+c}
      \big)_{\!\mathrm{gap}}
      \times
      [0,1]
      \ar[
        rr,
        "{ \eta }"
      ]
      &&
      \BoundedOperators\big(
        \mathbb{C}^{v+c}
      \big)_{\!\mathrm{gap}}
      \\
      (H,t)
      &\longmapsto&
      H 
        \circ
      \Big(
        (1-t) \sqrt{H^2}^{-1}
        +
        t
      \Big)
      \mathrlap{\,,}
    \end{tikzcd}
  \end{equation}
  whose equivariance is again evident.
\end{proof}

In summary, we have seen that:
\begin{standout}
\textit{The $G$-symmetry protected $v$-band topological phases of $d$-dimensional crystalline gapped quantum materials with access to $c$ conduction bands are classified by $G$-equivariant nonabelian cohomology of the Brillouin torus with coefficients in $\mathrm{Gr}_v^{v+c}$.}
\end{standout}
in that, with \cref{SpaceOfUnstablePhasesViaMaps}:
\begin{equation}
  \label{ClassificationOfUnstablePhasesByNonabelianCohomology}
  (v,c)\mathrm{Phases}(d)^G
  \simeq
  H_G\big(
    \widehat{T}{}^d
    ,\,
    \mathrm{Gr}_{v}^{v+c}
  \big)
  \mathrlap{\,.}
\end{equation}
This classification is ``unstable'' (called ``fragile'' or ``delicate'', cf. \cite{nLab:UnstableTopologicalPhaseOfMatter}) in that in its assumption of fixed finite numbers $v$ and $c$ of valence and conductions bands accessible by the system, its classification may break down when more bands become accessible to the system, in particular in the stable K-theoretic limit where $v,c \to \infty$.

\begin{example}[Fragile crystalline 2-band insulator phases]
  \label[example]{FragileCrystalline2D2BandChernPhases}
  The prominent case of crystalline \emph{2-band Chern phases} (cf. \parencites{AndoFu2015}[\S8-9]{Sergeev2023} such as the \emph{Haldane model}, cf. \cite[\S 8.3]{Sergeev2023}) corresponds to setting $v \defneq c \defneq 1$ in \cref{ClassificationOfUnstablePhasesByNonabelianCohomology} and $\sigma \defneq \mathrm{e}$ in \cref{TheGradingOnSymmetryGroup}.
  In this case the fragile classifying space \cref{HomotopyEquivalenceBetweenGappedBlochHamiltoniansAndLinearSubspaces} happens to be given by the 2-sphere 
  \begin{equation}
    \mathrm{Gr}_{1}^2 
    \simeq
    \mathbb{C}P^1
    \simeq
    S^2
    \mathrlap{\,,}
  \end{equation}
  whence \cref{ClassificationOfUnstablePhasesByNonabelianCohomology} says \parencites[(17)]{SS25-FQAH}{SS25-CrystallineChern} that the most fine-grained \emph{fragile crystalline Chern-phases} are classified by the equivariant form \cite{Cruickshank2003, SS20-Tad} of \emph{Cohomotopy} cohomology theory (cf. \parencites[\S VII]{STHu59}[Ex. 2.7]{FSS23-Char}) 
  \begin{equation}
    \pi^n(-)
    :=
    \pi_0 \mathrm{Map}\big(
      -, S^2
    \big)
  \end{equation}
  in ``degree'' $n = 2$:
  \begin{equation}
    \label{CrystallineChernPhasesInEquivariantCohomotopy}
    \begin{aligned}
    (1,1)\mathrm{Phases}(d)^G
    & \simeq
    \pi^2_G\big(
      \widehat{T}{}^d
    \big)
    \\
    & \defneq
    \pi_0
    \mathrm{Map}\big(
      \widehat{T}{}^d
      ,\,
      \mathbb{C}P^1
    \big)^G.
    \end{aligned}
  \end{equation}
\end{example}

This is the example on which we will focus now. While Chern phases have received considerable attention, in particular in their 2-band form, actual analysis of their fragile crystalline phases \cref{CrystallineChernPhasesInEquivariantCohomotopy} seems to have found little to no attention before we brought up the issue in \parencites[(17)]{SS25-FQAH}{SS25-CrystallineChern}.

\subsubsection{Mass Terms gapping Nodal Lines}
\label{GappedNodalLinesIn2BandChernInsulators}

Following the discussion in \cref{GappedNodesAndProbeBranes}, we go one step further and take account of the topological class of the process by which the 2-band Chern insulator phases  \cref{FragileCrystalline2D2BandChernPhases}are obtained by ``gapping out''  \emph{nodal lines}. A nodal line in a \emph{topological semimetal} (cf. \cite{nLab:TopologicalSemiMetal}) is a curve in the Brillouin torus \cref{BrillouinTorus} over which the bulk energy gap between the valence and the conduction bands closes. A deformation of the material, reflected in a \emph{mass term} being added to its Bloch Hamiltonian, may lift the band energy degeneracy over the previous nodal curve to turn the topological semimetal into a topological insulator \cref{SpaceOfGappedFiniteRankNlochHamiltonians}.

\smallskip 
Typically, such band nodes are ``protected'' by a symmetry \cref{UnitaryRepresentationOfPointGroup}, such as by $\mathbb{Z}_2$ ``mirror symmetry'' (cf. \parencites[\S 2.A]{Fang2016}{Ma2018}), which acts on the Brillouin torus by reflection of the $x$-coordinate (say)
\begin{equation}
  \begin{tikzcd}[row sep=-2pt, column sep=2pt]
    \mathbb{Z}_2
    \times
    \widehat{T}{}^3
    \ar[
      rr
    ]
    &&
    \widehat{T}{}^3
    \\
    \big(
      [1]
      ,\,
      [k_x, k_y, k_z]
    \big)
    &\longmapsto&
    {[-k_x, k_y, k_z]}
  \end{tikzcd}
\end{equation}
and acts on the bands as the $X$ Pauli matrix \cref{MatrixRepresentationOfQuaternions}
\begin{equation}
  -\mathrm{i}
  \CliffordElement(\mathbf{i})
  \defneq
  \left(
  \begin{matrix}
    1 & ~0
    \\
    0 & -1
  \end{matrix}
  \right)
  \mathrlap{.}
\end{equation}
The nodal line then lies in one of the mirror planes $[k_x] = 0$, say at $k_x = 0$, where the band symmetry \cref{SymmetryTransformationOfBlochHamiltonian} implies that the Bloch Hamiltonian, which may generally be expanded in the Pauli matrices \cref{MatrixRepresentationOfQuaternions} as
\begin{equation}
  \begin{aligned}
  &
  H_{[\vec k]}
  =
  h_{[\vec k]}
  -
  \mathrm{i}\Big(
  h^{x}_{[\vec k]}
  \CliffordElement(\mathbf{i})
  +
  h^{y}_{[\vec k]}
  \CliffordElement(\mathbf{j})
  +
  h^{z}_{[\vec k]}
  \CliffordElement(\mathbf{k})
  \Big)
  \in
  \BoundedOperators\big(
    \mathbb{C}^{2}
  \big)
  \,,
  \\
  &
  h^{(-)}_{[\vec k]} \in \mathbb{R}
  \mathrlap{\,,}
  \end{aligned}
\end{equation}
is of the form
\begin{equation}
  H_{[0, k_y, k_z]}
  =
  h_{[0, k_y, k_z]}
  -
  \mathrm{i}
  h^{x}_{[0, k_y, k_z]}
  \CliffordElement(\mathbf{i})
  \mathrlap{\,,}
\end{equation}
with vanishing energy gap on the nodal line $\phi : \!\begin{tikzcd}[sep=small] S^1 \ar[r, hook] & \widehat{T}{}^3\end{tikzcd}$,
\begin{equation}
  h^x_{[0,k_y, k_z]}
  = 0
  \;\;\;
  \Leftrightarrow
  \;\;\;
  [0,k_y, k_z]
  \in
  \phi(S^1) 
  \subset
  \widehat{T}{^3}
  \mathrlap{.}
\end{equation}

Conversely, the \emph{mass term} perturbation which gaps this nodal curve (cf. \cite[p. 23]{SS22-Ord}) typically needs to anti-commute (cf. \parencites[(4)]{Morimoto2013}[p. 8]{Chiu2014}[Lem 9.55]{Freed2021}) with the Bloch Hamiltonian there, breaking the mirror symmetry. With the mass term relevant only on a tubular neighborhood of the nodal curve, we may therefore consider it as localized on the nodal curve, where it hence is of the form
\begin{equation}
  \label{MassTermOnNodalLine}
  \begin{aligned}
    &
    M_{s} 
      = 
    -\mathrm{i}
    \big(
      m^y_s \, \CliffordElement(\mathbf{j})
      +
      m^z_s \, \CliffordElement(\mathbf{k})
    \big)
    \\
    &
    (m^x_s,m^y_s) 
      \in 
    \mathbb{R}^2 \setminus \{0\} 
    \\
    &
    s \in S^1
    \mathrlap{\,,}
  \end{aligned}
\end{equation}
hence equivalently
\begin{equation}
  \label{MassTermRevolvingAroundNodalLine}
  \begin{aligned}
    &
    M_s
    =
    -
    m_s \, 
    e^{ 
      \tfrac{\alpha_s}{2}
      \CliffordElement(\mathbf{i})
    }
    \cdot
    \mathrm{i}
    \CliffordElement(\mathbf{j})
    \cdot
    e^{ 
      -
      \tfrac{\alpha_s}{2}  
      \CliffordElement(\mathbf{j})
    }
    \mathrlap{\,,}
    \\
    &
    m_s \in \mathbb{R}_{> 0}
    \mathrlap{\,,}
    \\
    &
    \alpha_s \in \mathbb{R} \;.
  \end{aligned}
\end{equation}

\subsubsection{Classifying Fibration for Mass Terms}
\label{ClassifyingFibrationForMassTerms}

The above analysis
\cref{MassTermOnNodalLine} shows (which may not previously have been appreciated) that there is topology in the choice of mass term, classified by a winding number. (Mathematically, this is of just the form familiar from the \emph{Su-Schrieffer-Heeger model}, cf. \parencites[\S 1]{Asboth2016}{nLab:SSHModel}, but the physical phenomena described in both cases are different.)

Concretely, the above formula \cref{MassTermRevolvingAroundNodalLine} makes manifest that, after picking any Bloch basis state for the valence bundle over the basepoint $[0] \in S^1$, the variation of the mass term along the nodal curve is equivalently reflected by the correspondingly varying family of its Bloch eigenstates, say 
\begin{equation}
  \label{FamilyOfStatesOfAMassTerm}
  s
\;  \longmapsto \;
  e^{%
    \tfrac{\alpha_s}{2}
    \CliffordElement(\mathbf{i})
  }
  \left(
  \begin{matrix}
    +1
    \\
    -1
  \end{matrix}
  \right)
  =
  \left(
  \begin{matrix}
    +
    e^{+\mathrm{i}\tfrac{\alpha_s}{2} }
    \\
    -
    e^{-\mathrm{i}\tfrac{\alpha_s}{2} }
  \end{matrix}
  \right).
\end{equation}

We observe now that this family of Bloch states over the (previous, now gapped) nodal curve, reflecting the choice of gapping process (mass term) giving rise to a fully gapped 2-band insulator phase $H_{(-)} : \begin{tikzcd}[sep=small]\widehat{T}{}^d \ar[r] & \mathbb{C}P^1,\end{tikzcd}$ is exactly the choice of a dashed arrow making the following diagram commute
\begin{equation}
  \label{MassTermAsRelativeLiftOfBlochHamiltonian}
  \begin{tikzcd}[
    ampersand replacement=\&,
    column sep=80pt
  ]
    S^1
    \ar[
      d, 
      hook,
      "{ \phi }"{swap}
    ]
    \ar[
      r, 
      dashed,
      "{
        s 
         \,\mapsto\,
        \scalebox{0.8}{$\left(
        \begin{matrix}
          +e^{+\mathrm{i}\tfrac{\alpha_{s}}{2}}
          \\
          -e^{-\mathrm{i}\tfrac{\alpha_{s}}{2}}
        \end{matrix}
        \right)
        \cdot 
        \mathbb{Z}_2
        $}
      }"
    ]
    \&
    \mathbb{R}P^3
    \ar[
      d, 
      ->>,
      "{ 
        \FactoredCHopfFibration
      }"
    ]
    \&[-80pt]
    v \cdot \mathbb{R}
    \ar[
      d,
      shorten=5pt,
      |->
    ]
    \\
    \widehat{T}{}^d
    \ar[
      r,
      "{
        H_{(-)}
      }"
    ]
    \&
    \mathbb{C}P^1
    \&
    v \cdot \mathbb{C}
    \mathrlap{\,,}
  \end{tikzcd}
\end{equation}
where the map on the right is the factor of the complex Hopf fibration from \cref{FactoringTheCHopfFibration}.

To note here how the commutativity of this diagram accurately reflects the aspects of the gapping process relevant for topological classification: The composite left-bottom map identifies the fibers of the valence bundle (after the gapping) over the locus where the nodal gap closure had been, and the composite top-right map identifies these with exactly the fibers that the mass term has produced by shifting away the degenerate conduction band. 

This indicates that the classification of topological insulator phases which takes into account their origin by gapping a given topological semimetal phase with a nodal curve $\phi$, is given not just by plain Cohomotopy as in \cref{CrystallineChernPhasesInEquivariantCohomotopy}, but by its relative twisted version
$
  H\big(
    \phi, 
    \FactoredCHopfFibration
  \big)^G
$ from \cref{MeasuringRelativeChargesInQCohom}.

Note that the above dashed map \cref{MassTermAsRelativeLiftOfBlochHamiltonian} has codomain $\mathbb{R}P^3 \simeq S(\mathbb{C}^2)/\mathbb{Z}_2$ instead of $S(\mathbb{C}^2)$ itself because with the mass term \cref{MassTermRevolvingAroundNodalLine} being a well defined function on $S^1$, the Bloch states \cref{FamilyOfStatesOfAMassTerm} are in general only periodic up to sign inversion.
Explicitly, iff the winding number of the mass term \cref{MassTermRevolvingAroundNodalLine} is \emph{even} then the dashed map \cref{MassTermAsRelativeLiftOfBlochHamiltonian}
 factors via $\begin{tikzcd}[sep=small]S(\mathbb{C}^2) \ar[r,->>] & \mathbb{R}P^3\end{tikzcd}$ \cref{FactoringTheCHopfFibration} through the actual $\mathbb{C}$-Hopf fibration
\begin{equation}
  \mbox{Even mass term winding}
  \;\;\;
  \Rightarrow
  \;\;\;
  \begin{tikzcd}[
    ampersand replacement=\&,
    column sep=80pt
  ]
    S^1
    \ar[
      d, 
      hook,
      "{ \phi }"{swap}
    ]
    \ar[
      r, 
      dashed,
      "{
        s 
         \,\mapsto\,
      \scalebox{0.8}{$  \left(
        \begin{matrix}
          +e^{+\mathrm{i}\tfrac{\alpha_{s}}{2}}
          \\
          -e^{-\mathrm{i}\tfrac{\alpha_{s}}{2}}
        \end{matrix}
        \right)
        $}
      }"
    ]
    \&
    S(\mathbb{C}^2)
    \ar[
      d, 
      ->>,
      "{ \ComplexHopfFibration }"
    ]
    \\
    \widehat{T}{}^d
    \ar[
      r,
      "{
        H_{(-)}
      }"
    ]
    \&
    \mathbb{C}P^1
    \mathrlap{\,.}
  \end{tikzcd}
\end{equation}

In conclusion, we find that 2-band gapped topological phases together with the gapping process that made room for it, of nodal curves $\phi$ in a parent semi-metal phase, are jointly classified by the relative twisted generalization of the (equivariant) Cohomotopy classification from \cref{FragileCrystalline2D2BandChernPhases} which is classified by the fibration $\FactoredCHopfFibration$ \cref{FactoringTheCHopfFibration}
\begin{equation}
  \label{RelatveTwistedCohomologyForNodalLines}
  H_G\big(
    \phi
    ;\,
    \FactoredCHopfFibration
  \big)
  =
  \pi_0
  \left\{
  \adjustbox{raise=3pt}{
  \begin{tikzcd}
    \Sigma^1
    \ar[
      rr,
      dashed
    ]
    \ar[
      d,
      hook,
      "{ \phi }"{swap}
    ]
    &&
    \mathbb{R}P^3
    \ar[
      d,
      ->>,
      "{ \FactoredCHopfFibration }"
    ]
    \\
    \widehat{T}{}^d
    \ar[
      rr,
      dashed,
      "{  H_{(-)} }"
    ]
    &&
    \mathbb{C}P^1
  \end{tikzcd}
  }
 \right\}^G.
\end{equation}
This describes the (potentially) practically relevant physical situation where the topological insulator phase remains close enough to the parent nodal line semimetal phase that its deformations cannot reach insulator phases that arise from topologically distinct gapping processes.

\subsubsection{Stabilization to K-Theory}

Still more popular in the current literature than unstable/fragile band topology \cref{ClassificationOfUnstablePhasesByNonabelianCohomology} is (cf. \cite{nLab:KTheoryClassOfTopPhases})
the coarser but \emph{stable} classification of crystalline topological phases in equivariant K-theory \cref{OrbiKReducingToEquivariantK}.

Our construction in \cref{TheEquivariantOrientation} now gives, first of all, the coarsening cohomology operation from the fine-grained but fragile crystalline phases classified by equivariant 2-Cohomotopy \cref{CrystallineChernPhasesInEquivariantCohomotopy} to coarse but stable classification by equivariant K-theory \cref{OrbiKReducingToEquivariantK}:

Given a rank=2 unitary representation $U$  \cref{UnitaryRepresentationOfPointGroup} of the crystal point group $G$, and hence the induced stable $G$-action on $\GradedFredholmOperators$ as in \cref{StableActionOnFred}
\begin{equation}
  \begin{tikzcd}[sep=15pt]
    G
    \ar[
      r,
      "{ U }"
    ]
    &
    \mathrm{U}(2)
    \ar[
      rr,
      "{
        \bigoplus_{\mathbb{N}}
        \,
        \gamma
      }"
    ]
    &&
    \mathrm{U}(\HilbertSpace)
    \ar[r]
    &
    \mathrm{Aut}\big(
      \GradedFredholmOperators
    \big)
    \mathrlap{\,,}
  \end{tikzcd}
\end{equation}
encoding the given band symmetry as per \cref{SymmetryTransformationOfBlochHamiltonian},
composition with the map which is the tautological $\mathbb{C}P^1$-parameterized Fredholm operator (as in \cref{EquivarianceOfTautologicalFred}) gives the cohomology operation from fragile equivariant 2-Cohomotopy to stable equivariant K-theory, shown at the bottom here:
\begin{equation}
  \begin{tikzcd}[
    column sep=50pt,
    row sep=0pt
  ]
    \mathbb{C}P^1
    \ar[
      out=60,
      in=180-60,
      looseness=4,
      "{
        \,\mathclap{G}\,
      }"{description}
    ]
    \ar[
      rr,
      "{
        F
          ^{\mathrm{taut}}
          _{\mathbb{C}P^1}
        \,\ominus\,
        F
          ^{\mathbb{C}^1}
          _{\mathrm{triv}}
      }"
    ]
    &&
    \GradedFredholmOperators
    \ar[
      out=60,
      in=180-60,
      looseness=4,
      "{
        \,\mathclap{G}\,
      }"{description}
    ]
    \\
    \pi_0 \, \mathrm{Map}\big(
      \widehat{T}{}^d
      ,\,
      \mathbb{C}P^1
    \big)^G
    \ar[
      d,
      equals
    ]
    \ar[
      rr,
      "{
        \big(
        F
          ^{\mathrm{taut}}
          _{\mathbb{C}P^1}
        \,\ominus\,
        F
          ^{\mathbb{C}^1}
          _{\mathrm{triv}}
        \big)_\ast
      }"{description}
    ]
    &&
    \pi_0 \, \mathrm{Map}\big(
      \widehat{T}{}^d
      ,\,
      \GradedFredholmOperators
    \big)^G
    \ar[
      d,
      equals
    ]
    \\[8pt]
    \pi^2_G\big(
      \widehat{T}{}^d
    \big)
    \ar[
      rr,
      "{
        \mathrm{ch}^{\pi/K}
      }"
    ]
    &&
    \mathrm{KU}_G\big(
      \widehat{T}{}^d
    \big)
    \mathrlap{\,.}
  \end{tikzcd}
\end{equation}

Recalling the construction of this map from \cref{TheEquivariantOrientation}, at $P \in \mathbb{C}P^1$ the tautological Fredholm operator is $F^{\mathrm{taut}}_P = F_{P}$ \cref{TheTautologicalParameterizedFredholmOperator}, which is the Fredholm operator incarnation \cref{FredholmOperatorFromProjector} of the projector $P$. Unwinding the definitions, this recovers the normalized Bloch Hamiltonian in its Pauli matrix expansion (cf. \cite[\S 6.2.3]{SS25-Bun}).

Secondly, we obtain the refinement of this stabilization construction to the relative twisted classification \cref{RelatveTwistedCohomologyForNodalLines} of 2-band insulator phases sensitive to the gapping procedure form their nodal line semimetal parent phase: The corresponding cohomology operation is now given by forming the pasting composite according to \cref{MeasuringRelativeChargesInQCohom} with the $\mathbb{R}P^3$-relative form of the homotopy $h_1^{\mathrm{KU}}$ from \cref{PullbackOfTautologicalFredholmOperatorAlongQuaternionicHopfFibrationTrivializes}:
\begin{equation}
  \begin{array}{c}
  \begin{tikzcd}[column sep=40pt]
    \hspace{-30pt}
    H_G\big(
      \phi;
      \,
      \FactoredCHopfFibration
    \big)
    \ar[rr]
    &\phantom{---}&
    H_G\big(
      \phi;
      \,
      \gamma_1^{\mathrm{KU}}
    \big)
  \end{tikzcd}
  \\
  \left(
  \adjustbox{raise=3pt}{
  \begin{tikzcd}[
    column sep=25pt,
    row sep=25pt
  ]
    \Sigma^1
    \ar[
      d,
      hook,
      "{ 
        \phi 
      }"{swap}
    ]
    \ar[
      r,
      dashed,
      "{ M_{(-)} }"
    ]
    &
    \mathbb{R}P^3
    \ar[
      d,
      ->>,
      "{
        \FactoredCHopfFibration
      }"{swap}
    ]
    \\
    \widehat{T}^{d}
    \ar[
      r,
      dashed,
      "{
        H_{(-)}
      }"
    ]
    &
    \mathbb{C}P^1
  \end{tikzcd}
  }
  \right)
  \longmapsto
  \left(
  \adjustbox{raise=3pt}{
  \begin{tikzcd}[
    column sep=34pt,
    row sep=25pt
  ]
    \Sigma^1
    \ar[
      d,
      hook,
      "{ 
        \phi 
      }"{swap}
    ]
    \ar[
      r,
      dashed,
      "{ M_{(-)} }"
    ]
    &
    \mathbb{R}P^3
    \ar[
      d,
      ->>,
      "{
        \FactoredCHopfFibration
      }"{swap}
    ]
    \ar[
      r,
      hook
    ]
    &[+10pt]
    \mathbb{R}P^4
    \ar[
      d,
      "{
        (f^{\mathbb{R}}_{\mathbb{C}})^\ast
        \gamma_1^{\mathrm{KU}}
      }"
    ]
    \ar[
      dl,
      shorten=10pt,
      Rightarrow,
      "{
        h_1^{\mathrm{KU}}
      }"{sloped}
    ]
    \\
    \widehat{T}^{d}
    \ar[
      r,
      dashed,
      "{
        H_{(-)}
      }"{swap}
    ]
    &
    \mathbb{C}P^1
    \ar[
      r,
      "{
        F^{\mathrm{tau}}_{\mathbb{C}P^1}
        \ominus
        F^{\mathbb{C}^1}_{\mathrm{triv}}
      }"{swap}
    ]
    &
    \GradedFredholmOperators
  \end{tikzcd}
  }
 \! \right)
  .
  \end{array}
\end{equation}

\subsection{Revisiting Brane Charges}
\label{RevisitingBraneCharges}

We here discuss (along the lines of \cref{MeasuringRelativeChargesInQCohom,TwistedEquivariantOrientation}) the measurement in topological K-theory of the charges on M5-brane worldvolumes (such as sourced by the singular self-dual string), which microscopically are in twisted relative Cohomotopy. This will also give us occasion to make precise and complete an old suggestion of \parencites[(3.7, 3.17)]{Horava1998}[p. 6-8]{Witten2001} (see \cref{ComparingToHoravaWitten} below) for how to exhibit D6-brane charge in K-theory, namely for how to construct an explict map from $S^3 \simeq \mathbb{R}^3_{\cpt}$ to the space of self-adjoint Fredholm operators which represents, under \cref{SubspacesOfFredAsClassifyingSpacesForK},  a generator of $\mathrm{KU}^1(S^3) \simeq \mathbb{Z}$.

\subsubsection{Recap of abelian (D-)brane charges}
\label{RecapOfAbelianBraneCharges}

In order to make contact with traditional discussion, we briefly recall the gist of the traditional idea of measuring (D-)brane charges in abelian cohomology (cf. \cite{nLab:DBraneChargeQuantizationInKTheory}). While ``well known'', in its totality the following is not always easy to glean from existing literature.

So let $E$ be an abelian cohomology theory (cf. \cref{TableOfNotionsOfCohomology}) such as:
\begin{itemize}
\item $E^n(-) \simeq H^n(-;A)$ --- ordinary cohomology,
\item $E^n(-) \simeq \mathrm{KU}^n(-)$--- complex topological K-theory.
\end{itemize}

Then:

\begin{enumerate}
\item A \textbf{singular magnetic brane} is a substantial source of flux. In the absence of twistings, the total flux through any sphere enclosing the brane is the same (Gau{\ss} law), hence the flux \emph{density} on the enclosing spheres scales with a negative power of their radius and thus diverges at the would-be locus of the brane, which hence is a \emph{singularity}. Therefore the charge/total flux of singular branes is to be measured on the non-singular  \emph{complement} of their worldvolume. For flat branes, this complement is homotopy equivalent to the enclosing sphere:
\begin{equation}
  \label{EnclosingSphereAroundSingularBrane}
  \mathbb{R}^{d} \setminus \mathbb{R}^p   \;\simeq\; 
  \mathbb{R}^p \times \mathbb{R}_{>0} \times S^{d-p-1} 
  \underset{
    \mathrm{hmtpy}
  }{\sim}
  S^{d-p-1}
  \mathrlap{\,.}
\end{equation}

The original (theoretical) example is the \emph{Dirac magnetic monopole} with $d = 3$ and $p = 0$, and with the magnetic flux that is sourced by the monopole measured in ordinary integral 2-cohomology:
\begin{equation}
  H^2\big(
    \mathbb{R}^3 
      \setminus
    \mathbb{R}^0
    ;
    \mathbb{Z}
  \big)
  \simeq
  H^2(S^2; \mathbb{Z})
  \simeq
  \mathbb{Z}
  \mathrlap{\,.}
\end{equation}

\begin{table}[htb]
\caption{
  \label{TransverseTopologyOfFlatDBranes}
  Topologies on which to measure the total charge associated with singular (middle row) and probe D-branes branes (bottom row). 
}
\vspace{-2mm} 
\hspace*{.7cm}
\scalebox{0.95}{$
\def\arraystretch{1.6}
\def\tabcolsep{4pt}
\begin{tabular}{|c||c|c|c|c|c||c|c|c|c|c|c|}
  \hline
  & 
  \multicolumn{5}{c||}{\bf Type IIA}
  &
  \multicolumn{6}{c|}{\bf Type IIB}
  \\
  \hline
  \begin{tabular}{@{}c@{}}
    D$p$-brane
    species
  \end{tabular}%
  &
  D0
  &
  D2
  & 
  D4
  &
  D6
  &
  D8
  &
  D(-1)
  &
  D1
  & 
  D3
  & 
  D5
  & 
  D7
  & 
  D9
  \\
  \hline
  \hline
  \def\arraystretch{.9}%
  \begin{tabular}{@{}c@{}}
    Enclosing sphere 
    \\
    \scalebox{.9}{$
      S^{9-p-1}
      \sim
      \mathbb{R}^{9} \setminus \mathbb{R}^p
    $}
  \end{tabular}%
  &
  $S^8$
  &
  $S^6$
  &
  $S^4$
  &
  $S^2$
  &
  $S^0$
  &
  $S^9$
  &
  $S^7$
  &
  $S^5$
  &
  $S^3$
  &
  $S^1$
  &
  $S^{-1}$
  \\
  \hline
  \def\arraystretch{.9}%
  \begin{tabular}{@{}c@{}}
    Transverse space
    \\
    \scalebox{.9}{$
     S^{9-p}
     \sim
     \mathbb{R}^{9-p}_{\cpt}
       \times 
     \mathbb{R}^{p}
    $
    }
  \end{tabular}
  & 
  $S^9$
  & 
  $S^7$
  & 
  $S^5$
  & 
  $S^3$
  & 
  $S^1$
  &
  $S^{10}$
  &
  $S^8$
  &
  $S^6$
  &
  $S^4$
  &
  $S^2$
  &
  $S^0$
  \\
  \hline
\end{tabular}
$}
\end{table}

Often overlooked is that the analogue remains true for singular type IIA/B D-branes (cf. \cite[Rem. 4.5]{SS23-Defect}), hence for $d = 9$ and $p = 2k$ or $p = 2k + 1$, respectively, with the total RR-flux that they source measured in $\mathrm{KU}^0$ or $\mathrm{KU}^1$, respectively (cf. the second row of \cref{TransverseTopologyOfFlatDBranes}): 
\begin{equation}
  \begin{tikzcd}[sep=-2pt]
  \mathrm{KU}^0\big(
    \mathbb{R}^9 \setminus
    \mathbb{R}^{p=2k}
  \big)
  &\simeq&
  \mathrm{KU}^0\big(
    S^{9-2k-1}
  \big)
  &\simeq&
  \mathbb{Z}
  \mathrlap{\,,}
  \\
  \mathrm{KU}^1\big(
    \mathbb{R}^9 \setminus
    \mathbb{R}^{p=2k+1}
  \big)
  &\simeq&
  \mathrm{KU}^1\big(
    S^{9-2k}
  \big)
  &\simeq&
  \mathbb{Z}
  \mathrlap{\,.}
  \end{tikzcd}
\end{equation}

Instead, the traditional literature insists that D-brane charge in type IIA/B is measured in $\mathrm{KU}^1$/$\mathrm{KU}^0$, respectively (degrees reversed) --- but this statement refers to the charge not of substantial singular branes but of ``probe branes'', which is of different conceptual nature:

\item An \textbf{electric probe brane} is like a fundamental particle in perturbation theory, hence not ``back-reacted''. The spacetime topology is hence unaffected by their presence, and their charge/number is the integral, in $E^{d-p}$, of a density on the compactified transverse space. For flat branes, this is homotopy equivalent to \emph{another} sphere:
\begin{equation}
  \mathbb{R}^p \times 
  \big(
    \mathbb{R}^{d-p}_{\cpt}
  \big)
  \simeq
  \mathbb{R}^p \times 
  \times
  S^{d-p}
  \underset{
    \mathrm{hmtpy}
  }{\sim}
  S^{d-p}
  \mathrlap{\,.}
\end{equation}

The archetypical example is electric fundamental particles (like electrons) whose total charge/number is measured in ordinary integral cohomology:
\begin{equation}
  H^3\big(
    \mathbb{R}^0
    \times 
    \mathbb{R}^3_{\cpt}
    ;
    \mathbb{Z}
  \big)
  \simeq
  H^3\big(
    S^3
    ;
    \mathbb{Z}
  \big)
  \simeq
  \mathbb{Z}
  \mathrlap{\,.}
\end{equation}

Applied to D-branes (cf. the bottom row in \cref{TransverseTopologyOfFlatDBranes}), this gives the formulas found in the traditional literature \cite{nLab:DBraneChargeQuantizationInKTheory} on ``D-brane charge'' :
\begin{equation}
  \begin{tikzcd}[sep=0pt]
    \mathrm{KU}^{9-2k}\Big(
      \mathbb{R}^{2k}
      \times
      \mathbb{R}^{9-2k}_{\cpt}
    \Big)
    &\simeq&
    \mathrm{KU}^1\big(
      S^{9-2k}
    \big)
    &\simeq&
    \mathbb{Z}
    \\
    \mathrm{KU}^{9-2k-1}\Big(
      \mathbb{R}^{2k+1}
      \times
      \mathbb{R}^{9-2k-1}_{\cpt}
    \Big)
    &\simeq&
    \mathrm{KU}^0\big(
      S^{9-2k-1}
    \big)
    &\simeq&
    \mathbb{Z}
  \end{tikzcd}
\end{equation}

It is, tacitly, this second set of formulas, for \emph{probe} D-branes, which led to the proposal \parencites[(3.7, 3.17)]{Horava1998}[pp. 6-8]{Witten2001} for formulas for D6-brane charge in $\mathrm{KU}^1(S^3)$ in terms of Fredholm operators parameterized over their transverse space $\mathbb{R}^3_{\cpt}$.
\end{enumerate}

Below in \cref{3SphereOfSelfAdjointFredholmOps} we complete these old arguments. But what actually motivates us here is that (\cref{ChargeOnM5ProbesOfFlatSpace}) the same mathematics also describes stable measurement of charges \emph{on M5-branes} (such as of the self-dual 1-brane there), along the lines of \cref{OrientationsMeasuringRelativeCharges}.

\subsubsection{Charge on M5 probes of flat space}
\label{ChargeOnM5ProbesOfFlatSpace}

Consider the simple but important case of \cref{TwistedEquivariantOrientation} where an M5 brane worldvolume $\Sigma^{1,5}$ probes flat Minkowski spacetime.\footnote{
  But the following depends only on the pullback of the bulk Cohomotopy charge to the M5-worldvolume being trivial, $\phi^\ast (F_4^{\pi}, F_7^\pi) \simeq 0$, which is the case for instance also for holographic embeddings of M5-branes into $\mathrm{AdS}_7 \times S^4$, cf. \cite{GSS25-Embedding}. 
}

Since this ambient bulk spacetime is contractible, its microscopic bulk charge is trivial and equivalently represented by a map constant on a point in $S^4$, whence the relative cohomology on the brane worldvolume is microscopically in the plain 3-Cohomotopy of $\Sigma^5$. Measuring this charge in relative $\mathrm{KU}^0$ is equivalent to measuring it in $\Omega \mathrm{KU}^0 \simeq \mathrm{KU}^1$ by pushforward along the unit 
$\begin{tikzcd}[sep=small] \Sigma^3 1^{\mathrm{KU}} : S^3 \ar[r] & \mathrm{KU}^1\end{tikzcd}$
\begin{equation}
  \begin{aligned}
  &
  \begin{tikzcd}[row sep=21pt, column sep=14pt]
    \Sigma^5
    \ar[
      r,
      dashed
    ]
    \ar[
      dd,
      hook,
      "{ \phi }"
    ]
    &
    S^3
    \ar[r, hook]
    \ar[
      dd
    ]
    \ar[
      ddr,
      phantom,
      "{ \lrcorner }"{pos=.1}
    ]
    & 
    S^7
    \ar[
      dd,
      "{
        p_{\mathbb{H}}
      }"{description, pos=.4}
    ]
    \ar[rr]
    &&
    \ast
    \ar[
      dd
    ]
    \ar[
      ddll,
      Rightarrow,
      shorten=6pt,
      "{
        h_3^{\mathrm{KU}}
      }"{description}
    ]
    \\
    \\
    \mathbb{R}^{10}\
    \ar[
      r,
      "{ \sim }"
    ]
    &
    \ast
    \ar[r]
    &
    S^4
    \ar[
      rr,
      "{
        \Sigma^4 1^{\mathrm{KU}}~
      }"{swap}
    ]
    &&
    \mathrm{KU}^0
  \end{tikzcd}
 \;\simeq \;
  \begin{tikzcd}
   [
     row sep=12.5pt, column sep=small
   ]
    \Sigma^5
    \ar[
      r,
      dashed,  
    ]
    \ar[
      dd,
      hook,
      "{ \phi }"
    ]
    &
    S^3
    \ar[rrrr]
    \ar[
      dd
    ]
    \ar[
      drr,
      shorten=-4pt,
      "{ 
        \Sigma^3 1^{\mathrm{KU}} 
      }"{sloped, swap}
    ]
    & &&
    &[-8pt]
    \ast
    \ar[
      dd
    ]
    \\
    & &&
    \mathrm{KU}^1
    \ar[urr]
    \ar[dll]
    \ar[
      drr,
      phantom,
      "{ \lrcorner }"{pos=.1}
    ]
    \\
    \mathbb{R}^{10}\
    \ar[
      r,  
      "{ \sim }"
    ]
    &
    \ast
    \ar[rrrr]
    &
    &&&
    \mathrm{KU}^0
    \mathrlap{.}
  \end{tikzcd}
  \end{aligned}
\end{equation}

\newpage 
\begin{example}
  The primary \emph{singular} brane (cf. \cref{RecapOfAbelianBraneCharges}) \emph{on} the M5-worldvolume is the 1-brane known as the \emph{self-dual string} or \emph{M-string} \cite{nLab:MString}, which is the source of the self-dual $H_3$-flux (cf. \cite[\S 3.3]{GSS25-M5}): Its enclosing sphere \cref{EnclosingSphereAroundSingularBrane} is the 3-sphere
  \begin{equation}
    \Sigma^5 
    \defneq
    \mathbb{R}^5 \setminus \mathbb{R}^1
    \simeq
    \mathbb{R}^1 \times \mathbb{R}_{\> 0}
    \times S^{3}
    \underset{
      \mathrm{hmtpy}
    }{\sim}
    S^3
    \mathrlap{\,,}
  \end{equation}
  so that its microscopic brane charge in Cohomotopy is 
  \begin{equation}
    \pi_0\big\{
      \begin{tikzcd}[sep=small]
        \Sigma^5 
        \ar[r, dashed]
        &
        S^3
      \end{tikzcd}
    \big\}
    \simeq
    \mathbb{Z}
    \mathrlap{\,,}
  \end{equation}
  as it should be. And, in fact, this is already equal to the stable brane charge as seen in K-theory, in that the coarsening cohomology operation is an isomorphism, in this case:
  \begin{equation}
    \pi^3(\Sigma^5)
    \simeq
    \pi_0\big\{
      \begin{tikzcd}[sep=small]
        \Sigma^5 
        \ar[r, dashed]
        &
        S^3
      \end{tikzcd}
    \big\}
    \begin{tikzcd}
      {}
      \ar[
        rr,
        "{
          (\Sigma^3 1^{\mathrm{KU}})_\ast
        }",
        "{ \sim }"{swap}
      ]
      &&
      {}
    \end{tikzcd}
    \pi_0\big\{
      \begin{tikzcd}[sep=small]
        \Sigma^5 
        \ar[r, dashed]
        &
        \mathrm{KU}^1
      \end{tikzcd}
    \big\}
    \simeq
    \mathrm{KU}^1(S^3)
    \mathrlap{\,.}
  \end{equation}
\end{example}

In order to understand this M5-worldvolume charge measurement in K-theory more generally and more deeply, we proceed to construct an explicit model for $\Sigma^3 1^{\mathrm{KU}}$ as a map to self-adjoint Fredholm operators:

\subsubsection{The 3-sphere of self-adjoint Fredholm operators}
\label{3SphereOfSelfAdjointFredholmOps}

By \cref{SubspacesOfFredAsClassifyingSpacesForK} we have
\begin{equation}
  \big\{
    \begin{tikzcd}[sep=small]
      S^3
      \ar[r, dashed]
      &
      \FredholmOperators_{\mathbb{C}}^+
    \end{tikzcd}
  \big\}
  \simeq
  \pi_0
  \,
  \mathrm{KU}^1(S^3)
  \simeq
  \mathbb{Z}
  \mathrlap{\,.}
\end{equation}

We ask now for an essentially explicit formula for the map on the left that corresponds to $1 \in \mathbb{Z}$ on the right. Basic as this question is, it does not seem to have been discussed in the mathematical literature. But a partial suggestion for how to go about this construction may be understood to have been proposed by \parencites[(3.7, 3.17)]{Horava1998}[pp. 6-8]{Witten2001}, there with the aim of describing the K-theory charge of D6-branes (whose transversal space is $\simeq S^3$). We will now complete this suggestion to a rigorous solution. The key to that is the following \cref{AtiyahSingerExponentialMap} 
(of which \parencites[(3.7)]{Horava1998} might have been an echo).

Before discussing this proposition and its consequences, we first recall some basic operator algebraic background:

\begin{definition}[Spectrum of self-adjoint operators]
\label[definition]{SepctrumOfSAOperators}
The \emph{spectrum} 
$\mathrm{spec}(A) \subset \mathbb{C}$
of a linear operator $A \in \BoundedOperators(\HilbertSpace)$ is the subset of $\lambda \in \mathbb{C}$ for which $A - \lambda \cdot \mathrm{id}$ does not have a bounded linear inverse. For self-adjoint operators, where the spectrum is real
\begin{equation}
  A^\dagger = A
  \;\;\;
  \Rightarrow
  \;\;\;
  \mathrm{spec}(A) \subset \mathbb{R}
  \mathrlap{\,,}
\end{equation}
this is (cf. \cite[Rem. 9.15    ]{Moretti2017}) the union of
\begin{enumerate}
\item the \emph{discrete spectrum} $\mathrm{spec}_{\mathrm{dsc}}(A) \subset \mathrm{spec}(A)$, consisting of those eigenvalues $\lambda$ of $A$ which are
\begin{enumerate}
  \item isolated: there exists $\epsilon \in \mathbb{R}_{> 0}$ such that $(\lambda-\epsilon, \lambda+ \epsilon) \subset \mathbb{R}$ contains no other point of $\mathrm{spec}(A)$,
  \item of finite geometric multiplicity: the eigenspace $\mathrm{ker}(A - \lambda \cdot \mathrm{id})$ is finite-dimensional,
\end{enumerate}
\item 
the \emph{essential spectrum} which is remainder:
\begin{equation}
  \label{EssentialSpectrumOfSAOperator}
  \mathrm{spec}_{\mathrm{ess}}(A)
  =
  \mathrm{spec}(A) 
    \setminus 
  \mathrm{spec}_{\mathrm{dsc}}(A)
  \mathrlap{\,.}
\end{equation}
\end{enumerate}
\end{definition}

\begin{definition}
  \label[definition]{UnitaryFredholmGroup}
  The \emph{unitary Fredholm group} (cf. \cite{nLab:FredholmGroup}) is the subgroup of unitary operators on $\HilbertSpace$ \cref{TheHilbertSpace} which differ from the identity by a compact operator:
  \begin{equation}
    \label{TheUnitaryFredholmGroup}
    \mathrm{U}^c(\HilbertSpace)
    :=
    \big\{
      U \in \UH
      \;\big\vert\;
      U - \mathrm{id}
      \in 
      \CompactOperators(\HilbertSpace)
    \big\}\,.
  \end{equation}
  We will also be interested in the space
  \begin{equation}
    \label{MinusTheUnitaryFredholmGroup}
    -\mathrm{U}^c(\HilbertSpace)
    :=
    \big\{
      U \in \UH
      \;\big\vert\;
      U + \mathrm{id}
      \in 
      \CompactOperators(\HilbertSpace)
    \big\}
    \mathrm{\,,}
  \end{equation}
  which is not a group, but as a space is homeomorphic to the unitary Fredholm group.
\end{definition}
\begin{lemma}
  \label[lemma]
    {WHEfromStableUnitaryToUnitaryFredholm}
  The group homomorphism from the stable unitary group to the unitary Fredholm group \cref{TheUnitaryFredholmGroup}, given by \emph{stabilizing} unitary matrices
  \begin{equation}
    \begin{tikzcd}[row sep=-3pt, column sep=0pt]
      \bigcup_{n \in \mathbb{N}}
      \,
      \mathrm{U}(\mathbb{C}^n)
      \ar[
        rr
      ]
      &&
      \mathrm{U}^c(\HilbertSpace)
      \\
      (U,n)
        &\longmapsto&
      U
      \oplus
      \bigoplus_{\mathbb{N}_{>0}}
      \mathrm{id}_{\mathbb{C}^n}
    \end{tikzcd}
  \end{equation}
  is a weak homotopy equivalence, in that it induces isomorphisms on all homotopy groups.
\end{lemma}

\begin{notation}
Recall from \cref{DisjointComponentsOfSAFred} that the non-contractible component of the space self-adjoint complex Fredholm operators, $\FredholmOperators_{\mathbb{C}}^+$, is that where these have both positive and negative essential spectrum (\cref{SepctrumOfSAOperators}). Denote the further subspaces with essential spectrum concentrated on $\pm 1$ (\cite[below (2.5)]{AtiyahSinger1969}) and with actual spectrum concentrated on $\pm 1$ as follows, respectively:
\begin{equation}
  \label{SAFredWithPositiveAndNegativeEssentialSpectrum}
  \begin{tikzcd}[row sep=small, 
    column sep=0pt]
  \mathcal{F}_0
  \ar[d, hook]
  &:=&
  \Big\{
    f 
      \in
    \FredholmOperators^+_{\mathbb{C}}
    \;\big\vert\;
      \substack{
      \mathrm{spec}(f)
      \\
      =
      \{+1, -1\}
      }
  \Big\}
  \\
  \mathcal{F}
  \ar[
    d, hook
  ]
  &
  :=
  &
  \Big\{
    f 
      \in
    \FredholmOperators^+_{\mathbb{C}}
    \;\big\vert\;
      \substack{
      \mathrm{spec}_{\mathrm{ess}}(f)
      =
      \{+1, -1\}
      \\
      \vert f \vert = 1
      }
  \Big\}
  \\
  \FredholmOperators^+_{\mathbb{C},\ast}
  &:=&
  \Big\{
    f 
      \in
    \FredholmOperators^+_{\mathbb{C}}
    \;\big\vert\;
    \substack{
      \mathrm{spec}_{\mathrm{ess}}(f)
      \cap \mathbb{R}_{> 0}
      \neq \varnothing
      \\
      \mathrm{spec}_{\mathrm{ess}}(f)
      \cap \mathbb{R}_{< 0}
      \neq \varnothing
    }
  \Big\}.
  \end{tikzcd}
\end{equation}
\end{notation}
\begin{remark}
  By functional calculus, the elements $f \in \mathcal{F}_0$ \cref{SAFredWithPositiveAndNegativeEssentialSpectrum} satisfy: 
  \begin{equation}
    \label{fInF0SquareToIdentity}
    f \in \mathcal{F}_0
    \;\;\;\;
    \Rightarrow
    \;\;\;\;
    f^2 = \mathrm{id}
    \mathrlap{\,.}
  \end{equation}
\end{remark}

\begin{example}
Let 
$
  D(\mathbb{H}_{\mathrm{im}}) \simeq D^3
$ denote the closed unit ball in the space of imaginary quaternions (\cref{Quaternions}). Under the representation of such quaternions $x \in \mathbb{H}_{\mathrm{im}}$ by Pauli matrices $\CliffordElement_x$ (\cref{PauliMatricesAsStarRepresentation}) these give self-adjoint elements $ - \mathrm{i} \CliffordElement_x \in \BoundedOperators(\mathbb{C}^2)$ of operator norm $\vert \CliffordElement_x \vert = \vert x \vert \in [0,1]$. Therefore, \emph{stabilizing} these matrices by forming their direct sum with infinitely many copies of $\mathrm{diag}(+1, -1) \in \BoundedOperators(\mathbb{C}^2)$, under \cref{HilbertSpaceAsDirectSumOfComplexVectorSpaces}, they constitute  elements of the space $\mathcal{F}$ \cref{SAFredWithPositiveAndNegativeEssentialSpectrum}. Precisely when $x \in D(\mathbb{H}_{\mathrm{im}})$ is on the boundary, hence if $\vert x \vert = 1$, then this stabilized element is actually in $\mathcal{F}_0 \subset \mathcal{F}$. Hence we have a commuting diagram of maps of this form:
\begin{equation}
  \label{UnitBallQuaternionsAsFredOperators}
  \begin{tikzcd}[
    sep=0pt
  ]
    x
    \ar[
      rr,
      shorten=12pt,
      |->
    ]
    &\phantom{--}&
    -\mathrm{i} \CliffordElement_x
    \oplus
    \bigoplus_{\mathbb{N}_{>0}}
    \scalebox{.8}{$
    \left(
    \begin{matrix}
      + 1 \! & \, 0
      \\
      \; 0 \!& -1
    \end{matrix}
    \right)
   $}
   \\
    D^3 
    \simeq
    D(\mathbb{H}_{\mathrm{im}})
    \ar[
      rr,
      "{ \mathrm{stab} }"
    ]
    &&
    \mathcal{F}
    \\[10pt]
    S^2 
      \simeq 
    \partial D(\mathbb{H}_{\mathrm{im}})
    \ar[
      rr
    ]
    \ar[
      u,
      shorten=-1.5pt,
      hook
    ]
    &&
    \mathcal{F}_0
    \mathrlap{\,,}
    \ar[
      u, 
      hook
    ]
  \end{tikzcd}
\end{equation}
and therefore a map of quotient spaces
\begin{equation}  \label{UnitBallQuaternionsAsFredOperatorsModuloBoundary}
  \begin{tikzcd}
    S^3 
      \simeq
    D(\mathbb{H}_{\mathrm{im}})/\partial
    \ar[
      rr,
      "{
        \mathrm{stab}
      }"
    ]
    &&
    \mathcal{F}/\mathcal{F}_0
    \mathrlap{\,.}
  \end{tikzcd}
\end{equation}
\end{example}
Our next goal is to show that this map represents the generator of $\mathrm{KU}^1(S^3) \simeq \mathbb{Z}$.

\begin{lemma}
  The second inclusion 
  in \cref{SAFredWithPositiveAndNegativeEssentialSpectrum}
  is a homotopy equivalence, as is the coprojection of the quotient by the first inclusion:
  \begin{equation}
    \label{ZigZagFromFModF0ToFredPlus}
    \begin{tikzcd}[column sep=large]
      \mathcal{F}/\mathcal{F}_0
      &
      \mathcal{F}
      \ar[
        l,
        ->>,
        "{\sim}"{swap},
        "{ \;\; \mathrm{hmtpy} }"
      ]
      \ar[
        r, 
        hook,
        "{ \sim }",
        "{ \mathrm{hmtpy} }"{swap}
      ]
      &
      \FredholmOperators^+_{\mathbb{C},\ast}
      \mathrlap{\,.}
    \end{tikzcd}
  \end{equation}
\end{lemma}
\begin{proof}
  The first statement (concerning the right map) is \cite[below (2.5)]{AtiyahSinger1969}. We proceed to prove the second statement (concerning the left map). For that we invoke some basic homotopy theory which we have not reviewed here, but which may be found in standard textbooks, such as \cite{AguilarGitlerPrieto2002}.

  Namely, we claim that the first inclusion in \cref{SAFredWithPositiveAndNegativeEssentialSpectrum} is a \emph{Hurewicz cofibration} (\cite[Def. 4.1.5]{AguilarGitlerPrieto2002}) including a \emph{contractible space}:
  \begin{equation}
    \label{F0IsCofibrationOfContractibleIntoF}
    \begin{tikzcd}[column sep=large]
      \ast
      \ar[
        r, 
        <-, 
        "{ \sim }"{pos=.6},
        "{ \mathrm{hmtpy} }"{swap}
      ]
      &
      \mathcal{F}_0
      \ar[
        r,
        hook,
        "{
          \in \mathrm{Cof}
        }"{swap}
      ]
      &
      \mathcal{F}
      \mathrlap{\,.}
    \end{tikzcd}
  \end{equation}

  To see that $\mathcal{F}_0$ is contractible, note with \cref{fInF0SquareToIdentity} that, under passage to eigenspaces, the operators in this space correspond to choices of direct sum decompositions $\HilbertSpace \simeq \HilbertSpace_+ \oplus \HilbertSpace_-$. This shows that $\mathcal{F}_0$ is a Grassmannian homeomorphic to
  \begin{equation}
    \mathcal{F}_0
    \simeq
    \UH 
      / 
    \big( 
      \UH \times \UH 
    \big)
    \mathrlap{\,,}
  \end{equation}
  and hence its contractibility is a consequence of Kuiper's theorem \cref{KuiperTheorem}.

  To see that we have a cofibration: Since all spaces of operators in question are metric spaces (as subspaces of bounded operators with the operator norm topology) they are \emph{perfectly normal Hausdorff} spaces (cf. \cite[Def. 4.1.13]{AguilarGitlerPrieto2002}). Moreover, $\mathcal{F}_0 \subset \mathcal{F}$ is a closed subspace (being the preimage of $\{0\}$ under the map $\begin{tikzcd}[sep=small] \mathcal{F} \ar[r] & \BoundedOperators(\HilbertSpace) : f \mapsto f^2 - \mathrm{id}\end{tikzcd}$). Therefore it is sufficient (by \cite[Thm. 4.1.14]{AguilarGitlerPrieto2002}) to see that the inclusion is a \emph{strong deformation retract of a neighborhood} \cite[Def. 4.1.11]{AguilarGitlerPrieto2002}. That neighborhood may be taken to be the invertible operators among $\mathcal{F}$, and the retraction may then be given by functional calculus, shifting all points in the spectrum to their sign in $\{+1,-1\}$. 

  This implies the claim by the general fact (by \cite[Thm. 4.2.1]{AguilarGitlerPrieto2002}) that the quotient coprojection of a cofibrantly included contractible space is a homotopy equivalence.
\end{proof}

\begin{proposition}[The Atiyah-Singer exponential map]
\label[proposition]{AtiyahSingerExponentialMap}
  The following exponential map from the self-adjoint Fredholm operators in $\mathcal{F}$ \cref{SAFredWithPositiveAndNegativeEssentialSpectrum} to minus the unitary Fredholm group $-\mathrm{U}^c(\HilbertSpace)$ \cref{MinusTheUnitaryFredholmGroup} is a homotopy equivalence:
  \begin{equation}
    \label{TheAtiyahSingerExponentialMap}
    \exp\big(
      \mathrm{i}
      \pi
      (-)
    \big)
    :
    \begin{tikzcd}[column sep=large]
      \mathcal{F}
      \ar[
        r, 
        "{ \sim }",
        "{ \mathrm{hmtpy} }"{swap}
      ]
      &
      -
      \mathrm{U}^c(\HilbertSpace)
      \,,
    \end{tikzcd}
  \end{equation}
  which as such descends to the quotient by $\mathcal{F}_0$:
  \begin{equation}
    \label{DescendedAtiyahSingerExponentialMap}
    \exp\big(
      \mathrm{i}
      \pi
      (-)
    \big)
    :
    \begin{tikzcd}[column sep=large]
      \mathcal{F}/\mathcal{F}_0
      \ar[
        r, 
        "{ \sim }",
        "{ \mathrm{hmtpy} }"{swap}
      ]
      &
      -
      \mathrm{U}^c(\HilbertSpace)
      \,.
    \end{tikzcd}
  \end{equation}
\end{proposition}
\begin{proof}
  The first statement is 
  \cite[Prop. 3.3]{AtiyahSinger1969}. To see that the map as such passes to the quotient, as claimed in the second statement, recall that $f \in \mathcal{F}_0$ implies that $f^2 = \mathrm{id}$ \cref{fInF0SquareToIdentity}, whence Euler's formula gives that:
  \begin{equation}
    f^2 = \mathrm{id}
    \;\;\;\;
    \Rightarrow
    \;\;\;\;
    \left\{
    \begin{aligned}
    \exp\big(
      \mathrm{i}\pi\, f
    \big)
    &=
    \mathrm{cos}(\pi)
    \cdot
    \mathrm{id}
    + 
    \mathrm{sin}(\pi)
    \cdot
    f
    \\[-2pt]
    &=
    -\mathrm{id}
    \mathrlap{\,.}
    \end{aligned}
    \right.
  \end{equation}
  Therefore we have a commuting diagram
  \begin{equation}
    \begin{tikzcd}
      \mathcal{F}
      \ar[
        rr, 
        "{ \sim }",
        "{ \mathrm{hmtpy} }"{swap}
      ]
      \ar[
        d,
        ->>,
        "{
          \sim
        }"{sloped, swap, pos=.4},
        "{
          \mathrm{hmtpy}
        }"{sloped, swap, rotate=180, pos=.4}
      ]
      &&
      -\mathrm{U}^c(\HilbertSpace)
      \\[+10pt]
      \mathcal{F}/\mathcal{F}_0
      \mathrlap{\,,}
      \ar[
        urr
      ]
    \end{tikzcd}
  \end{equation}
  where the top map is a homotopy equivalence by \cref{TheAtiyahSingerExponentialMap} while the left map is a homotopy equivalence by \cref{ZigZagFromFModF0ToFredPlus}. By the ``2 out of 3''-property of homotopy equivalence, this implies the claim that also the diagonal map is a homotopy equivalence.
\end{proof}

Now we are ready to prove:
\begin{proposition}
  \label[proposition]
    {RepresentingTheGeneratorOfKU1S3}
  Under the equivalence of \cref{SubspacesOfFredAsClassifyingSpacesForK}, the generator of $\mathrm{KU}^1(S^3)$ is represented by forming stabilized Pauli matrices \cref{UnitBallQuaternionsAsFredOperatorsModuloBoundary}:
  \begin{equation}
    \begin{tikzcd}[row sep=-3pt, column sep=0pt]
    \pi_0\Big\{
        S^3
        \dashrightarrow
        \mathcal{F}/\mathcal{F}_0
    \Big\}
    &
      \underset{\mathclap{
        \scalebox{.7}{%
          \cref{ZigZagFromFModF0ToFredPlus}%
        }%
      }}{
        \simeq
      }
    &
    \pi_0\Big\{
        S^3
        \dashrightarrow
        \FredholmOperators
          ^+
          _{\mathbb{C}}
    \Big\}
    &
      \underset{\mathclap{
        \scalebox{.7}{%
         \cref{MatchingKTheoryToFredSubspaces}%
        }%
      }}{
        \simeq
      }
    &
    \mathrm{KU}^1(S^3)
    &\simeq&
    \mathbb{Z}
    \\
    \Big[
    D(\mathbb{H}_{\mathrm{im}})
      \overset{
        \mathrm{stab}
      }{\dashrightarrow}
    \mathcal{F}
    \Big]
    \ar[
      rrrrrr,
      |->,
      shorten=14pt
    ]
    &&&&&&
    1\,.
    \end{tikzcd}
  \end{equation}
\end{proposition}
\begin{proof}
Since the exponential map \cref{DescendedAtiyahSingerExponentialMap} is a homotopy equivalence, it is sufficient to see that the composite
\begin{equation}
  \label{TheMapFromS3ToUnitaryFredholmGroup}
  \begin{tikzcd}[
    row sep=-1pt, column sep=3pt,
    ampersand replacement=\&
  ]
    S^3
    \simeq
    D(\mathbb{H}_{\mathrm{im}})/
    \partial
    \ar[
      rr,
      "{ \mathrm{stab} }"
    ]
    \&\& 
    \mathcal{F}/\mathcal{F}_0
    \ar[
      rr,
      "{
        \exp\left(
          \mathrm{i}\pi(-)
        \right)
      }"
    ]
    \&\&
    -\mathrm{U}^c(\HilbertSpace)
    \\
    x 
      \&\longmapsto\& 
    -\mathrm{i} \CliffordElement_x
    \oplus
    \bigoplus_{\mathbb{N}}
    \scalebox{.8}{$
    \left(
    \begin{matrix}
      + 1 & ~0
      \\
      ~ 0 & -1
    \end{matrix}
    \right)
    $}
    \&\longmapsto\&
    e^{\pi \CliffordElement_x}
    \oplus
    (- \mathrm{id}
      _{\bigoplus_{\mathbb{N}}\mathbb{C}^2})
  \end{tikzcd}
\end{equation}
represents the generator of $\pi_3\big(\mathrm{U}^c(\HilbertSpace)\big) \simeq \mathbb{Z}$. But, as seen from its components shown in the second row, this map is, up to stabilization, just the exponential map from the unit ball in the Lie algebra $\mathfrak{su}(2)$ onto $\mathrm{SU}(2) \simeq S^3$ and as such represents the generator of $\pi_0\big\{S^3 \dashrightarrow S^3\big\} \simeq \mathbb{Z}$, by \cref{WHEfromStableUnitaryToUnitaryFredholm}.
\end{proof}

\begin{remark}
  \label[remark]
    {ComparingToHoravaWitten}
  The terms $-\mathrm{i}\gamma_x$ and  $e^{\pi \gamma_x}$ in \eqref{TheMapFromS3ToUnitaryFredholmGroup} clearly want to correspond with the expressions in \parencites[(3.7, 3.17)]{Horava1998}[(3.2)]{Witten2001}, under the syntactic translation
  $$
    -\mathrm{i}\gamma_x
    \;\leftrightarrow\;
    \frac{\vec \sigma \cdot \vec x}
    {\vert x \vert}
    \mathrlap{\,.}
  $$
  In \cite{Horava1998}, it had remained open in which spaces these expressions are meant to take values and how they actually represent the claimed K-theory classes. In particular, \cite{Witten2001} pointed out that some argument for trivialization of these expressions at $\vert x \vert  = 1$ was missing, and suggested that the resolution has to do with Kuiper's theorem. But the spaces in which this would happen were still not declared, nor a reason given for why the result represents a class in $\mathrm{KU}^1(S^3)$, and specifically the generator. We suggest that \cref{RepresentingTheGeneratorOfKU1S3} fills these gaps and thereby completes this old argument.  
\end{remark}



\section{Conclusion and Outlook}

Motivated by a survey (\cref{OverviewChargesInCohomology}) of
\begin{enumerate}
\item
how \emph{fragile} topological phases and \emph{microscopic} brane charges are classified by twisted relative \emph{nonabelian}/\emph{unstable} generalized orbi-cohomology, 
\\
specifically by \emph{Cohomotopy}, for 2-band insulators and for M-branes;

\item 
how coarsened measurement of these in \emph{abelian}/\emph{stable} generalized orbi-cohomology is equivalent to four/ten-dimensional universal \emph{complex orientations}, 
\end{enumerate}
we laid out a detailed and explicit unraveling of what this means in the case of measuring twisted relative Cohomotopy charges in relative topological K-theory.

To this end we:
\begin{enumerate}
\item
gave a streamlined pedagogical introduction to topological groupoids, stacks and orbifolds (\cref{TopologicalGroupoidsAndStacks})

\item 
culminating in a general discussion of (generalized nonabelian) orbifold cohomology (\cref{OrbifoldCohomology}),

\item
specialized this to a new elegant model for \emph{twisted orbi-orientifold K-theory}, based on classifying spaces of Fredholm operators equivariant under quantum symmetries (\cref{OrbifoldKTheory}),

\item
spelled out in this model explicitly the four/ten-dimensional equivariant complex/quaternionic orientation, by lifting the tautological line bundles over the $\mathbb{C}$/$\mathbb{H}$ Hopf fibrations to Fredholm operators (\cref{TheEquivariantLineBundle,TheEquivariantOrientation}).
\end{enumerate}

As examples and applications we discussed:
\begin{enumerate}
\item
a refined classification of 2-band topological insulators taking into account the topology of the gapping process of their nodal line semi-metal parent phase (\cref{RevisitingFragileTopologicalPhases}),
\item
the measurement of brane charges on/in M5- probe worldvolumes (such as sourced by \emph{M-strings}), in 3-Cohomotopy and coarsened in $\mathrm{KU}^1$ (\cref{RevisitingBraneCharges}, in the process completing an old argument for how to describe D6-brane charge in $\mathrm{KU}^1$, in \cref{3SphereOfSelfAdjointFredholmOps}).
\end{enumerate}

These developments serve to help bridge that gap between:
\begin{enumerate}
\item the relatively more popular discussion of topological phases/charges measured in abelian/stable generalized cohomology theories such as K-theory,

\item their fine-grained fragile/microscopic classification in nonabelian/unstable generalized cohomology theories such as Cohomotopy.
\end{enumerate}

In particular, the motivating result --- that measurement of twisted relative Cohomotopy charges in stable cohomology $E$ corresponds to low-dimensional complex $E$-orientation --- indicates that the role of topological K-theory here is generally to be regarded as just the first approximation step in a \emph{chromatic tower} of complex-oriented cohomology theories, which continues with \emph{elliptic cohomology}, proceeds over a tower of \emph{Morava K-theories}, and ultimately culminates in \emph{complex cobordism cohomology}. 

This perspective, rooted in unstable Cohomotopy, may help better understand existing proposals for the role of higher chromatic cohomology in M-theory and, by the general correspondence to topological quantum systems, may bring these higher chromatic concepts to bear also on the analysis of topological phases of matter.



\clearpage

\printbibliography

\end{document}